\documentclass[final]{ua-thesis_custom}
\pdfoutput=1
\usepackage{graphicx}
\usepackage{epstopdf}
\usepackage{ enumerate, cite}

\usepackage{verbatim}
\usepackage{amssymb,amsmath,amsthm}
\usepackage[mathscr]{eucal}

\usepackage{makeidx}
\numberwithin{equation}{section}
\usepackage{graphicx}
\usepackage{psfrag}
\usepackage{afterpage}
\usepackage{subfigure}
\usepackage{float}
\usepackage{hyperref}
\usepackage{ifpdf}

\usepackage{pdfpages}

\usepackage{appendix}
\usepackage{chngcntr}
\usepackage{etoolbox}

\AtBeginEnvironment{subappendices}{%
\addcontentsline{toc}{chapter}{Appendices}
\counterwithin{figure}{section}
\counterwithin{table}{section}
}

\usepackage{ mathrsfs}

\usepackage{MnSymbol}

\newcommand{\grad}{\operatorname{grad}}
\newcommand{\diag}{\operatorname{diag}}

\newtheorem{definition}{Definition}[section]
\newtheorem{proposition}{Proposition}[section]
\newtheorem{theorem}{Theorem}[section]
\newtheorem{corollary}{Corollary}[section]
\newtheorem{lemma}{Lemma}[section]
\newcommand{\req}[1]{Eq.\,(\ref{#1})}

\newcommand{\beqn}{\begin{equation}}
\newcommand{\eeqn}{\end{equation}}
\newcommand{\GeV}{\text{ GeV}}
\newcommand{\MeV}{\text{ MeV}}
\newcommand{\keV}{\text{ keV}}
\newcommand{\eV}{\text{ eV}}
\newcommand{\meV}{\text{ meV}}

\director{Johann Rafelski}
\author{Jeremiah Birrell}
\title{Non-Equilibrium Aspects of Relic Neutrinos: From Freeze-out to the Present Day}
\date{2014}
\makeindex

\ifpdf
\fi

\begin{document}

\maketitle

\chapter*{Acknowledgments}
 Many thanks go to my advisor, Dr. Johann Rafelski for his patience, encouragement, and dedication in helping me achieve this goal.  This dissertation could not have happened without his great insight into so many areas of physics.  

Thanks to my committee, professors Sean Fleming, David Glickenstein, and Douglas Pickrell, as well as professors Moysey Brio and Jan Wehr for being willing and able to assist whenever I had a question. Thanks also to Dr. Berndt Muller for reviewing this dissertation.

Thanks to Cheng-Tao Yang, a student from National Taiwan University who visited U of A and worked on evaluating matrix elements for neutrino processes, supervisor Dr. Pisin Chen, and with whom we wrote several papers. 

Finally, thanks to Dr. Michael Tabor for his advice and guidance during my time at the University of Arizona and for giving me the opportunity to pursue this goal.

This work was supported in part by a grant from the U.S. Department of Energy, DE-FG02-04ER41318 (PI Dr. Johann Rafelski). Fall 2011 through spring 2014 this work was conducted with Government support under and awarded by DoD, Air Force Office of Scientific Research, National Defense Science and Engineering Graduate (NDSEG) Fellowship, 32 CFR 168a.

\chapter*{Dedication}
\thispagestyle{topright}
\begin{center}To my parents for their invaluable emotional support during the highs and lows.\end{center}

\tableofcontents
\listoffigures
\listoftables

\begin{abstract}
In this dissertation, we study the evolution and properties of the relic (or cosmic) neutrino distribution from neutrino freeze-out at $T=O(1)$ MeV through the free-streaming era up to today, focusing on the deviation of the neutrino spectrum from equilibrium and in particular we demonstrate the presence of chemical non-equilibrium that continues to the present day.  The work naturally separates into two parts.  The first, which constitutes chapters \ref{ch:intro} through \ref{ch:nu_today}, focuses on aspects of the relic neutrinos that can be explored using conservation laws.  The second part, chapters \ref{ch:vol_forms} through \ref{ch:param_studies}, studies the neutrino distribution using the full general relativistic Boltzmann equation.

 Chapter \ref{ch:intro} begins with a brief overview of the Friedmann$-$Lemaitre$-$Robertson$-$Walker metric and its use in cosmology. With this background, we give a broad overview of the history of the Universe, from just prior to neutrino freeze-out up through the present day, placing the history of cosmic neutrino evolution in its proper context.

 Motivated by the Planck CMB measurements of the effective number of neutrinos, $N_\nu$, chapter \ref{ch:model_ind} focuses on the distinction between chemical and kinetic equilibrium and freeze-out. Using these concepts, we derive those properties of neutrino freeze-out that depend only on conservation laws and are independent of the details of the scattering processes. In particular, we characterize the dependence of both $N_\nu$ and the deviation of the neutrino distribution from chemical equilibrium on the neutrino kinetic freeze-out temperature. Part one ends with chapter \ref{ch:nu_today}, which connects the freeze-out era with the current era by characterizing the present day neutrino spectrum as seen from the Earth. It includes the velocity and de Broglie wavelength distributions and a computation the drag force on a coherent  detector due to neutrino scattering.

We now begin the second part of this dissertation, where the focus is on properties of cosmic neutrinos that depend on the details of the neutrino reactions, as is necessary for modeling the non-thermal distortions from equilibrium and computing freeze-out temperatures. As a preliminary, in chapter \ref{ch:vol_forms} we develop some geometry background concerning volume forms and integration on submanifolds that is helpful in computations. 

 In chapter \ref{ch:boltz_orthopoly} we recall a spectral method, adapted to near chemical equilibrium, that has been used in prior works to study neutrino freeze-out. We then detail a new spectral method, based on a dynamical basis of orthogonal polynomials. This method was designed extend the regime of applicability to systems far from chemical equilibrium and/or that undergo significant reheating, that is a temperature dependence that does follow a simple scaling law. In the process, we also improved the speed of the method. The method is validated on an exactly solvable model problem.  

In chapter \ref{ch:coll_simp} we list the reactions that neutrinos participate in while freezing out and detail an improved procedure for analytically simplifying the corresponding scattering integrals for subsequent numerical computation.  This procedure relies on some of the concepts introduced in chapter \ref{ch:vol_forms}. Using these scattering integral computations,  we solve the Boltzmann equation through the neutrino freeze-out period using both spectral methods from chapter \ref{ch:boltz_orthopoly}. We show numerically that our new method agrees with the prior method when both are applicable and also find that our method significantly reduces the required computer time -- by a factor 20 or more.

 Finally, in chapter \ref{ch:param_studies} we use this novel approach to perform parametric studies of the dependence of the neutrino freeze-out on the Weinberg angle, weak force interaction strength, the strength of gravity, and electron mass in order to constrain time and/or temperature variation of these parameters using measurements of $N_\nu$. This exploration is performed with the aim of recognizing mechanisms in the neutrino freeze-out process that are capable of leading to the measured value of  $N_\nu$ in the environment of a hot Universe in which freeze-out  occurs. 

\end{abstract}

\part{Neutrino Freeze-out via Conservation Laws}
\chapter{Introduction to Cosmology and the Relic Neutrino Background}\label{ch:intro}

At a temperature of $5$ MeV the Universe consisted of a plasma of $e^\pm$, photons, and neutrinos.  At around $1$ MeV neutrinos stop interacting, or freeze-out, and begin to free-stream through the Universe. Today they comprise the relic neutrino background. Photons freeze-out around $0.25$ eV and today they make up the Cosmic Microwave Background (CMB), currently at $T_{\gamma,0}=0.235$ meV.  Relic neutrinos have not been directly measured, but their impact on the speed of expansion of the Universe is imprinted on the CMB.  Indirect measurements of the relic neutrino background, such as by the Planck satellite \cite{Planck},  constrain neutrino properties such as mass and number of massless degrees of freedom.

 In later chapters, we will study the details of the neutrino freeze-out process and their impact on observables in detail but first we present an overview of cosmology, from just prior to neutrino freeze-out until today, putting the relic neutrinos in their proper context. Much of this material, including most figures, was adapted from our paper \cite{ErasOfUniverse}.

\section{Standard Cosmology}\label{cosmo}
 To follow the history of the relic neutrino distribution, one must first understand the relation between the expansion dynamics of the Universe, its energy content, and the connection to the photon and neutrino temperature. For this purpose we need some preparation in the  Friedmann$-$Lemaitre$-$Robertson$-$Walker (FRW) cosmological  model, see for example \cite{hartle2003gravity,hobson,misner1973gravitation}. Assuming a homogeneous, isotropic Universe, one arrives at the spacetime metric
\beqn\label{metric}
ds^2=dt^2-a^2(t)\left[ \frac{dr^2}{1-kr^2}+r^2(d\theta^2+\sin^2(\theta)d\phi^2)\right]
\eeqn
characterized  by the scale parameter $a(t)$.  $a(t)$ determines the distance between objects at rest in the Universe frame, otherwise known as comoving observers. The geometric parameter $k=-1,0,1$ identifies the geometry of the spacial hypersurfaces defined by comoving observers. Space is a flat-sheet for the observationally preferred value $k=0$ \cite{Planck}, hyperbolic for $k=-1$, and spherical for $k=1$.

The dynamics are governed by the Einstein equations
\beqn\label{Einstine}
G^{\mu\nu}=R^{\mu\nu}-\left(\frac R 2 -\Lambda\right) g^{\mu\nu}=-\frac{1}{M_p^2} T^{\mu\nu},  
\quad R= g_{\mu\nu}R^{\mu\nu}
\eeqn
where $M_p\equiv 1/\sqrt{8\pi G_N}$ is the Planck mass, $G_N$ is the gravitational constant, and we work in units where $\hbar=c=1$. Recall that the Einstein tensor $G^{\mu\nu}$ is divergence free and hence so is the total stress energy tensor, $T^{\mu\nu}$.  Note that our definition of $M_p$, while more convenient in cosmology, differs by a factor of $1/\sqrt{8\pi}$ from the particle physics convention.  Finally, we point out that there are several sign conventions in use regarding the definition of geometrical quantities and Einstein's equation that are clarified in appendix \ref{app:conventions}.

 In a homogeneous isotropic spacetime, the matter content is necessarily characterized by two quantities, the energy density $\rho$ and isotropic pressure $P$
\begin{equation}
  T^\mu_\nu =\mathrm{diag}(\rho, -P, -P, -P).
\end{equation}
 It is common to absorb the Einstein cosmological constant $\Lambda$ into $\rho$ and $P$ by defining
\beqn\label{EpsLam}
\rho_\Lambda=M_p^2\Lambda, \qquad P_\Lambda=-M_p^2 \Lambda.
\eeqn
We implicitly consider this done from now on.

The global Universe dynamics can be characterized by two  quantities, the Hubble parameter  $H$, a strongly time dependent quantity on cosmological time scales,  and the deceleration parameter $q$
\beqn\label{dynamic}
\frac{\dot a }{a}\equiv H(t) ,\quad 
q\equiv -\frac{a\ddot a}{\dot a^2}.
\eeqn
We note the relations
\beqn
\quad \frac{\ddot a}{a}=-qH^2,\quad \dot H=-H^2(1+q). 
\eeqn

Two dynamically independent equations arise using the metric \req{metric} in \req{Einstine}
\beqn\label{hubble}
\frac{8\pi G_N}{3} \rho =  \frac{\dot a^2+k}{a^2}
=H^2\left( 1+\frac { k }{\dot a^2}\right),
\qquad
\frac{4\pi G_N}{3} (\rho+3P)  =-\frac{\ddot a}{a}=qH^2.
\eeqn
We can eliminate the strength of the interaction, $G_N$,  solving both these equations for ${8\pi G_N}/{3}$, and equating the result to find a relatively simple constraint for the deceleration parameter
\beqn\label{qparam}
q=\frac 1 2 \left(1+3\frac{P}{\rho}\right)\left(1+\frac{k}{\dot a^2}\right).
\eeqn
 From this point on, we work within the  flat cosmological model with $k=0$ and so $q$ is determined entirely by the matter content of the Universe
\begin{equation}\label{qparam}
q=\frac 1 2 \left(1+3\frac{P}{\rho}\right).
\end{equation}

As must be the case for any solution of Einstein's equations,   \req{hubble} implies that the energy momentum tensor of matter is divergence free
\beqn\label{divTmn}
\nabla_\nu T^{\mu\nu} =0 \Rightarrow -\frac{\dot\rho}{\rho+P}=3\frac{\dot a}{a}=3H.
\eeqn
 The same relation also follows from  conservation of entropy, $dE+PdV=TdS=0,\  dE=d(\rho V),\  dV=d(a^3)$. Given an equation of state $P(\rho)$, solution of \req{divTmn} describes the dynamical evolution of matter in the Universe. Combined with the Hubble equation
\begin{equation}\label{Hubble_eq}
H^2=\frac{\rho}{3M_p}
\end{equation}
this allows us to solve for the large scale dynamics of the Universe. 

Using the flat FRW model of cosmology outlined above, we now present several perspectives on the history of the Universe.  First we focus on the reheating history. 

\section{Reheating History of the Universe}\label{Eralink}

At times where dimensional scales are irrelevant, entropy conservation means that  temperature scales inversely with the scale factor $a(t)$. This follows from \req{divTmn} when $ \rho\simeq 3P   \propto T^4$. However, as the temperature drops and at their respective $m\simeq T$ scales, successively less massive particles annihilate and disappear from the thermal Universe. Their entropy reheats the other degrees of freedom and thus in the process, the entropy originating in a massive degree of freedom is shifted into the effectively massless degrees of freedom that still remain.  This causes the  $T\propto 1/a(t)$ scaling to break down; during each of these `reorganization' periods the drop in temperature is slowed by the concentration of entropy in fewer degrees of freedom, leading to a change in the reheating ratio, $R$, defined as
\begin{equation}\label{redshiftratio}
R\equiv \frac{1+z}{ T_\gamma/T_{\gamma,0}}, \qquad 1+z\equiv \frac{a_{0}}{a(t)}.
\end{equation}
The reheating ratio connects the photon temperature redshift to the geometric redshift, where $a_0$ is the scale factor today (often normalized to $1$) and quantifies the deviation from the scaling relation between $a(t)$ and $T$.

As we will see, the change in $R$ can be computed by the drop in the number of degrees of freedom.  At a temperature on the order of the top quark mass, when all standard model particles were in thermal equilibrium, the Universe was pushed apart by 28 bosonic and 90 fermionic degrees of freedom. The total number of degrees of freedom can be computed as follows.  

For bosons we have the following: the doublet of charged Higgs particles has $4=2\times2=1+3$  degrees of freedom -- three will migrate to the longitudinal components of $W^\pm, Z$ when the electro-weak vacuum freezes and the EW symmetry breaking arises, while one is retained in the one single dynamical charge neutral Higgs component. In the massless stage, the SU(2)$\times$U(1) theory has 4$\times$2=8 gauge degrees of freedom where the first coefficient  is  the number of particles $(\gamma, Z, W^\pm)$ and each massless gauge boson has  two transverse polarizations. Adding in $8_c\times2_s=16$ gluonic degrees of freedom we obtain 4+8+16=28  bosonic degrees of freedom. 

The count of fermionic degrees of freedom includes three $f$ families, two spins $s$, another factor two for particle-antiparticle duality. We have in each family of flavors a doublet of $2\times 3_c$ quarks, 1-lepton and 1/2 neutrinos (due left-handedness which was not implemented counting spin). Thus we find that a total $3_f\times 2_p\times 2_s\times(2\times 3_c+1_l+1/2_\nu)=90$ fermionic degrees of freedom. We further recall that massless fermions contribute 7/8 of that of bosons in both pressure and energy density. Thus the total number of massless Standard Model particles at a temperature above the top quark mass scale, referring by convention to bosonic degrees of freedom, is $g_{\rm SM}=28+90\times 7/8=106.75$

In figure~\ref{fig:dof}  we show the cube of the reheating ratio \req{redshiftratio} as a function of photon temperature $T_\gamma$ from the primordial high temperature  early Universe on the right to the present on the left, where $R$  must be by definition unity.  The periods of change seen in figure \ref{fig:dof} come when the temperature crosses the mass of a particle species that is in equilibrium. One can see drops corresponding to the disappearance of particles as indicated.   After $e^+e^-$ annihilation on the left, there are no significant degrees of freedom remaining to annihilate and feed entropy into photons, and so $R$  remains constant until today. We show the result using a Fermi gas model with a very rough model for the QGP phase transition and hadronization period near $O(100\MeV)$. The fermi gas model is a poor approximation above the QGP phase transition; a more precise model using lattice QCD, see e.g. \cite{Borsanyi:2013bia}, together with a high temperature perturbative QCD expansion, see e.g. \cite{Letessier:2002gp}, would be needed to improve on this situation but the details do not impact the neutrino freeze-out period near $1\MeV$ which is our primary concern, and so we do not consider these issues further here.

\begin{figure} 
\centerline{\hspace*{0.4cm}\includegraphics[height=6.6cm]{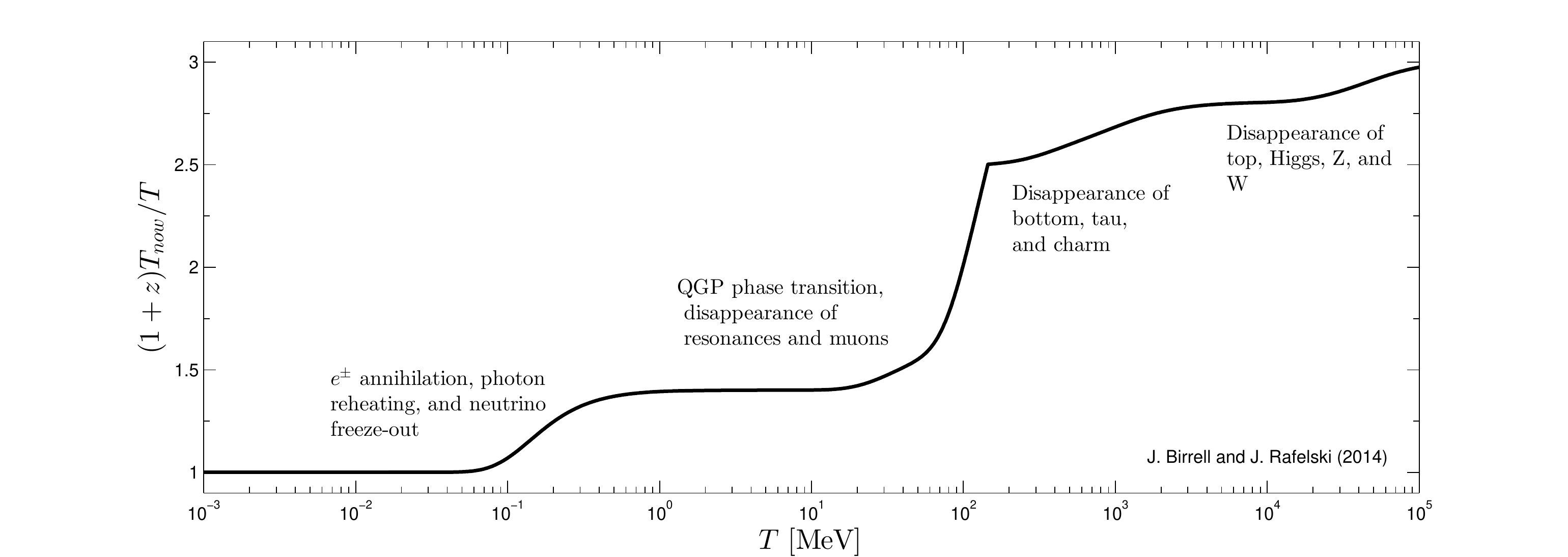}}
\caption{Disappearance of degrees of freedom. The Universe volume inflated approximately by a factor of 27 above the thermal red shift scale as massive particles disappeared successively from the inventory.\label{fig:dof}}
 \end{figure}

As long as the dynamics are at least approximately entropy conserving, the total drop in $R$ is entirely determined by entropy conservation. Namely, the magnitude of the drop in $R$ figure~\ref{fig:dof} is a measure of the number of degrees of freedom that have disappeared from the Universe. Consider   two times $t_1$ and $t_2$ at which all particle species that have not yet annihilated are effectively massless.  By conservation of comoving entropy and  scaling $T\propto 1/a$ we have
\begin{equation}\label{r_ratio}
1=\frac{a_1^3S_{1}}{a_2^3 S_2}=\frac{a_1^3\sum_ig_i T_{i,1}^3}{a_2^3\sum_j g_j T_{j,2}^3},\qquad \left(\frac{R_1}{R_2}\right)^3=\frac{\sum_ig_i (T_{i,1}/T_{\gamma,1})^3}{\sum_j g_j (T_{j,2}/T_{\gamma,2})^3}
\end{equation}
where the sums are over the total number of degrees of freedom present at the indicated time and the degeneracy factors $g_i$ contain the $7/8$ factor for fermions. In the second form    we divided the numerator and denominator by $a_{0}T_{\gamma,0}$. We distinguish between the temperature of each particle species and our reference temperature, the photon temperature.  This is important since today neutrinos are colder than photons, due to photon reheating from  $e^\pm$ annihilation occurring after neutrinos decoupled (this is only an approximation, a point we will study in detail in subsequent chapters).  By conservation of entropy one obtains the neutrino to photon temperature ratio of
\begin{equation}\label{T_nu_T_gamma}
T_\nu/T_\gamma=({4}/{11})^{1/3}.
\end{equation}
We will call this the reheating ratio in the decoupled limit.  For details on the derivation of this standard result, see for example our paper in appendix \ref{app:model_ind}, where it is obtained as a special case of a more general analysis. 

Using \req{r_ratio}  we  compute the total drop in $R^3$ shown in figure \ref{fig:dof}.  At $T=T_\gamma=\mathcal{O}(100\GeV)$ the number of active degrees of freedom is slightly below $g_{\rm SM}=106.75$ due to the partial disappearance of top quarks, but this approximation will be good enough for our purposes.  At this time, all the species are in thermal equilibrium with photons and so $T_{i,1}/T_{\gamma,1}=1$ for all $i$.  Today we have $2$ photon and $7/8\times 6$ neutrino degrees of freedom and a  neutrino to photon temperature ratio \req{T_nu_T_gamma}.  Therefore we have
\begin{equation}
\left(\frac{R_{100GeV}}{R_{now}}\right)^3= \frac{g_{SM}}{g_{\rm now}}=\frac{106.75}{2+\frac{7}{8}\times 6\times \frac{4}{11}}\approx 27.3
\end{equation}
which is the  fractional change we see in the fermi gas model curve in figure \ref{fig:dof} (as mentioned above, the QCD model is reduced due to interactions). The meaning of this factor is that the Universe approximately inflated by a factor 27 above the thermal red shift scale as massive particles disappeared successively from the inventory.

\section{Composition of the Universe}
From the perspective of reheating, the history of the Universe from the end of $e^\pm$ annihilation until today has been uneventful.  We can shed additional light on this period and others by looking at the composition of the Universe as a function of temperature

\begin{figure}
\centerline{\hspace*{0.4cm}\includegraphics[height=7.6cm]{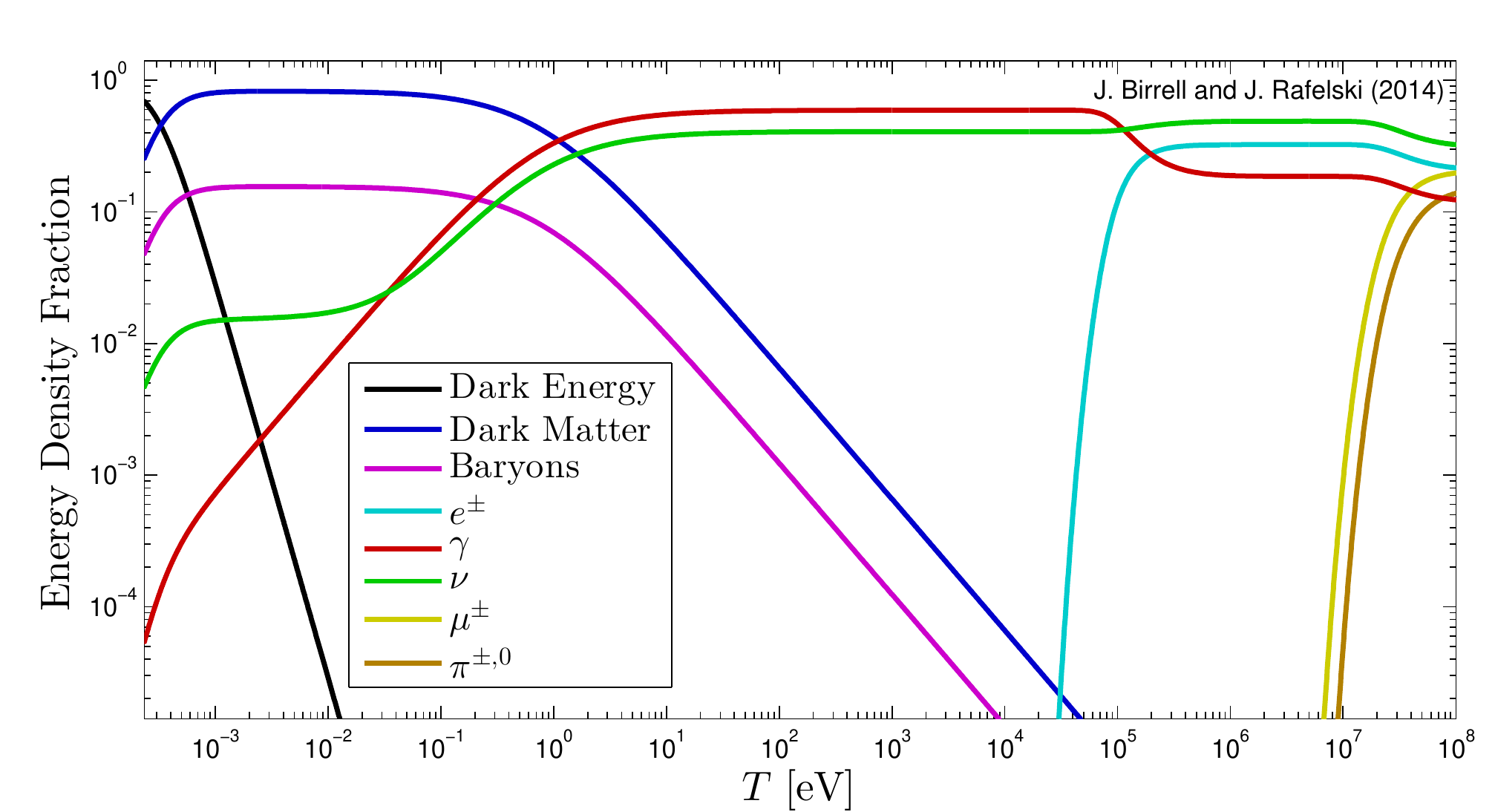}}\label{fig:energy_frac}
\caption{Current era: $69\%$ dark energy, $26\%$ dark matter, $5\%$ baryons, $<1\%$ photons and neutrinos, $1$ massless and $2\times .1$ eV neutrinos (Neutrino mass choice is just for illustration.  Other values are possible).}
 \end{figure}
In figure \ref{fig:energy_frac} we begin on the right at the end of the hadron era with the disappearance of muons and pions.  This constitutes a reheating period, with energy and entropy from these particles being transfered to the remaining $e^\pm$, photon, neutrino plasma.  Continuing to $T=O(1)$ MeV, we come to the annihilation  of $e^\pm$ and the photon reheating period.  Notice that only the photon energy density fraction increases here.  As discussed above, a common simplifying assumption is that neutrinos are already decoupled at this time and hence do not share in the reheating process, leading to a difference in photon and neutrino temperatures \req{T_nu_T_gamma}.

After passing through a long period, from $T=O(1)$ MeV until $T=O(1)$ eV, where the energy density is dominated by photons and free-streaming neutrinos, we then come to the beginning of the matter dominated regime, where the energy density is dominated by dark matter and baryonic matter.  This transition is the result of the redshifting of the photon and neutrino energy, $\rho\propto T^4$, whereas for non-relativistic matter $\rho\propto a^{-3}\propto T^3$.  Note that our inclusion of neutrino mass causes the leveling out of the neutrino energy density fraction during this period, as compared to the continued redshifting of the photon energy.

Finally, as we move towards the present day CMB temperature of $T_{\gamma,0}=0.235$ meV on the left hand side, we have entered the dark energy dominated regime.  For the present day values, we have used the fits from the Planck data \cite{Planck} of  $69\%$ dark energy, $26\%$ dark matter and $5\%$ baryons (and zero spatial curvature).  The photon energy density is fixed by the CMB temperature $T_{\gamma,0}$ and the neutrino energy density is fixed by $T_{\gamma,0}$ along with the photon to neutrino temperature ratio.  Both constitute $<1\%$ of the current energy budget.

\section{Deceleration Parameter}
We conclude our overview of cosmology with one final perspective, the Universe as seen through the deceleration parameter.  The deceleration parameter is another indicator of the transition between different eras of the Universe's history.  Recall the relation \req{qparam} (for $k=0$)  between deceleration parameter and matter content of the Universe. In particular we have the regimes

\begin{itemize}
\item Radiation dominated Universe: $P=\rho/3 \implies q=1$.\\

\item  (Non-relativistic) Matter dominated Universe: $P\ll\rho \implies q=1/2$.\\

\item Dark energy ($\Lambda$) dominated Universe: $P=-\rho \implies q=-1$.\\

\end{itemize}
We use $q$ first to characterize the era from today back to the end of neutrino freeze-out and then from freeze-out until the end of the hadron era.

\subsection{Back in time to Neutrino Freeze-out}\label{recomb}
In the following we use the mix of matter  (31\%) and dark energy (69\%) with photon and neutrino backgrounds favored by the latest Planck results \cite{Planck}, where we gave two neutrino species mass of $m_\nu=30\meV$ and a third neutrino remains  massless.  This is a different mass value than used above and again, it is only for illustration-- other mass choices are possible within present day constraints and will impact to some degree where exactly matter dominance emerges from the radiative Universe.  We presume  that neutrino kinetic freeze-out completed before the onset of $e^\pm$-annihilation into  photons, leading to the neutrino to photon temperature ratio \req{T_nu_T_gamma}. Again, this is a common simplifying assumption.  Much of the remainder of this work will involve improving on this approximation, but for the purposes of this overview it is sufficient.

Figure \ref{fig:today} shows in the left frame the temperature  (left axis) and deceleration parameter (right axis)  from shortly after the completion neutrino freeze-out until today.  The horizontal dot-dashed lines show  the pure radiation-dominated value of $q=1$ and the matter-dominated value of $q=1/2$. The expansion in this era starts off as radiation-dominated, but transitions to matter-dominated starting around $T=\mathcal{O}(10\eV)$ and begins to transition to a dark energy dominated era at $T=\mathcal{O}(1\meV)$. We are still in the midst of this transition today. The vertical dot-dashed lines show  the time of recombination at $T\simeq0.25\eV$, when the Universe became transparent to photons, and reionization at $T\simeq {\cal O}(1\meV)$, when hydrogen in the Universe was again ionized due to light from the first galaxies \cite{Zaroubi:2012in}. 

On the right in figure  \ref{fig:today}  we show the Hubble parameter $H$ and redshift $z+1\equiv a_0/a(t)$. We can see in figure \ref{fig:today} a visible deviation from power law behavior due to the transitions from radiation to matter dominated and from matter to dark energy dominated expansion.  These transitions are accentuated and more easily visualized in the form of the deceleration parameter $q$. The time span covered by the figure  \ref{fig:today} is in essence the entire lifespan of the Universe, but of course on a logarithmic time scale there is a lot of room for interesting physics in the tiny blip that happened beforehand.

\begin{figure}
\begin{minipage}{\linewidth}
\makebox[0.5\linewidth]%
{\includegraphics[keepaspectratio=true,scale=0.52]{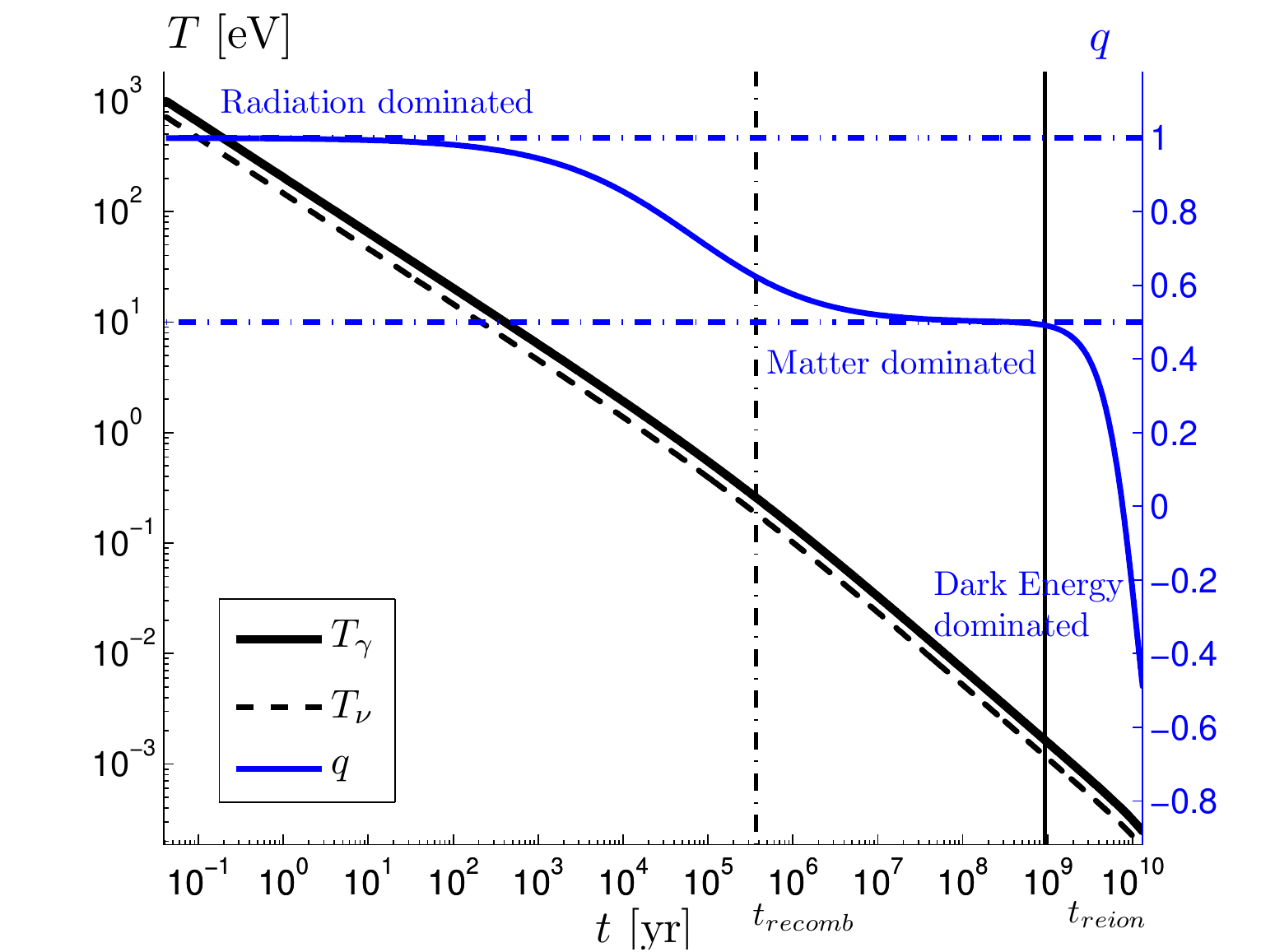}}
\makebox[0.5\linewidth]%
{\includegraphics[keepaspectratio=true,scale=0.52]{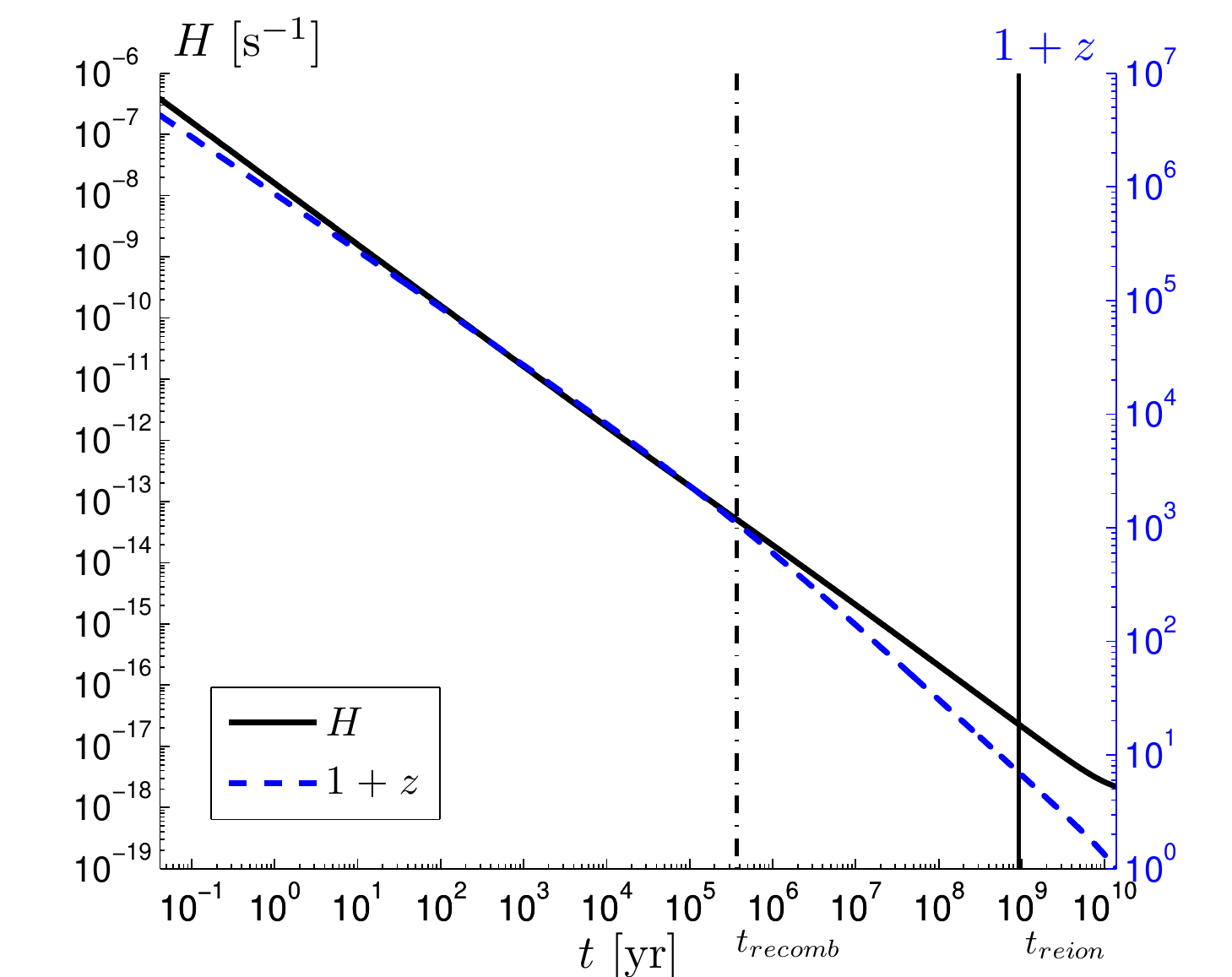}}
\caption{Transition periods in the composition of the Universe: on left -- evolution of temperature $T$  and deceleration parameter $q$; on right --  evolution of the Hubble parameter $H$ and redshift $z$.
\label{fig:today} }
\end{minipage}
\end{figure}

\subsection{Neutrino Freeze-out Era }\label{nudecoup}
The era separating the photon-neutrino-matter-dark energy Universe we just described from the end of the hadron Universe is quite complex in its evolution.   We begin when the number of $e^\pm$-pairs has decayed to the same abundance as the number of baryons in the Universe at the temperature  $T=\mathcal{O}(10\keV)$ and reach back to $T={\cal O}(30\MeV)$ where muons and pions are disappearing from the Universe.

\begin{figure}
\begin{minipage}{\linewidth}
\makebox[0.5\linewidth]%
{\includegraphics[keepaspectratio=true,scale=0.54]{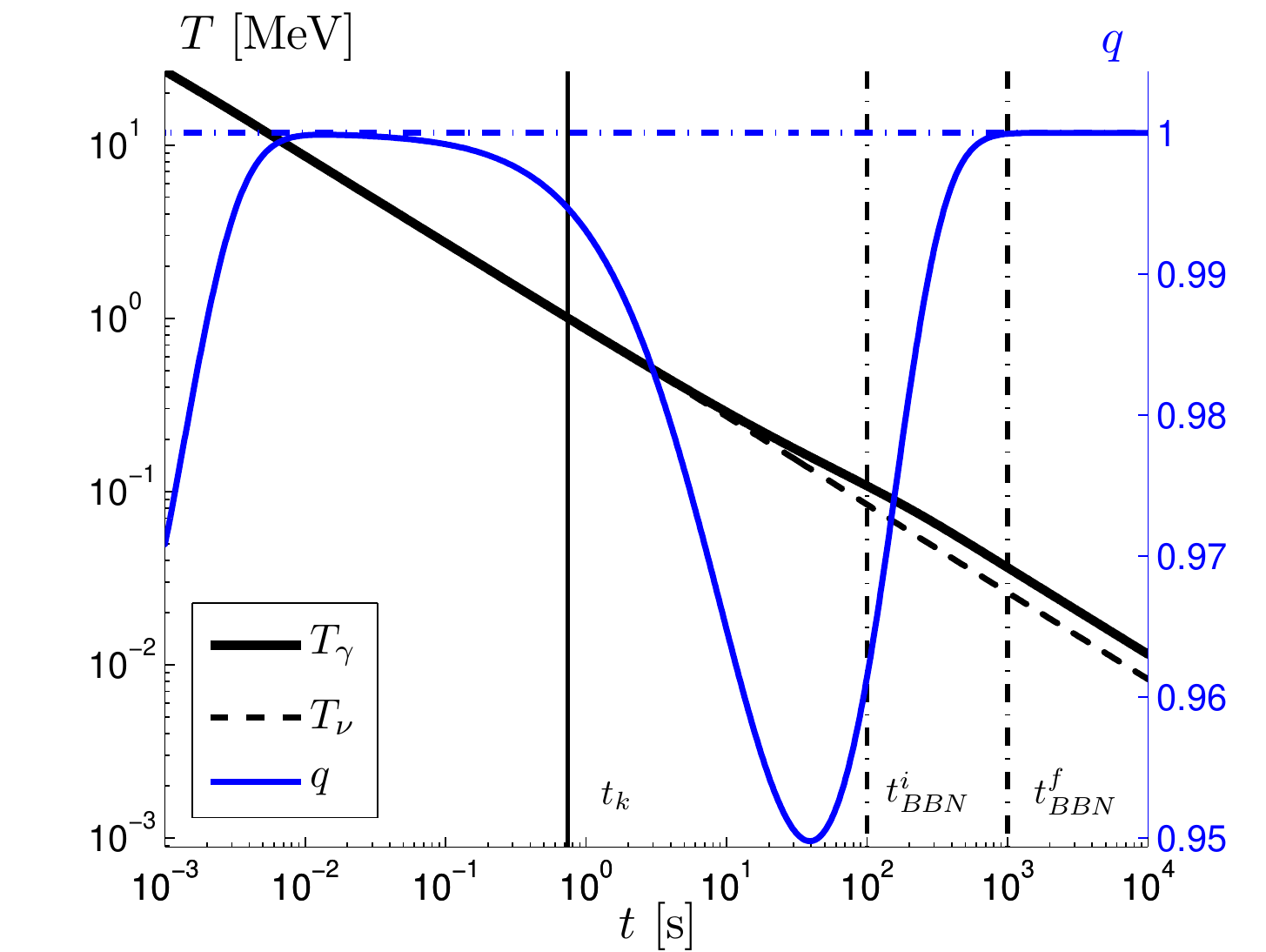}} 
\makebox[0.5\linewidth]%
{\includegraphics[keepaspectratio=true,scale=0.54]{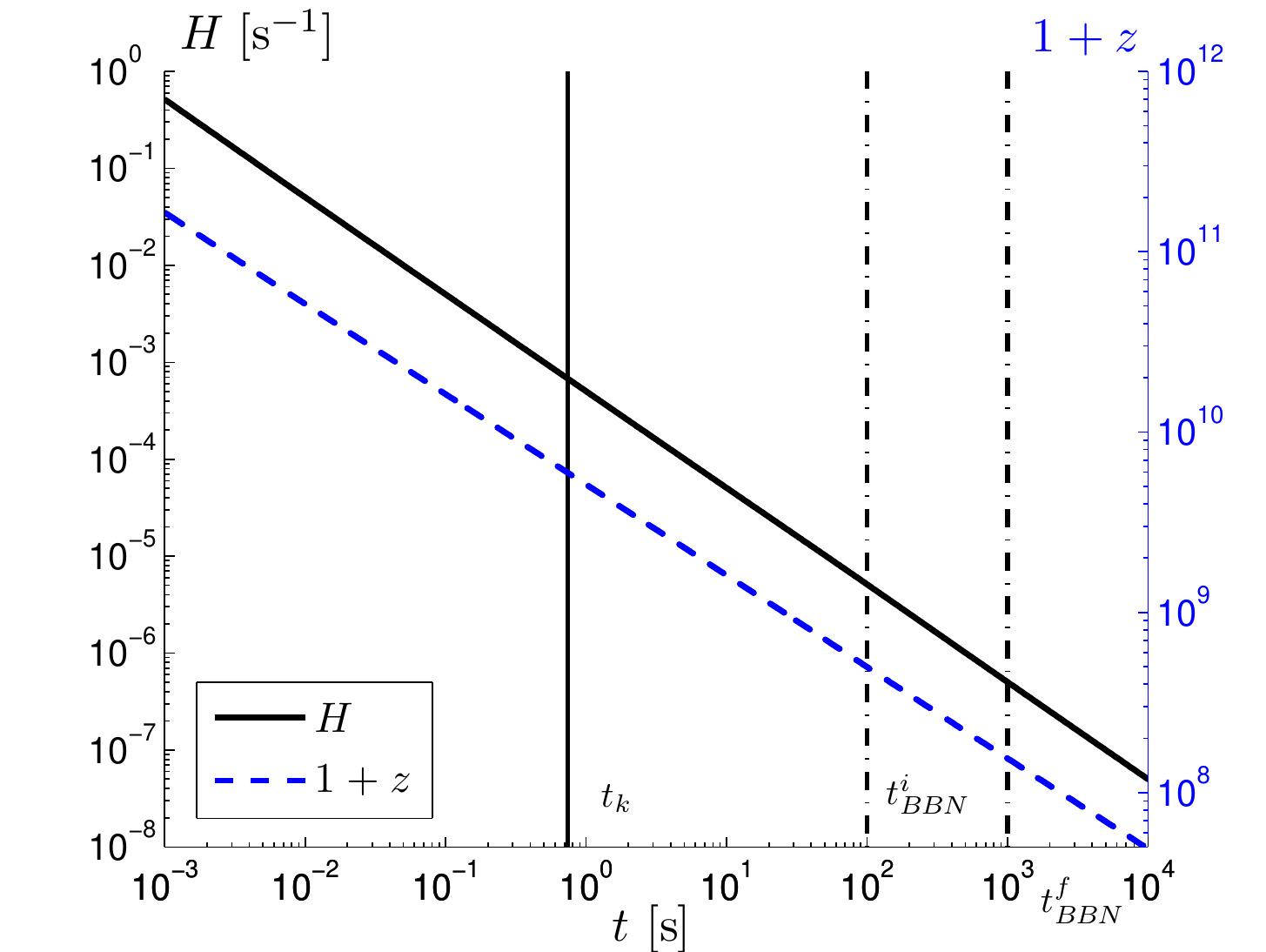}} 
\caption{From the end of baryon antimatter annihilation through BBN, see figure \ref{fig:today}.
\label{fig:BBN}  }
\end{minipage}
\end{figure}

 In figure~\ref{fig:BBN} the horizontal dot-dashed line for $q=1$  shows the pure radiation dominated value with two exceptions. First, the presence of massive pions  and muons reduce  the value of $q$ near to the maximal temperature shown.  Second, when the temperature is near the value of the electron mass, the $e^\pm$-pairs are not yet fully depleted but already sufficiently non-relativistic to cause another dip in $q$.  These are not large drops; the expansion is still predominately radiation dominated.  But $q$ provides a sensitive measure of when various mass scales become relevant and is a good indicator of the presence of a reheating period.

 The dashed line shows the neutrino temperature, which decouples from the $e^\pm$ and photon temperature at $T={\cal O}(1\MeV)$ when neutrinos freeze-out and begin free streaming. In figure~\ref{fig:BBN} the unit of time is seconds and the range spans the domain from fractions of a millisecond to a few hours. After neutrino freeze-out we come to Big Bang Nucleosynthesis, the period when the lighter elements were synthesized in a hot but relatively dilute plasma \cite{Iocco:2008va}. We left some time gap between this and the domain shown in figure \ref{fig:today}  describing the current era -- there is an uneventful evolution between the two domains. 

\section{Focusing on Neutrino Freeze-out}
Neutrino freeze-out is, as far as we know, the unique era in the history of the Universe when a significant matter fraction froze out at the same time that a reheating period was beginning, namely the start of $e^\pm$ annihilation.  It is this coincidence that makes neutrino freeze-out a rich and complicated period to study as compared to the many other reheating periods in the history of the Universe. This period has been studied before \cite{Madsen,Dolgov_Hansen,Gnedin,Esposito2000,Mangano2002,Mangano2005}, but the Planck satellite results \cite{Planck} motivate a reinvestigation of this period of cosmology.  We therefore make the interplay of neutrino freeze-out and reheating from $e^\pm$ annihilation the primary focus of the remainder of this work.

\begin{subappendices}
\section{Conventions}\label{app:conventions}

There are several sign conventions in use in general relativity.  As discussed in \cite{hobson}, these conventions differ by the sign factors $S1$, $S2$, $S3$, which appear in the following objects:
\vspace{3mm}

Metric Signature: $\eta^{\mu\nu}=(S1)\text{Diag}(-1,1,1,1)$
\vspace{3mm}

Riemann Tensor: $R^\mu_{\alpha\beta\gamma}=(S2)(\partial_{\beta}\Gamma^\mu_{\alpha\gamma}-\partial_{\gamma}\Gamma^\mu_{\alpha\beta}+\Gamma^\mu_{\sigma\beta}\Gamma^\sigma_{\gamma\alpha}-\Gamma^\mu_{\sigma\gamma}\Gamma^\sigma_{\beta\alpha})$
\vspace{3mm}

Einstein Equation: $G_{\mu\nu}-(S3)\Lambda g_{\mu\nu}=(S3)8\pi G_NT_{\mu\nu}$
\vspace{3mm}

Ricci Tensor: $R_{\mu\nu}=(S2)(S3)R^\alpha_{\mu\alpha\nu}$
\vspace{3mm}

The sign $S3$ comes from the choice of what index is contracted in forming the Ricci tensor.  Since that sign factor appears in both $R_{\mu\nu}$ and $R$ it affects the overall sign of $G_{\mu\nu}$ and therefore Einstein's equation as shown above. In this dissertation we will use the $(-,+,-)$ convention.

\end{subappendices}

\chapter{ Study of Neutrino Distribution using Conservation Laws}\label{ch:model_ind}

\section{Effective Number of Neutrinos}\label{sec:N_nu}
In the previous chapter we gave an overview of cosmology that included a simple model of neutrino freeze-out, wherein neutrinos decouple prior to the $e^\pm$ annihilation reheating period, leading to the reheating ratio in the decoupled limit \req{T_nu_T_gamma} which we now denote by $R_\nu$. However, as we mentioned several times this is only an approximate model. In reality, the freeze-out and reheating periods overlap to some degree which greatly complicates the picture, as some energy and entropy from the annihilating $e^\pm$ goes into neutrinos.  This overlap has observable consequences, as any extra energy in neutrinos impacts the speed of expansion of the Universe, through the Hubble equation \req{Hubble_eq}.

The additional energy and entropy fed into neutrinos is typically quantified by the effective number of neutrinos, $N_\nu$, defined by comparing the total neutrino energy density to the energy density of a massless fermion with two degrees of freedom and neutrino to photon temperature ratio $R_\nu$,
\begin{equation}
N_{\nu}=\frac{\rho_\nu}{\frac{7}{120}\pi^2  \left(R_\nu T_\gamma\right)^4}.
\end{equation}
 By definition, any transfer of energy from $e^\pm$ into neutrinos results in $N_\nu>N_\nu^f=3$, the number of physical neutrino flavors.  $N_\nu$ can be  measured by fitting to observational data, such as the Planck CMB measurements. A numerical computation based on the Boltzmann equation with two body scattering~\cite{Mangano2005} gives to $N_{\nu}^{\rm th}=3.046$. However the Planck CMB results contain several fits~\cite{Planck} based on different data sets which suggest that $N_\nu$ is in the range $3.30\pm 0.27$ to $3.62\pm0.25$ ($68\%$ confidence level). 

This tension between the Planck results and theoretical reheating studies motivates our work. This tension has inspired various theories, such as \cite{Weinberg:2013kea}, where it is postulated to be due to the presence of as yet undiscovered particle species. In this work, we avoid postulating the existence of additional particles, but rather explore the possibility that the increase in $N_\nu$ is the consequences of additional energy and entropy being transfered into neutrinos during $e^\pm$ annihilation.  In other words, we postulate additional neutrino reheating and explore its consequences.

\section{Matter Content}
In this work, matter will be modeled by a particle distribution function $f(t,x,p)$ that, roughly speaking, gives the probability of finding a particle per unit spacial volume per unit momentum space volume at a given time.  The distribution function gives the stress energy tensor, particle four-current, and entropy four-current via 
\begin{align}
T^{\mu,\nu}(t,x)=&\frac{g_p}{(2\pi)^3}\int p^\mu p^\nu f(t,x,p) \sqrt{|g|}\frac{d^3p}{p_0},\\
n^\nu(t,x)=&\frac{g_p}{(2\pi)^3}\int p^\nu f(t,x,p) \sqrt{|g|}\frac{d^3p}{p_0},\\
s^\nu(t,x)=&-\frac{g_p}{(2\pi)^3}\int(f\ln(f)\pm(1\mp f)\ln(1\mp f))p^\mu\sqrt{|g|}\frac{d^3p}{p_0}
\end{align}
where the upper signs are for fermions, the lower for bosons, $g_p$ is the degeneracy of the particle, and $g$ is the determinant of the metric.  In an FRW Universe, the expressions for the  energy density, pressure, number density, and entropy density of a particle of mass $m$ are
\begin{align}\label{moments}
\rho=&\frac{g_p}{(2\pi)^3}\int f(t,x,p)Ed^3p,\\
\rho=&\frac{g_p}{(2\pi)^3}\int f(t,x,p)\frac{p^2}{3E}d^3p,\\
n=&\frac{g_p}{(2\pi)^3}\int f(t,x,p) d^3p, \hspace{2mm} E=\sqrt{m^2+p^2},\\
s=&-\frac{g_p}{(2\pi)^3}\int (f\ln(f)\pm(1\mp f)\ln(1\mp f)) d^3p.
\end{align}

The dynamics of the distribution function,  and therefore the precise nature of neutrino freeze-out and the energy and entropy transfered into the neutrino sector, are governed by the Boltzmann equation
\begin{equation}
p^\alpha\partial_{x^\alpha} f-\Gamma^{j}_{\mu,\nu}p^\mu p^\nu\partial_{p^j}f=C[f]
\end{equation}
where repeated Greek indices indicate a sum over $0,...,3$ and Roman indices indicate a sum over the spacial components $1,...,3$.  The right hand side is the collision operator and incorporates the physics of any short range interactions that the particles participate in. The left hand side gives the dynamics under any long range forces. For us the only long range force will be gravity, encoded in the Christoffel symbols $\Gamma^j_{\mu\nu}$, and so the Boltzmann equation expresses the fact that particles undergo geodesic motion in between collisions. For much greater detail on the definition of the distribution function in a general spacetime, the geometric origin of the Boltzmann equation, and various properties and relations satisfied by moments of the distribution function, see for example \cite{andre,cercignani,bruhat,ehlers,kolb,bernstein2004kinetic}.

We will study neutrino freeze-out in detail using the Boltzmann equation in the second part of this dissertation, starting in chapter \ref{ch:boltz_orthopoly}. However, in this chapter we persue a model independent approach wherein we assume instantaneous chemical/kinetic equilibrium and sharp freeze-out transitions between them.  Though limited in the kinds of questions we can address and answer, this approach makes up for these limitations by letting us derive several important properties that are independent of microscopic dynamics, i.e. independent of $C[f]$, so long as these assumptions are sufficiently accurate.  The dynamics will be derived from conservation laws involving the moments \ref{moments}, but first we must describe the distinction between chemical and kinetic equilibrium.

\subsection{Chemical and Kinetic Equilibrium}
At sufficiently high temperatures, such as existed in the early Universe, both particle creation and annihilation (i.e. chemical) processes and momentum exchanging (i.e. kinetic) scattering processes can occur sufficiently rapidly to establish complete thermal equilibrium of a given particle species. The most probable canonical distribution function $f_{ch}^\pm$ of  fermions (+) and bosons (-) in both chemical and kinetic equilibrium is found by maximizing entropy subject to energy being conserved
\begin{equation}\label{ch_eq}
f_{ch}^\pm=\frac{1}{\exp(E/T)\pm 1}, \hspace{2mm} T>T_{ch}
\end{equation}
where $E$ is the particle energy, $T$ the temperature, and $T_{ch}$ the chemical freeze-out temperature.

For a physical system comprising {\em interacting} particles whose temperature is decreasing with time, there will be a period where the temperature is greater than the kinetic freeze-out temperature, $T_k$, but below chemical freeze-out. During this period, momentum exchanging processes continue to maintain an equilibrium distribution of energy among the available particles, which we call kinetic equilibrium, but particle number changing processes no longer occur rapidly enough to keep the equilibrium particle number yield, i.e. for $T<T_{ch}$ the particle number changing processes have `frozen-out'. In this condition the momentum distribution, which is in kinetic equilibrium but chemical non-equilibrium, is obtained by maximizing  entropy subject to  particle number and energy constraints and thus two parameters appear
\begin{equation}\label{k_eq}
f_{k}^\pm=\frac{1}{\Upsilon^{-1} \exp(E/T)\pm 1},\hspace{2mm} T_k<T\leq T_{ch}.
\end{equation}
The need to preserve the total particle number within the distribution introduces an additional parameter $\Upsilon$ called fugacity.

The fugacity, $\Upsilon(t)\equiv e^{\sigma(t)}$, controls the occupancy of phase space and is necessary once $T(t)<T_{ch}$ in order to conserve particle number.  A fugacity different from $1$ implies an over-abundance ($\Upsilon>1$) or under-abundance ($\Upsilon<1$) of particles compared to chemical equilibrium and in either of these  situations one speaks of chemical non-equilibrium. 

The effect of $\sigma$ is similar after that of chemical potential $\mu$, except that $\sigma$ is equal for particles and antiparticles, and not opposite. This means $\sigma>0$ ($\Upsilon>1$) increases the density of both particles and antiparticles, rather than increasing one and decreasing the other as is common when the chemical potential is associated with conserved quantum numbers.  Similarly, $\sigma<0$ $(\Upsilon<1)$ decreases both. The fact that $\sigma$ is not opposite for particles and antiparticles reflects the fact that both  the number of particles and the number of antiparticles are conserved after chemical freeze-out, and not just their difference.  Ignoring the small particle antiparticle asymmetry their equality reflects the fact that any process that modifies  the distribution would affect both particle and antiparticle distributions in the same fashion.   Such an asymmetry would be incorporated by replacing $\Upsilon\rightarrow \Upsilon e^{\pm\mu/T}$ where $\mu$ is the chemical potential, but we ignore it in this work as the matter antimatter asymmetry is on the order of $1$ part in $10^9$.

 We also emphasize that the fugacity is time dependent and not just an initial condition.  At high temperatures $\Upsilon=1$ and we will find that $\Upsilon<1$ emerges dynamically as a result of the freeze-out process. The importance of fugacity was first introduced in \cite{PhysRevLett.48.1066} in the context of quark-gluon plasma.  Its presence in cosmology was noted in  \cite{Bernstein:1985,Dolgov:1993} but its importance has been largely forgotten and the consequences unexplored in the literature.

Once the temperature drops below the kinetic freeze-out temperature $T_k$ we reach  the free streaming period where  particle scattering processes have completely frozen out and the resultant distribution is obtained by solving the collisionless Boltzmann equation with initial condition as given by the chemical non-equilibrium   distribution \req{k_eq}.  As already indicated, the two transitions between these three regimes constitute  the freeze-out process -- first we have at $T_{ch}$ the chemical freeze-out and at lower $T_k$ the kinetic freeze-out.

\subsection{Entropy Conservation}
In this section we show that in an FRW Universe and under the assumption of chemical or kinetic equilibrium, the total comoving entropy of all particle species is conserved. More specifically, we will consider a collection of particles with distinct fugacities $\Upsilon_i$, all of which are in kinetic equilibrium at a common temperature $T$.   For the following derivation, it is useful to define $\mu_i=\sigma_i T$.  This gives the expressions a familiar thermodynamic form with $\mu$ playing the role of chemical potential and helps with the calculations, but should not be confused with a chemical potential as discussed above.  

Integration by parts establishes the following identities for the kinetic equilibrium distribution \req{k_eq}
\begin{equation}\label{identities}
s_i=\frac{\partial P_i}{\partial T}=(P_i+\rho_i-\mu_i n_i)/T, \hspace{3mm} n_i=\frac{\partial P_i}{\partial \mu_i}.
\end{equation}
 Using \req{divTmn} and \req{identities}, we calculate $d/dt(a^3s)$ where $s=\sum_i s_i$ is the total entropy density.
\begin{align}\frac{1}{a^3}\frac{d}{dt}(a^3sT)&=\frac{1}{a^{3}}\frac{d}{dt}(a^3(P+\rho-\sum_i \mu_i n_i))\\
&=\dot{P}+\dot{\rho}-\sum_i \left(\dot{\mu_i}n_i+\mu_i\dot{n_i}\right)+3\left(P+\rho-\sum_i \mu_i n_i\right)\dot{a}/a\notag\\
&=\frac{\partial P}{\partial T} \dot{T}+\sum_i\frac{\partial P_i}{\partial \mu_i} \dot{\mu_i}-\sum_i \left(\dot{\mu_i}n_i+\mu_i\dot{n_i}+3\mu_i n_i \dot{a}/a\right)+\nabla_\mu \mathcal{T}^{\mu 0}\notag\\
&=s\dot{T}-\sum_i \left(\mu_i\dot{n_i}+3\mu_i n_i \dot{a}/a\right)\notag\\
&=s\dot{T}- a^{-3}\sum_i\mu_i\frac{d}{dt}(a^3n_i).
\end{align}
Therefore
\begin{align}\label{S_n_eq}
\frac{d}{dt}(a^3s)=&\frac{1}{T}\frac{d}{dt}(a^3sT)-a^3s\frac{\dot T}{T}=-\sum_i\sigma_i\frac{d}{dt}(a^3n_i).
\end{align}
If every particle is either in chemical equilibrium (i.e. $\sigma_i= 0$) or has frozen out chemically, and thus has a conserved comoving particle number, then this implies comoving entropy conservation.  

This observation completely fixes the dynamics of the system in the chemical or kinetic equilibrium regimes.  The dynamical quantities are the scale factor $a(t)$, the common temperature $T(t)$, and the fugacities of each particle species $\Upsilon_i(t)$ that is not in chemical equilibrium.  The dynamics are given by the Einstein equation, conservation of the total comoving entropy of all particle species, and conservation of comoving particle number for each species not in chemical equilibrium (otherwise $\Upsilon_i=1$ is constant)
\begin{equation}\label{eq_dynamics}
H^2=\frac{\rho}{3M_p^2}, \hspace{2mm} \frac{d}{dt}(a^3s)=0,\hspace{2mm} \frac{d}{dt}(a^3n_i)=0 \text{ when } \Upsilon_i\neq 1.
\end{equation}

\section{Key Results From our Study of Neutrino Freeze-out}\label{nu_freezeout_summary}
Using the dynamical equations \req{eq_dynamics} we studied the neutrino distribution after freeze-out under the instantaneous equilibrium approximation in the papers \cite{Birrell2013} and \cite{Birrell:2013_2}, attached as appendices \ref{app:chem_freezeout} and \ref{app:model_ind} respectively.  In these works we assumed that the chemical freeze-out occurs before reheating begins and hence the system is in kinetic but not chemical equilibrium from the beginning of reheating until kinetic freeze-out.  In \cite{Birrell2013} we showed numerically that this is the case for the reaction $e^+e^-\rightarrow \nu_e\bar\nu_e$ to high accuracy.  In  \cite{Birrell:2013_2} we characterized the dependence of the neutrino distribution on the kinetic freeze-out temperature $T_k$.  If one is interested in a model where the chemical freeze-out temperature also varies significantly, then the analysis presented in these papers should be repeated with both the chemical and kinetic freeze-out temperatures treated as free parameters. Below we give some of the key results from our analysis.

As discussed in  \cite{Birrell:2013_2}, a deviation from $\Upsilon=1$ and the reheating ratio in the decoupled limit, \req{T_nu_T_gamma}, is a necessary result of the transfer of entropy from the annihilating $e^\pm$ into neutrinos.  In that paper we used conservation laws to analytically derive an approximate relation between the fugacity $\Upsilon=e^\sigma$ and the photon to neutrino temperature ratio
\begin{align}\label{Upsilon_ratio}
\frac{T_\gamma}{T_\nu}&=a\Upsilon^{b}\left(1+c\sigma^2+O(\sigma^3)\right),\\
\label{value_a}
a&=\left(1+\frac{7}{8}\frac{g_{e^\pm}}{g_\gamma}\right)^{1/3}=\left(\frac{11}{4}\right)^{1/3}=R_\nu^{-1}\approx 1.4010,\\
\label{value_b}
b&\approx 0.367,\\
c&\approx -0.0209.
\end{align}
An approximate power law fit was first obtained numerically in \cite{Birrell2013}. In \cite{Birrell:2013_2} we also derived a relation between the effective number of neutrinos and the fugacity $\Upsilon=e^\sigma$ that results from neutrino freeze-out
\begin{equation}\label{N_nu_approx}
N_\nu=\frac{360}{7\pi^4}\frac{e^{-4b\sigma}}{(1+c\sigma^2)^4}\int_0^\infty \frac{u^3}{e^{u-\sigma}+1}du\left(1+O(\sigma^3)\right).
\end{equation}

\begin{figure}\label{fig:Tk_dependence}
\begin{minipage}{\linewidth}
\makebox[0.5\linewidth]%
{\includegraphics[height=5.8cm]{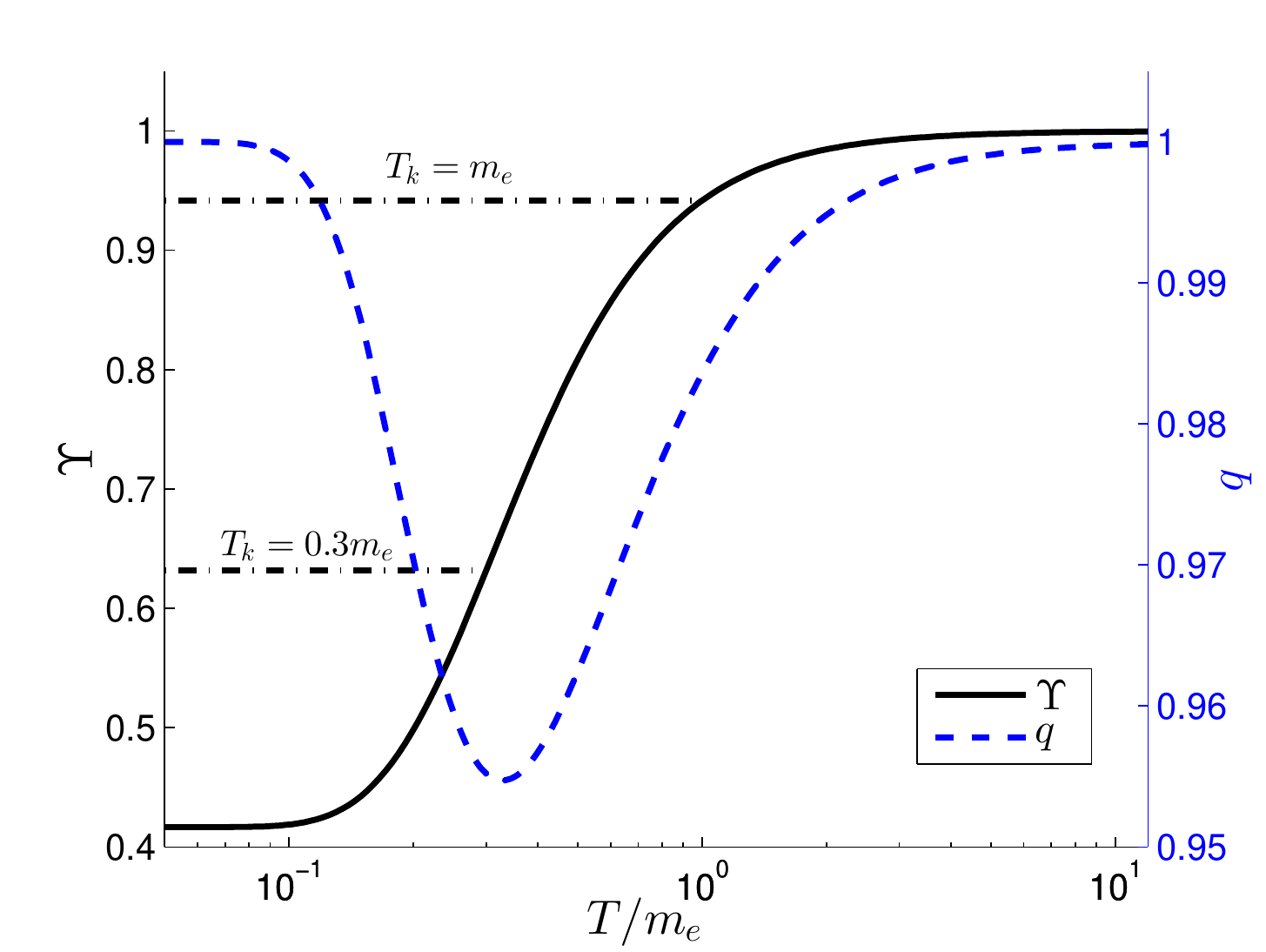}}
\makebox[0.5\linewidth]%
{\includegraphics[height=5.8cm]{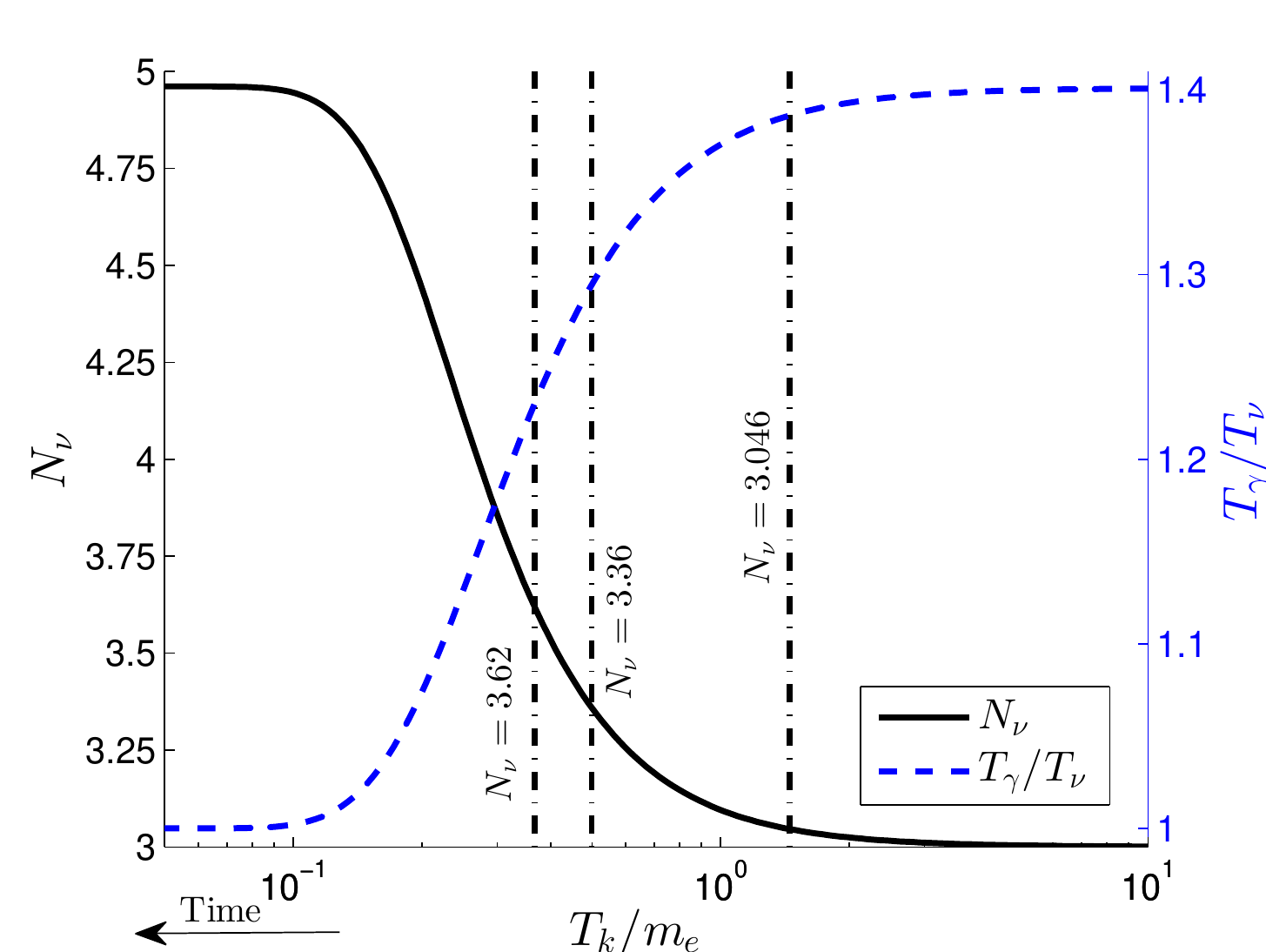}}
\caption{Dependence of neutrino fugacity (left) and effective number of neutrinos and reheating ratio (right) on the neutrino kinetic freeze-out temperature. We also show the evolution of the deceleration parameter through the freeze-out period (left).}
\end{minipage}
 \end{figure}

These two papers also contain several figures that show other relations relation between the quantities $T_k$, $N_\nu$, $\Upsilon$, and $T_\gamma/T_\nu$ for which we do not have simple analytic relations. In figure \ref{fig:Tk_dependence} we give slightly modified versions of two of these plots, showing the dependence of $N_\nu$, $\Upsilon$, and $T_\gamma/T_\nu$ on $T_k$.  In particular, the fugacity evolves following the solid black curve in the lefthand plot until it reaches the kinetic freeze-out temperature, at which point the neutrinos decouple and $\Upsilon$ remains constant thereafter, as shown in the dashed black curves for two sample values of $T_k$.

We showed in \cite{Birrell:2013_2} that after kinetic freeze-out, the free-streaming neutrino momentum distribution takes the form
\begin{equation}\label{neutrino_dist}
f(t,E)=\frac{1}{\Upsilon^{-1}e^{p/T_\nu}+ 1}
\end{equation}
where the neutrino effective temperature is redshifted as the universe expands
\begin{equation}\label{Tneutrino_dist}
T_\nu(t)\propto \frac{1}{a(t)}
\end{equation}
and the value of the fugacity  that developed during the freeze-out process is frozen into the distribution and remains constant while free-streaming. The resulting expressions for the energy density, pressure, and number density in the rest frame of the neutrino background are
\begin{align}
\rho&=\frac{g_\nu}{2\pi^2}\!\int_0^\infty\!\!\frac{\left(m_\nu^2+p^2\right)^{1/2}p^2dp }{\Upsilon^{-1}e^{p/T_\nu}+ 1},\label{neutrino_rho}\\[0.2cm]
P&=\frac{g_\nu}{6\pi^2}\!\int_0^\infty\!\!\frac{\left(m_\nu^2+p^2\right)^{-1/2}p^4dp }{\Upsilon^{-1} e^{p/T_\nu}+ 1},\label{neutrino_P}\\[0.2cm]
n&=\frac{g_\nu}{2\pi^2}\!\int_0^\infty\!\!\!\frac{p^2dp }{\Upsilon^{-1}e^{p/T_\nu}+ 1}.
\label{num_density}
\end{align}

  Finally, in  \cite{Birrell:2013_2} we presented for the first time a  physically consistent derivation of the equation of state of free-streaming neutrinos, including dependence on both $N_\nu$ and neutrino mass ($\beta=m_\nu/T_\gamma$). 
\begin{align}
&\rho^{EV}/\rho_0= N_\nu+0.1016\sum_i\beta_i^2+0.0015\delta N_\nu\sum_i\beta_i^2\notag\\
&-0.0001\delta N_\nu^2\sum_i\beta_i^2-0.0022\sum_i\beta_i^4,\\
&P^{EV}/P_0= N_\nu-0.0616\sum_i\beta_i^2-0.0049\delta N_\nu\sum_i\beta_i^2\notag\\
&+0.0005\delta N_\nu^2\sum_i\beta_i^2+0.0022\sum_i\beta_i^4.\label{tau_Ups}
\end{align}
The inclusion of fugacity was a crucial aspect in obtaining a physically consistent description, as it was in all of the above results.

\chapter{Neutrinos Today}\label{ch:nu_today}
Among the great science and technology challenges of this century is the development of the experimental capability to detect cosmic background neutrinos~\cite{Stodolsky:1975,Cabibbo:1982,Shvartsman,Langacker:1982,Smith,Ferreras:1995wf,Hagmann:1999kf,Duda:2001hd,Gelmini,Ringwald:2009,Liao:2012,Hedman}. With the  recently proposed PTOLEMY experiment, which aims to detect relic  elctron-neutrino capture by tritium \cite{PTOLEMY}, the characterization of the relic neutrino background is increasingly relevant.  Using our  characterization of the neutrino distribution after freeze-out and the subsequent free-streaming dynamics from \cite{Birrell:2013_2} and summarized in section \ref{nu_freezeout_summary}, we lay groundwork for a characterization of the present day relic neutrino spectrum, which we explore  from the  perspective of an observer moving relative to the neutrino background, including the dependence on neutrino mass and $N_\nu$. Beyond consideration of the observable neutrino distributions, we evaluate the ${\cal O}(G_F^2)$ mechanical drag force acting on the moving observer.  The work presented here can be found in our paper \cite{nu_today}

\section{Neutrino Distribution in a Moving Frame}

The neutrino background and the cosmic microwave background (CMB) were in equilibrium until decoupling (called freeze-out) at $T_k\simeq {\cal O}{\rm (MeV)}$, hence one surmises that an observer would have the same relative velocity relative to the relic neutrino background  as with CMB. As a particular example in considering the spectrum, we present in more detail the case of an observer comoving with  Earth velocity $v_\oplus=300$\,km/s relative to the CMB,  modulated by orbital velocity ($\pm29.8$\,km/s).  We will write velocities in units of $c$, though our specific results will be presented in km/s.

In the cosmological setting, for $T<T_k$ the neutrino spectrum evolves according to the well known Fermi-Dirac-Einstein-Vlasov (FDEV) free-streaming  distribution~\cite{Langacker:1982,bruhat,Wong,Birrell:2013_2}.  By casting it in a relativistically invariant form we can then make a transformation to the rest frame of an observer moving with relative velocity $v_{\text{rel}}$ and obtain
\begin{align}\label{neutrino_dist}
f(p^\mu)=&\frac{1}{\Upsilon^{-1} e^{\sqrt{(p^\mu U_\mu)^2-m_\nu^2}/T_\nu}+1}.
\end{align}
The 4-vector characterizing the rest frame of the neutrino FDEV distribution is
\begin{equation}\label{4_vel}
U^\mu=(\gamma,0,0,v_{\text{rel}}\gamma),\hspace{2mm} \gamma={1}/{\sqrt{1-v_{\text{rel}}^2}},
\end{equation} 
where we have chosen coordinates so that the relative motion is in the $z$-direction. The neutrino effective temperature $T_\nu(t)= T_k\,(a(t_k)/a(t))$ is the scale-shifted freeze-out temperature $T_k$. Here $a(t)$ is the cosmological scale factor where $\dot a(t)/a(t)\equiv H$ is the observable Hubble parameter. $\Upsilon$ is the  fugacity factor, here describing the underpopulation of neutrino phase space that was frozen into the neutrino FDEV distribution in the process of decoupling from the $e^\pm,\gamma$-QED background  plasma.

There are several available bounds on neutrino masses. Neutrino energy and pressure components are important before photon freeze-out and thus $m_\nu$ impacts Universe dynamics. The analysis of CMB data alone leads to $\sum_i m_\nu^i<0.66$eV ($i=e,\mu,\tau$) and including Baryon Acoustic Oscillation (BAO) gives $\sum m_\nu<0.23$eV~\cite{Planck}.  {\small PLANCK CMB} with lensing observations~\cite{Battye:2013xqa} lead to  $\sum m_{\nu}=0.32\pm0.081$ eV. Upper bounds have been placed on the electron neutrino mass in direct laboratory measurements  $m_{\bar\nu_e}<2.05$eV~\cite{Beringer:1900zz,Aseev}.   In the subsequent analysis we will focus on the neutrino mass range $0.05$eV to $2$eV in order to show that direct measurement sensitivity allows the exploration of a wide mass range.

 The relations derived in~\cite{Birrell:2013_2} and restated in section \ref{nu_freezeout_summary} determine $T_\nu/T_\gamma$ and  $\Upsilon$ in terms of the measured  value of  $N_\nu$ under the assumption of a strictly SM-particle inventory.  In the following we treat $N_\nu$  as a variable model parameter and use the above mentioned relations to characterize our results in terms of $N_\nu$.

\section{Velocity, Energy, and Wavelength Distributions}

Using \req{neutrino_dist}, the normalized FDEV velocity distribution for an observer in relative motion has the form
\begin{align} \label{fvdistrib}
&f_v=\frac{g_\nu}{n_\nu 4\pi^2}\!\!\!\int_0^\pi \!\!\!\!\frac{ p^2dp/dv\sin(\phi) d\phi}{\Upsilon^{-1}e^{\sqrt{( E-v_{\text{rel}} p \cos(\phi))^2\gamma^2-m_\nu^2}/T_\nu}+1},\notag\\
&p(v)=\frac{m_\nu v}{\sqrt{1-v^2}},\qquad \frac{dp}{dv}=\frac{m_\nu}{(1-v^2)^{3/2}}.
\end{align}
The normalization $n_\nu$ depends on $N_\nu$ but not on $m_\nu$ since decoupling occurred at $T_k\gg m_\nu$. For each neutrino flavor (all flavors are equilibrated by oscillations) we have, per neutrino or antineutrino and at non-relativistic relative velocity,
\begin{equation}\label{nnu}
n_\nu=[-0.3517\delta N_\nu^2+6.717\delta N_\nu+56.06]\,{\rm cm}^{-3}
\end{equation}
($\delta N_\nu\equiv N_\nu-3$), compare to Eq.(55) in Ref.~\cite{Birrell:2013_2}.

We show $f_v$ in figure \ref{fig:rel_v_dist_300}   for several values of the neutrino mass, $v_{\text{rel}}=300$ km/s, and $N_\nu=3.046$ (solid lines) and $N_\nu=3.62$ (dashed lines). As expected, the lighter the neutrino, the more $f_v$  is weighted towards higher velocities with the velocity becoming visibly peaked about $v_{\text{rel}}$ for $m_\nu=2$ eV. 
\begin{figure}
\centerline{\includegraphics[height=6cm]{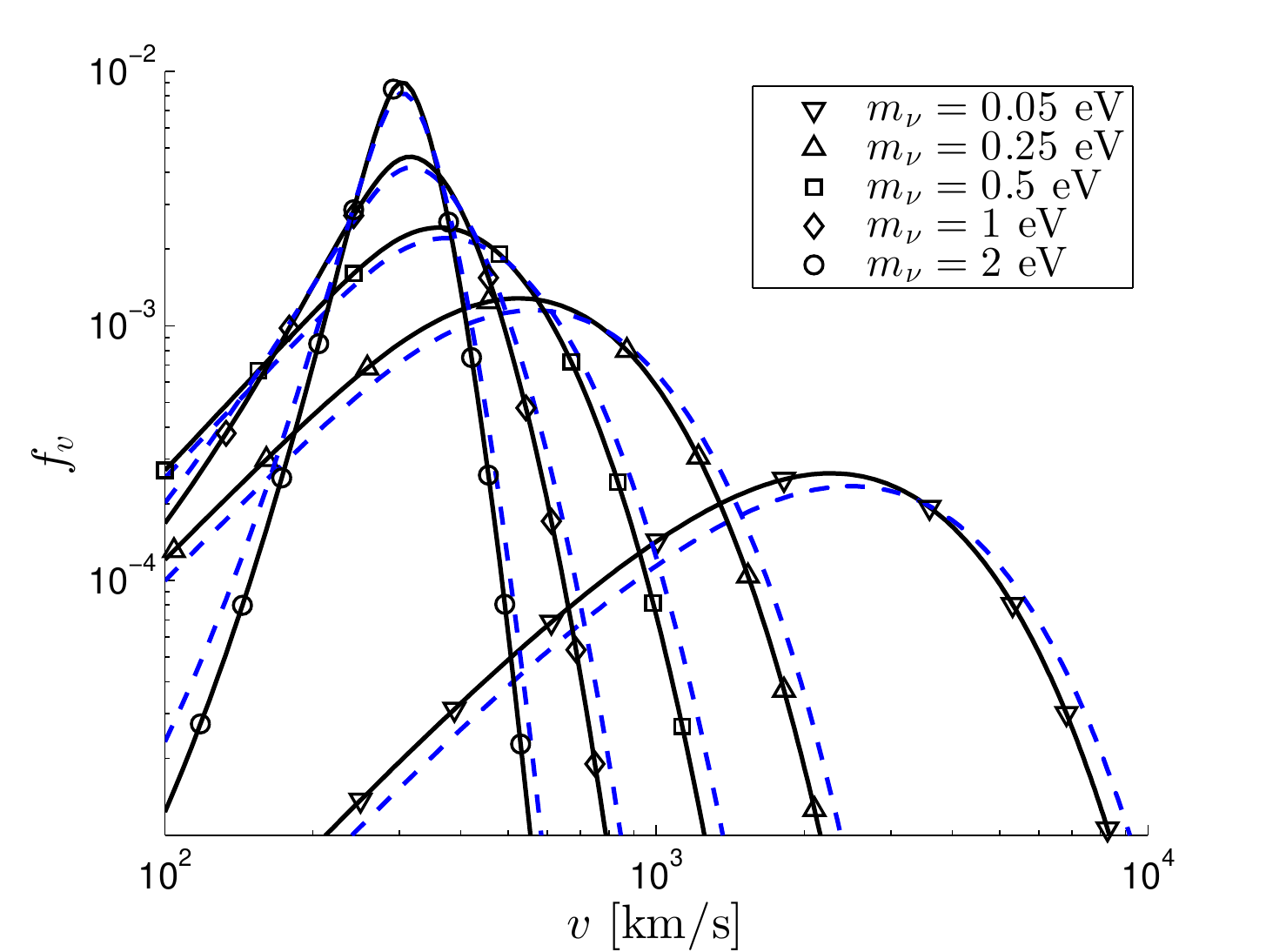}}
\caption{Normalized neutrino FDEV velocity distribution in the Earth frame. We show the distribution for $N_\nu=3.046$ (solid lines) and $N_\nu=3.62$ (dashed lines).}\label{fig:rel_v_dist_300}
 \end{figure}

A similar procedure produces the normalized FDEV energy distribution $f_E$.  In \req{fvdistrib} we replace $dp/dv\to dp/dE$ where it is understood that 
\begin{equation}
p(E)=\sqrt{E^2-m_\nu^2},\qquad \frac{dp}{dE}=\frac{E}{p}.
\end{equation}
We show $f_E$ in figure \ref{fig:E_dist_300}  for several values of the neutrino mass, $v_{\text{rel}}=300$ km/s, and $N_\nu=3.046$ (solid lines) and $N_\nu=3.62$ (dashed lines). The width of the FDEV energy distribution is on the micro-eV scale and the kinetic energy $T=E-m_\nu$ is peaked about $T=\frac{1}{2}m_\nu v_{\text{rel}}^2$, implying that the relative velocity between the Earth and the CMB is the dominant factor for $m_\nu>0.1$ eV.

\begin{figure}
\centerline{\includegraphics[height=6cm]{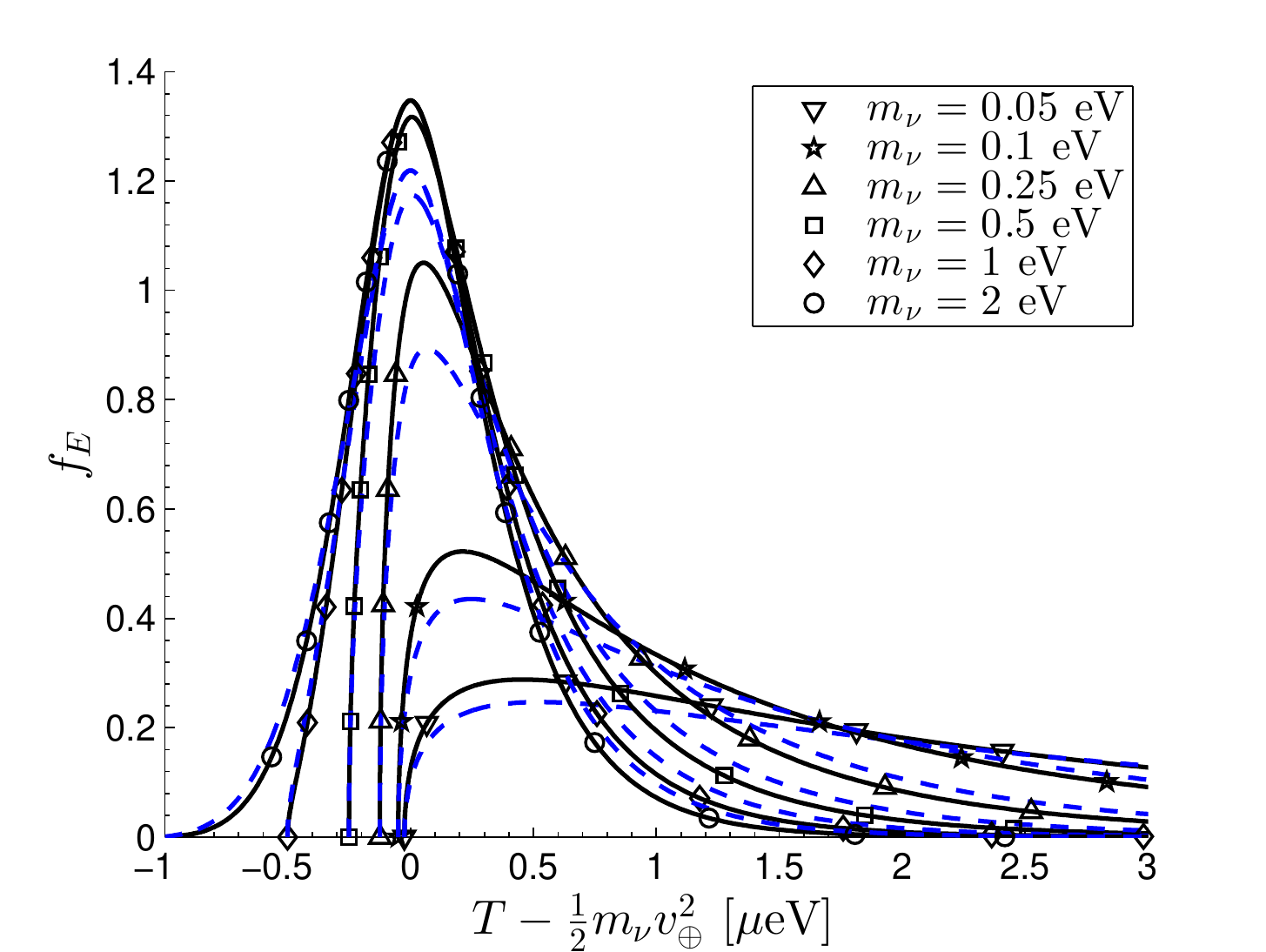}}
\caption{Neutrino FDEV energy distribution in the Earth frame. We show the distribution for $N_\nu=3.046$ (solid lines) and $N_\nu=3.62$ (dashed lines). }\label{fig:E_dist_300}
 \end{figure}

By multiplying $f_E$ by the neutrino velocity and number density for a single neutrino flavor (without anti-neutrinos) we obtain the particle flux density,
 \begin{equation}
 \frac{dJ}{dE}\equiv\frac{dn}{dAdtdE},
\end{equation} 
shown in figure \ref{fig:flux_dist}. We show the result for $N_\nu=3.046$ (solid lines) and $N_\nu=3.62$ (dashed lines). The flux is normalized in these cases to a local denisty $56.36$~cm${}^{-3}$ and $60.10$~cm${}^{-3}$ respectively. The precise neutrino flux in the Earth frame is significant for efforts to detect relic neutrinos, such as the PTOLEMY experiment~\cite{PTOLEMY}. The energy dependence of the flux shows a large sensitivity to the mass. However, the maximal fluxes do not vary significantly with $m$. In fact the maximum values are independent of $m$ when $v_{\text{rel}}=0$, as follows from the fact that $v=p/E=dE/dp$.  In the Earth frame, where $0<v_\oplus\ll c$, this translates into only a small variation in the maximal flux.
\begin{figure}
\centerline{\includegraphics[height=7cm]{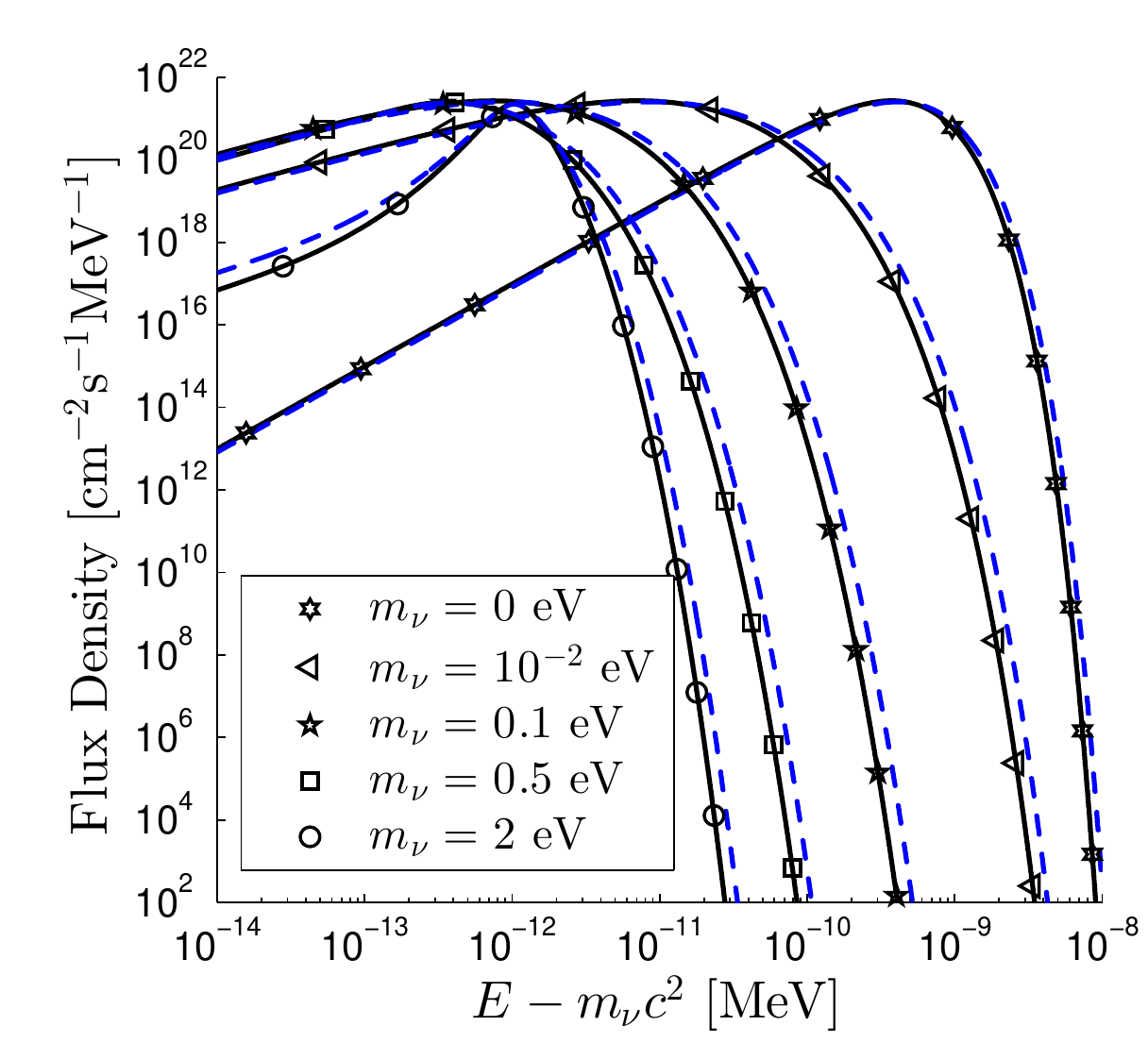}}
\caption{Neutrino flux density in the Earth frame. We show the result for $N_\nu=3.046$ (solid lines) and $N_\nu=3.62$ (dashed lines) for an observer moving with $v_\oplus=300$\,km/s.}\label{fig:flux_dist}
 \end{figure}

Using $\lambda=2\pi/p$ we find  in the the normalized FDEV de Broglie wavelength distribution
\begin{equation}
f_\lambda=\frac{ 2\pi g_\nu}{n_\nu\lambda^4}\!\!\int_0^\pi\!\!\! \!\frac{\sin(\phi) d\phi}{\Upsilon^{-1}e^{\sqrt{( E-v_{\text{rel}} p \cos(\phi))^2\gamma^2-m_\nu^2}/T_\nu}\!\!+\!1}
\end{equation}
shown in figure \ref{fig:deBrogle_300} for $v_{\text{rel}}=300$ km/s and for several values $m_\nu$ comparing  $N_\nu=3.046$ with $N_\nu=3.62$. 
\begin{figure}
\begin{minipage}{\linewidth}
\makebox[0.5\linewidth]%
{\includegraphics[height=5.8cm]{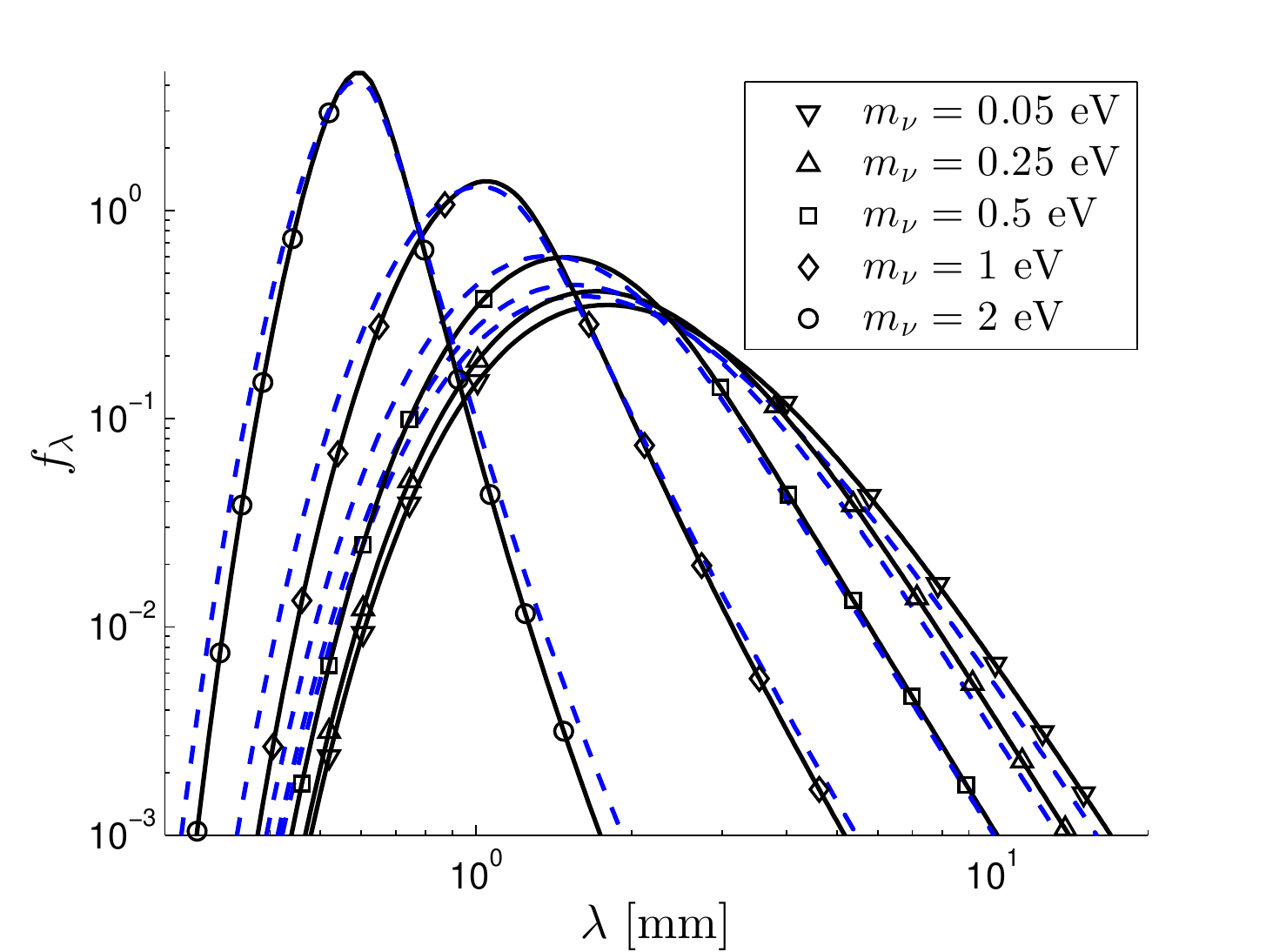}}
\makebox[0.5\linewidth]%
{\includegraphics[height=5.8cm]{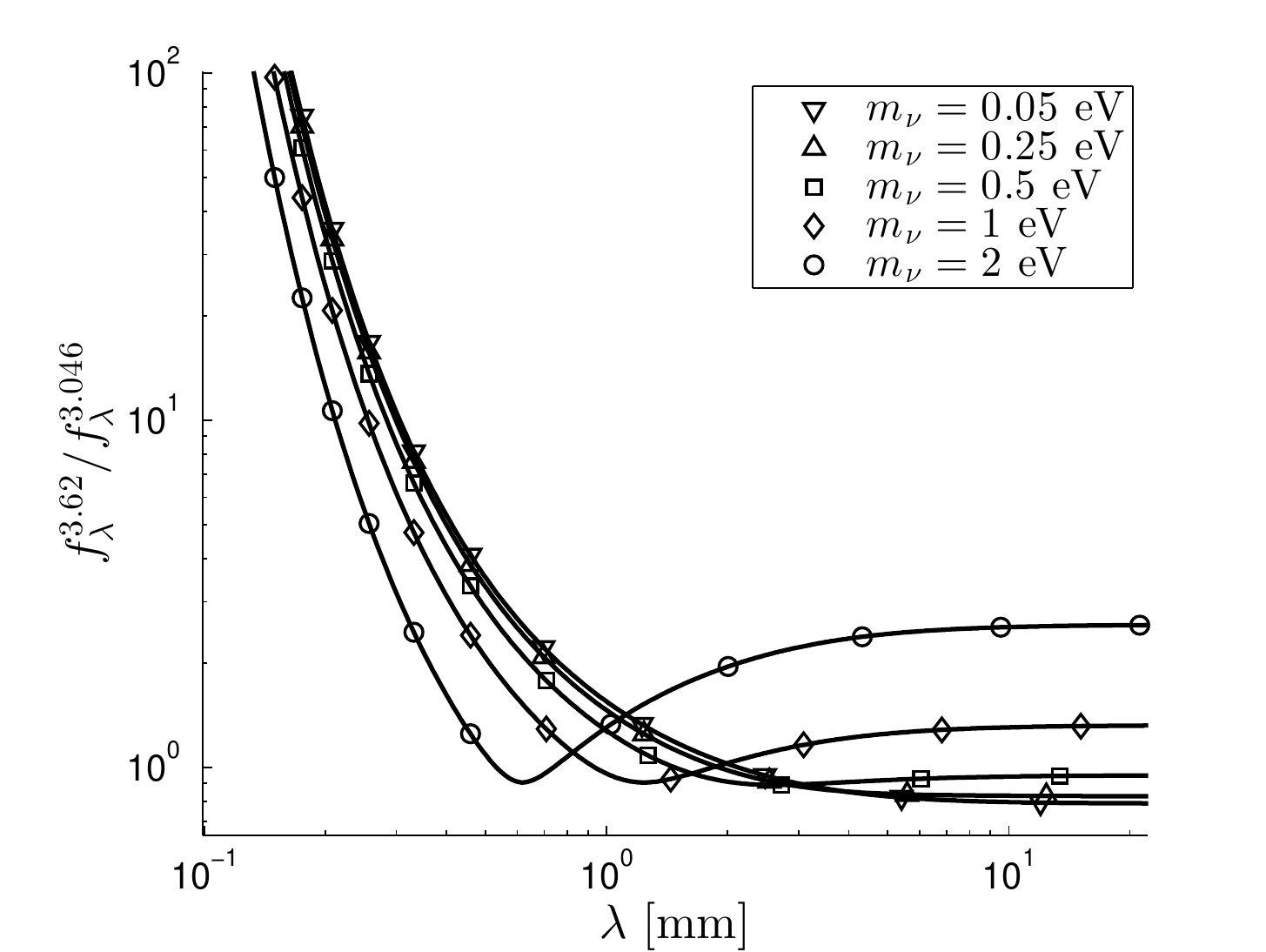}}
\end{minipage}
\caption{Neutrino  FDEV de Broglie wavelength  distribution in the Earth frame. We show in top panel the distribution for $N_\nu=3.046$ (solid lines) and $N_\nu=3.62$ (dashed lines) and in bottom panel their ratio.}\label{fig:deBrogle_300}
 \end{figure}

\section{Drag Force}
Given the neutrino distribution, we evaluate the drag force due to the anisotropy of the neutrino distribution in the rest frame of the moving object for $N_\nu=3.046$. The relic neutrinos will undergo potential scattering with the scale of the potential strength being
\begin{equation}\label{V0}
V_0=CG_F\rho_{N_c},\hspace{2mm} \rho_N\equiv N_c/R^3
\end{equation}
where $R$ is the linear size of the detector.  When the detector size is smaller than the quantum de Broglie wavelength of the neutrino, all scattering centers are added coherently to for the target effective `charge' $N_c$.  $\rho_{N_c}$ is the charge density, and C=O(1) and is depending on material composition of the object. Such considerations are of interest both for scattering from terrestrial detectors, as well as for ultra-dense objects of neutron star matter density, e.g.   strangelet  CUDOS~\cite{CUDO} - recall that such nuclear matter fragments with $R<\lambda$  despite their small size would have a mass rivaling that of large meteors. We find $V_0\simeq 10^{-13}$ eV for normal matter densities, but for nuclear target density a potential well with $V_0\simeq {\cal O}(10 {\rm eV})$.  

We consider relic neutrino potential scattering to obtain the average momentum transfer to the target and hence the drag force. The particle flux per unit volume in momentum space is
\begin{equation}\label{dn_quantum}
\frac{dn}{dtd Ad^3{\bf p}}({\bf p})=2f({\bf p})p/m_\nu,\hspace{2mm} p\equiv |{\bf p}|
\end{equation}
where the factor of two comes from combining neutrinos and anti-neutrinos of a given flavor. Our use of non-relativistic velocity is justified by figure \ref{fig:rel_v_dist_300}.  The recoil change in detector momentum per unit time is  
\begin{align}
\frac{d{\bf p}}{dt}=& \int  {\bf q}A \frac{dn}{dtdAd^3p}({\bf p})d^3p,\\
{\bf q}A\equiv &\int ({\bf p}-{\bf p_f})\frac{d\sigma}{ d\Omega}({\bf p_f},{\bf p})d\Omega.
\end{align}
Here ${\bf p}$ and ${\bf p}_f$, the incoming and outgoing momenta respectively, have the same magnitude. $qA$ is the momentum transfer times area, averaged over outgoing momenta, and $d\Omega$ is the solid angle for to ${\bf p}_f$.  

For a spherically symmetric potential the differential cross section depends only on the incoming energy and the angle $\phi$ between ${\bf p}$ and ${\bf p_f}$.  Therefore, for each ${\bf p}$ the integral over $d\Omega$ of the components orthogonal to ${\bf p}$ is zero by symmetry.  This implies
\begin{align}
{\bf q}A\equiv &2\pi{\bf p}\int(1-\cos(\phi))\frac{d\sigma}{ d\Omega}(p,\phi)\sin(\phi)d\phi.
\end{align}
The only angular dependence in the neutrino distribution is in ${\bf p}\cdot {\bf\hat z}$ and therefore the components of the force orthogonal to ${\bf \hat z}$ integrate to zero, giving
\begin{align}\label{drag}
\frac{d{\bf p}}{dt}=&\frac{8\pi^2{\bf\hat z}}{m_\nu } \int p^4g(p) f(p,\tilde\phi) \cos(\tilde\phi) \sin(\tilde\phi) dpd\tilde\phi,\\
g(p)\equiv& \int_0^\pi(1-\cos(\phi))\frac{d\sigma}{ d\Omega}(p,\phi)\sin(\phi)d\phi.\label{geq}
\end{align}

For the case of normal density matter, the Born approximation is valid due to the weakness of the potential compared to the neutrino energy seen in figure \ref{fig:E_dist_300}. To obtain an order of magnitude estimate, we take a Gaussian potential \begin{equation}
V(r)=V_0e^{-r^2/R^2}
\end{equation}
for which the differential cross section in the Born approximation can be analytically evaluated
\begin{align}
&\frac{d\sigma}{ d\Omega}(p,\phi)=\frac{\pi m_\nu^2V_0^2R^6}{4}e^{-q^2R^2/2},\notag\\
&q=|{\bf p}-{\bf p}_f|=2p\sin(\phi/2).
\end{align}
The integral over $\phi$ in \req{geq} can also be done analytically,
\begin{align}
g(p)=&\pi m_\nu^2V_0^2R^6\frac{1-(2R^2p^2+1)e^{-2R^2p^2}}{4R^4p^4}.
\end{align}
 In the long and short wavelength limit we have 
\begin{align}
\label{long_wave_drag}
&g(p)\simeq\frac{\pi}{2} m_\nu^2V_0^2R^6,\quad pR\ll 1\\
&F_L\simeq 4\pi^3m_\nu V_0^2R^6 \int p^4 f(p,\tilde\phi) \cos(\tilde\phi) \sin(\tilde\phi) dpd\tilde\phi\notag\\
\label{short_wave_drag}
&g(p)\simeq\frac{\pi m_\nu^2V_0^2R^2}{4p^4}, 
\quad pR\gg 1\\
&F_S\simeq 2\pi^3 m_\nu V_0^2R^2 \int f(p,\tilde\phi) \cos(\tilde\phi) \sin(\tilde\phi) dpd\tilde\phi\notag.
\end{align}
 We also note that in the short wavelength limit, our coherent scattering treatment is only applicable to properly prepared structured targets \cite{Liao:2012}.

Inserting \req{V0} we see that this force is $O(G_F^2)$, see also \cite{Shvartsman,Smith,Gelmini}, as compared to the $O(G_F)$ effects debated in  \cite{Opher,Lewis,Opher2,Cabibbo:1982,Langacker:1982,Smith,Ferreras:1995wf}. In long wavelength limit the size $R$ cancels, in favor of $N_c^2$ which explicitly shows that scattering is on the square of the charges of the target. This results in an enhancement of the force by a factor of $N_c$ over the incoherent scattering case, due to $V_0^2$ scaling with $N_c^2$. This effect exactly parallels the proposed detection of supernovae MeV energy scale neutrinos by means of collisions with the entire atomic nucleus~\cite{Divari:2012zz}.

Fits to the integrals in the above force formulas \req{long_wave_drag} and \req{short_wave_drag} can be obtained in the region $0.005 \text{eV}\leq m_\nu\leq 0.25\text{eV}$, $v_\text{rel}\leq 300$km/s, yielding
\begin{align}\label{F_L}
F_L\!=&2\,10^{-31}{\rm N}\!\left(\!\frac{m_\nu}{0.1 \text{eV}}\!\right)^{\!\!2}\!\! \left(\!\frac{V_0}{1\text{peV}}\!\right)^{\!\!2}\!\!\left(\!\frac{R}{1\text{mm}}\!\right)^{\!\! 6} \!\!\frac{v_{\text{rel}}}{v_\oplus},\\[0.2cm]
F_S=&4\, 10^{-33}{\rm N}\!\left(\!\frac{m_\nu}{0.1\text{eV}}\!\right)^{\!\!2}\! \left(\!\frac{V_0}{1\text{peV}}\!\right)^{\!\!2}\!\!\left(\!\frac{R}{1\text{mm}}\!\right)^{\!\!2}\times\notag \\
&\hspace*{1.5cm}\times\frac{v_{\text{rel}}}{v_\oplus}\!\left(\!1\!-\!0.2\frac{m_\nu}{0.1\text{eV}}\frac{v_{\text{rel}}}{v_\oplus}\right).
\end{align}
We emphasize that they are not valid in the limit as $m_\nu\rightarrow 0$. Considering that the current frontier of precision force measurements at the level of individual ions is on the order of $10^{-24}$N \cite{Biercuk}, the ${\cal O}(G_F^2)$ force on a coherent mm-sized terrestrial detector is negligible, despite the factor of $N_c$ enhancement. 

We now consider scattering from nuclear matter density $\rho_N\simeq 3 \, 10^{8}{\rm kg/mm}^3$ objects where $V_0={\cal O}(10\text{eV})$ is effectively infinite compared to the neutrino energy unless the object velocity relative to the neutrino background is ultra-relativistic.  Therefore we are in the hard `ball' scattering limit. As with the analysis for normal matter density, we will investigate both the long and short wavelength limits. In the long wavelength limit, only the S-wave contributes to hard sphere scattering and $d\sigma/d\Omega=R^2$, independent of angle. Using \req{drag} and a similar fit to \req{F_L} gives
\begin{align}\label{FLHard}
F_L=&\frac{16\pi^2R^2}{m_\nu} \int p^4 f(p,\tilde\phi) \cos(\tilde\phi) \sin(\tilde\phi) dpd\tilde\phi\\
\simeq &\, 4\, 10^{-20}{\rm N}\left(\frac{R}{1\text{mm}}\right)^2\frac{v_{\text{rel}}}{v_\oplus}.
\end{align}
In particular the force is independent of $m_\nu$.  We also note that at high velocity, \req{FLHard}  underestimates the drag force. The resulting acceleration is
\begin{equation}
a=10^{-28}\frac{m}{s^2}\frac{v_{\text{rel}} }{v_\oplus}\! \left(\frac{R}{1\text{mm}}\right)^{-1}\!\!\left(\frac{\rho}{\rho_N}\right)^{-1}.\!\!
\end{equation}
The Newtonian drag time constant, $v_{\rm rel}/a$, is
\begin{equation}
\tau= 10^{26}\text{yr}\frac{R}{1\text{mm}}\,\frac{\rho}{\rho_N}\!
\end{equation}
which suggests that the compact object produced early on in stellar evolution remain largely unaltered.

The last case to consider is the short wavelength hard sphere scattering limit.  This limit is classical and so we no longer treat it as quantum mechanical potential scattering, but rather as elastic scattering of point particle neutrinos from a hard sphere of radius $R$.  For a single scattering event where the component of the momentum normal to the sphere is ${\bf p}^\perp=({\bf p}\cdot \hat {\bf r}) \hat {\bf r}$, the change in particle momentum is  $\Delta {\bf p}=-2{\bf p}^\perp$. The particle flux per unit volume in momentum space at a point ${\bf r}$ on a radius $R$ sphere $S_R^2$ and inward pointing momentum ${\bf p}$ (i.e. ${\bf p}\cdot \hat{\bf  r}<0$) is
\begin{equation}\label{dn_classical}
\frac{dn}{dtd Ad^3{\bf p}}({\bf x},{\bf p})=2f({\bf p})|{\bf v}\cdot \hat{\bf r}|
\end{equation}
where the factor of two comes from combining neutrinos and anti-neutrinos of a given flavor.  Note that for point particles the flux is proptional to the normal component of the velocity, as opposed to wave scattering where it is proportional to the magnitude of the velocity, seen in \req{dn_quantum}.

Using \req{dn_classical}, the recoil change in momentum per unit time is  
\begin{equation}
\frac{d{\bf p}}{dt}= -\int_{{\bf p}\cdot \hat{\bf r}<0} \!\!\!\!\!\!2\Delta {\bf p}  f({\bf p}) \frac{1}{m_\nu}|{\bf p}\cdot \hat{\bf r}| d^3{\bf p}R^2d\Omega.
\end{equation}
The only angular dependence in $f$ is through ${\bf p}\cdot \hat {\bf z}$ so by symmetry, the ${\bf \hat x}$ and ${\bf \hat y}$ components integrate to $0$.  Therefore we have
\begin{equation}
\frac{d{\bf p}}{dt}=\frac{-4R^2\hat{\bf z}}{m_\nu}\int_{{\bf p}\cdot \hat{\bf r}<0}\!\!\!   f({\bf p}) ({\bf p}\cdot \hat{\bf r})^2 \hat {\bf r}\cdot\hat{\bf z}\, d^3{\bf p}d\Omega.  
\end{equation}

We perform this integration in spherical coordinates for ${\bf r}$ and in the spherical coordinate vector field basis for ${\bf p}=p_r\hat{\bf r} +p_\theta\hat{\bf r}_\theta+p_\phi\hat{\bf r}_\phi,\hspace{2mm} p_r<0$, where we recall
\begin{align}
&\hat{\bf r}=\cos  \theta \sin \phi \, \hat{\bf x}+\sin \theta \sin \phi \hat{\bf y}+\cos  \phi\,\hat{\bf z}\notag,\\
&\hat{\bf r}_\theta=-\sin \theta \hat{\bf x}+\cos \theta \hat{\bf y},\\
&\hat{\bf r}_\phi=\cos \theta \cos \phi\,\hat{\bf x}+\sin \theta \cos \phi \,\hat{\bf y}-\sin \phi \, \hat{\bf z}\notag.
\end{align}
Therefore the force per unit surface area is
\begin{align}\label{drag_eq}
\frac{1}{A}\frac{d{\bf p}}{dt}=&\frac{-2}{m_\nu}\int_0^\pi\!\!\int_{p_r<0}\!\!\!\!\!\!\!\! f({\bf p}) p_r^2  d^3{\bf p}\cos \phi \sin \phi  d\phi\hat{\bf z},\notag\\
f({\bf}p)=&\frac{1}{  \Upsilon^{-1}e^{\sqrt{( E-V_\oplus {\bf p}\cdot \hat {\bf z})^2\gamma^2-m_\nu^2}/T_\nu}+1 },\\
          &{\bf p}\cdot\hat{\bf z}=p_r\cos \phi-p_\phi\sin\phi .\notag
\end{align}

We obtain an approximation over the range $v_{\text{rel}}\leq v_\oplus ;\ 0.05\text{eV}\leq m_\nu\leq 0.25\text{eV}$ given by
\begin{align}\label{drag_fit}
&F_S =  10^{-20} \text{N}\left(\frac{R}{1{\rm mm}}\right)^2\frac{v_{\text{rel}}}{ v_\oplus}.
\end{align}
This is a similar result to the long wavelength hard sphere limit \req{FLHard}, but the fact that it is only applicable to objects larger than the neutrino wavelength means that the acceleration it generates is negligible on the timescale of the Universe.

\section{Summary and Conclusion}
In summary, we have characterized the relic cosmic neutrinos and their velocity, energy, and de Broglie wavelength distributions in a frame of reference moving relative to the neutrino background. We have shown explicitly the mass $m_\nu$ dependence and the dependence on neutrino reheating expressed by $N_\nu$, choosing a range within the experimental constraints. This is a necessary input for the measurement of $N_\nu$ and neutrino mass by future detection efforts.  Finally, we have discussed in detail the $O(G_F^2)$ mechanical drag force  originating in the dipole anisotropy induced by motion relative to the neutrino background.  Despite enhancement with the total target charge found within the massive neutrino wavelength, the magnitude of the force is found to be well below the reach of current  precision force measurements.

Our results are derived under the assumption that $N_\nu$ is due entirely to SM neutrinos, with no contribution from new particle species. In principle future, relic neutrino detectors, such as PTOLEMY~\cite{PTOLEMY}, will be able to distinguish between these alternatives since the effect of $N_\nu$ as presented here is to increase neutrino flux~\cite{Birrell:2013_2}, see \req{nnu}. However, to this end one must gain precise control over the enhancement of neutrino galactic relic density due to  gravitational effects~\cite{Ringwald:2004np} as well as the annual modulation~\cite{Safdi}.

\part{General Relativistic Boltzmann Equation and Neutrino Freeze-out }
\chapter{Geometry Background: Volume Forms on Submanifolds}\label{ch:vol_forms}
We now begin our transition from chemical and kinetic equilibrium matter models to full nonequilibrium.  This requires the machinery of general relativistic kinetic theory.  Before we make that transition, in this chapter we address the purely geometrical problem of inducing volume forms on submanifolds, a problem that will be of interest to us in concrete computations later on and is necessary for the definition of kinetic theory in its usual form. This chapter is much more mathematical than the remainder of this dissertation and, when standard, we use geometrical language and notation here without further explanation.  See for example \cite{lee2003introduction,lee1997riemannian}.  We found this formalism to be a great aid, especially in our development of an improved method for computing scattering integrals, presented in chapter \ref{ch:coll_simp}.  However, if one is content with simply using the results then this chapter is non-essential. See also our paper \cite{Birrell:2014uka}.

\section{Inducing Volume Forms on Submanifolds}

Given a Riemannian manifold $(M,g)$ with volume form $dV_g$ and a  hypersurface $S$, the standard Riemannian hypersurface area form, $dA_g$ is defined on $S$ as the volume form of the pullback metric tensor on $S$.   Given vectors $v_1,...,v_k$ we define the interior product (i.e. contraction)  operator acting on a form $\omega$ of degree $n\geq k$ as the $n-k$ form 
\begin{equation}
i_{(v_1,...,v_k})\omega=\omega(v_1,...,v_k,\cdot).
\end{equation}
 With this notation, the hypersurface area form can equivalently be computed as
\begin{equation}
dA_g=i_v dV_g
\end{equation}
where $v$ is a unit normal vector to $S$.  This method extends to submanifolds of codimension greater than one as well as to semi-Riemannian manifolds, as long as the metric restricted to the submanifold is non-degenerate. 

However, there are many situations where one would like to define a natural volume form on a submanifold that is induced by a volume form in the ambient space, but where the above method is inapplicable, such as defining a natural volume form on the light cone or other more complicated degenerate submanifolds in relativity. In this section, we will describe a method for inducing volume forms on regular level sets of a function that is applicable in cases where there is no metric structure and show its relation to more widely used semi-Riemannian case.  We prove analogues of the coarea formula and Fubini's theorem in this setting. 

Let $M$, $N$ be smooth manifolds, $c$ be a regular value of a smooth function $F:M\rightarrow N$, and $\Omega^M$ and $\Omega^N$ be volume forms on $M$ and $N$ respectively.  Using this data, we will be able to induce a natural volume form on the level set $F^{-1}(c)$.  The absence of a metric on $M$ is made up for by the additional information that the function $F$ and volume form $\Omega^N$ on $N$ provide. The following proposition makes our definition precises and proves the existence and uniqueness of the induced volume form.

\begin{proposition}\label{induced_vol_form}
Let $M$, $N$ be $m$ (resp. $n$)-dimensional smooth manifolds with volume forms $\Omega^M$ (resp. $\Omega^N$). Let $F:M\rightarrow N$ be smooth and $c$ be a regular value.  Then there is a unique volume form $\omega$  (also denoted $\omega^M$) on $F^{-1}(c)$ such that $\omega_x=i_{(v_1,...,v_n)}\Omega^M_x$ whenever $v_i\in T_xM$ are such that 
\begin{equation}\label{unit_volume}
\Omega^N(F_*v_1,...,F_* v_n)=1.
\end{equation}
We call $\omega$ the {\bf volume form induced by $F:(M,\Omega^M)\rightarrow (N,\Omega^N)$}.
\end{proposition}
\begin{proof}
$F_*$ is onto $T_{F(x)}N$ for any $x\in F^{-1}(c)$.  Hence there exists $\{v_i\}_1^n\subset T_xM$ such that $\Omega^N(F_*v_1,...,F_* v_n)=1$.  In particular, $F_* v_i$ is a basis for $T_{F(x)} N$.  Define $\omega_x=i_{(v_1,...,v_n)}\Omega_x$. This is obviously a nonzero $m-n$ form on $T_xF^{-1}(c)$ for each $x\in F^{-1}(c)$.  We must show that this definition is independent of the choice of $v_i$ and the result is smooth.\\

 Suppose $F_*v_i$ and $F_*w_i$ both satisfy \req{unit_volume}.  Then $F_*v_i=A_i^jF_*w_j$ for $A\in SL(n)$. Therefore $v_i-A_i^jw_j\in \ker F_{*x}$.  This implies
\begin{equation}
i_{(v_1,...,v_n)}\Omega^M_x=\Omega^M_x(A_1^{j_1}w_{j_1},...,A_n^{j_n}w_{j_n},\cdot)
\end{equation}
since the terms involving $\ker F_*$ will vanish on $T_x F^{-1}(c)=\ker F_{*x}$.  Therefore
\begin{align}\label{ind_of_v_proof}
i_{(v_1,...,v_n)}\Omega^M_x&=A_1^{j_1}...A_n^{j_n}\Omega^M_x(w_{j_1},...,w_{j_n},\cdot)\\
&=\sum_{\sigma\in S_n} \pi(\sigma)A_1^{\sigma(1)}...A_n^{\sigma(n)}\Omega^M_x(w_1,...,w_n,\cdot)\\
&=\det(A)i_{(w_1,...,w_n)}\Omega^M_x\\
&=i_{(w_1,...,w_n)}\Omega^M_x.
\end{align}
This proves that $\omega$ is independent of the choice of $v_i$.  If we can show $\omega$ is smooth then we are done.  We will do better than this by proving that for any  $v_i\in T_xM$ the following holds
\begin{equation}
i_{(v_1,...,v_n)}\Omega^M_x=\Omega^N(F_*v_1,...,F_*v_n)\omega_x.
\end{equation}
To see this, take $w_i$ satisfying \req{unit_volume}.  Then $F_*v_i=A_i^j F_*w_j$. This determinant can be computed from
\begin{align}
\Omega^N(F_*v_1,...,F_*v_n)=\det(A)\Omega^N(F_*w_1,...,F_*w_n)=\det(A).
\end{align}
 Therefore, the same computation as \req{ind_of_v_proof} gives
\begin{align}
i_{(v_1,...,v_n)}\Omega^M_x=\det(A)\omega_x=\Omega^N(F_*v_1,...,F_*v_n)\omega_x
\end{align}
as desired.  To prove that $\omega$ is smooth, take a smooth basis of vector fields $\{V_i\}_1^m$ in a neighborhood of $x$.  After relabeling, we can assume $\{F_*V_i\}_1^n$ are linearly independent at $F(x)$ and hence, by continuity, they are linearly independent at $F(y)$ for all $y$ in some neighborhood of $x$.  In that neighborhood, $\Omega^N(F_*V_1,...,F_*V_n)$ is non-vanishing and therefore
\begin{equation}
\omega=(\Omega^N(F_*V_1,...,F_*V_n))^{-1}i_{(V_1,...,V_n)}\Omega
\end{equation} 
which is smooth.
\end{proof}

\begin{corollary}\label{induced_vol_eq}
For any  $v_i\in T_xM$ the following holds
\begin{equation}\label{vol_formula1}
i_{(v_1,...,v_n)}\Omega^M_x=\Omega^N(F_*v_1,...,F_*v_n)\omega_x.
\end{equation}
\end{corollary}

\begin{corollary}
If $\phi:M\rightarrow\mathbb{R}$ is smooth and $c$ is a regular value then by equipping $\mathbb{R}$ with its canonical volume form we have 
\begin{equation}
\omega_x=i_v\Omega^M_x
\end{equation}
where $v\in T_xM$ is any vector satisfying $d\phi(v)=1$.
\end{corollary}

It is useful to translate  \req{vol_formula1} into a form that is more readily applicable to computations in coordinates.  Choose arbitrary coordinates $y^i$ on $N$ and write $\Omega^N=h^N(y) dy^n$. Choose coordinates $x^i$ on $M$ such that $F^{-1}(c)$ is the coordinate slice
\begin{equation}
F^{-1}(c)=\{x:x^1=...=x^n=0\}
\end{equation}
and write $\Omega^M=h^M(x)dx^m$. The coordinate vector fields $\partial_{x^i}$ are transverse to $F^{-1}(c)$ and so
\begin{equation}
\Omega^N(F_*\partial_{x^1},...,F_*\partial_{x^n})=h^N(F(x))\det \left(\frac{\partial F^i}{\partial x^j}\right)_{i,j=1..n}
\end{equation}
and
\begin{equation}
i_{(\partial_{x^1},...,\partial_{x^n})}\Omega^M=h^M(x) dx^{n+1}...dx^m.
\end{equation}
Therefore we obtain
\begin{equation}\label{vol_form_coords}
\omega_x=\frac{h^M(x)}{h^N(F(x))}\det \left(\frac{\partial F^i}{\partial x^j}\right)^{-1}_{i,j=1..n}dx^{n+1}...dx^m.
\end{equation}

Just like in the (semi)-Riemannian case, the induced measure allows us to prove a coarea formula where we break integrals over $M$ into slices. In this theorem and the remainder of the section, we consider integration with respect to the density defined by any given volume form i.e. we ignore the question of defining consistent orientations.
\begin{theorem}[Coarea formula]\label{vol_form_coarea}
Let $M$ be a smooth manifold with volume form $\Omega^M$, $N$ a smooth manifold with volume form $\Omega^N$ and $F:M\rightarrow N$ be a smooth map.  If $F_*$ is surjective at a.e. $x\in M$ then for $f\in L^1(\Omega^M)\bigcup L^+(M)$
\begin{equation}\label{coarea_formula}
\int_Mf(x) \Omega^M(dx)=\int_{N}\int_{F^{-1}(z)} f(y)\omega^M_z(dy) \Omega^N(dz)
\end{equation}
where $\omega^M_z$ is the volume form induced on $F^{-1}(z)$ as in lemma \ref{induced_vol_form}.
\end{theorem}
\begin{proof}
First suppose $F$ is a submersion. By the rank theorem there exists a countable collection of charts $(U_i,\Phi_i)$ that cover $M$ and corresponding charts $(V_i,\Psi_i)$ on $N$ such that 
\begin{align}
\Psi_i\circ F\circ \Phi_i^{-1}(y^1,...,y^{m-n},z^1,...,z^n)=(z^1,...,z^n).
\end{align}
Let $\sigma_i$ be a partition of unity subordinate to $U_i$.  For each $i$ and $z$ we have $\Phi_i(U_i\cap F^{-1}(z))=\left(\mathbb{R}^{m-n}\times\{\Psi_i(z)\}\right)\cap \Phi_i(U_i)$.  We can assume that the $\Phi_i(U_i)=U_i^1\times U_i^2\subset \mathbb{R}^{m-n}\times \mathbb{R}^n$ and therefore each $\Phi_i$ is a slice chart for $F^{-1}(z)$ for all $y$ such that $F^{-1}(z)\cap U_i\neq \emptyset$.  In other words, $\Phi_i(U_i\cap F^{-1}(z))= U_i^1\times \{\Psi(z)\}$.  This lets us compute the left and right hand sides of \req{coarea_formula} for $f\in L^+(M)$
\begin{align}
\int_Mf(x) \Omega^M(dx)&=\sum_i\int_{U_i}(\sigma_if)(x) \Omega^M(dx)\\
&=\sum_i\int_{\Phi_i(U_i)}(\sigma_if)\circ \Phi^{-1}(y,z) \Phi^{-1*}\Omega^M(dy,dz)\\
&=\sum_i\int_{\Phi_i(U_i)}(\sigma_if)\circ \Phi^{-1}(y,z)|g^M(y,z)| dy^{m-n}dz^n\\
&=\sum_i\int_{U_i^2}\left[\int_{U_i^1}(\sigma_if)\circ \Phi^{-1}(y,z)|g^M(y,z)| dy^{m-n}\right]dz^n\\
&\text{where }\Omega^M=g^M dy^1\wedge...\wedge dy^{m-n}\wedge dz^1\wedge...\wedge dz^n.
\end{align}
and
\small
\begin{align}
&\int_{N}\int_{F^{-1}(z)} f(y)\omega^M_z(dy) \Omega^N(dz)\\
=&\sum_i \int_{N}\left[\int_{\Phi_i(U_i\cap F^{-1}(z))} (\sigma_if)\circ\Phi_i^{-1}(y,\Psi(z))\Phi_i^{-1*}\omega^M_z(dy)\right] \Omega^N(dz)\\
=&\sum_i \int_{V_i}\left[\int_{\Phi_i(U_i\cap F^{-1}(z))} (\sigma_if)\circ\Phi_i^{-1}(y,\Psi(z))\Phi_i^{-1*}\omega^M_z(dy)\right] \Omega^N(dz)\\
=&\sum_i \int_{\Psi_i(V_i)}\left[\int_{\Phi_i(U_i\cap F^{-1}(\Psi^{-1}(z))} (\sigma_if)\circ\Phi_i^{-1}(y,z)\Phi_i^{-1*}\omega^M_z(dy)\right] \Psi^{-1*}\Omega^N(dz)\\
=&\sum_i \int_{U_i^2}\left[\int_{U_i^1\times \{z\}} (\sigma_if)\circ\Phi_i^{-1}(y,z)|g^M_z(y)| dy^{m-n}\right] |g^N(z)| dz^n\\
&\text{where }\omega^M_z=g^M_z dy^1\wedge...\wedge dy^{m-n} \text{ and }\Omega^N=g^N dz^1\wedge...\wedge dz^n \text{ for } g_1^M,g_N>0.
\end{align}
\normalsize

Therefore, if we can show $|g^M(y,z)|=|g_z^M(y)g^N(z)|$ on $U_i^1\times U_i^2$ we are done. From corollary \ref{induced_vol_eq} we have
\begin{equation}
(-1)^{n(m-n)} g^M(y,z)=\Omega^M(\partial_{z^1},...,\partial_{z^n},\partial_{y^1},...,\partial_{y^{m-n}})=\Omega^N(F_*\partial_{z^n},...,F_*\partial_{z^n})g_z^M(y).
\end{equation}
Since $\Psi\circ F\circ\Phi^{-1}=\pi_2$ we have $F_*\partial_{z^j}=\partial_{z_j}$ and so $\Omega^N(F_*\partial_{z^n},...,F_*\partial_{z^n})=g^N$ which completes the proof in the case where $F$ is a submersion.  The generalization to the case where $F_*$ is surjective a.e. follows from Sard's theorem and the fact that the set of $x\in M$ at which $F_*$ is surjective is open.

\end{proof}
\section{Comparison to Riemannian Coarea Formula}
We now recall the classical coarea formula for semi-Riemannian metrics, see \cite{chavel1995riemannian} for example,  and give its relation to theorem \ref{vol_form_coarea}.
\begin{definition}
Let $F:(M,g)\rightarrow (N,h)$ be a smooth map between semi-Riemannian manifolds.  The {\bf normal Jacobian} of $F$ is
\begin{equation}
NJF(x)=|\det(F_*|_x(F_*|_x)^T)|^{1/2}
\end{equation}
where where $(F_*|_x)^T$ denotes the adjoint map $T_xN\rightarrow T_xM$ obtained pointwise from the pullback $T^*N\rightarrow T^*M$ combined with the tangent-cotangent bundle isomorphisms defined by the metrics.
\end{definition}

\begin{lemma}
The normal Jacobian has the following properties.
\begin{itemize}
\item $(F_*|_x)^T:T_{F(x)}N\rightarrow (\ker F_*|_x)^\perp$.
\item If $F_*|_x$ is surjective then $(F_*|_x)^T$ is 1-1.
\item In coordinates
\begin{equation}
NJF(x)=\left|\det\left(h_{ik}(F(x))\frac{\partial F^k}{\partial x^l}(x)g^{lm}(x)\frac{\partial F^j}{\partial x^m}(x)\right)\right|^{1/2}.
\end{equation}
\item  If $F_*|_x$ is surjective and $g$ is nondegenerate on $ker F_*|_x$ then $F_*|_x(F_*|_x)^T$ is invertible.
\item If $c\in N$ is a regular value of $F$ and $g$ is nondegenerate on $F^{-1}(c)$ then $NJF(x)$ is non-vanishing and smooth on $F^{-1}(c)$.
\end{itemize}
\end{lemma}

Combining these lemmas with the rank theorem, one can prove the standard semi-Riemannian coarea formula
\begin{theorem}[Coarea formula]
Let $F:(M,g)\rightarrow (N,h)$ be a smooth map between semi-Riemannian manifolds such that $F_*$ is surjective at a.e. $x\in M$ and $g$ is nondegenerate on $F^{-1}(c)$ for a.e $c\in N$.  Then for $\phi\in L^1(dV_g)$ we have
\begin{equation}
\int_M\phi(x)dV_g=\int_{y\in N}\int_{x\in F^{-1}(y)}\frac{\phi(x)}{NJF(x)}dA_g dV_h
\end{equation}
where $dA_g$ is the volume measure induced on $F^{-1}(y)$ by pulling back the metric $g$.  In particular, if $N=\mathbb{R}$ with its canonical metric then $NJF=|\nabla F|$ and 
\begin{equation}
\int_M \phi dV_g=\int_\mathbb{R}\int_{F^{-1}(r)}\frac{\phi(x)}{|\nabla F(x)|} dA_g dr.
\end{equation}
\end{theorem}

The relation between the Riemannian coarea theorem and theorem \ref{vol_form_coarea} follows from the following proposition.
\begin{proposition}
Let $F:(M,g)\rightarrow (N,h)$  be a smooth map between semi-Riemannian manifolds and $c$ be a regular value.  Suppose $g$ is nondegenerate on $F^{-1}(c)$.  Let $\omega$ be the volume form on $F^{-1}(c)$ induced by $F:(M,dV_g)\rightarrow (N,dV_h)$.  Then
\begin{equation}
\omega=NJF^{-1}dA_g
\end{equation}
as densities.
\end{proposition}
\begin{proof}
By corollary \ref{induced_vol_eq}, for any  $v_i\in T_xM$ we have
\begin{equation}
i_{(v_1,...,v_n)}\Omega^M_x=dV_h(F_*v_1,...,F_*v_n)\omega_x.
\end{equation}
If we let $v_i$ be an orthonormal basis of vectors orthogonal to $F^{-1}(c)$ at $x$ then $F_*v_i$ are linearly independent and so
\begin{align}
\omega=&(dV_h(F_*v_1,...,F_*v_n))^{-1}i_{(v_1,...,v_n)}dV_g\\
=&(dV_h(F_*v_1,...,F_*v_n))^{-1}dA_g.
\end{align}
Choose coordinates about $x$ and $F(x)$ so that $\partial_{x^i}=v_i$ for $i=1...n$, $\{\partial_{x^i}\}_{n+1}^m$ span $\ker F_*$, and $\partial_{y_i}$ are orthonormal.  Then 
\begin{align}
dV_h(F_*v_1,...,F_*v_n)&=\sqrt{|\det(h)|}\frac{\partial F^{j_1}}{\partial x^1}...\frac{\partial F^{j_n}}{\partial x^n}dy^1\wedge...\wedge dy^n(\partial_{y^{j_1}},...,\partial_{y^{j_n}})\\
&=\det\left(\frac{\partial F^{j}}{\partial x^i}\right)_{i,j=1}^n.
\end{align}
$F_*\partial_{x^i}=0$ for $i=n+1...m$ and so $\frac{\partial F^j}{\partial x_i}=0$ for $i=n+1...m$.  Letting $\eta=\diag(\pm 1)$ be the signature of $g$, we find
\begin{align}
NJF(x)=&\left|\det\left(h_{ik}(F(x))\frac{\partial F^k}{\partial x^l}(x)g^{lm}(x)\frac{\partial F^j}{\partial x^m}(x)\right)\right|^{1/2}\\
=&\left|\det\left(\sum_{l,m=1}^n\frac{\partial F^k}{\partial x^l}(x)\eta^{lm}(x)\frac{\partial F^j}{\partial x^m}(x)\right)\right|^{1/2}\\
=&\left|\det\left(\frac{\partial F^k}{\partial x^l}\right)_{k,l=1}^n\det(\eta^{lm})_{l,m=1}^n\det\left(\frac{\partial F^j}{\partial x^m}\right)_{j,m=1}^n\right|^{1/2}\\
=&\left|\det\left(\frac{\partial F^k}{\partial x^l}\right)_{k,l=1}^n\right|\\
=&|dV_h(F_*v_1,...,F_*v_n)|.
\end{align}
Therefore 
\begin{equation}
\omega=NJF^{-1}dA_g
\end{equation}
as densities.
\end{proof}
In particular, this shows that even though $NJF$ and $dA_g$ are undefined individually when $g$ is degenerate on $F^{-1}(c)$, one can make sense of their ratio in this situation as the induced volume form $\omega$.

\section{Delta Function Supported on a Level Set}
 The induced measure defined above allows for a coordinate independent definition of a delta function supported on a regular level set.  Such an object is of great use in performing calculations in relativistic phase space.  We give the definition and prove several properties that justify several common formal manipulations that one would like to make with such an object.
\begin{definition}
Motivated by the coarea formula, we define the composition of the {\bf Dirac delta function} supported on $c\in N$ with a smooth map $F:M\rightarrow N$ such that $c$ is a regular value of $F$ by
\begin{equation}\label{delta_def}
 \delta_c(F(x))\Omega^M \equiv \omega^M
\end{equation}
on $F^{-1}(c)$.  This is just convenient shorthand, but it commonly used in the physics literature (typically without the justification presented above or in the following results).   For $f\in L^1(\omega^M)$ we will write 
\begin{equation}
\int_M f(x)\delta_c(F(x))\Omega^M(dx)
\end{equation} 
in place of 
\begin{equation}
\int_{F^{-1}(c)} f(x) \omega^M(dx).
\end{equation}

More generally, if the subset of $F^{-1}(c)$ consisting of critical points, a closed set whose complement we call $U$, has $\dim M-\dim N$ dimensional Hausdorff measure zero in $M$ then we define
\begin{equation}
\int_M f(x)\delta_c(F(x))\Omega^M(dx)=\int_{F|_U^{-1}(c)} f(x)\omega^M.
\end{equation}
This holds, for example, if $U^c$ is contained in a submanifold of dimension less than  $\dim M-\dim N$.  

Equivalently, we can replace $U$ in this definition with any {\bf open} subset of $U$ whose complement still has $\dim M-\dim N$ dimensional Hausdorff measure zero. In this situation, we will say $c$ is a regular value except for a lower dimensional exceptional set.  Note that while Hausdorff measure depends on a choice of Riemannian metric on $M$, the measure zero subsets are the same for each choice.
\end{definition}

Using \req{vol_form_coords}, along with the coordinates described there, we can (at least locally) write the integral with respect to the delta function in the more readily usable form
\begin{equation}\label{delta_integral_coords}
\int_M f(x)\delta_c(F(x))\Omega^M=\int_{F^{-1}(c)} f(x)\frac{h^M(x)}{h^N(F(x))}\bigg|\det \left(\frac{\partial F^i}{\partial x^j}\right)^{-1}\bigg|dx^{n+1}...dx^m.
\end{equation}
The absolute value comes from the fact that we use $\delta_c(F(x))\Omega^M$ to define the orientation on $F^{-1}(c)$.

As expected, such an operation behaves well under diffeomorphisms.
\begin{lemma}\label{diffeo_property}
Let $c$ be a regular value of $F:M\rightarrow N$ and $\Phi:M^{'}\rightarrow M$ be a diffeomorphism.  Then the delta functions induced by $F:(M,\Omega^M)\rightarrow (N,\Omega^N)$ and $F\circ\Phi:(M^{'},\Phi^*\Omega^M)\rightarrow  (N,\Omega^N)$ satisfy
\begin{equation}
\delta_c(F\circ\Phi)(\Phi^*\Omega^M)=\Phi^*(\delta_c(F)\Omega^M).
\end{equation}
\end{lemma}

\begin{lemma}
Let $c$ be a regular value of $F:(M,\Omega^M)\rightarrow (N,\Omega^N)$ and $\Phi:N\rightarrow (N^{'},\Omega^{N^{'}})$ be a diffeomorphism where $\Phi^*\Omega^{N^{'}}=\Omega^N$.  Then the delta functions induced by $F:(M,\Omega^M)\rightarrow (N,\Omega^N)$ and $\Phi\circ F:(M,\Omega^M)\rightarrow (N^{'},\Omega^{N^{'}})$ satisfy
\begin{equation}
\delta_c(F)\Omega^M=\delta_{\Phi(c)}(\Phi\circ F)\Omega^M.
\end{equation}
\end{lemma}

We also have a version of Fubini's theorem.

\begin{theorem}[Fubini's Theorem for Delta functions]
Let $M_1,M_2,N$ be smooth manifolds with volume forms $\Omega_1,\Omega_2, \Omega^N$. Let $M\equiv M_1\times M_2$ and $\Omega\equiv \Omega_1\wedge\Omega_2$. Suppose that the set of $(x,y)\in F^{-1}(c)$ such that $F|_{M_1\times\{y\}}$ is not regular at $x$ has $\dim M_1+\dim M_2-\dim N$ dimensional Hausdorff measure zero in $M_1\times M_2$ (we denote the complement of this closed set by $U$).  Then for $f\in L^1(\omega)\bigcup L^+(F^{-1}(c))$ we have
\begin{equation}\label{Fubini_eq}
\int_Mf(x,y)\delta_c(F(x,y)) \Omega(dx,dy)=\int_{M_2}\left[\int_{U^y} f(x,y) \delta_c(F(x,y))\Omega_1(dx) \right]\Omega_2(dy)
\end{equation}
where $U^y=\{x\in M_1:(x,y)\in U\}$.
\end{theorem}
\begin{proof}
Our assumption about $F|_{M_1\times\{y\}}$ implies that $c$ is a regular value of $F:M_1\times M_2\rightarrow N$ except for the lower dimensional exceptional set $U^c$
and for $y\in M_2$, $c$ is also a regular value of $F|_{U^y\times\{y\}}$, hence both sides of \req{Fubini_eq} are well defined (note this is open in $M_1$).  Rewriting \req{Fubini_eq} without the delta function, we then need to show that 
\begin{equation}
\int_{F|_U^{-1}(c)} f(x,y) d\omega=\int_{M_2}\left[\int_{F|_{U^y\times\{y\}}^{-1}(c)} f(x,y) \omega^1_{c,y}(dx)\right]\Omega_2(dy).
\end{equation}
where $\omega^1_{c,y}$ is the induced volume form on $F|_{U^y\times\{y\}}^{-1}(c)$.  

Consider the projection map restricted to the $c$-level set, $\pi_2:F|_U^{-1}(c)\rightarrow M_2$.  By assumption, $F|_{M_1\times\{y\}}$ is regular at $x$ for all $(x,y)\in F|_U^{-1}(c)$. For such an $(x,y)$, take a basis $w_i\in T_yM_2$. Since $F|_{M_1\times\{y\}}$ has full rank at $x$, for each $i$ there exists $v_i\in T_xM_1$ such that $F(\cdot,y)_*v_i=F_*(0,w_i)$.  Therefore $(-v_i,w_i)\in \ker F_*|_{(x,y)}=T_{(x,y)}F|_U^{-1}(c)$.  Hence $w_i\in\pi_{2*} T_{(x,y)}F^{-1}(c)$ and so $\pi_2:F|_U^{-1}(c)\rightarrow M_2$ is regular at $(x,y)$.  

Since $\pi_2$ is regular for all $(x,y)\in F|_U^{-1}(c)$ the coarea formula applies, giving
\begin{align}
\int_{F|_U^{-1}(c)}f d\omega=&\int_{M_2}\left[\int_{\pi_2^{-1}(y)}f\tilde{\omega}_{c,y}^1\right]\Omega_2(dy)
\end{align}
for all $f\in L^1(\omega)\bigcup L^+(F^{-1}(c))$, where $\tilde{\omega}_{c,y}^1$ is the volume form on $\pi_2^{-1}(y)$ induced by $\pi_2:(F|_U^{-1}(c),\omega)\rightarrow (M_2,\Omega_2)$.

As a point set, $\pi_2^{-1}(y)=F|_{ U^y\times\{y\}}^{-1}(c)$ and both are embedded submanifolds of $M_1\times M_2$ for a.e. $y\in M_2$, hence are equal as manifolds.  So if we can show $\tilde{\omega}_{c,y}^1=\omega^1_{c,y}$ as densities whenever $F|_{M_1\times\{y\}}$ is regular at $x$ for some $(x,y)$ then we are done.  

Given any such $(x,y)$, take $v_i\in T_xM_1$ such that $\Omega^N(F(\cdot,y)_*v_i)=1$.  By definition, $\omega_{c,y}^1=i_{(v_1,...,v_n)}\Omega_1$.  We also have $(v_i,0)\in T_{(x,y)}M_1\times M_2$ and $\Omega^N(F_*(v_i,0))=1$.  Hence 
\begin{align}
\omega=&i_{((v_1,0),...,(v_n,0))}(\Omega_1\wedge\Omega_2)\\
=&(i_{((v_1,0),...,(v_n,0))}\Omega_1)\wedge\Omega_2.
\end{align}
Let $w_i\in T_y M_2$ such that $\Omega_2(w_1,...,w_{m_2})=1$.  By the same argument as above, there exists $\tilde{v}_i\in T_xM_1$ such that $(\tilde{v}_i,w_i)\in \ker F_*=T_{(x,y)}F^{-1}(c)$.  $\pi_{2*}(\tilde{v}_i,w_i)=w_i$ and $\Omega_2(w_1,...,w_{m_2})=1$ so by definition,
\begin{equation}
\tilde{\omega}_{c,y}^1=i_{((\tilde{v}_1,w_1),...,(\tilde{v}_{m_2},w_{m_2}))}\omega
\end{equation}
Since any term containing $\Omega_2$ will vanishes on $T_F(\cdot,y)^{-1}(c)\subset T M_1$, we have  
\begin{align}
\tilde{\omega}_{c,y}^1=&(-1)^{m_1-n}i_{((v_1,0),...,(v_n,0))}\Omega_1\\
=&(-1)^{m_1-n}\omega_{c,y}^1\wedge\left(i_{((\tilde{v}_1,w_1),...,(\tilde{v}_{m_2},w_{m_2}))}\Omega_2\right)\\
=&(-1)^{m_1-n}\omega_{c,y}^1.
\end{align}
As we are integrating with respect to the densities defined by $\omega_{c,y}^1$ and $\tilde{\omega}_{c,y}^1$ we are done.  
\end{proof}

Before moving on, we give a few more useful identities.

\begin{proposition}\label{delta_associative}
Let $(c_1,c_2)$ be a regular value of $F\equiv F_1\times F_2:(M,\Omega^M)\rightarrow (N_1\times N_2,\Omega^{N_1}\wedge\Omega^{N_2})$.  Then $c_2$ is a regular value of $F_2$, $c_1$ is a regular value of $F_1|_{F_2^{-1}(c_2)}$ and we have
\begin{equation}
\delta(F)\Omega^M=\delta(F_1)(\delta(F_2)\Omega^M)
\end{equation}
\end{proposition}
\begin{proof}
$(c_1,c_2)$ is a regular value of $F$, hence there exists $v_i$, $w_i$ such that $F_* v_i=(\tilde{v}_i,0)$, $F_* w_i=(0,\tilde{w}_i)$ satisfy 
\begin{equation}
\Omega^{N_1}\wedge\Omega^{N_2}( (\tilde{v}_1,0),..., (0,\tilde{w}_1),...)=1.
\end{equation}
After rescaling, we can assume
\begin{equation}
\Omega^{N_1}( \tilde{v}_1,...,\tilde{v}_{n_1})=1,\hspace{2mm} \Omega^{N_2}(\tilde{w}_1,...,\tilde{w}_{n_2})=1.
\end{equation}
Therefore $c_2$ is a regular value of $F_2$ and 
\begin{equation}
\delta(F_2)\Omega^M=i_{w_1,...,w_n}\Omega^M.
\end{equation}
The tangent space to $F_2^{-1}(c_2)$ is $\ker (F_2)_*$ which contains $v_i$.  Hence $c_1$ is a regular value of $F_1|_{F_2^{-1}(c_2)}$ and 
\begin{align}
\delta(F_1)(\delta(F_2)\Omega^M)=&i_{v_1,...,v_n}\delta(F_2)\Omega^M\\
=&\pm i_{v_1,...,v_n,w_1,...,w_n}\Omega^M
\end{align}
and so they agree as densities.
\end{proof}

\begin{proposition}
Let $c_i\in N_i$ be regular values of $F_i:M_i\rightarrow N_i$ and define $F=F_1\times F_2:M_1\times M_2\rightarrow N_1\times N_2$, $c=(c_1,c_2)$.  If $\Omega^{M_i}$ and $\Omega^{N_i}$ are volume forms on $M_i$ and $N_i$ respectively then 
\begin{equation}
\delta_c( F) \left(\Omega^{M_1}\wedge\Omega^{M_2}\right)=\left(\delta_{c_1}( F_1)\Omega^{M_1}\right)\wedge\left(\delta_{c_2}( F_2)\Omega^{M_2}\right)
\end{equation}
as densities.
\end{proposition}
\begin{proof}
Our assumptions ensure that both sides are $m_1+m_2-n_1-n_2$-forms on $F_1^{-1}(c_1)\times F_2^{-1}(c_2)$.  Choose $v_i^j\in TM_i$ that satisfy $\Omega^{N_i}(F_{i*}v^1_i,...,F_{i*}v^{n_i}_i)=1$ then
\begin{align}
&\Omega^{N_1}\wedge \Omega^{N_2}(F_*(v_1^1,0),...,F_*(v_1^{n_1},0),F_*(0,v_2^1),...,F_*(0,v_2^{n_2}))\\
=&\Omega^{N_1}\wedge \Omega^{N_2}(F_{1*}v_1^1,...,F_{2*}v_2^{n_2})\\
=&\Omega^{N_1}(v_1^1,...,v_1^{n_1})\Omega^{N_2}(v_2^1,...,v_2^{n_2})\\
=&1.
\end{align}
Therefore, by definition
\begin{align}
\delta_c\circ F \left(\Omega^{M_1}\wedge\Omega^{M_2}\right)=&i_{(v_1^1,0),...,(v_1^{n_1},0),(0,v_2^1),...,(0,v_2^{n_2})}\left(\Omega^{M_1}\wedge\Omega^{M_2}\right)\\
=&(-1)^{n_2}\left(i_{v_1^1,...,v_1^{n_1}}\Omega^{M_1}\right)\wedge\left(i_{v_2^1,...,v_2^{n_2}}\Omega^{M_2}\right)\\
=&(-1)^{n_2}\left(\delta_{c_1}\circ F_1\right)\wedge\left(\delta_{c_2}\circ F_2\right).
\end{align}
Therefore they agree as densities.
\end{proof}

\begin{proposition}\label{delta_product}
Let $F_i:M_i\rightarrow N_i$ and $g:N_1\times N_2\rightarrow K$ be smooth.  Let $\Omega^{M_i}$, $\Omega^{N_1}$, $\Omega^K$ be volume forms on $M_i$, $N_1$, $K$ respectively.  Suppose $c$ is a regular value of $F_1$ and $d$ is a regular value of $g(c,F_2)$ and of $g\circ F_1\times F_2$. Then 
\begin{equation}
\delta_c(F_1)\left[\delta_d( g\circ F_1\times F_2)\left(\Omega^{M_1}\wedge\Omega^{M_2}\right)\right]=\left(\delta_c(F_1)\Omega^{M_1}\right)\wedge\left(\delta_d(g(c, F_2))\Omega^{M_2}\right).
\end{equation}
\end{proposition}
\begin{proof}
 Let $(x,y)\in (f\circ F_1\times F_2)^{-1}(d)$ with $x\in F^{-1}(c)$. For any $w\in T_c N_1$ there exists $v\in T_x M_1$ such that $F_{1*}v=w$.  $d$ is a regular value of $g(c,F_2)$ hence there exists $\tilde{v}$ such that $g(c,F_2)_*\tilde{v}=(g\circ F_1\times F_2)_*(v,0)$.  Therefore $(g\circ F_1\times F_2)_*(v,-\tilde{v})=0$ and $F_1*(v,-\tilde{v})=w$.  This proves $c$ is a regular value of $F_1$ on $(g\circ F_1\times F_2)^{-1}(d)$.  This proves both sides are defined and are forms on $F^{-1}(c)\times g(c,F_2)^{-1}(d)$.\\

  Let $x\in F^{-1}(c)$ and $y\in  g(c,F_2)^{-1}(d)$ and choose $v_i$, $w_j$ such that
\begin{equation}
\Omega^{N_1}(F_{1*}v_1,...,F_{1*}v_{n_1})=1, \hspace{2mm} \Omega^K(g(c,F_2)_*w_1,...,g(c,F_2)_*w_k)=1.
\end{equation}
Then 
\begin{equation}
\Omega^K((g\circ F_1\times F_2)_*(0,w_1),...,(g\circ F_1\times F_2)_*(0,w_k))=1
\end{equation}
and so 
\begin{align}
\delta_d( g\circ F_1\times F_2)\left(\Omega^{M_1}\wedge\Omega^{M_2}\right)&=i_{(0,w_1),...,(0,w_k)}\left(\Omega^{M_1}\wedge\Omega^{M_2}\right)\\
&=\Omega^{M_1}\wedge\left(i_{w_1,...,w_k}\Omega^{M_2}\right)\\
&=\Omega^{M_1}\wedge\left(\delta_d(g(c,F_2))\Omega^{M_2}\right).
\end{align}
By the same argument as above, we get $\tilde{v}_i$ such that $(v_i,\tilde{v}_i)\in T_{(x,y)} (g\circ F_1\times F_2)^{-1}(d)$.  Hence
\begin{equation}
\delta_c(F_1)\left[\delta_d( g\circ F_1\times F_2)\left(\Omega^{M_1}\wedge\Omega^{M_2}\right)\right]=i_{(v_1,\tilde{v}_1),...,(v_{n_1},\tilde{v}_{n_1})}\left[\Omega^{M_1}\wedge\left(i_{w_1,...,w_k}\Omega^{M_2}\right)\right].
\end{equation}
The only non-vanishing term is
\begin{equation}
\left(i_{(v_1,\tilde{v}_1),...,(v_{n_1},\tilde{v}_{n_1})}\Omega^{M_1}\right)\wedge\left(i_{w_1,...,w_k}\Omega^{M_2}\right)=\left(i_{v_1,...,v_{n_1}}\Omega^{M_1}\right)\wedge\left(i_{w_1,...,w_k}\Omega^{M_2}\right)
\end{equation}
since the other terms all contain a $m_1-n_1+l$ form on the $m_1-n_1$-dimensional manifold $F^{-1}(c)$ for some $l>0$.  This proves the result.
\end{proof}

Sometimes it is convenient to use the delta function to introduce ``dummy integration variables", by which we mean utilizing the following simple corollary of the coarea formula.
\begin{corollary}\label{dummy_int}
Let $\Omega^M$ be a volume form on $M$, $F:M\rightarrow (N,\Omega^N)$ be smooth, and $f:N\times M\rightarrow \mathbb{R}$  such that $f(F(\cdot),\cdot)\in L^1(\Omega^M)\bigcup L^+(M)$.  If $F_*$ is surjective at a.e. $x\in M$ then
\begin{equation}
\int_M f(F(x),x)\Omega^M(dx)=\int_N\int_{F^{-1}(z)} f(z,x)\delta_z(F)\Omega^M(dx) \Omega^N(dz).
\end{equation}
\end{corollary}

\section{Applications}
\subsection{Relativistic Volume Element}\label{rel_vol_form}
We now discuss an application of the above results to the single particle phase space volume element. We first define it in the massive case, where the semi-Riemannian method of defining volume forms is applicable.  In such and approach, the massless case is often handled via a limiting argument \cite{tsamparlis}.  We then show how our method is able to handle both the massive and massless case in a unified and simpler manner.

 Given a time oriented $n+1$ dimensional semi-Riemannian manifold $(M,g)$, there is a natural induced metric $\tilde{g}$ on the tangent bundle, called the diagonal lift.  At a given point $(x,p)\in TM$ its coordinate independent definition is
\begin{align}
\tilde{g}_{(x,p)}(v,w)=g_x(\pi_{*} v,\pi_{*} w)+g_x(D_t \gamma_v, D_t \gamma_w)
\end{align}
where $\gamma_v$ is any curve in $TM$ with tangent $v$ at $x$, $\pi:TM\longrightarrow M$ is the projection, and $D_t\gamma_v$ is the covariant derivative of $\gamma_v$, treated as a vector field along the curve $\pi\circ\gamma_v$, and similarly for $\gamma_w$, see e.g. \cite{pettini}. The result can be shown to be independent of the choice of curves.  In a coordinate system on $M$ where the the first coordinate is future timelike and the Christoffel symbols are $\Gamma^\beta_{\sigma\eta}$, consider the  induced coordinates $(x^\alpha,p^\alpha), \hspace{2mm}\alpha=0,...,n$ on $TM$.  In these coordinates we have 
\begin{equation}
\tilde{g}_{(x^\alpha,p^\alpha)}=g_{\beta,\delta}(x^\alpha)dx^\beta\otimes dx^\delta +g_{\beta,\delta}(x^{\alpha})\epsilon^\beta\otimes \epsilon^\delta, \hspace{2mm} \epsilon^\beta=dp^\beta+p^\sigma\Gamma^\beta_{\sigma\eta}(x^{\alpha})dx^\eta.
\end{equation}
The vertical and horizontal subspaces are spanned by
\begin{equation}\label{horizontal_subspace}
V_\alpha=\partial_{p^\alpha}, \hspace{2mm} H_\alpha=\partial_{x^\alpha}-p^\sigma\Gamma_{\sigma\alpha}^\beta\partial_{p^\beta}
\end{equation}
respectively.  The horizontal vector fields satisfy
\begin{equation}
\tilde{g}(H_\alpha,H_\beta)=g_{\alpha\beta}.
\end{equation}
For any manifold (oriented or not), the tangent bundle has a canonical orientation.  With this orientation, the volume form on $TM$ induced by $\tilde{g}$ is
\begin{equation}
\widetilde{dV}_{(x^\alpha,p^{\alpha})}=|g(x^\alpha)|dx^0\wedge...\wedge dx^n\wedge dp^0\wedge...\wedge dp^n.
\end{equation}
where $|g(x^\alpha)|$ denotes the absolute value of the determinant of the component matrix of $g$ in these coordinates.

Of primary interest in kinetic theory for a particle of mass $m\geq 0$ is the mass shell bundle
\begin{equation}
P_m=\{p\in TM :g(p,p)=m^2,\hspace{1mm}  p\text{ future directed}\}
\end{equation}
and it will be necessary to have a volume form on $P_m$.  $P_m$ is a connected component of the zero set of the of the smooth map 
\begin{equation}\label{defining_function}
h:TM\setminus \{0_x:x\in M\}\longrightarrow \mathbb{R},\hspace{2mm} h(x,p)= \frac{1}{2}(g_x(p,p)-m^2).  
\end{equation} 
We remove the image of the zero section to avoid problems when $m=0$.  Its differential is
\begin{equation}\label{dh}
dh=\frac{1}{2}\frac{\partial g_{\sigma\delta}}{\partial x^\alpha}p^\sigma p^\delta dx^\alpha+g_{\sigma\delta}p^\sigma dp^\delta=g_{\sigma\delta}p^\sigma\epsilon^\delta.
\end{equation}
$g$ is nondegenerate, so for $p=p^{\alpha}\partial_{x^\alpha}\in TM_x\setminus{\{0_x\}}$ there is some $v=v^\alpha\partial{x^\alpha}\in TM_x$ with $g(v,p)\neq 0$.  Therefore
\begin{equation}
dh_{(x,p)}(v^\alpha\partial_{p^\alpha})=g(v,p)\neq 0.
\end{equation}
This proves $P_m$ is a regular level set of $h$, and hence is a closed embedded hypersurface of $TM\setminus \{0_x:x\in M\}$.  For $m\neq 0$ it is also closed in $TM$, but for $m=0$ every zero vector is a limit point of $P_m$.\\

\noindent{\bf Massive Case:}\\
For $m\neq 0$, we will show that $P_m$ is a semi-Riemannian hypersurface in $TM$ and hence inherits a volume form from $TM$. This is the standard method of inducing a volume form, as presented in \cite{tsamparlis}.  

The normal to $P_m$ is 
\begin{equation}
\grad h=\tilde{g}^{-1}(dh)=p^\alpha\partial_{p^\alpha}
\end{equation}
which has norm squared 
\begin{equation}
\tilde{g}(\grad h,\grad h)=g(p,p)=m^2.
\end{equation}
Therefore, for $m\neq 0$, $P_m$ has a unit normal $N=\grad h/m$ and so it is a semi-Riemannian hypersurface with volume form
\begin{equation}
\widetilde{dV}_m=i_N \widetilde{dV}=\frac{|g|}{m}dx^0\wedge...\wedge dx^n\wedge\left(\sum_\alpha (-1)^\alpha p^\alpha dp^0\wedge...\wedge\widehat{dp^\alpha}\wedge...\wedge dp^n\right)
\end{equation}
where $i_N$ denotes the interior product (or contraction) and a hat denotes an omitted term.  We are also interested in the volume form on $P_{m,x}$ the fiber of $P_m$ over a point $x\in M$.  We obtain this by contracting $\widetilde{dV}$ with an orthonormal basis of vector fields normal to $P_{m,x}$. Such a basis is composed of $N$ together with an orthonormalization of the basis of horizontal fields, $W_\alpha=\Lambda^\beta_\alpha H_\beta$, where $H_\beta$ are defined in \req{horizontal_subspace}. Therefore we have
\begin{equation}\label{contract_horiz}
\widetilde{dV}_{m,x}=i_{W_0}...i_{W_n}\widetilde{dV}_m.
\end{equation}
 We can simplify these expressions by defining a coordinate system on the momentum bundle, writing $p^0$ as a function of the $p^i$.  The details, which are standard, are carried  out in Appendix  \ref{coord_comp}.  The results are
\begin{equation}
\widetilde{dV}_m=\frac{m|g|}{p_0}dx^0\wedge...\wedge dx^n\wedge dp^1\wedge...\wedge dp^n,\\
\end{equation}
\begin{equation}
\widetilde{dV}_{m,x}=\frac{m|g|^{1/2}}{p_0}dp^1\wedge...\wedge dp^n.
\end{equation}
We define $\pi$ and $\pi_x$ by
\begin{equation}
\pi=\frac{1}{m}\widetilde{dV}_m=\frac{|g|}{p_0}dx^0\wedge...\wedge dx^n\wedge dp^1\wedge...\wedge dp^n,\\
\end{equation}
\begin{equation}\label{pi_x}
\pi_x=\frac{1}{m}\widetilde{dV}_{m,x}=\frac{|g|^{1/2}}{p_0}dp^1\wedge...\wedge dp^n.
\end{equation}
We will typically omit the subscript $x$ and let the context distinguish whether we are integrating over the full momentum bundle (i.e. both over spacetime and momentum variables) or just momentum space at a single point in spacetime.  \\

\noindent{\bf Massless Case:}\\
When $m=0$ the above construction fails.  However, we can use proposition \ref{induced_vol_form} to induce a volume form using the map \req{defining_function} defined above. Here we carry out the construction for the induced volume form on $P_{m,x}$ for any $m\geq 0$. The volume form on each tangent space $T_xM$ is
\begin{equation}
\tilde{dV}_x=|g(x)|^{1/2}dp^0\wedge...\wedge dp^n.
\end{equation}
We assume that the coordinates are chosen so that the vector field $\partial_{p^0}$ is timelike. By \req{dh} we find
\begin{equation}
dh(\partial_{p^0})=g_{\alpha 0}p^\alpha\neq 0
\end{equation}
on $P_{m,x}$.  Therefore, by corollary \ref{induced_vol_eq} the induced volume form is
\begin{align}\label{mass_shell_vol}
\omega=&\frac{1}{dh(\partial_{p^0})} i_{\partial_{p^0}} \tilde{dV}_x\\
=&\frac{|g|^{1/2}}{p_0}dp^1\wedge...\wedge dp^n.
\end{align}
We can also pull this back under the coordinate chart on $P_{m,x}$ defined in Appendix \ref{coord_comp} and obtain the same expression in coordinates. This result agrees with our prior definition of \req{pi_x} in the case where $m>0$ but is also able to handle the massless case in a uniform manner, without resorting to a limiting argument as $m\rightarrow 0$.

We also point out another convention in common use where $h$ is replaced by $2h$.  This leads to an additional factor of $1/2$ in the volume element, distinguishing this definition from the one based on semi-Riemannian geometry.  However, the convention
\begin{equation}
\omega=\frac{|g|^{1/2}}{2p_0}dp^1\wedge...\wedge dp^n
\end{equation}
 is in common use and will be employed in the scattering integral computations in chapter \ref{ch:coll_simp}.

\subsection{Relativistic Phase Space}
Here we justify several manipulations that are useful for working with relativistic phase space integrals.

\begin{lemma}\label{parallel_lemma}
Let $V$ be an $n$-dimensional vector space.  The subset of $\prod_1^N V\setminus\{0\}$ consisting of $N$-tuples of parallel vectors is an $n+N-1$ dimensional closed submanifold of $\prod_1^N V\setminus\{0\}$.
\end{lemma}
\begin{proof}
The map $V\times \mathbb{R}^{N-1}\rightarrow  \prod_1^N V\setminus\{0\}$ given by
\begin{equation}
F(p,a^2,...,a^N)=(p,a^2p,...,a^{N}p)
\end{equation}
is an injective immersion and maps onto the desired set.
\end{proof}
For reactions converting $k$ particles to $l$ particles, the relevant phase space is $3(k+l)-4$ dimensional and so for $k+l\geq 4$ (in particular for $2$-$2$ reactions), the set of parallel $4$-momenta is lower dimensional and can be ignored. This will be useful as we proceed.

\begin{lemma}
Let $N\geq 4$. Then
\begin{equation}
\prod_i \delta(p_i^2-m_i^2)d^4p_i=\left(\prod_i \delta(p_i^2-m_i^2)\right)\prod_i d^4p_i
\end{equation}
 and 
\begin{equation}
\delta(\Delta p)\left[\left(\prod_i \delta(p_i^2-m_i^2)\right)\prod_i d^4p_i\right]=\left(\delta(\Delta p)\prod_i \delta(p_i^2-m_i^2)\right)\prod_i d^4p_i
\end{equation}
where each $d^4p_i$ is the standard volume form on future directed vectors, $\{p:p^2\geq 0, p^0>0\}$, we give $\mathbb{R}$ its standard volume form, and $\Delta p=a^ip_i$, $a^i=1$, $i=1,...,l$, $a^i=-1$, $i=l,...,N$. 
\end{lemma}
\begin{proof}
Let $F_1(p_i)=(p_1^2,...,p_N^2)$ and $F_2(p_i)=(\Delta p,F_1(p_i))$.  We need to show that $(m_1^2,...,m_N^2)$ is a regular value of $F_1$ and $(0,m_1^2,...,m_k^2)$ is a regular value of $F_2$.  The result then follows from proposition \ref{delta_associative}.

It holds for $F_1$ since each $p_i\neq 0$. For $F_2$, the differential is
\begin{equation}
(F_2)*=\left( \begin{array}{cccc}
a^{1}I&a^{2}I&...& a^{N}I \\
2 \eta_{ij}p^j_1&0&...&0\\
\vdots&&&\vdots\\
0&...&0&2 \eta_{ij}p^j_N\\
\end{array} \right)
\end{equation}
where $I$ is the $4$-by-$4$ identity.  The fact that $(F_1)_*$ is onto means that we need only show $(F_2)_*$ maps onto $\mathbb{R}^4\times(0,...,0)$.  

By lemma \ref{parallel_lemma} we assume there exists $i,j$ such that $p_i,p_j$ are not parallel. We are done if for each standard basis vector $e_k\in\mathbb{R}^4$ there exists $q\in\mathbb{R}^4$ such that
\begin{equation}
p_i\cdot q=\frac{1}{a^j}p_i\cdot e_k,\hspace{2mm} p_j\cdot q=0.
\end{equation}
If $p_j$ is null then there is a $c$ such that $q=c p_j$ satisfies these conditions. If $p_j$ is non-null then complete it to an orthonormal basis.  $p_i$ must have a component along the orthogonal complement of $p_j$ and we can take $q$ to be proportional to that component.

\end{proof}

\begin{subappendices}

\section{Volume Form in Coordinates}\label{coord_comp}
Here we derive a useful formula for the volume form on the momentum bundle in a simple coordinate system, continuing from section \ref{rel_vol_form}.  We begin in a coordinate system $x^\alpha$ on $U\subset M$ and the induced coordinates $p^\alpha$ on $TM$ where our only assumption is that the $0$'th coordinate direction is future timelike, and so $g_{00}>0$.  For any $v^i\in \mathbb{R}^n$, let $v^0=-g_{0i}v^i/g_{00}$.  $v^\alpha$ is orthogonal to the $0$'th coordinate direction, and therefore spacelike. Hence 
\begin{equation}
0\geq g_{\alpha \beta}v^\alpha v^\beta=-(g_{0i}v^i)^2/g_{00}+g_{ij}v^iv^j.
\end{equation}
and is zero iff $v^\alpha=0$. Therefore, the following map is well defined
\begin{align}
(x^\alpha,p^j)&\longrightarrow (x^\alpha,p^0(x^\alpha,p^j),p^1,...,p^n),  \hspace{2mm} \alpha=0...n, \hspace{1mm} j=1...n \notag\\
 p^0&=-g_{0j}p^j/g_{00}+\left((g_{0j}p^j/g_{00})^2+(m^2-g_{ij}p^ip^j)/g_{00}\right)^{1/2}.
\end{align}
and smooth on $\mathbb{R}^{n+1}\times\mathbb{R}^n$ if $m\neq 0$, and on $\mathbb{R}^{n+1}\times\left(\mathbb{R}^n\setminus{0}\right)$ if $m=0$.  We also have $g_{00}p^0+g_{0j}p^j>0$ under either of these cases, and so the resulting element of $TM$ is future directed and has squared norm $m^2$, so it maps into $P_m$.  It is a bijection and has full rank, hence it is a coordinate system on $P_m$.  In these coordinates, the volume form is
\begin{align}
\widetilde{dV}_m=&\frac{|g|}{m}dx^0\wedge...\wedge dx^n\wedge\left(p^0dp^1\wedge...\wedge dp^n+\sum_j (-1)^j p^j dp^0\wedge...\wedge\widehat{dp^j}\wedge...\wedge dp^n\right),\\
dp^0=&\partial_{x^\alpha} p^0dx^\alpha+\partial_{p^j}(p^0) dp^j.\notag
\end{align}
The terms in $dp^0$ involving $dx^\alpha$ drop out once they are wedged with $dx^0\wedge...\wedge dx^n$, hence
\small
\begin{align}
\widetilde{dV}_m=&\frac{|g|}{m}dx^0\wedge...\wedge dx^n\wedge\left(p^0dp^1\wedge...\wedge dp^n+\sum_{i,j} (-1)^j p^j \partial_{p^i}p^0 dp^i\wedge...\wedge\widehat{dp^j}\wedge...\wedge dp^n\right)\notag\\
=&\frac{|g|}{m}\left(p^0-\sum_{j}p^j \partial_{p^j}(p^0) \right)dx^0\wedge...\wedge dx^n\wedge dp^1\wedge...\wedge dp^n\notag
\end{align}
\normalsize

\begin{align}
p^0-p^j\partial_{p^j}(p^0)=& p^0+g_{0j}p^j/g_{00}-\frac{(g_{0j}p^j/g_{00})^2-g_{ij}p^ip^j/g_{00}}{\left((g_{0j}p^j/g_{00})^2+(m^2-g_{ij}p^ip^j)/g_{00}\right)^{1/2}}\notag\\
=&\frac{1}{p_0}\left(\frac{1}{g_{00}}(g_{00}p^0+g_{0,j}p^j)^2-(g_{0j}p^j)^2/g_{00}+g_{ij}p^ip^j\right)=\frac{m^2}{p_0}.
\end{align}
Therefore
\begin{equation}
\widetilde{dV}_m=\frac{m|g|}{p_0}dx^0\wedge...\wedge dx^n\wedge dp^1\wedge...\wedge dp^n.
\end{equation}

To compute the volume form on $P_{m,x}$, recall  that 
\begin{equation}\label{contract_horiz}
\widetilde{dV}_{m,x}=i_{W_0}...i_{W_n}\widetilde{dV}_m.
\end{equation}
Where $W_i$ is an orthonormalization of the basis of horizontal fields, $W_\alpha=\Lambda^\beta_\alpha H_\beta$, where $H_\beta$ are defined in \req{horizontal_subspace}. All of the contractions in \req{contract_horiz} that involve the $dp^\alpha$'s will be zero when restricted to $P_{m,x}$ since the $dx^\alpha$ are zero there. Hence we obtain
\begin{align}\label{dV_x}
\widetilde{dV}_{m,x}=&\frac{|g|}{m}\left(p^0-\sum_{j}p^j \partial_{p^j}(p^0) \right)dx^0\wedge...\wedge dx^n\left(W_0,...,W_n)\right) dp^1\wedge...\wedge dp^n\notag\\
=&\frac{|g|\det(\Lambda)}{m}\left(p^0-\sum_{j}p^j \partial_{p^j}(p^0) \right)dx^0\wedge...\wedge dx^n\left(H_0,...,H_n)\right) dp^1\wedge...\wedge dp^n\notag\\
=&\frac{|g|^{1/2}}{m}\left(p^0-\sum_{j}p^j \partial_{p^j}(p^0) \right) dp^1\wedge...\wedge dp^n
\end{align}
since $\det(\Lambda^\sigma_\alpha g_{\sigma\delta}\Lambda^\delta_\beta)=1$.
 In the coordinate system on $P_{m,x}$
\begin{align}
(p^j)&\longrightarrow (p^0(x^\alpha,p^j),p^1,...,p^n),\notag\\
 p^0&=-g_{0j}(x)p^j/g_{00}(x)+\left((g_{0j}(x)p^j/g_{00}(x))^2+(m^2-g_{ij}(x)p^ip^j)/g_{00}(x)\right)^{1/2}.
\end{align}
 the same calculation as above gives the formula
\begin{equation}
\widetilde{dV}_{m,x}=\frac{m|g|^{1/2}}{p_0}dp^1\wedge...\wedge dp^n.
\end{equation}
\end{subappendices}

\chapter{Boltzmann Equation Solver Adapted to Emergent Chemical Non-equilibrium}\label{ch:boltz_orthopoly}

Having completed our brief geometrical interlude, we now return to our study of neutrino freeze-out. The analysis in chapter \ref{ch:model_ind} was based on exact chemical and kinetic equilibrium and sharp freeze-out transitions at $T_{ch}$ and $T_k$, but these are  only approximations.  The  Boltzmann equation is a more precise model of the dynamics of the freeze-out process and furthermore, given the collision dynamics it is capable of capturing in a {\em quantitative manner} the non-thermal distortions from equilibrium, for example the emergence of actual distributions and the approximate values  of $T_{ch}$, $T_k$, and $\Upsilon$.  Indeed,  in  such a dynamical description no hypothesis about the presence of kinetic or chemical (non) equilibrium needs to be made, as the distribution close to \req{k_eq} with   $\Upsilon\ne  1$ emerges naturally as the outcome of collision processes, even when the particle system approaches the freeze-out temperature domain  in chemical equilibrium.

Considering the natural way in which chemical non-equilibrium emerges from chemical equilibrium during freeze-out, it is striking that the literature on Boltzmann solvers does not reflect on the accommodation of emergent chemical non-equilibrium into the method of solution. For an all-numerical solver this may not be a necessary step as long as there are no constraints that preclude development of a general non-equilibrium solution. However, when strong chemical non-equilibrium is present either in the intermediate time period or/and at the end of the evolution, a brute force approach could be very costly in computer time. Motivated by this circumstance and past work with physical environments in which chemical non-equilibrium arose,  we introduce here a  spectral method for solving the Boltzmann equation that utilizes a dynamical basis of orthogonal polynomials which is adapted to the case of emerging chemical non-equilibrium. We validate our method via a  model problem  that captures the essential physical characteristics of interest and use it to highlight the type of situation where this new method exhibits its advantages.

In the cosmological neutrino freeze-out context, the general relativistic Boltzmann equation has been used to study neutrino freeze-out in the early universe and has been successfully solved using both discretization in momentum space \cite{Madsen,Dolgov_Hansen,Gnedin,Mangano2005} and a spectral method based on a fixed basis of orthogonal polynomials \cite{Esposito2000,Mangano2002}.    In Refs.\cite{Wilkening,Wilkening2} the non-relativistic Boltzmann equation was solved via a spectral method similar in  one important mathematical idea to the approach we present here.  For near equilibrium solutions, the spectral methods have the advantage of requiring a relatively small number of modes to obtain an accurate solution, as opposed to momentum space discretization which in general leads to a large highly coupled nonlinear system of odes irrespective of the near equilibrium nature of the system.  

The efficacy of the spectral method used in \cite{Esposito2000,Mangano2002} can largely be attributed to the fact that, under the conditions considered there, the true solution is very close to a chemical equilibrium distribution, \req{ch_eq}, where the temperature is controlled by the dilution of the system. However, as we have discussed, the Planck CMB results \cite{Planck} indicate the possibility that neutrinos participated in reheating to a greater degree than previously believed, leading to a more pronounced chemical non-equilibrium and reheating. Efficiently obtaining this emergent chemical non-equilibrium within realm of kinetic theory motivates the development of a new numerical method that adapts to this new circumstance.

In section \ref{boltzmann_basics} we give a basic overview of the relativistic Boltzmann equation in an FRW Universe.  In section \ref{the_method} we discuss our modified spectral method in detail.  In subsection \ref{free_stream_approach} we recall the orthogonal polynomial basis used in \cite{Esposito2000,Mangano2002} and in subsection \ref{kinetic_eq_approach} we introduce our modified basis and characterize precisely the differences in the method  we propose. We compare these two bases in subsection \ref{basis_comparison}. In subsection \ref{dynamics_sec} we use the Boltzmann equation to derive the dynamics of the mode coefficients and identify physically motivated evolution equations for the effective temperature and fugacity.  In section \ref{validation} we validate the method using a model problem.  In  \ref{orthopoly_app} we give further details on the construction of the parametrized family orthogonal polynomials we use to solve the Boltzmann equation.  The work presented in this chapter can be found in our paper \cite{Birrell_orthopoly}.

\section{Relativistic Boltzmann Equation }\label{boltzmann_basics}

Recall the general relativistic Boltzmann equation introduced in chapter \ref{ch:model_ind}
\begin{equation}\label{boltzmann}
p^\alpha\partial_{x^\alpha}f-\Gamma^j_{\mu\nu}p^\mu p^\nu\partial_{p^j}f=C[f].
\end{equation}
As discussed above, the left hand side expresses the fact that particles undergo geodesic motion in between point collisions. The term $C[f]$ on the right hand side of the Boltzmann equation is called the collision operator and models the short range scattering processes that cause deviations from geodesic motion. For $2\leftrightarrow 2$ reactions between fermions, such as neutrinos and $e^\pm$, the collision operator takes the form
\begin{align}\label{coll}
C[f_1]=&\frac{1}{2}\int F(p_1,p_2,p_3,p_4) S |\mathcal{M}|^2(2\pi)^4\delta(\Delta p)\prod_{i=2}^4\delta_0(p_i^2-m_i^2)\frac{d^4p_i}{(2\pi)^3},\\
F=&f_3(p_3)f_4(p_4)f^1(p_1)f^2(p_2)-f_1(p_1)f_2(p_2)f^3(p_3)f^4(p_4),\notag\\
f^i=&1- f_i.\notag
\end{align}
Here $|\mathcal{M}|^2$ is the process amplitude or matrix element, $S$ is a numerical factor that incorporates symmetries and prevents over-counting, $f^i$ are the fermi blocking factors, $\delta(\Delta p)$ enforces four-momentum conservation in the reactions, and the $\delta_0(p_i^2-m_i^2)$ restrict the four momenta to the future timelike mass shells.

The matrix element for a $2-2$ reaction is some function of the Mandelstam variables $s, t, u$, of which only two are independent, defined by 
\begin{align}\label{Mandelstam}
&s=(p_1+p_2)^2=(p_3+p_4)^2,\\
&t=(p_3-p_1)^2=(p_2-p_4)^2,\\
&u=(p_3-p_2)^2=(p_1-p_4)^2.\\
&s+t+u=\sum_i m_i^2.
\end{align}
We will return to a study of $2$-$2$ scattering kernels for neutrino processes in chapter \ref{ch:coll_simp}.  When testing our method in this chapter, we will use a simplified model to avoid any application specific details.

We now restrict our attention to  systems of fermions under the assumption of homogeneity and isotropy. We assume that the particle are effectively massless,  i.e. the temperature is much greater than the mass scale.  Homogeneity and isotropy imply that the distribution function of each particle species under consideration has the form $f=f(t,p)$ where $p$ is the magnitude of the spacial component of the four momentum.  In a flat FRW universe the Boltzmann equation reduces to
\begin{equation}\label{boltzmann_p}
\partial_t f-pH \partial_p f=\frac{1}{E}C[f],\hspace{2mm} H=\frac{\dot{a}}{a}.
\end{equation}

The Boltzmann equation \req{boltzmann_p} can be simplified by the method of characteristics. Writing $f(p, t)=g(a(t)p,t)$ and reverting back to call the new distribution $g\to f$, the 2nd term in \req{boltzmann_p} cancels out and the evolution in time can be studied directly.  Using the formulas for the moments of $f$ \req{moments}, this transformation implies for the rate of change in the   number density and energy density  
\begin{align}\label{n_div}
\frac{1}{a^3}\frac{d}{dt}(a^3n_1)=&\frac{g_p}{(2\pi)^3}\int C[f_1] \frac{d^3p}{E}.\\
\label{rho_div}
\frac{1}{a^4}\frac{d}{dt}(a^4\rho_1)=&\frac{g_p}{(2\pi)^3}\int C[f_1] d^3p .
\end{align} 
For free-streaming particles the vanishing of the collision operator implies conservation of `comoving' particle number of species 1. From the associated powers of $a$ in \req{n_div} and \req{rho_div} we see that the energy per free streaming particle as measured by an observer scales as $1/a$, a manefestation or redshift.

\section{Spectral Methods}\label{the_method}
\subsection{Polynomials for systems close to kinetic and chemical equilibrium}\label{free_stream_approach}
Here we outline the approach for solving \req{a_vars} used in \cite{Esposito2000,Mangano2002} in order to contrast it with our approach as presented in subsection \ref{kinetic_eq_approach}.  As just discussed, the Boltzmann equation  is a linear first order partial differential equation and can be reduced using a new variable $y=a(t)p$  via the method of characteristics and exactly solved in the collision free ($C[f]=0)$ limit.   This motivates a change of variables from $p$ to $y$ which eliminates the momentum derivative, leaving the simplified equation
\begin{equation}\label{a_vars}
\partial_tf=\frac{1}{E} C[f].
\end{equation}

We let $\hat\chi_i$ be the orthonormal polynomial basis on the interval $[0,\infty)$ with respect to the weight function
\begin{equation}\label{free_stream_weight}
f_{ch}=\frac{1}{e^y+1},
\end{equation}
constructed as in  \ref{orthopoly_app}. $f_{ch}$ is the Fermi-Dirac chemical equilibrium distribution for massless fermions and temperature $T=1/a$.  Therefore this ansatz is well suited to distributions that are manifestly in chemical equilibrium ($\Upsilon=1$) or remain close and with $T\propto 1/a$, which we call dilution temperature scaling.  Assuming that $f$ is such a distribution  motivates the decomposition
\begin{equation}\label{free_stream_ansatz}
f=f_{ch}\chi,\qquad \chi=\sum_i d^i\hat\chi_i.
\end{equation}

Using this ansatz  equation \req{a_vars} becomes
\begin{equation}
\dot{d}^k=\int_0^\infty\frac{1}{E}\hat{\chi}_k C[f]dy.
\end{equation}
Because of \req{free_stream_ansatz}, we call this the chemical equilibrium method.

We also have the following expressions for the particle number density and energy density
\begin{align}\label{free_stream_moments}
n&=\frac{g_p}{2\pi^2 a^3}\sum_0^2 d^i\int_0^\infty f_{ch}\hat\chi_i y^2dy,\\
\rho&=\frac{g_p}{2\pi^2a^4}\sum_0^3 d^i\int_0^\infty f_{ch}\hat\chi_i y^3dy.
\end{align}

Note that the sums truncate at $3$ and $4$ terms respectively, due to the fact that $\hat\chi_k$ is orthogonal to all polynomials of degree less than $k$. This implies that in general, at least four modes are required to capture both the particle number and energy flow. More modes are needed if the non-thermal distortions are large and the back reaction of higher modes on lower modes is significant.

\subsection{Polynomials for systems   not close to chemical equilibrium}\label{kinetic_eq_approach}
Our primary interest is in solving \req{T_vars} for systems close to the kinetic equilibrium distribution \req{k_eq} but not necessarily in chemical equilibrium, a task for which the method in the previous section is not well suited in general. For a general kinetic equilibrium distribution, the temperature does not necessarily scale as $T\propto 1/a$ i.e. the temperature is not controlled solely by dilution.  For this reason, we will find it more useful to make the change of variables $z=p/T(t)$ rather than the scaling used in \req{a_vars}.  Here $T(t)$ is to be viewed as the time dependent effective temperature of the distribution $f$, a notion we will make precise later.  With this change of variables, the Boltzmann equation becomes
\begin{equation}\label{T_boltzmann}
\partial_t f-z\left(H+\frac{\dot T}{T}\right)\partial_z f=\frac{1}{E}C[f].
\end{equation}

 To model a distribution close to kinetic equilibrium at temperature $T$ and fugacity $\Upsilon$, we assume
\begin{equation}\label{kinetic_approx}
f(t,z)= f_\Upsilon (t,z)\psi(t,z),\hspace{2mm} f_\Upsilon(z)=\frac{1}{\Upsilon^{-1}e^z+1}
\end{equation}
where the kinetic equilibrium distribution $f_\Upsilon $ depends on $t$ because we are assuming $\Upsilon$ is time dependent (with dynamics to be specified later). 

We will solve \req{T_boltzmann} by expanding $\psi$ in the basis of orthogonal polynomials generated by the parametrized weight function
\begin{equation}\label{weight}
w(z)\equiv w_\Upsilon(z)\equiv z^2f_\Upsilon (z)=\frac{z^2}{\Upsilon^{-1} e^z+1}
\end{equation}
on the interval $[0,\infty)$. See \ref{orthopoly_app} for details on the construction of these polynomials and their dependence on the parameter $\Upsilon$. This choice of weight is physically motivated by the fact that we are interested in solutions that describe massless particles not too far from kinetic equilibrium, but (potentially) far from chemical equilibrium. We call this the chemical non-equilibrium method.

We emphasize that we have made three important changes as compared to  the chemical equilibrium method:
\begin{enumerate}
\item  We allow a general time dependence of the effective temperature parameter $T$ i.e. we do not assume dilution temperature scaling $T=1/a$.
\item We have replaced the chemical  equilibrium distribution in the weight \req{free_stream_weight}  with a chemical non-equilibrium distribution  $f_\Upsilon $ i.e. we introduced $\Upsilon$.
\item We have introduced an additional factor of $z^2$ to the functional form of the weight as proposed in a different context in Refs.\cite{Wilkening,Wilkening2}. 
\end{enumerate} 
We note that the authors of \cite{Esposito2000} did consider the case of fixed chemical potential imposed as an initial condition. This is not the same as an emergent chemical non-equilibrium, i.e. time dependent $\Upsilon$, that we study here, nor do they consider a $z^2$ factor in the weight. We borrowed the idea for the $z^2$ prefactor from   Ref.\cite{Wilkening2}, where it was found that including a $z^2$ factor along with the non-relativistic chemical equilibrium distribution in the weight improved the accuracy of their method. Fortuitously,  this will also allow us to capture the particle number and energy flow with fewer terms than required by the chemical equilibrium method. A suitably modified weight and method allows us to maintain  these advantages when a particle mass scale becomes relevant. However, there are some additional qualifications and subtleties that arise in such a program and so we return to this problem in a subsequent work.

\subsection{Comparison of Bases}\label{basis_comparison}

Before deriving the dynamical equations for the method outlined in section \ref{kinetic_eq_approach}, we illustrate the error inherent in approximating the chemical non-equilibrium distribution \req{k_eq}  with a  chemical equilibrium distribution \req{ch_eq} whose temperature is $T=1/a$.   Given a chemical non-equilibrium distribution 
\begin{equation}\label{zeroth_approx}
f_\Upsilon (y)=\frac{1}{\Upsilon^{-1}e^{y/(aT)}+1},
\end{equation}
 we can attempt to write it as a perturbation of the chemical equilibrium distribution,  
\begin{equation}\label{chi_def}
f_\Upsilon=f_{ch}\chi
\end{equation} as we would need to when using the method of section \ref{free_stream_approach}.  We expand $\chi=\sum_i d^i\hat\chi_i$ in the orthonormal basis generated by $f_{ch}$ and, using $N$ terms, form the $N$-mode approximation $f_\Upsilon^N$ to $f_\Upsilon$.  The $d^i$ are obtained by taking the $L^2(f_{ch}dy)$ inner product of $\chi$ with the basis function $\hat\chi_i$,
\begin{equation}
d^i=\int\hat\chi_i \chi f_{ch}dy=\int\hat\chi_i  f_\Upsilon dy.
\end{equation}
 Figures \ref{fig:free_stream_f0_approx_Ups_5} and \ref{fig:free_stream_f0_approx_Ups_1_5} show the normalized $L^1(dx)$ errors between $f_\Upsilon^N$ and $f_\Upsilon$, computed via
\begin{equation}
\text{error}_N=\frac{\int_0^\infty |f_\Upsilon -f_\Upsilon ^N|dy}{\int_0^\infty |f_\Upsilon |dy}.
\end{equation}

We note the appearance of the reheating ratio
\begin{equation}\label{reheat}
 R\equiv aT  
\end{equation}
in the denominator of \req{zeroth_approx}, which comes from changing variables from $z=p/T$ in \req{weight} to $y=ap$ in order to compare with \req{free_stream_weight}.  Physically, $R$ is the ratio of the physical temperature $T$ to the dilution controlled temperature scaling  of $1/a$.   In physical situations, including cosmology, $R$ can vary from unity when dimensioned energy scales influence dynamical equations for $a$. From the error plots we see that for $R$ sufficiently close to $1$, the approximation performs well with a small number of terms, even with $\Upsilon\neq 1$.  

\begin{figure}[H]
 \begin{minipage}[b]{0.5\linewidth}
\centerline{\includegraphics[height=6.cm]{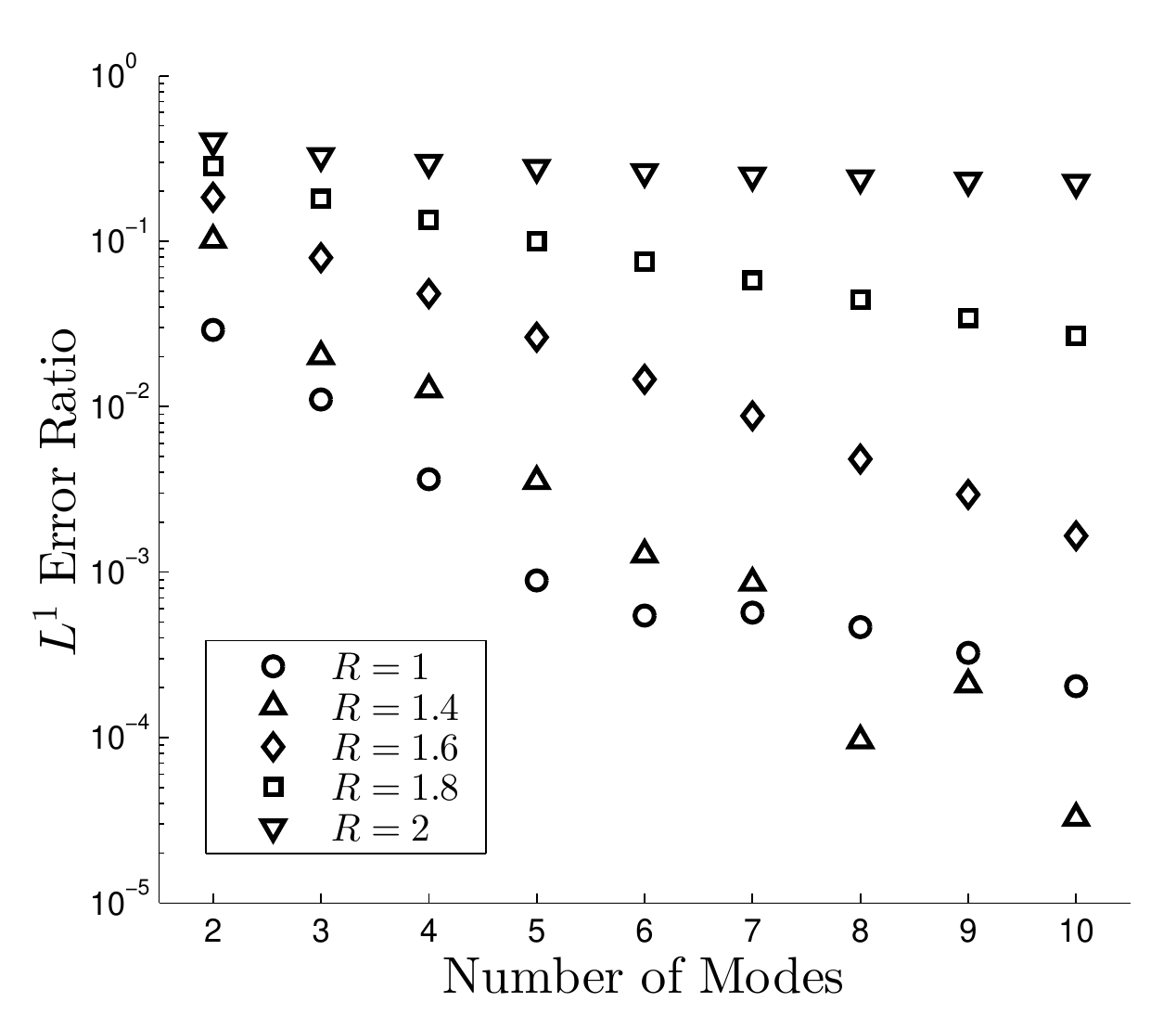}}
\caption{Errors in expansion of \req{zeroth_approx} as a function of number of modes, $\Upsilon=0.5$.}\label{fig:free_stream_f0_approx_Ups_5}
 \end{minipage}
 \hspace{0.5cm}
 \begin{minipage}[b]{0.5\linewidth}
\centerline{\includegraphics[height=6cm]{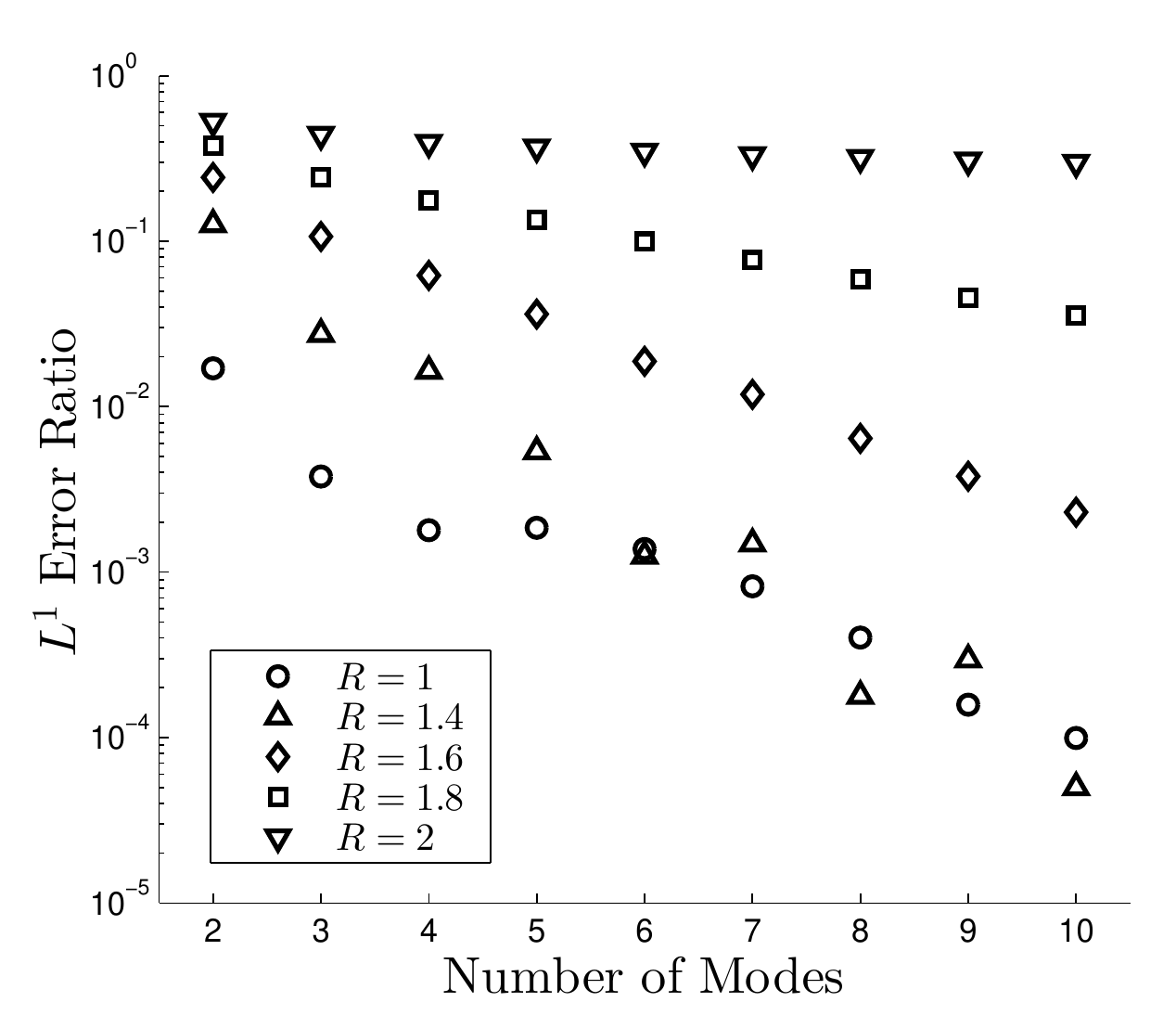}}
\caption{Errors in  expansion of \req{zeroth_approx} as a function of number of modes, $\Upsilon=1.5$.}\label{fig:free_stream_f0_approx_Ups_1_5}
 \end{minipage}
 \end{figure}

 In the case of large reheating, we find that when $R$ approaches and surpasses $2$, large spurious oscillations begin to appear in the expansion and they persist even when a large number of terms are used, as seen in figures  \ref{fig:free_stream_f0_approx_Ups_1_T_r_1_85} and \ref{fig:free_stream_f0_approx_Ups_1_T_r_2}, where we compare $f_\Upsilon/f_{ch}^{1/2}$ with $f_{\Upsilon}^N/f_{ch}^{1/2}$ for $\Upsilon=1$ and $N=20$.    See Ref.~\cite{Birrell_orthopoly} for futher discussion of the origin of these oscillations. This demonstrates that the chemical equilibrium method with dilution temperature scaling will  perform extremely poorly in situations that experience a large degree of reheating. For $R\approx 1$, the benefit of including fugacity is not as striking, as the chemical equilibrium basis is able to approximate \req{zeroth_approx} reasonably well.  However, for more stringent error tolerances including $\Upsilon$ can reduce the number of required modes in cases where the degree of chemical non-equilibrium is large.
\begin{figure}[H]
 \begin{minipage}[b]{0.5\linewidth}
\centerline{\includegraphics[height=6.cm]{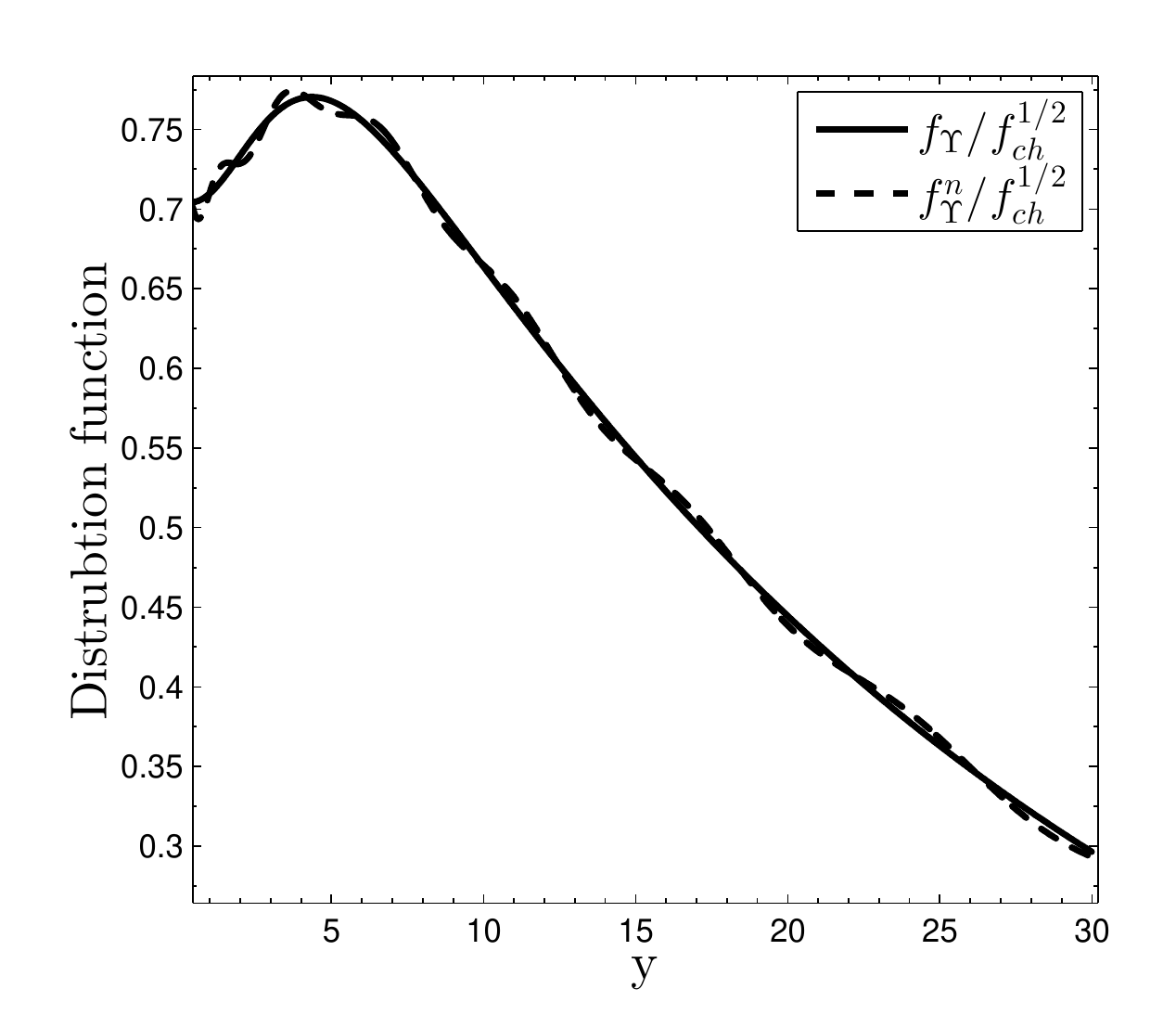}}
\caption{Approximation to \req{zeroth_approx} for $\Upsilon=1$ and $R=1.85$ using the first $20$ basis elements generated by \req{free_stream_weight}.}\label{fig:free_stream_f0_approx_Ups_1_T_r_1_85}
 \end{minipage}
 \hspace{0.5cm}
 \begin{minipage}[b]{0.5\linewidth}
\centerline{\includegraphics[height=6.cm]{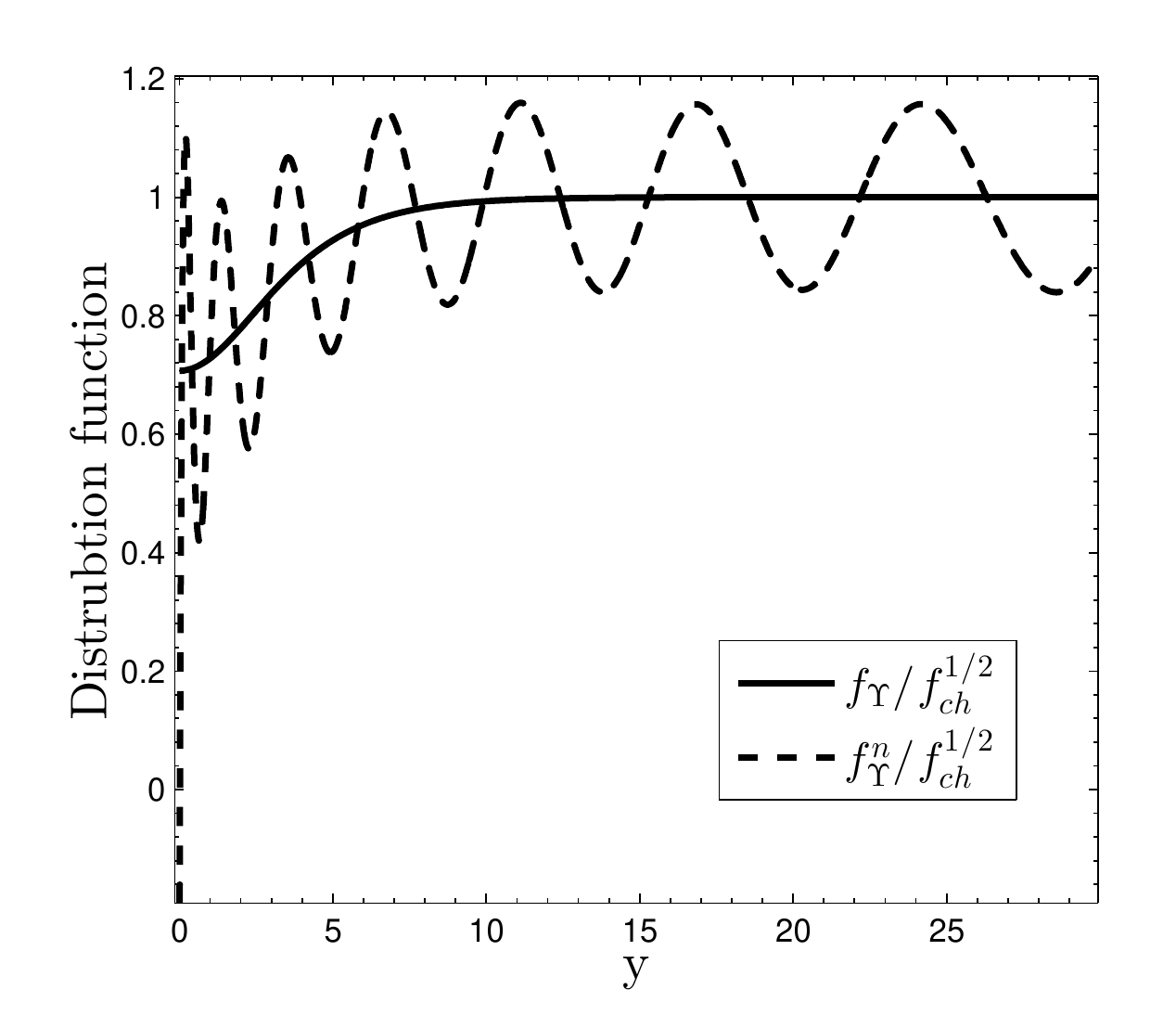}}
\caption{Approximation to \req{zeroth_approx} for $\Upsilon=1$ and $R=2$ using the first $20$ basis elements generated by \req{free_stream_weight}.}\label{fig:free_stream_f0_approx_Ups_1_T_r_2}
 \end{minipage}
 \end{figure}

\subsection{Dynamics}\label{dynamics_sec}
In this section we derive the dynamical equations for the  method outlined in section \ref{kinetic_eq_approach}.  In particular, we identify physically motivated dynamics for the effective temperature and fugacity.  Using \req{T_boltzmann} and the definition of $\psi$ from \req{kinetic_approx} we have
\begin{align}\label{near_equilib_eq}
\partial_t \psi+\frac{1}{f_\Upsilon }\frac{\partial f_\Upsilon }{\partial\Upsilon}\dot\Upsilon\psi-\frac{z}{f_\Upsilon }\left(H+\frac{\dot{T}}{T}\right)\left(\psi\partial_zf_\Upsilon +f_\Upsilon \partial_z \psi\right)=\frac{1}{f_\Upsilon E}C[f_\Upsilon \psi].
\end{align}
 Denote the monic orthogonal polynomial basis generated by the weight \req{weight} by $\psi_n$, $n=0,1,...$ where $\psi_n$ is degree $n$ and call the normalized versions  $\hat{\psi}_n$. Recall that $\hat\psi_n$ depend on $t$ due to the $\Upsilon$ dependence of the weight function used in the construction. Consider the space of polynomial of degree less than or equal to $N$, spanned by $\hat\psi_n$, $n=0,...,N$.   For $\psi$ in this subspace, we expand $\psi=\sum_{j=0}^Nb^j\hat\psi_j$ and use \req{near_equilib_eq}  to obtain
\begin{align}\label{T_vars}
\sum_i \dot{b}^i\hat\psi_i=&\sum_ib^i\frac{z}{f_\Upsilon }\left(H+\frac{\dot{T}}{T}\right)\left(\partial_z(f_\Upsilon )\hat\psi_i+f_\Upsilon \partial_z\hat\psi_i\right)\\
&-\sum_ib^i\left(\dot{\hat{\psi}}_i+\frac{1}{f_\Upsilon }\frac{\partial f_\Upsilon }{\partial\Upsilon}\dot\Upsilon\hat\psi_i\right)+\frac{1}{f_\Upsilon E}C[f].
\notag
\end{align}
From this we see  that the equations obtained from the Boltzmann equation by projecting onto the finite dimensional subspace are

\begin{align}
\dot b^k=& \sum_i b^i\left(H+\frac{\dot{T}}{T}\right)\left(\langle\frac{z}{f_\Upsilon }\hat\psi_i\partial_zf_\Upsilon ,\hat\psi_k\rangle+\langle z\partial_z \hat\psi_i,\hat\psi_k\rangle\right) \\
&-\sum_i b^i\dot{\Upsilon}\left(\langle\frac{1}{f_\Upsilon }\frac{\partial f_\Upsilon }{\partial\Upsilon}\hat\psi_i,\hat\psi_k\rangle+\langle\frac{\partial\hat{\psi}_i}{\partial \Upsilon},\hat\psi_k\rangle\right)+\langle\frac{1}{f_\Upsilon E}C[f],\hat\psi_k\rangle\notag
\end{align}

where $\langle\cdot,\cdot\rangle$ denotes the inner product defined by the weight function \req{weight}
\begin{equation}
\langle h_1,h_2\rangle=\int_0^\infty h_1(z)h_2(z)w_\Upsilon(z)dz.
\end{equation}
  The term in brackets comprises the linear part of the system, while the collision term contains polynomial nonlinearities when multiple coupled distribution are being modeled using a $2$-$2$ collision operator \req{coll}.  

To isolate the linear part, we define matrices
\begin{align}\label{A_B_matrices}
A^k_i(\Upsilon)\equiv&\langle\frac{z}{f_\Upsilon }\hat\psi_i\partial_zf_\Upsilon ,\hat\psi_k\rangle+\langle z\partial_z \hat\psi_i,\hat\psi_k\rangle,\\
B^k_i(\Upsilon)\equiv& C_i^k(\Upsilon)+D_i^k(\Upsilon),\hspace{2mm} C_i^k\equiv\Upsilon\langle\frac{1}{f_\Upsilon }\frac{\partial f_\Upsilon }{\partial\Upsilon}\hat\psi_i,\hat\psi_k\rangle,\hspace{2mm} D_i^k\equiv\Upsilon\langle\frac{\partial\hat{\psi}_i}{\partial \Upsilon},\hat\psi_k\rangle. 
\end{align}
 With these definitions, the equations for the $b^k$ become
\begin{align}\label{b_eq}
\dot b^k=& \left(H+\frac{\dot{T}}{T}\right)\sum_i A_i^k(\Upsilon)b^i-\frac{\dot{\Upsilon}}{\Upsilon}\sum_i B_i^k(\Upsilon)b^i+\langle\frac{1}{f_\Upsilon E}C[f],\hat\psi_k\rangle.
\end{align}
 See \ref{ortho-general} for details on how to recursively construct the $\partial_z\hat\psi_i$. We show how to compute the inner products $\langle\hat\psi_k,\partial_{\Upsilon}\hat\psi_k\rangle$ in  \ref{ortho-polynom-fam}. In \ref{lower_triang} we prove that that both $A$ and $B$ are lower triangular and show that the only inner products involving the $\partial_\Upsilon\hat{\psi}_i$ that are required in order to compute $A$ and $B$ are those the above mentioned diagonal elements, $\langle\hat\psi_k,\partial_{\Upsilon}\hat\psi_k\rangle$.

We fix the dynamics of $T$ and $\Upsilon$ by imposing the conditions
\begin{equation}\label{b_ics}
b^0(t)\hat\psi_0(t)=1,\hspace{2mm}b^1(t)=0.
\end{equation}
In other words,
\begin{equation}
f(t,z)=f_\Upsilon (t,z)\left(1+\phi(t,z)\right),\hspace{2mm} \phi=\sum_{i=2}^N b^i\hat\psi_i.
\end{equation}
This reduces the number of degrees of freedom in \req{b_eq} from $N+3$ to $N+1$.  In other words, after enforcing \req{b_ics}, \req{b_eq} constitutes $N+1$ equations for the remaining $N+1$ unknowns, $b^2,...,b^N$, $\Upsilon$, and $T$.  We will call $T$ and $\Upsilon$ the first two ``modes", as their dynamics arise from imposing the conditions \req{b_ics} on the zeroth and first order coefficients in the expansion. We will solve for their dynamics explicitly below.

To see the physical motivation for the choices \req{b_ics}, consider the particle number density and energy density.  Using orthonormality of the $\hat\psi_i$ and \req{b_ics} we have
\begin{align}
n=&\frac{g_pT^3}{2\pi^2}\sum_ib^i\int_0^\infty f_\Upsilon  \hat\psi_i z^2 dz=\frac{g_pT^3}{2\pi^2}\sum_ib^i\langle \hat\psi_i ,1\rangle\\
=&\frac{g_pT^3}{2\pi^2}b^0\langle \hat\psi_0 ,1\rangle=\frac{g_pT^3}{2\pi^2}\langle 1 ,1\rangle,\\
\rho=&\frac{g_pT^4}{2\pi^2}\sum_ib^i\int_0^\infty f_\Upsilon  \hat\psi_i z^3 dz=\frac{g_pT^4}{2\pi^2}\sum_ib^i\langle\hat\psi_i, z\rangle\\
=&\frac{g_pT^4}{2\pi^2}\left(b^0\langle\hat\psi_0, z\rangle+b^1\langle\hat\psi_1, z\rangle\right)=
\frac{g_pT^4}{2\pi^2}\langle 1,z\rangle.
\end{align}
 Using these together with the definition of the weight function \req{weight} we find
\begin{align}\label{th_eq_moments}
n=&\frac{g_pT^3}{2\pi^2}\int_0^\infty f_\Upsilon  z^2dz,\\
\label{th_eq_moments2}
\rho=&\frac{g_pT^4}{2\pi^2}\int_0^\infty f_\Upsilon  z^3dz.
\end{align}
Equations (\ref{th_eq_moments}) and (\ref{th_eq_moments2}) show  that the first two modes, $T$ and $\Upsilon$, with time evolution fixed by \req{b_ics} combine with the chemical non-equilibrium distribution $f_\Upsilon $ to capture the number density and energy density of the system exactly.  This fact is very significant, as it implies that within the chemical non-equilibrium approach as long as the back-reaction from the non-thermal distortions is small (meaning that the evolution of $T(t)$ and $\Upsilon(t)$ is not changed significantly when more modes are included), {\em all the effects relevant to the computation of  particle and energy flow are modeled by the time evolution of $T$ and $\Upsilon$ alone} and no further modes are necessary.  This gives a clear separation between the averaged physical quantities, characterized by $f_\Upsilon $, and the momentum dependent non-thermal distortions as contained in 
\begin{equation}
\phi=\sum_{i=2}^N b^i\hat\psi_i.
\end{equation}

One should contrast this chemical non-equilibrium behavior  with the chemical equilibrium situation, where a minimum of four modes is required to describe the number and energy densities, as shown in \req{free_stream_moments}.   Moreover we will show that convergence to the desired precision is faster in the chemical non-equilibrium approach as compared to chemical equilibrium. Due to the high cost of numerically integrating realistic collision integrals of the form \req{coll}, this fact can be very significant in applications. We remark that the relations \req{th_eq_moments} are the physical motivation for including the $z^2$ factor in the weight function. All three modifications we have made in constructing our new method, the introduction of an effective temperature i.e. $R\ne 1$, the generalization to chemical non-equilibrium $f_\Upsilon $, and the introduction of $z^2$ to the weight, \req{reheat}, were needed to obtain the properties \req{th_eq_moments}, but it is the introduction of $z^2$ that reduces the number of required modes and hence reduces the computational cost. 

With $b^0$ and $b^1$ fixed as in \req{b_ics} we can solve the equations for $\dot b^0$ and $\dot b^1$ from \req{b_eq} for $\dot\Upsilon$ and $\dot T$ to obtain

\begin{align}\label{Ups_T_eqs}
\dot{\Upsilon}/{\Upsilon}=&\frac{(Ab)^1\langle\frac{1}{f_\Upsilon E}C[f],\hat\psi_0\rangle-(Ab)^0\langle\frac{1}{f_\Upsilon E}C[f],\hat\psi_1\rangle }{[\Upsilon\partial_\Upsilon \langle1,1\rangle/(2||\psi_0||)+(Bb)^0](Ab)^1-(Ab)^0(Bb)^1},\\[0.5cm]
\dot{T}/T
=&\frac{(Bb)^1\langle\frac{1}{f_\Upsilon E}C[f],\hat\psi_0\rangle-\langle\frac{1}{f_\Upsilon E}C[f],\hat\psi_1\rangle[\Upsilon\partial_\Upsilon \langle1,1\rangle/(2||\psi_0||)+(Bb)^0]}{[\Upsilon\partial_\Upsilon \langle1,1\rangle/(2||\psi_0||)+(Bb)^0](Ab)^1-(Ab)^0(Bb)^1}-H\notag\\[0.3cm]
=&\frac{1}{(Ab)^1}\left((Bb)^1\dot{\Upsilon}/\Upsilon-\langle\frac{1}{f_\Upsilon E}C[f],\hat\psi_1\rangle\right)-H.\label{T_eq}
\end{align}

Here $(Ab)^n=\sum_{j=0}^NA^n_jb^j$ and similarly for $B$ and $||\cdot||$ is the norm induced by $\langle\cdot,\cdot\rangle$. In deriving this, we used
\begin{equation}
\dot{b}^0=\frac{1}{2||\psi_0||}\dot\Upsilon\partial_\Upsilon \langle1,1\rangle, \hspace{4mm} \partial_\Upsilon \langle1,1\rangle=\int_0^\infty \frac{z^2}{(e^{z/2}+ \Upsilon e^{-z/2})^2}dz
\end{equation}
which comes from differentiating \req{b_ics}. 
 
 It is easy to check that when the collision operator vanishes, then the above system is solved by 
\begin{equation}\label{free_stream_sol}
\Upsilon=\text{constant},\hspace{4mm} \frac{\dot T}{T}=-H,\hspace{2mm}  b^n=\text{constant},\hspace{1mm} n>2
\end{equation}
i.e. the fugacity and non-thermal distortions are `frozen' into the distribution and the temperature satisfies dilution scaling $T\propto 1/a$.

When the collision term becomes small, \req{free_stream_sol} motivates another change of variables. Letting $T=(1+\epsilon)/a$  gives the equation
\begin{equation}
\dot\epsilon=\frac{1+\epsilon}{(Ab)^1}\left((Bb)^1\dot{\Upsilon}/\Upsilon-\langle\frac{1}{f_\Upsilon E}C[f],\hat\psi_1\rangle\right).
\end{equation}
Solving this in place of \req{T_eq} when the collision terms are small avoids having to numerically track the free-streaming evolution.  In particular this will ensure conservation of comoving particle number, which equals a function of $\Upsilon$ multiplied by $(aT)^3$, to much greater precision in this regime as well as resolve the freeze-out temperatures more accurately.\\

\noindent{\bf Projected Dynamics are Well-defined:}\\
The following calculation shows that, for a distribution initially in kinetic equilibrium, the determinant factor in the denominator of \req{Ups_T_eqs} is nonzero and hence the dynamics for $T$ and $\Upsilon$, as well as the remainder of the projected system, are well-defined, at least for sufficiently small times. 

 Kinetic equilibrium implies the initial conditions $b^0=||\psi_0||$, $b^i=0$, $i>0$.  Therefore we have
\begin{align}
K\equiv& (\Upsilon\partial_\Upsilon \langle 1,1\rangle/(2||\psi_0||)+(Bb)^0)(Ab)^1-(Ab)^0(Bb)^1\\[0.3cm]
=&(C^0_0A^1_0-A^0_0C^1_0)(b^0)^2+\left[(D^0_0A^1_0-A^0_0D^1_0)(b^0)^2+\Upsilon\partial_\Upsilon \langle 1,1\rangle/(2||\psi_0||)A^1_0b^0\right]\notag\\[0.3cm]
\equiv & K_1+K_2.\notag
\end{align}
\begin{align}
K_1=&\langle \frac{1}{1+\Upsilon e^{-z}},1\rangle\langle \frac{-z}{1+\Upsilon e^{-z}}\hat\psi_1,\hat\psi_0\rangle-\langle\frac{-z}{1+\Upsilon e^{-z}},\hat\psi_0\rangle\langle\frac{1}{1+ \Upsilon e^{-z}}\hat\psi_1,1\rangle.
\end{align}
Inserting the formula for $\hat\psi_1$ from \req{poly_recursion} we find
\begin{align}
K_1=&-\frac{1}{||\psi_1||\,||\psi_0||}\left[\langle\frac{1}{1+ \Upsilon e^{-z}},\hat\psi_0\rangle\langle\frac{z^2}{1+\Upsilon e^{-z}},\hat\psi_0\rangle-\langle\frac{z}{1+\Upsilon e^{-z}},\hat\psi_0\rangle^2\right].
\end{align}
The Cauchy-Schwarz inequality  applied to the inner product with weight function
\begin{equation}
\tilde{w}=\frac{w}{1+\Upsilon e^{-z}}\hat\psi_0
\end{equation}
together with linear independence of $1$ and $z$ implies that the term in brackets is positive and so $K_1<0$ at $t=0$.  For the second term, noting that $D^1_0=0$ by orthogonality and using \req{norm_deriv_eq}, we have
\begin{align}
K_2=&[\langle\partial_\Upsilon\hat\psi_0,\hat\psi_0\rangle||\psi_0||+\partial_\Upsilon \langle 1,1\rangle/(2||\psi_0||)]\Upsilon A_0^1||\psi_0||\\
=&0.\notag
\end{align}
This proves that $K$ is nonzero at $t=0$.\\

\section{Validation}\label{validation}
We will validate our numerical method on an exactly solvable model problem
\begin{equation}\label{toy_eq}
\partial_t f-pH \partial_p f=\frac{1}{M}\left(\frac{1}{\Upsilon^{-1}e^{p/T_{eq}}+1}-f(p,t)\right), \hspace{2mm} f(p,0)=\frac{1}{e^{p/T_{eq}(0)}+1}
\end{equation}
where $M$ is a constant with units of energy and we choose units in which it is equal to $1$. This model describes a distribution that is attracted to a given equilibrium distribution at a prescribed time dependent temperature $T_{eq}(t)$ and fugacity $\Upsilon$. This type of an idealized scattering operator, without fugacity, was first introduced in \cite{Anderson_Witting}. By changing coordinates $y=a(t)p$ we find
\begin{equation}\label{free_stream_toy}
\partial_tf(y,t)=\frac{1}{\Upsilon^{-1}\exp[y/(a(t)T_{eq}(t))]+1}-f(y,t).
\end{equation}
 which has as solution
\begin{equation}\label{exact_sol}
f(y,t)=\int_0^t\frac{e^{s-t}}{\Upsilon^{-1}\exp[y/(a(s)T_{eq}(s))]+1}ds+\frac{e^{-t}}{\exp[y/(a(0)T_{eq}(0))]+1}.
\end{equation}
We now transform to $z=p/T(t)$ where the temperature $T$ of the distribution $f$ is defined as in section \ref{dynamics_sec}.  Therefore, we have the exact solution to
\begin{equation}\label{k_eq_toy}
\partial_tf-z\left(H+\frac{\dot{T}}{T}\right)\partial_zf=\frac{1}{\Upsilon^{-1}e^{zT/T_{eq}}+1}-f(z,t)
\end{equation}
given by
\begin{align}
f(z,t)=&\int_0^t\frac{e^{s-t}}{\Upsilon^{-1}\exp[a(t)T(t)z/(a(s)T_{eq}(s))]+1}ds\\
&+\frac{e^{-t}}{\exp[a(t)T(t)z/(a(0)T_{eq}(0))]+1}.\notag
\end{align}
We use this to test the chemical equilibrium and chemical non-equilibrium methods under two different conditions. 

\subsection{Reheating Test}
First we test the two methods we have outlined in a scenario that exhibits reheating.  Motivated by applications to cosmology, we choose a scale factor evolving as in the radiation dominated era, a fugacity $\Upsilon=1$, and choose an equilibrium temperature that exhibits reheating like behavior with $aT_{eq}$ increasing for a period of time,
\begin{align}\label{a_T_def}
a(t)=\left(\frac{t+b}{b}\right)^{1/2}\!\!\!,\ \  \ \
T_{eq}(t)=\frac{1}{a(t)}\left(1+\frac{1-e^{-t}}{e^{-(t-b)}+1}(R-1)\right)
\end{align}
where $R$ is the desired reheating ratio. Note that $(aT_{eq})(0)=1$ and $(aT_{eq})(t)\rightarrow R$ as $t\rightarrow\infty$. Qualitatively, this is reminiscent of the dynamics of neutrino freeze-out, but the range of reheating ratio for which we will test our method is larger than found there.

We solved \req{free_stream_toy} and \req{k_eq_toy} numerically using the chemical equilibrium and chemical non-equilibrium methods respectively for $t\in[0,10]$ and $b=5$ and the cases $R=1.1$, $R=1.4$, as well as the more extreme ratio of $R=2$.  The bases of orthogonal polynomials were generated numerically using the recursion relations from \ref{orthopoly_app}.  For the applications we are considering, where the solution is a small perturbation of equilibrium, only a small number of terms are required and so the numerical challenges associated with generating a large number of such orthogonal polynomials are not an issue.\\

\noindent{\bf Chemical Equilibrium Method:}\\
We solved \req{free_stream_toy} using the chemical equilibrium method, with the orthonormal basis defined by the weight function \req{free_stream_weight} for $N=2,...,10$ modes (mode numbers $n=0,...,N-1$) and prescribed  single step relative and absolute error tolerances of $10^{-13}$ for the numerical integration, and with asymptotic reheating ratios of $R=1.1$, $R=1.4$, and $R=2$.

In figures   \ref{fig:free_stream_num_err} and  \ref{fig:free_stream_E_err} we show the maximum relative error in the number densities and energy densities respectively over the time interval $[0,10]$ for various numbers of computed modes.  The particle number density and energy density are accurate, up to the integration tolerance level, for $3$ or more and $4$ or more modes respectively. This is consistent with \req{free_stream_moments} which shows the number of modes required to capture each of these quantities. However, fewer modes than these minimum values lead to a large error in the corresponding moment of the distribution function.

\begin{figure}[H]
 \begin{minipage}[b]{0.5\linewidth}
\centerline{\includegraphics[height=6.0cm]{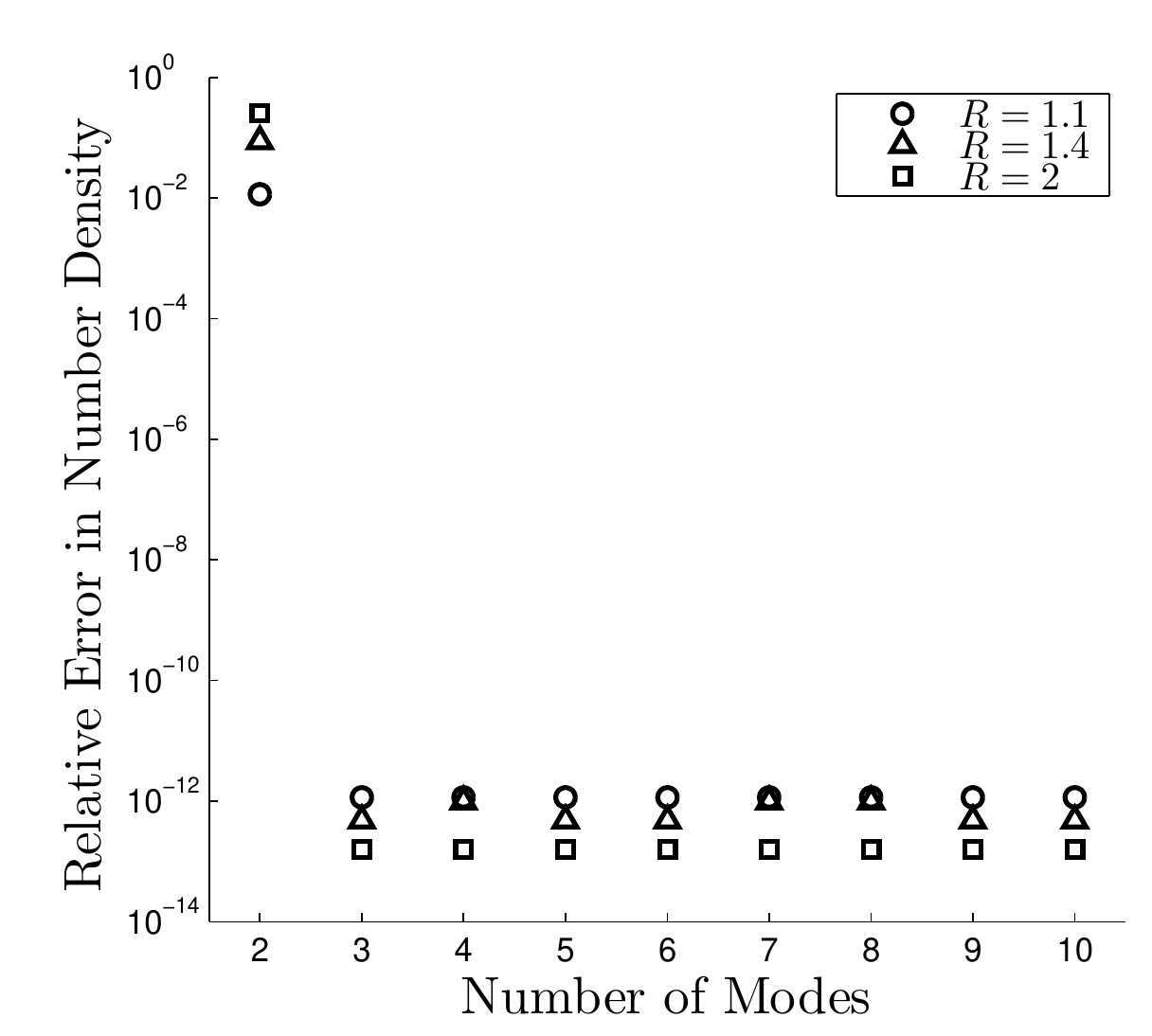}}
\caption{Maximum relative error in particle number density.}\label{fig:free_stream_num_err}
 \end{minipage}
 \hspace{0.5cm}
 \begin{minipage}[b]{0.5\linewidth}
\centerline{\includegraphics[height=6.0cm]{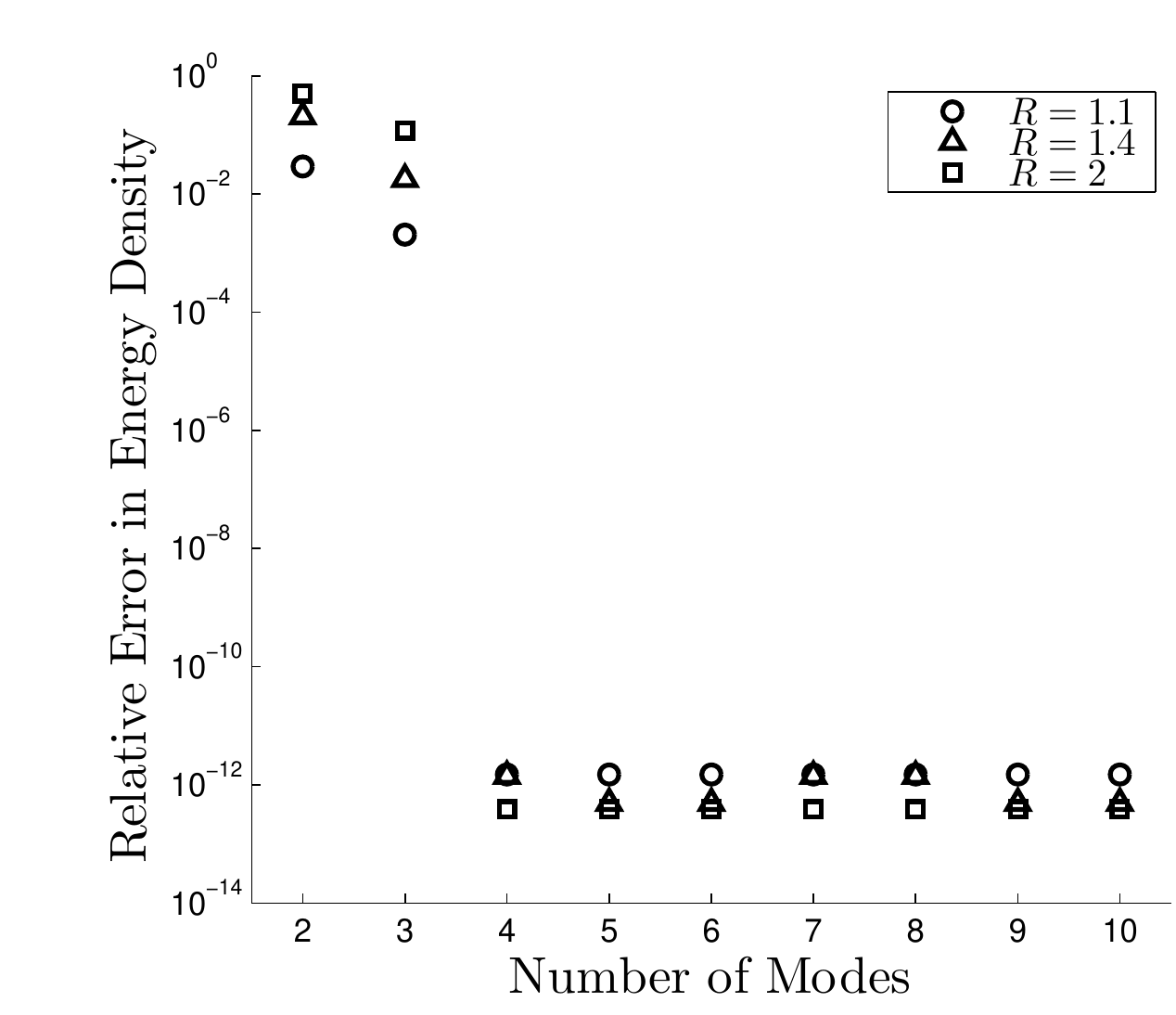}}
\caption{Maximum relative error in energy density.}\label{fig:free_stream_E_err}
 \end{minipage}
 \end{figure}

\begin{figure}[H]
\begin{minipage}[t]{0.5\linewidth}
\centerline{\includegraphics[height=6cm]{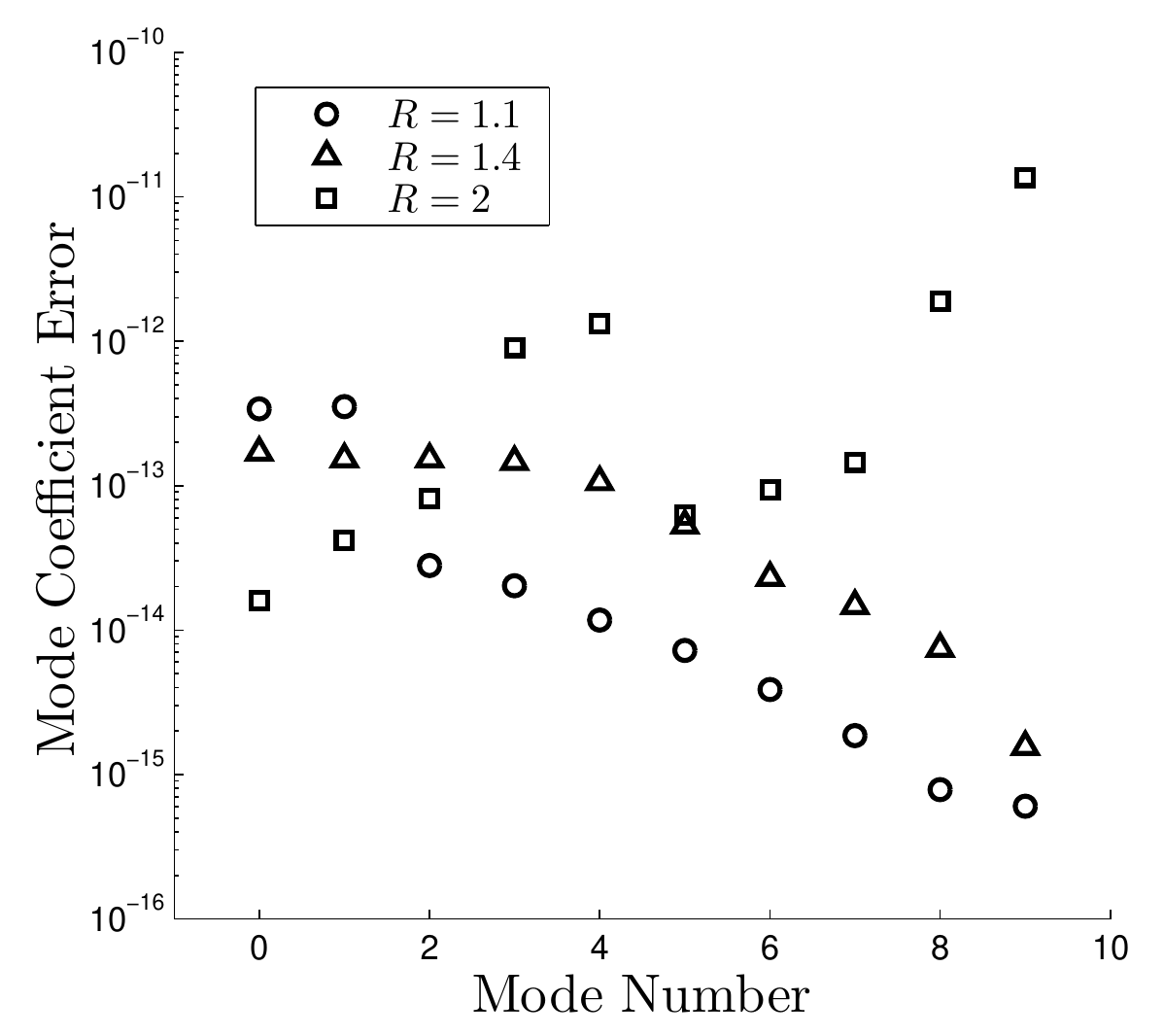}}
\caption{Maximum error in mode coefficients.}\label{fig:free_stream_b_err}
 \end{minipage}
 \hspace{0.5cm}
 \begin{minipage}[t]{0.5\linewidth}
\centerline{\includegraphics[height=6.0cm]{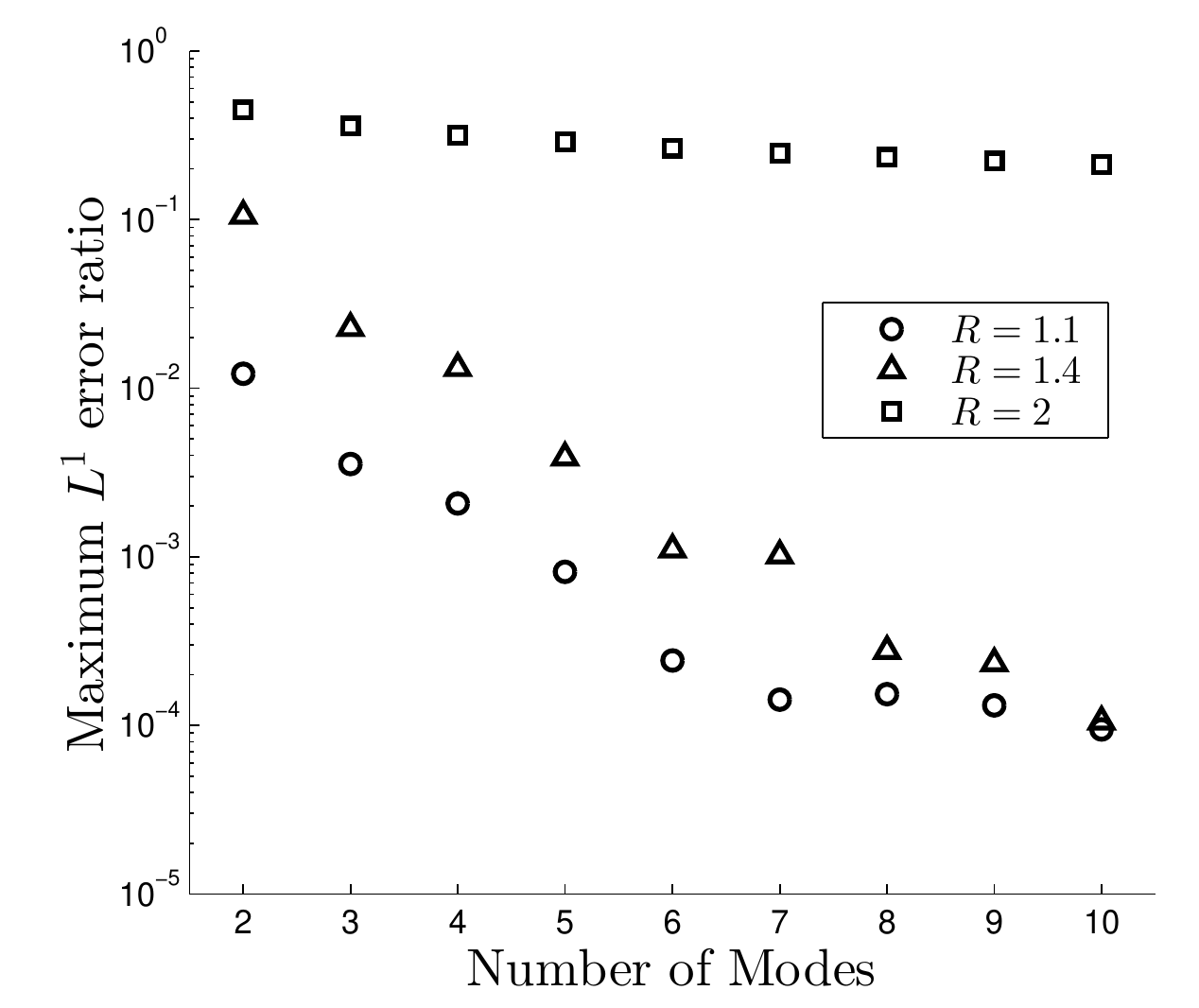}}
\caption{Maximum ratio  of $L^1$ error between computed and exact solutions to $L^1$ norm of the exact solution.}\label{fig:free_stream_L1_err}
 \end{minipage}
 \end{figure}

 To show that the numerical integration accurately captures the mode coefficients of the exact solution, \req{exact_sol}, we show the error between the computed coefficients and actual coefficients, denoted by $\tilde b_n$ and $b_n$ respectively
\begin{equation}\label{mode_err_def}
\text{error}_n=\max_{t} |\tilde{b}_n(t)-b_n(t)|,
\end{equation}
 in figure \ref{fig:free_stream_b_err}, where the evolution of the system was computed using $N=10$ modes.

In figure  \ref{fig:free_stream_L1_err} we show the error between the exact solution $f$, and the numerical solution $f^N$ computed using $N=2,...,10$ modes over the solution time interval, where we define the error by
\begin{equation}\label{f_err}
\text{error}_N=\max_{t} \frac{\int |f-f^N|dy}{\int |f|dy}.
\end{equation}
For $R=1$ and $R=1.4$  the chemical equilibrium method works reasonably well (as long as the number of modes is at least 4, so that the energy and number densities are properly captured) but for $R=2$ the approximate solution exhibits spurious oscillations, as seen in figure \ref{fig:free_stream_approx_T_r_2}, and has poor $L^1$ error.

\begin{figure}[H]
\begin{minipage}[t]{0.5\linewidth}
\centerline{\includegraphics[height=6.1cm]{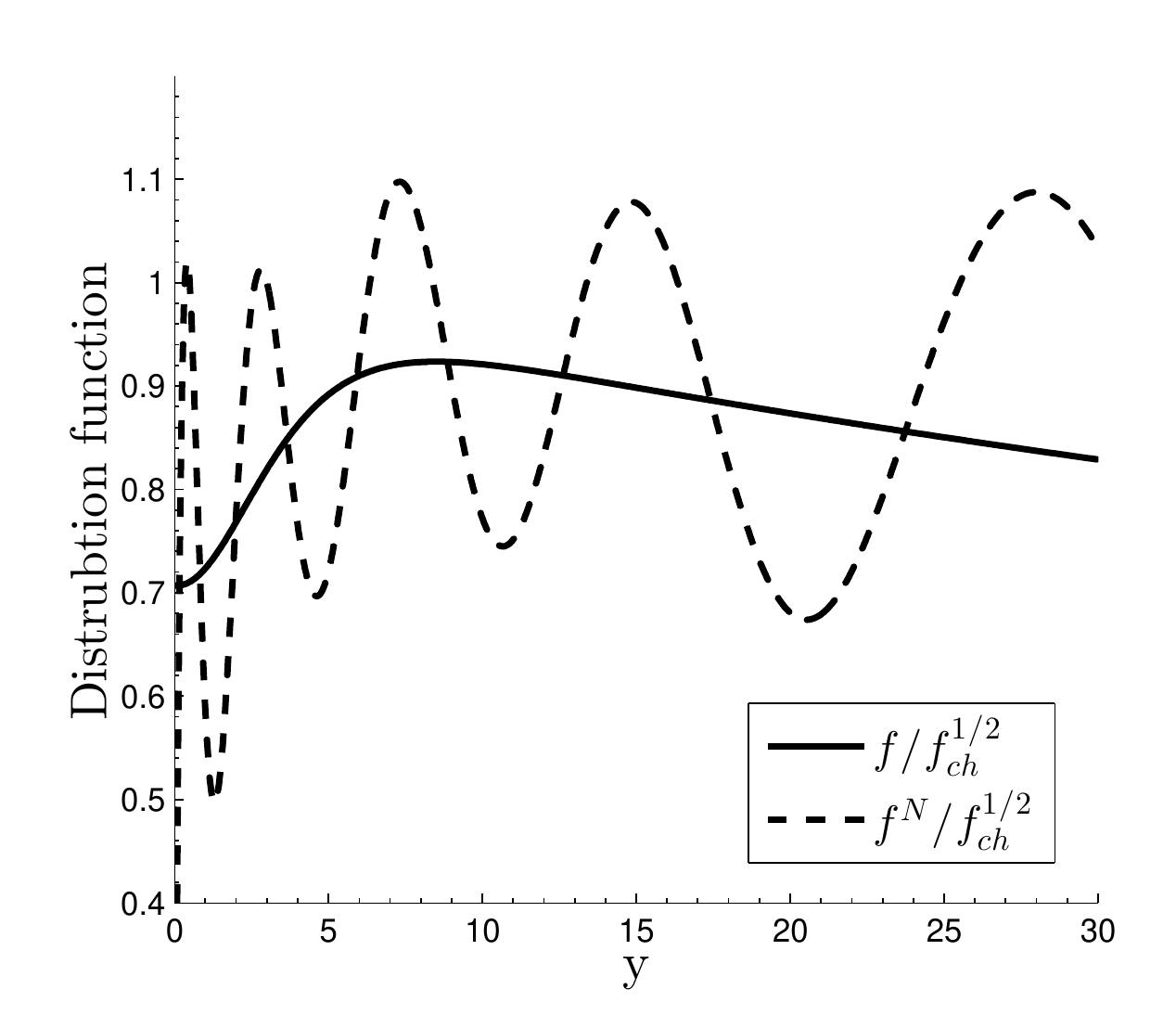}}
\caption{Approximate and exact solution for a reheating ratio $R=2$ and $N=10$ modes.}\label{fig:free_stream_approx_T_r_2}
 \end{minipage}
 \hspace{0.5cm}
 \begin{minipage}[t]{0.5\linewidth}
\centerline{\includegraphics[height=6.1cm]{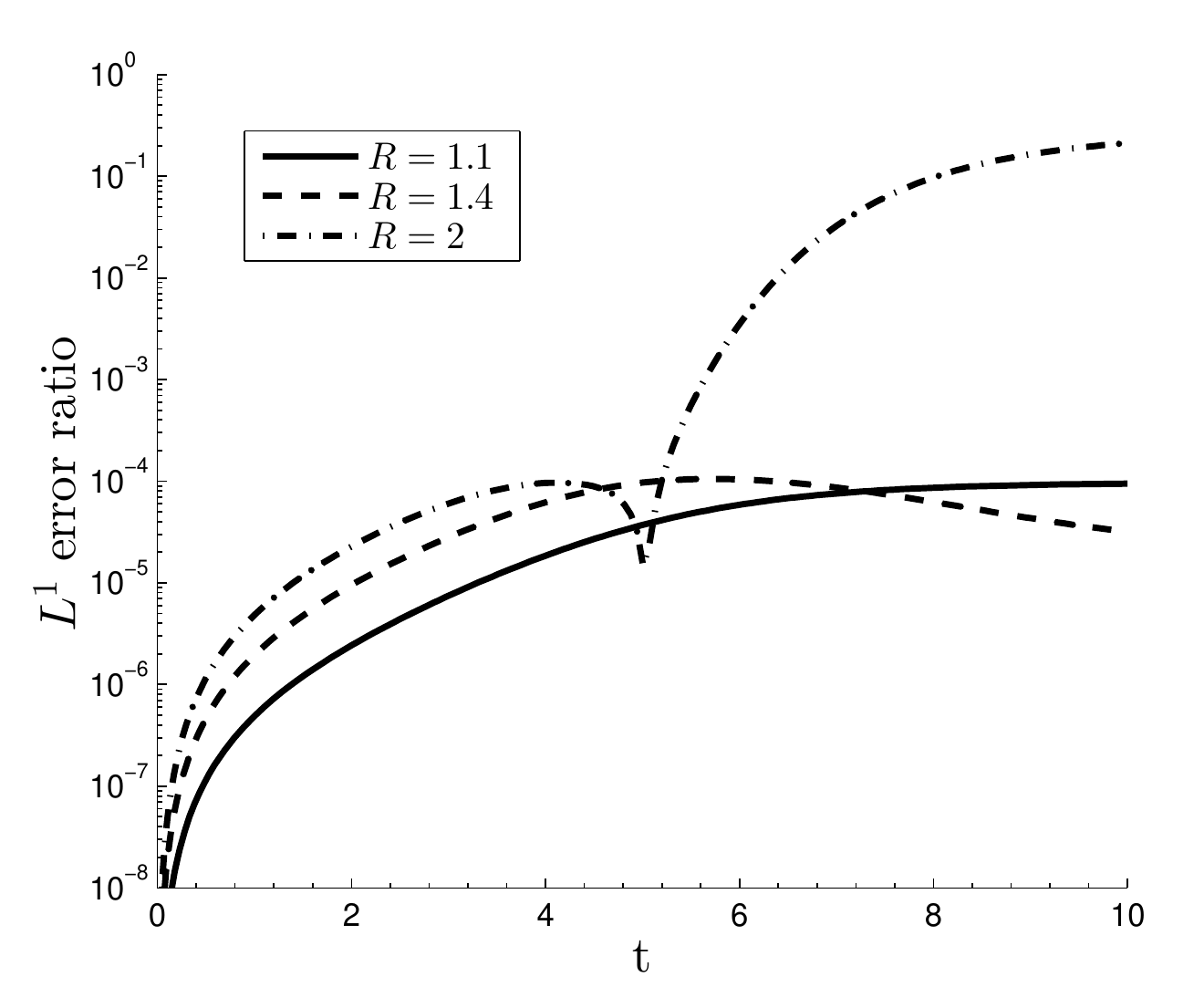}}
\caption{$L^1$ error ratio as a function of time for $N=10$ modes.}\label{fig:free_stream_L1_err_time}
 \end{minipage}
 \end{figure}

This poor behavior is expected based on the results in section \ref{basis_comparison}.  This is even clearer in figure \ref{fig:free_stream_L1_err_time} where we show the $L^1$ error ratio as a function of time for $N=10$ modes. In the $R=2$ case we see that the error increases as the reheating ratio approaches its asymptotic value of $R=2$ as $t\rightarrow\infty$.  As we will see, our methods achieves a much higher accuracy for a small number of terms in the case of large reheating ratio due to the replacement of dilution temperature scaling with the dynamical effective temperature $T$.

\noindent{\bf Chemical Non-Equilibrium Method:}\\
We now solve  \req{free_stream_toy} using the chemical non-equilibrium method, with the orthonormal basis defined by the weight function \req{weight} for $N=2,...,10$ modes, a prescribed numerical integration tolerance of $10^{-13}$, and asymptotic reheating ratios of $R=1.1$, $R=1.4$, and $R=2$.  Recall that we are referring to $T$ and $\Upsilon$ as the first two modes ($n=0$ and $n=1$). 

\begin{figure}[H]
 \begin{minipage}[b]{0.5\linewidth}
\centerline{\includegraphics[height=6.2cm]{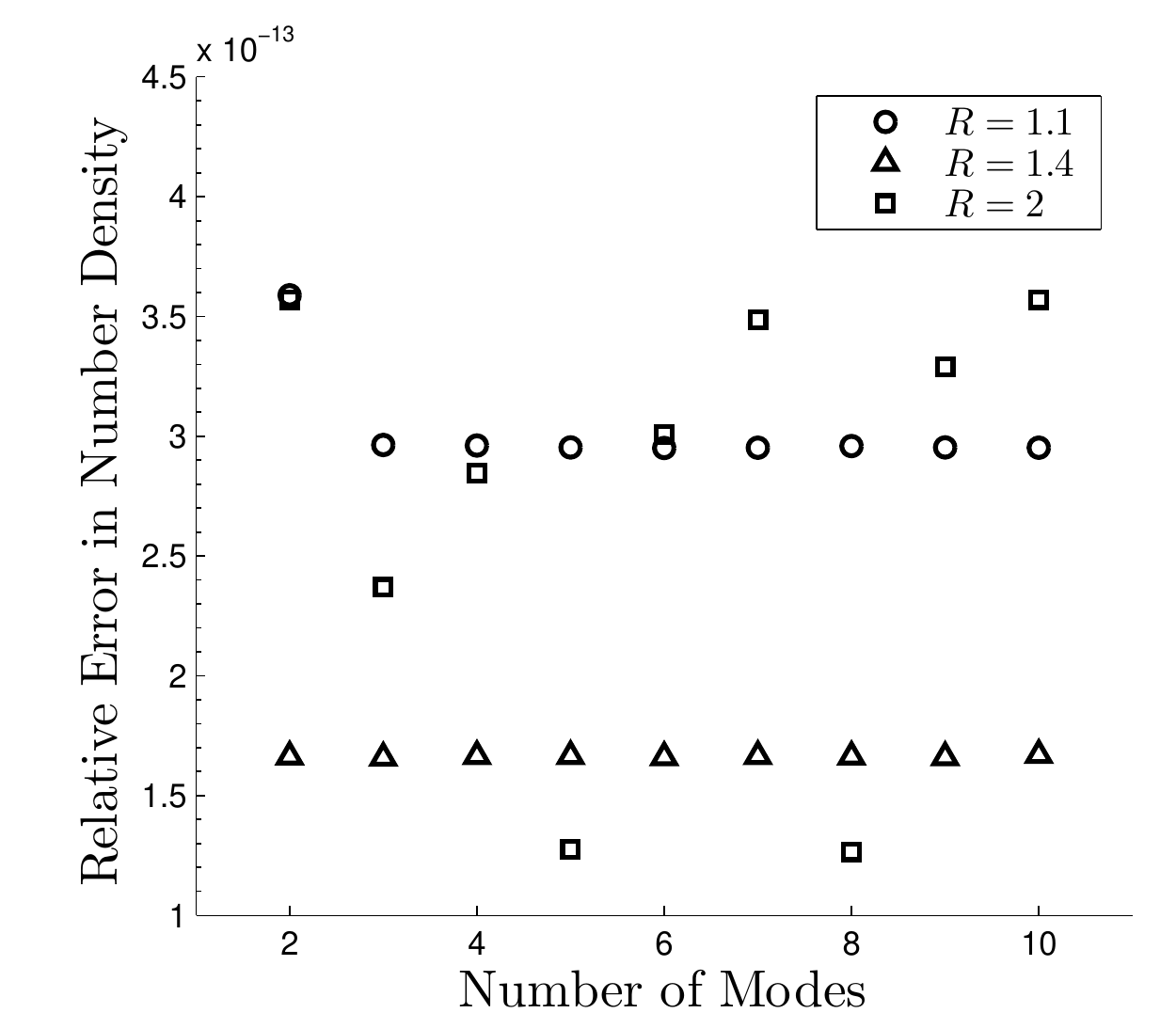}}
\caption{Maximum relative error in particle number density.}\label{fig:keq_num_err}
 \end{minipage}
 \hspace{0.5cm}
 \begin{minipage}[b]{0.5\linewidth}
\centerline{\includegraphics[height=6.2cm]{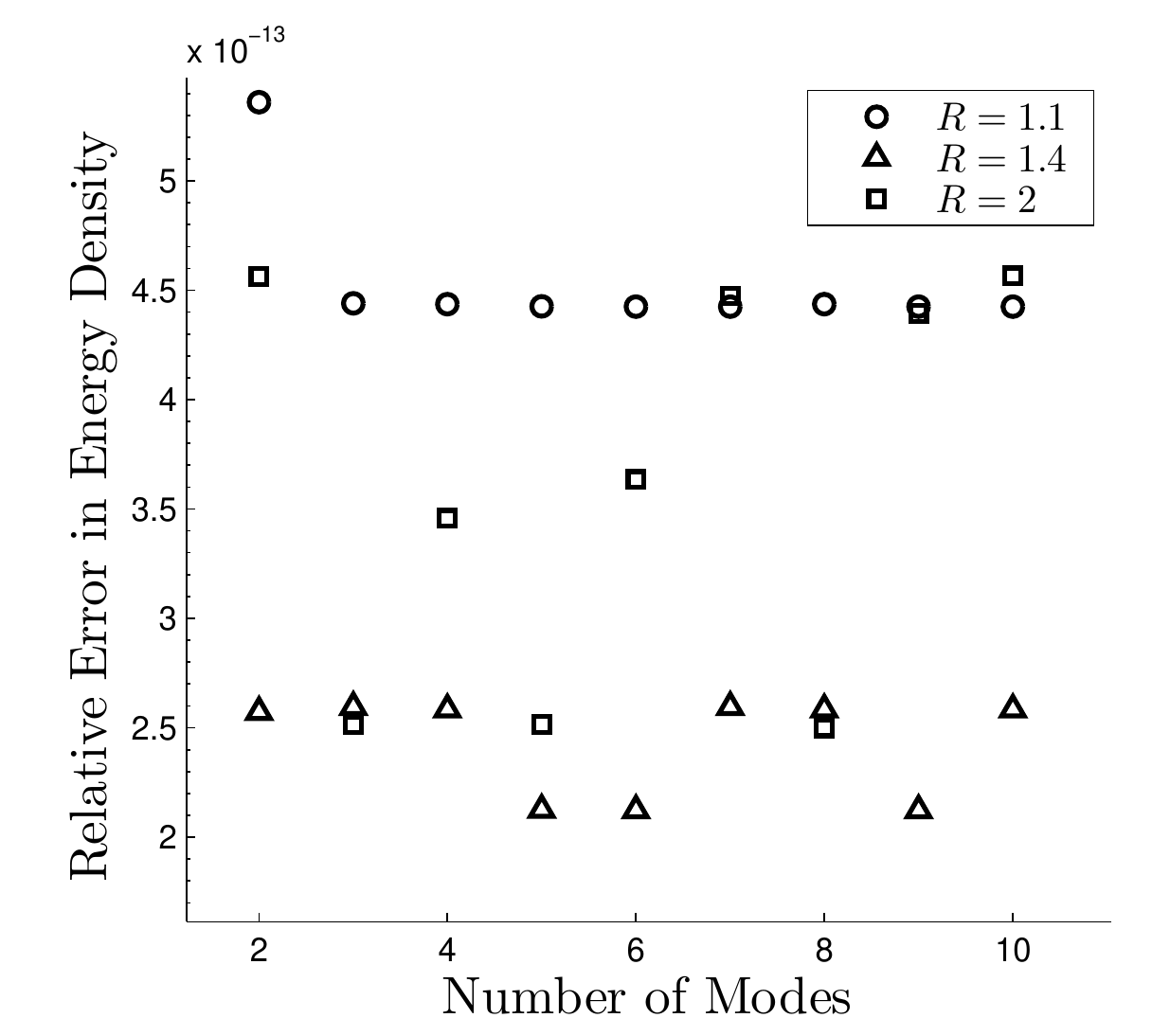}}
\caption{Maximum relative error in energy density.}\label{fig:keq_E_err}
 \end{minipage}
 \end{figure}

In figures \ref{fig:keq_num_err} and  \ref{fig:keq_E_err} we show the maximum relative error over the time interval $[0,10]$ in the number densities and energy densities respectively for various numbers of computed modes. Even for only $2$ modes, the number and energy densities are accurate up to the integration tolerance level.  This is in agreement with the analytical expressions in \req{th_eq_moments}.

\begin{figure}[H]
 \begin{minipage}[b]{0.5\linewidth}
\centerline{\includegraphics[height=6.2cm]{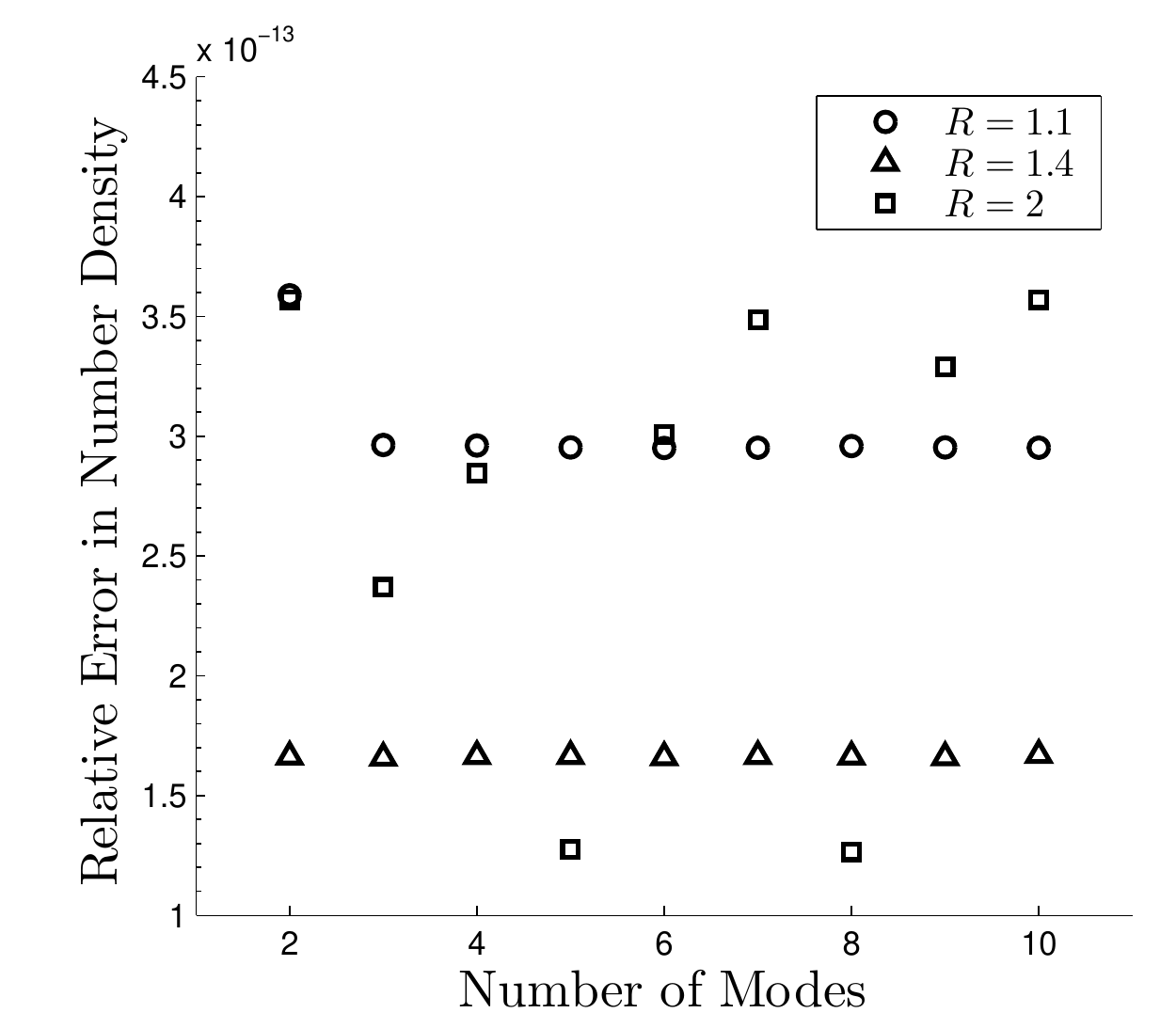}}
\caption{Maximum relative error in particle number density.}\label{fig:keq_num_err}
 \end{minipage}
 \hspace{0.5cm}
 \begin{minipage}[b]{0.5\linewidth}
\centerline{\includegraphics[height=6.2cm]{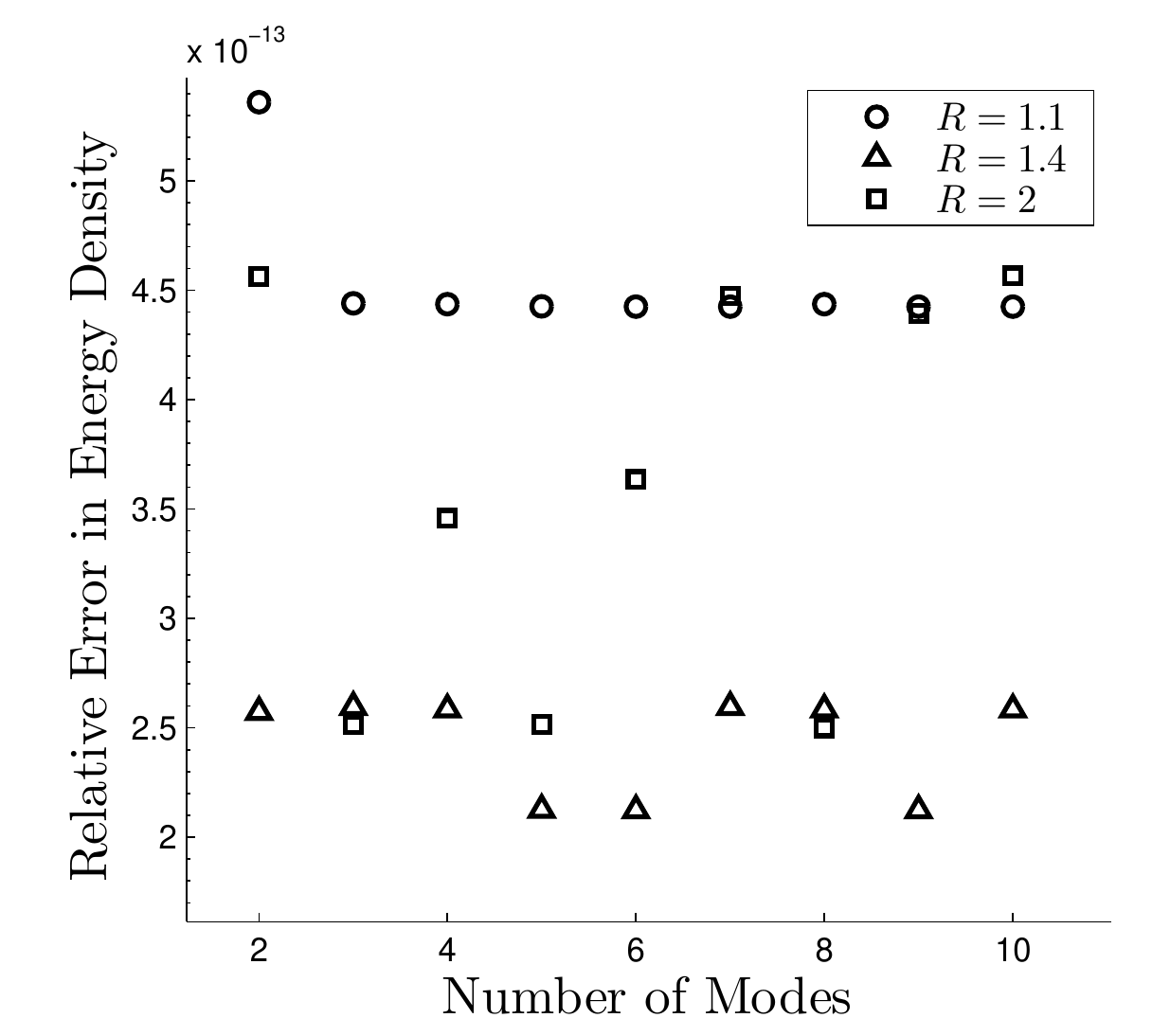}}
\caption{Maximum relative error in energy density.}\label{fig:keq_E_err}
 \end{minipage}
 \end{figure}

\begin{figure}[H]
\begin{minipage}[t]{0.5\linewidth}
\centerline{\includegraphics[height=6cm]{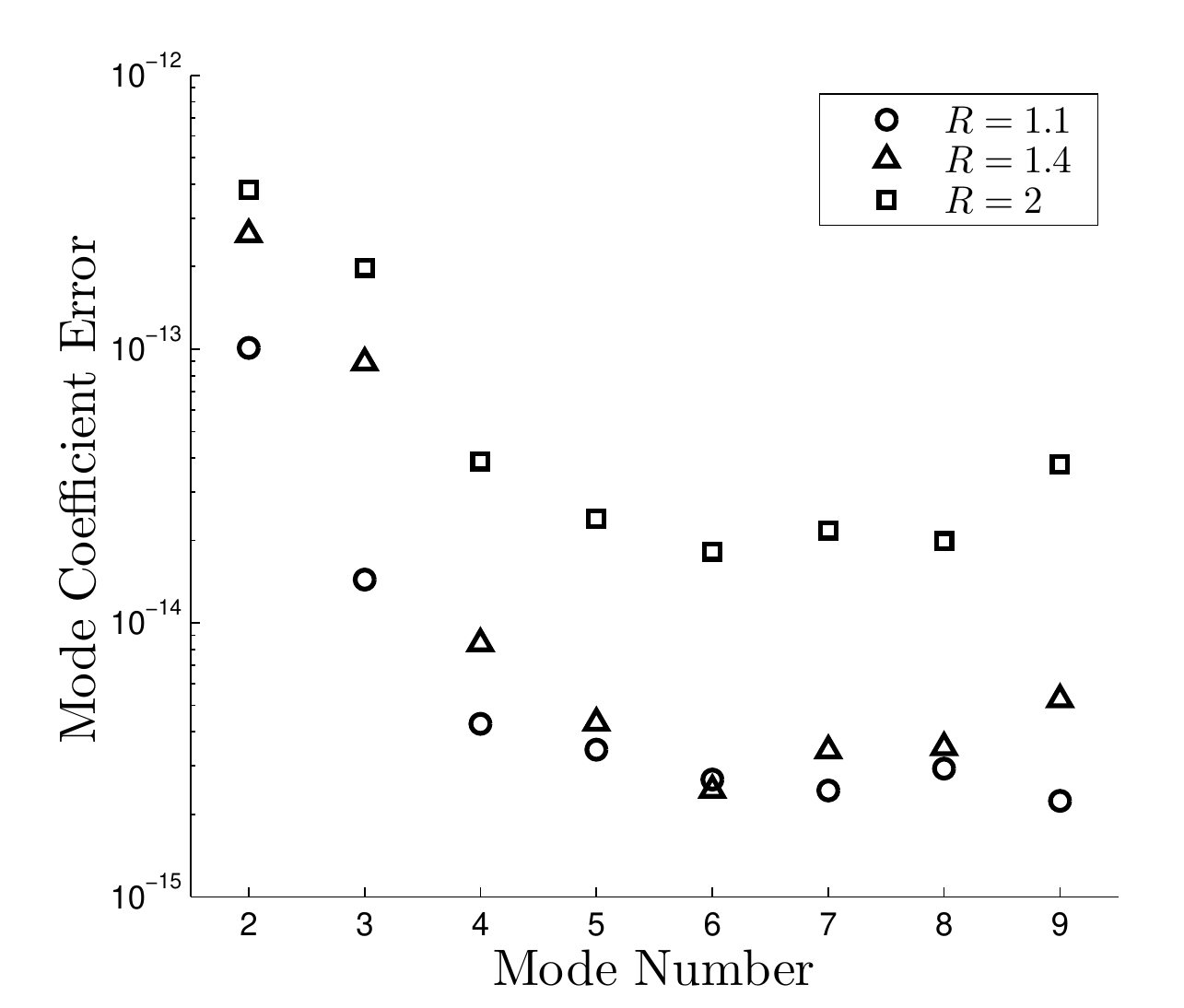}}
\caption{Maximum error in mode coefficients.}\label{fig:keq_b_err}
 \end{minipage}
 \hspace{0.5cm}
 \begin{minipage}[t]{0.5\linewidth}
\centerline{\includegraphics[height=6.1cm]{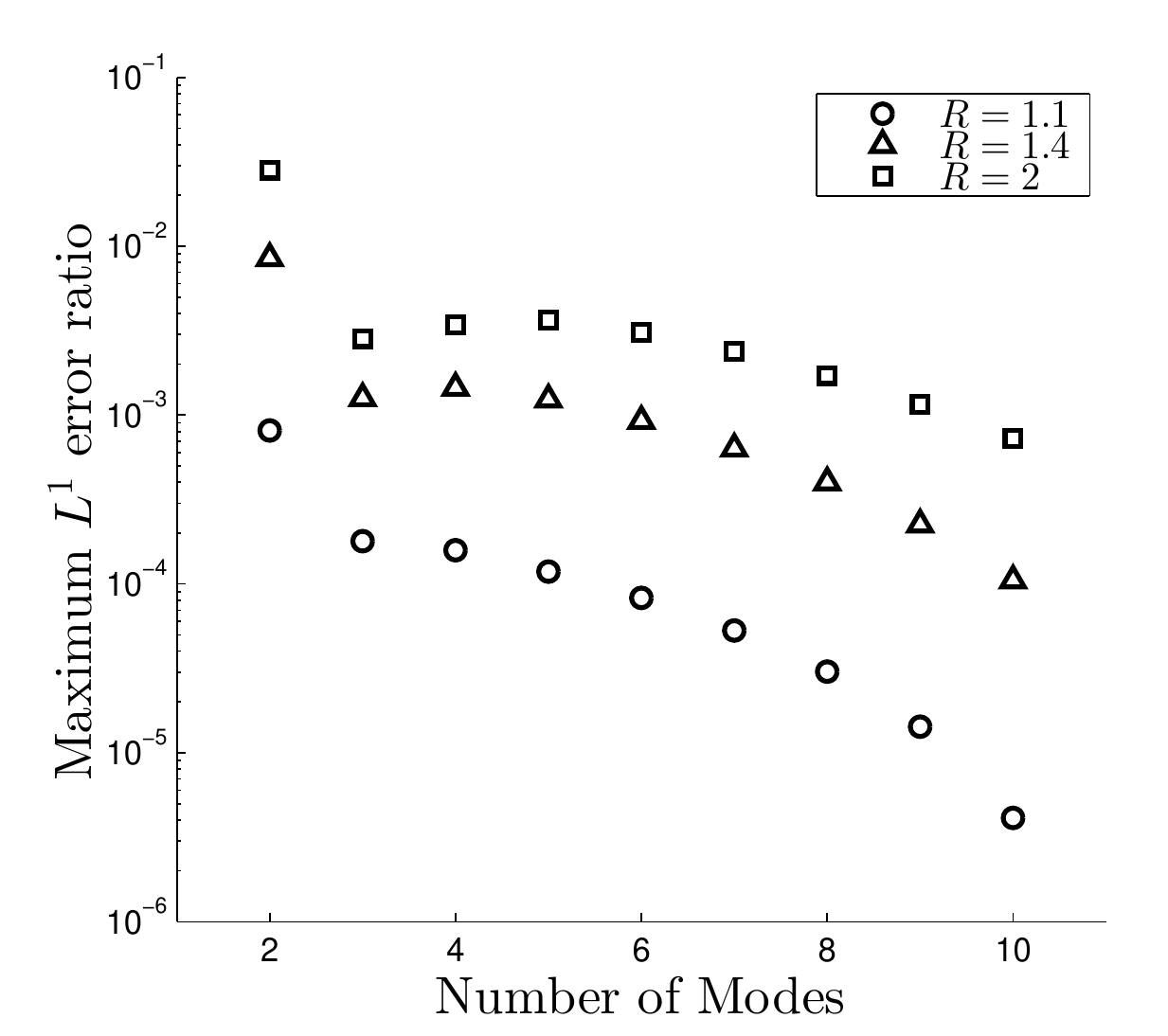}}
\caption{Maximum ratio of $L^1$ error between computed and exact solutions to $L^1$ norm of the exact solution.}\label{fig:keq_L1_err}
 \end{minipage}
 \end{figure}

 To show that the numerical integration accurately captures the mode coefficients of the exact solution, \req{exact_sol}, we give the error in the computed mode coefficients \req{mode_err_def}, where the evolution of the system was computed using $N=10$ modes, in figure \ref{fig:keq_b_err}.

In figure \ref{fig:keq_L1_err} we show the error between the approximate and exact solutions, computed as in \req{f_err} for $N=2,...,10$ and $R=1.1$, $R=1.4$, and $R=2$ respectively.  For most mode numbers and $R$ values, the error using $2$ modes is substantially less than the error from the chemical equilibrium method using $4$ modes.  The result is most dramatic for the case of large reheating, $R=2$, where the spurious oscillations from the chemical equilibrium solution are absent, as seen in figure \ref{fig:keq_approx_Tr_2}, as compared to the chemical equilibrium method in figure \ref{fig:free_stream_approx_T_r_2}.  Note that we plot from $z\in [0,15]$ in comparison to $y\in[0,30]$ in figure \ref{fig:keq_approx_Tr_2} due to the relation $z=y/R$ as discussed in section \ref{basis_comparison}. Additionally, the error no longer increases as $t\rightarrow\infty$, as it did for the chemical equilibrium method, see figure \ref{fig:keq_L1_err_time}.  In fact it decreases since the exact solution approaches chemical equilibrium at a reheated temperature and hence can be better approximated by $f_\Upsilon$. 

In summary, in addition to the reduction in the computational cost when going from $4$ to $2$ modes, we also reduce the error compared to the chemical equilibrium method, all while still capturing the number and energy densities.  We emphasize that the error in the number and energy densities is limited by the integration tolerance and {\emph not} the number of modes (so long as the neglected, higher modes have a negligible impact on the first two modes, as is often the case).  

\begin{figure}[H]
 \begin{minipage}[t]{0.5\linewidth}
\centerline{\includegraphics[height=6.2cm]{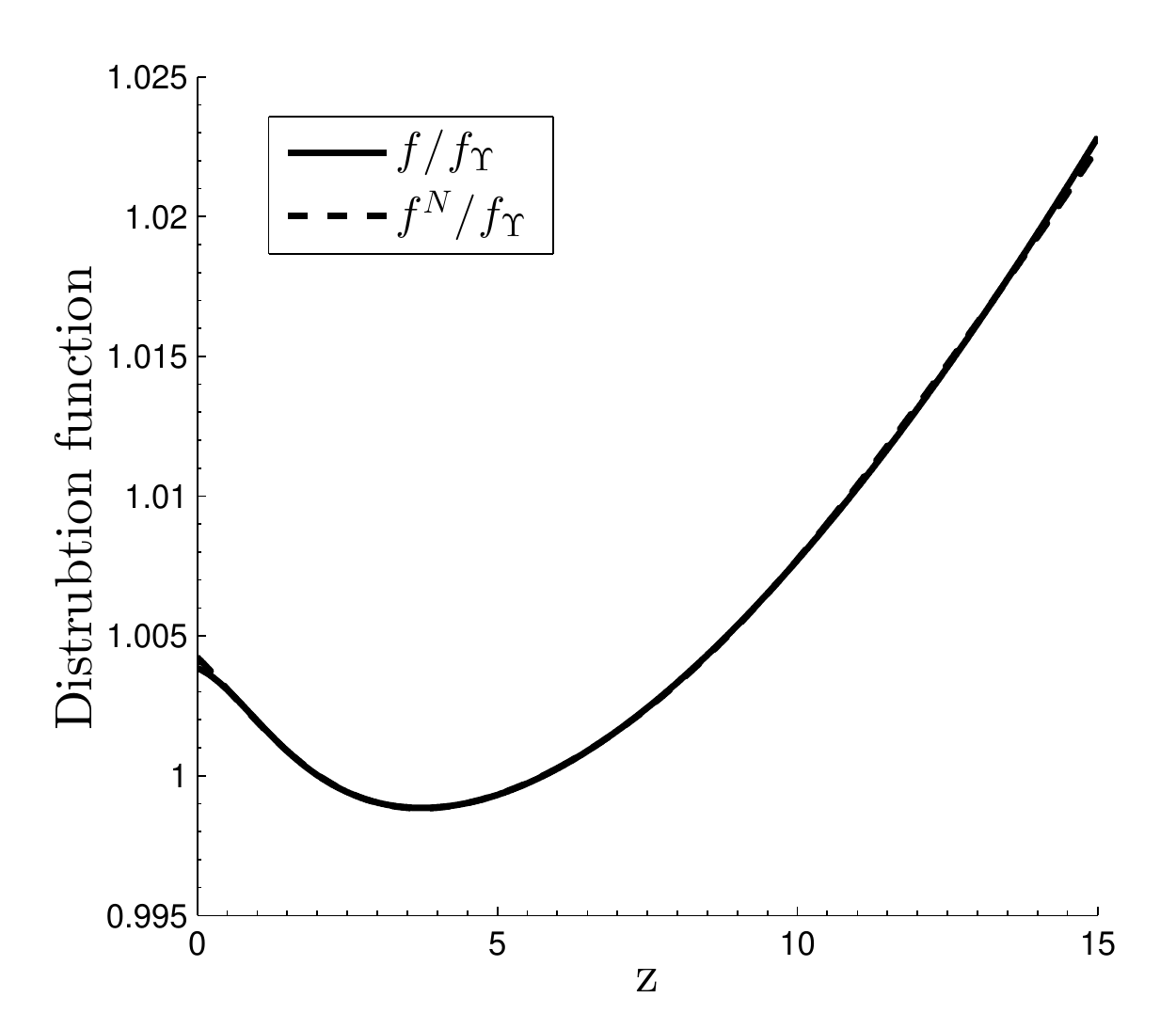}}
\caption{Approximate and exact solution for $R=2$ obtained with two modes.}\label{fig:keq_approx_Tr_2}
 \end{minipage}
 \hspace{0.5cm}
 \begin{minipage}[t]{0.5\linewidth}
\centerline{\includegraphics[height=6.2cm]{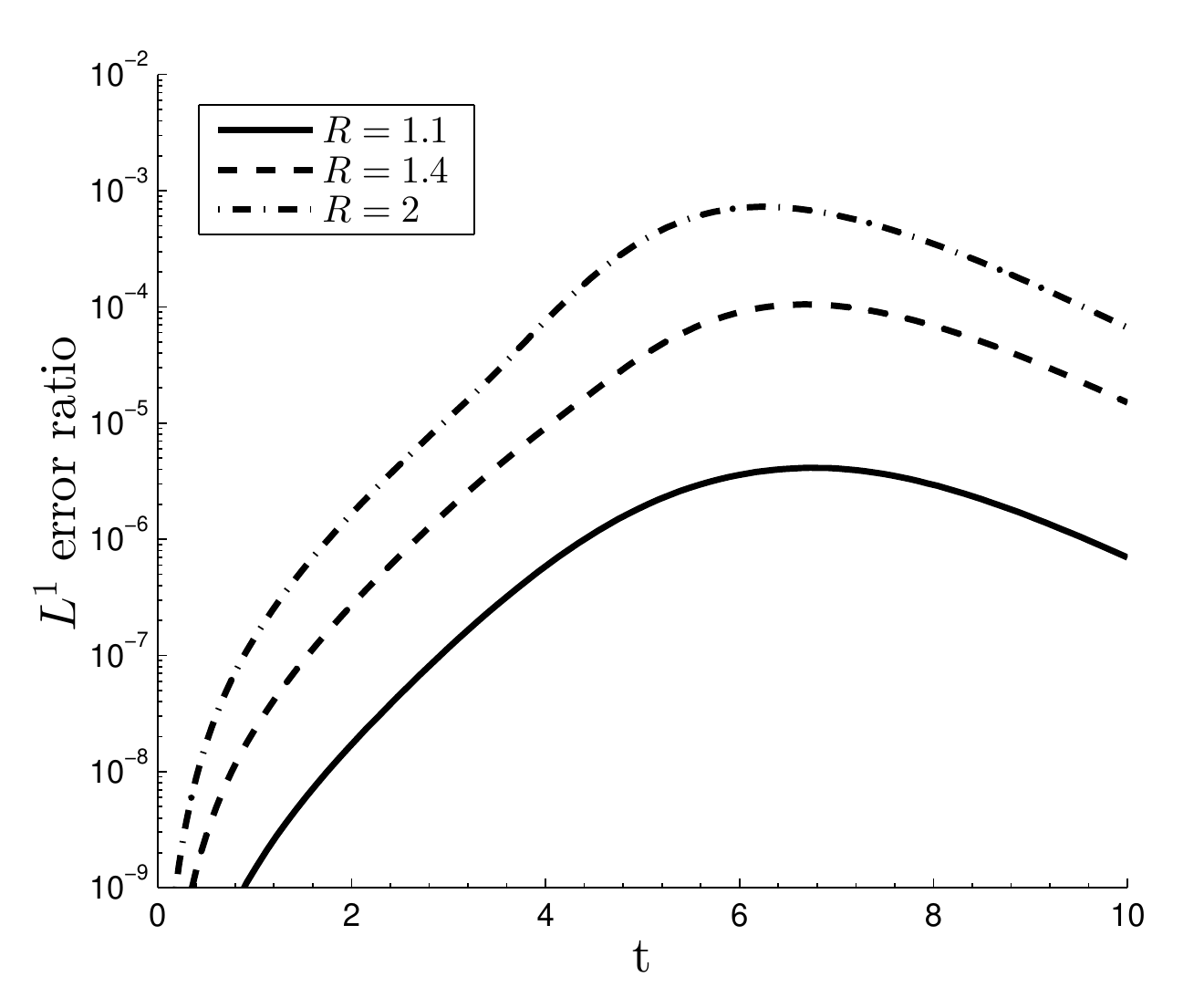}}
\caption{$L^1$ error ratio as a function of time for $n=10$ modes.}\label{fig:keq_L1_err_time}
 \end{minipage}
 \end{figure}

\subsection{Chemical Non-equilibrium Attractor}
The model problem above with $\Upsilon=1$  is reminiscent of the way that chemical equilibrium can emerge from chemical non-equilibrium in practice; the distribution of interest is attracted to some chemical equilibrium distribution but the particle creation/annihilation processes are not able to keep up with the momentum exchange and maintain an equilibrium particle yield, and so a fugacity $\Upsilon<1$ develops.  However, in order to isolate the effects of the fugacity on the solutions, we will now solve \req{toy_eq} under the  condition where our distribution is attracted to a fixed chemical non-equilibrium distribution.  More specifically, we take $T_{eq}(t)=1/a(t)$ and fugacities $\Upsilon=1.5$, $\Upsilon=0.9$, $\Upsilon=0.75$, and $\Upsilon=0.5$ with $a(t)$ defined as in \req{a_T_def}.

\begin{figure}[H]
 \begin{minipage}[t]{0.5\linewidth}
\centerline{\includegraphics[height=6.1cm]{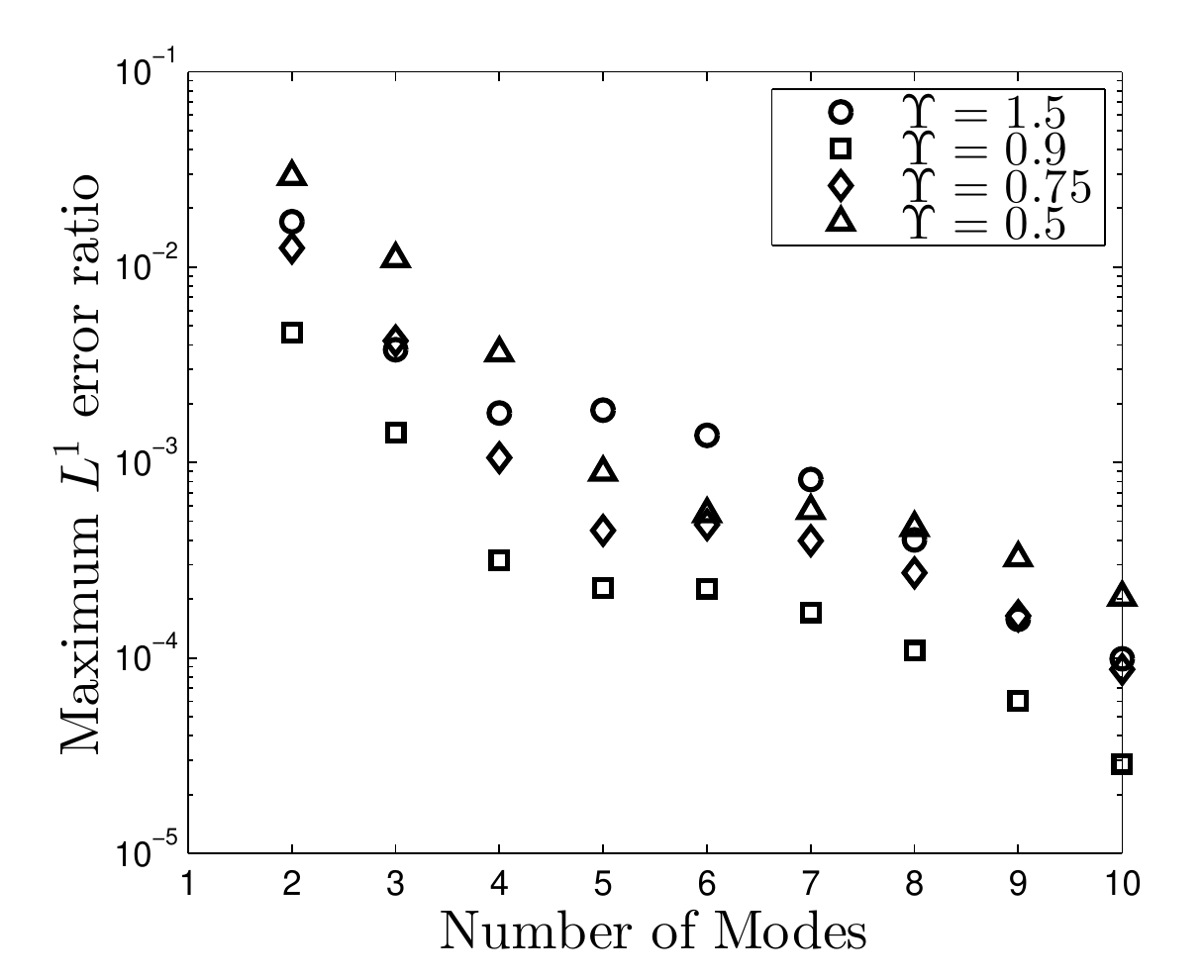}}
\caption{Maximum ratio  of $L^1$ error between computed and exact solutions to $L^1$ norm of the exact solution using the chemical equilibrium method.}\label{fig:free_stream_L1_err_Ups}
 \end{minipage}
 \hspace{0.5cm}
 \begin{minipage}[t]{0.5\linewidth}
\centerline{\includegraphics[height=6.1cm]{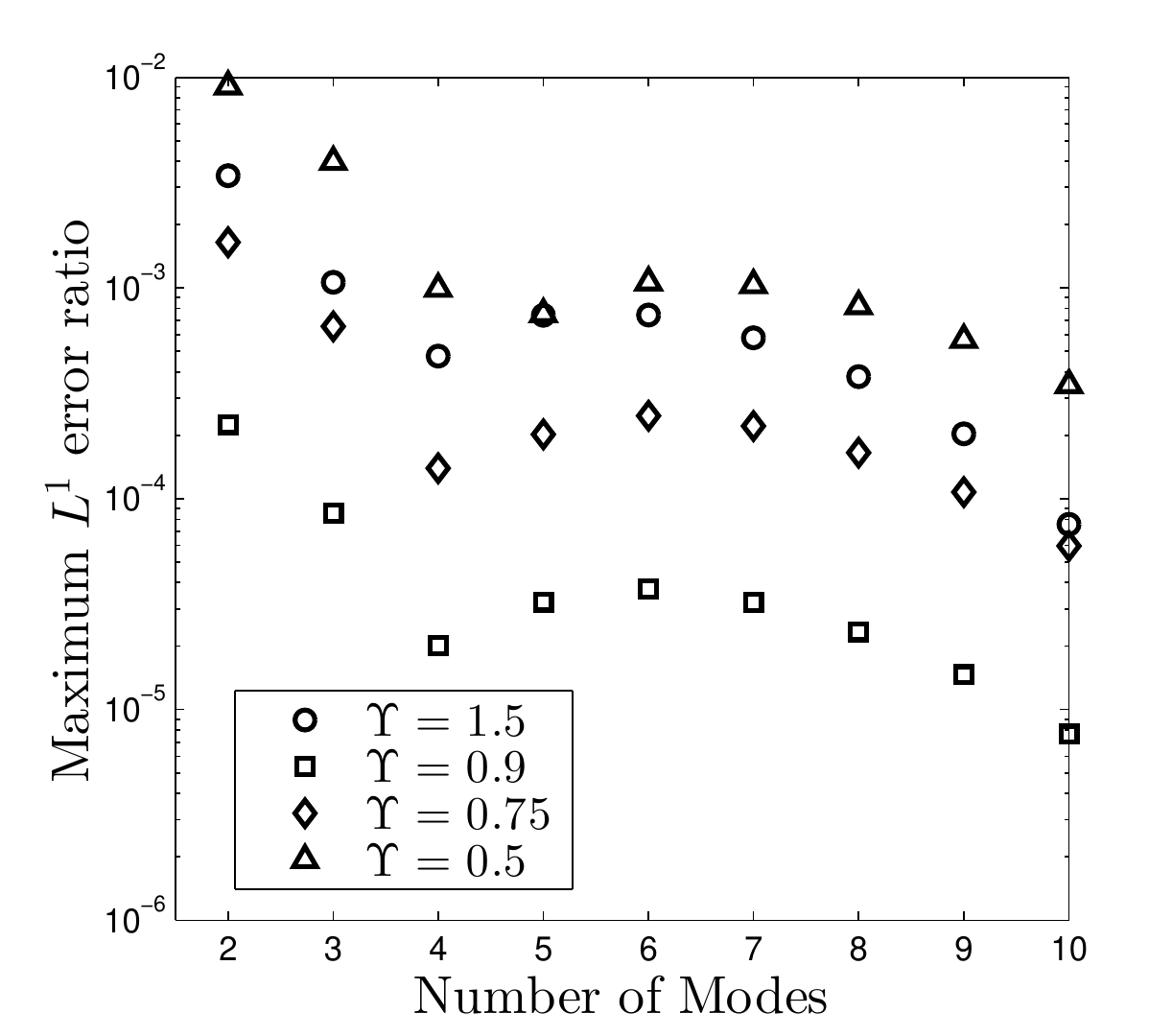}}
\caption{Maximum ratio of $L^1$ error between computed and exact solutions to $L^1$ norm of the exact solution using the chemical non-equilibrium method.}\label{fig:k_eq_L1_err_Ups}
 \end{minipage}
 \end{figure}

The behavior of the energy and number density errors are essentially the same as in the reheating test presented above and the mode coefficients are accurately captured, so we do not show these quantities here.  Instead  we show the maximum $L^1$ error, computed as in \req{f_err} in figure \ref{fig:free_stream_L1_err_Ups} for the chemical equilibrium method and in figure \ref{fig:k_eq_L1_err_Ups} for the chemical non-equilibrium method.  The error when using the latter method with only two terms is comparable to the former with four terms.

In figure \ref{fig:k_eq_L1_err_Ups_final} we show the final value of the relative $L^1$ error,
\begin{equation}\label{f_err}
\text{error}_{n}^f= \frac{\int |f-\tilde{f}|dy}{\int |f|dy}.
\end{equation}
where the distribution functions are evaluated at $t=t_f=10$. This figure is for the chemical non-equilibrium method applied to the chemical non-equilibrium attractor with tol=$10^{-10}$ -- the error tolerance was chosen small enough to disentangle integration error from the error due to the number of modes.  The chemical non-equilibrium ansatz is able to represent the final asymptotic state accurately and so the final error is much smaller than the maximum error in figure \ref{fig:k_eq_L1_err_Ups}.  For the chemical equilibrium ansatz the final error is nearly the same as the maximum error and so we don't show it here.

\begin{figure}[H]
\centerline{\includegraphics[height=6.1cm]{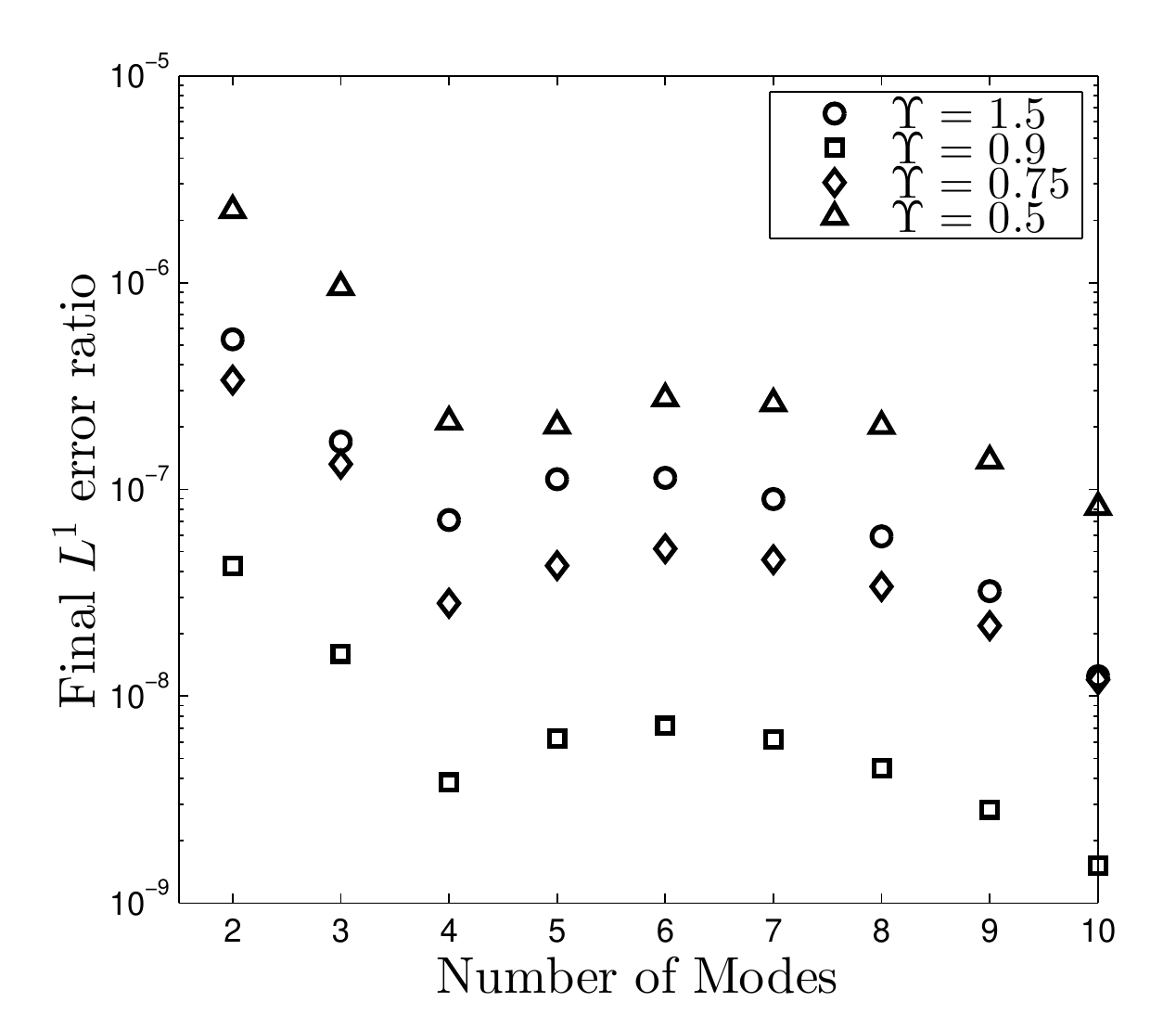}}
\caption{Final value of the ratio  of $L^1$ error between computed and exact solutions to $L^1$ norm of the exact solution using the chemical non-equilibrium method.}\label{fig:k_eq_L1_err_Ups_final}
 \end{figure}

\section{Summary and Outlook}
We have presented a spectral method for solving the Relativistic Boltzmann equation for a system  of massless fermions  diluting in time based on a dynamical basis of orthogonal polynomials.  The method is adapted to systems evolving near kinetic equilibrium, but allows for potentially strong chemical non-equilibrium in a transient and/or final state as well as strong reheating i.e. decoupling of temperature scaling from dilution scaling.  

The method depends on two time dependent parameters, the effective temperature $T(t)$ and phase space occupancy or fugacity $\Upsilon(t)$, whose dynamics are isolated by the requirement that the lowest modes capture the energy and particle number densities.  This gives the method a natural physical interpretation.  In particular, the dynamical fugacity is capable of naturally expressing the emergence of chemical non-equilibrium during the freeze-out process while the effective temperature captures any reheating phenomenon. Any system in approximate kinetic equilibrium that undergoes reheating and/or transitions to chemical non-equilibrium is a good match for this method.  In fact it is almost assured that our method will be considerably more computationally economical than the chemical equilibrium spectral method for any physical system in which the cost of computing the collision terms is high.

We validated the method on a model problem that exhibits the physical characteristics of reheating and chemical non-equilibrium.  We demonstrated that particle number and energy densities are captured accurately using only two degrees of freedom, the effective temperature and fugacity.  In general, this will hold so long as the back reaction from non-thermal distortions is small i.e. as long as kinetic equilibrium is a good approximation. 

The method presented here should be compared to the spectral method used in \cite{Esposito2000,Mangano2002}, which uses a fixed basis of orthogonal polynomials and is adapted to systems that are close to chemical equilibrium with dilution temperature scaling.  In addition to more closely mirroring the physics of systems that exhibit reheating and chemical non-equilibrium, the method presented here has a computational advantage over the chemical equilibrium method. Even when the system is close to chemical equilibrium with dilution temperature scaling, as is the case for the problem studied in \cite{Esposito2000,Mangano2002}, the method presented here reduces the minimum number of degrees of freedom needed to capture the particle number and energy densities from four to two.  In turn, this reduces the minimum number of collision integrals that must be evaluated by more than half.  

Numerical evaluation of collision operators for realistic interactions is a costly operation and so the new `emergent chemical non-equilibrium' approach we have presented here constitutes a significant reduction in the numerical cost of obtaining solutions. Moreover, even if the chemical equilibrium approach were to be properly modified to gain mathematical advantages we show  in our chemical non-equilibrium approach, it is not at all clear that the chemical equilibrium method can, with comparable numerical effort, achieve  a precise solution under conditions where transient or final chemical non-equilibrium and reheating are strong.

\begin{subappendices}
\section{Orthogonal Polynomials}\label{orthopoly_app}
\subsection{Generalities}\label{ortho-general}
Let $w:(a,b)\rightarrow [0,\infty)$ be a weight function where $(a,b)$ is a (possibly unbounded) interval and consider the Hilbert space $L^2(w dx)$.   We will consider weights such that $x^n\in L^2(wdx)$ for all $n\in\mathbb{N}$. We denote the inner product by $\langle\cdot,\cdot\rangle$, the norm by $||\cdot||$, and for a vector $\psi\in L^2$ we let $\hat{\psi}\equiv \psi/||\psi||$.  The classical three term recurrence formula can be used to define a set of orthonormal polynomials $\hat{\psi}_i$ using this weight function, for example see \cite{Olver},
\begin{align}\label{poly_recursion}
&\psi_0=1, \hspace{2mm} \psi_1=||\psi_0||(x-\langle x\hat\psi_0,\hat{\psi}_0\rangle)\hat{\psi}_0,\\
&\psi_{n+1}=||\psi_n||\left[\left(x-\langle x\hat\psi_n,\hat\psi_n\rangle\right)\hat\psi_n-\langle x\hat\psi_n,\hat\psi_{n-1}\rangle\hat\psi_{n-1}\right].
\end{align}
One can also derive recursion relations for the derivatives of $\psi_n$ with respect to $x$, denoted with a prime,
\begin{align}\label{deriv_rec}
&\psi_0^{'}=0, \hspace{2mm} \hat{\psi}_1^{'}=\frac{||\psi_0||}{||\psi_1||}\hat\psi_0,\\
&\hat{\psi}_{n+1}^{'}=\frac{||\psi_n||}{||\psi_{n+1}||}\left[\hat\psi_n+\left(x-\langle x\hat\psi_n,\hat\psi_n\rangle\right)\hat{\psi}_n^{'}-\langle x\hat\psi_n,\hat\psi_{n-1}\rangle\hat{\psi}_{n-1}^{'}\right].
\end{align}
Since $\hat{\psi}_n^{'}$ is a degree $n-1$ polynomial, we have the expansion 
\begin{equation}
\hat{\psi}_n^{'}=\sum_{k<n} a_n^k \hat{\psi}_k.
\end{equation}
Using \req{deriv_rec} we obtain a recursion relation for the $a_n^k$
\begin{equation}
a_{n+1}^k=\frac{||\psi_n||}{||\psi_{n+1}||}\left(\delta_{n,k}-\langle x\hat\psi_n,\hat\psi_{n}\rangle a_n^k-\langle x\hat\psi_n,\hat\psi_{n-1}\rangle a_{n-1}^k+\sum_{l<n}^la_n^l\langle x\hat\psi_l,\hat\psi_k\rangle\right),\notag
\end{equation}
\begin{equation}
a_1^0=\frac{||\psi_0||}{||\psi_1||}.
\end{equation}

\subsection{Parametrized Families of Orthogonal Polynomials}\label{ortho-polynom-fam}
Our method requires not just a single set of orthogonal polynomials, but rather a parametrized family of orthogonal polynomial generated by a weight function $w_t(x)$ that is a $C^1$ function of both $x\in(a,b)$ and some parameter $t$.  To emphasize this, we write $g_t(\cdot,\cdot)$ for $\langle\cdot,\cdot\rangle$.  We will assume that $\partial_t w$ is dominated by some $L^1(dx)$ function of $x$ only that decays exponentially as $x\rightarrow\pm\infty$ (if the interval is unbounded). In particular, this holds for the weight function \req{weight}.

Given the above assumption about the decay of $\partial_t w$, the dominated convergence theorem implies that $\langle p,q\rangle$ is a $C^1$ function of $t$ for all polynomials $p$ and $q$ and justifies  differentiation under the integral sign. By induction, it also implies implies that the $\hat\psi_i$ have coefficients that are $C^1$ functions of $t$. Therefore, for any polynomials $p$, $q$ whose coefficients are $C^1$ functions of $t$ we have
\begin{equation}
\frac{d}{dt}g_t( p,q)=\dot{g}_t(p,q)+g_t(\dot{p},q)+g_t( p,\dot{q})
\end{equation}
where a dot denotes differentiation with respect to $t$ and we use $\dot{g}_t(\cdot,\cdot)$ to denote the inner product with respect to the weight $\dot{w}$.  

\req{b_eq} for the mode coefficients requires us to compute $g(\dot{\hat\psi}_i,\hat\psi_j)$.  Differentiating the relation
\begin{equation}
\delta_{ij}=g_t(\hat\psi_i,\hat\psi_j)
\end{equation}
yields
\begin{equation}\label{ortho_deriv_eq}
0=\dot g_t(\hat\psi_i,\hat\psi_j)+g_t(\dot{\hat\psi}_i,\hat\psi_j)+g_t(\hat\psi_i,\dot{\hat\psi}_j).
\end{equation}
For $i=j$ we obtain
\begin{equation}\label{norm_deriv_eq}
g_t(\dot{\hat\psi}_i,\hat\psi_i)=-\frac{1}{2}\dot{g}_t(\hat\psi_i,\hat\psi_i).
\end{equation}
For $i<j$, $\dot{\hat\psi}_i$ is a degree $i$ polynomial and so it is orthogonal to $\hat\psi_j$. Therefore \req{ortho_deriv_eq} simplifies to
\begin{equation}
g_t(\dot{\hat\psi}_i,\hat\psi_j)=-\dot{g}_t(\hat\psi_i,\hat\psi_j),\hspace{2mm} i\neq j.
\end{equation}

\subsection{Proof of Lower Triangularity}\label{lower_triang}
Here we prove that the matrices that define the dynamics of the mode coefficients $b^k$ are lower triangular.  Recall the definitions
\begin{align}
A^k_i(\Upsilon)\equiv&\langle\frac{z}{f_\Upsilon }\hat\psi_i\partial_zf_\Upsilon ,\hat\psi_k\rangle+\langle z\partial_z \hat\psi_i,\hat\psi_k\rangle,\\
B^k_i(\Upsilon)\equiv &\Upsilon\left(\langle\frac{1}{f_\Upsilon }\frac{\partial f_\Upsilon }{\partial\Upsilon}\hat\psi_i,\hat\psi_k\rangle+\langle\frac{\partial\hat{\psi}_i}{\partial \Upsilon},\hat\psi_k\rangle\right).
\end{align}
Using integration by parts, we see that
\begin{equation}
A^k_i=-3\langle\hat\psi_i,\hat\psi_k\rangle-\langle \hat \psi_i,z\partial_z\hat\psi_k\rangle.
\end{equation}
Since $\hat\psi_i$ is orthogonal to all polynomials of degree less than $i$ we have $A^k_i=0$ for  $k<i$.  

$B^k_i$ can be simplified as follows.  First differentiate 
\begin{equation}
\delta_{ik}=\langle \hat\psi_i,\hat\psi_j\rangle
\end{equation}
with respect to $\Upsilon$ to obtain
\begin{align}
0=&\int \hat\psi_i\hat\psi_k\partial_{\Upsilon}wdz+\langle \partial_{\Upsilon}\hat\psi_i,\hat\psi_k\rangle+\langle \hat\psi_i,\partial_{\Upsilon}\hat\psi_k\rangle\\
=&\langle\frac{\hat\psi_i}{f_\Upsilon}\partial_{\Upsilon}f_\Upsilon,\hat\psi_k \rangle+\langle\partial_{\Upsilon}\hat\psi_i,\hat\psi_k\rangle+\langle \hat\psi_i,\partial_{\Upsilon}\hat\psi_k\rangle
\end{align}
Therefore 
\begin{equation}
B^k_i=-\Upsilon\langle\hat\psi_i,\partial_{\Upsilon}\hat\psi_k\rangle.
\end{equation}
$\partial_\Upsilon \hat\psi_k$ is a degree $k$ polynomial, hence $B_i^k=0$ for $k<i$ as desired.

\end{subappendices}

\chapter{ Collision Integrals}\label{ch:coll_simp}

\section{Collision Integral Inner Products}\label{coll_simp_sec}
Having detailed our method for solving the Boltzmann equation in chapter \ref{ch:boltz_orthopoly}, we must now address the computation of  collision integrals for neutrino processes.   See also our paper \cite{Birrell:2014uka}. To solve for the mode coefficients using \req{b_eq}, we must evaluate the collision operator inner products
\begin{align}\label{collision_integrals}
R_k\equiv&\langle\frac{1}{f_\Upsilon E_1}C[f_1],\hat\psi_k\rangle=\int_0^\infty \hat\psi_k(z_1)C[f_1](z_1) \frac{z_1^2}{E_1}dz_1\\
=&\frac{1}{2}\int \hat\psi_k(z_1)\int\left[f_3(p_3)f_4(p^4)f^1(p_1)f^2(p_2)-f_1(p_1)f_2(p_2)f^3(p_3)f^4(p^4)\right]\\
&\hspace{30mm}\times S |\mathcal{M}|^2(s,t)(2\pi)^4\delta(\Delta p)\prod_{i=2}^4\frac{d^{3}p_i}{2(2\pi)^3E_i}\frac{z_1^2}{E_1}dz_1,\notag\\
=&\frac{2(2\pi)^3}{8\pi}T_1^{-3}\int G_k(p_1,p_2,p_3,p_4)S |\mathcal{M}|^2(s,t)(2\pi)^4\delta(\Delta p)\prod_{i=1}^4\frac{d^{3}p_i}{2(2\pi)^3E_i},\\
=&2\pi^2T_1^{-3}\int G_k(p_1,p_2,p_3,p_4)S |\mathcal{M}|^2(s,t)(2\pi)^4\delta(\Delta p)\prod_{i=1}^4 \delta_0(p_i^2-m_i^2)\frac{d^4p_i}{(2\pi)^3},\\
G_k=&\hat\psi_k(z_1)\left[f_3(p_3)f_4(p_4)f^1(p_1)f^2(p_2)-f_1(p_1)f_2(p_2)f^3(p_3)f^4(p_4)\right],\hspace{2mm} f^i=1- f_i.
\end{align}
Note that $R_k$ only uses information about the distributions at a single spacetime point, and so we can work in a local orthonormal basis for the momentum.  Among other things, this implies that $p^2=p^\alpha p^\beta\eta_{\alpha\beta}$ where $\eta$ is the Minkowski metric
\begin{equation}
\eta_{\alpha\beta}=\diag(1,-1,-1,-1).
\end{equation}

From \req{collision_integrals}, we see that a crucial aspect of our spectral method is the ability to numerically compute  integrals of the type
\begin{align}\label{coll_ip}
&M\equiv\int G(p_1,p_2,p_3,p_4) S |\mathcal{M}|^2(s,t) (2\pi)^4\delta(\Delta p)\prod_{i=1}^4 \delta_0(p_i^2-m_i^2)\frac{d^4p_i}{(2\pi)^3},\\
&G(p_1,p_2,p_3,p_4)=g_1(p_1)g_2(p_2)g_3(p_3)g_4(p_4)
\end{align}
for some functions $g_i$.

Even after eliminating the delta functions in \req{coll_ip}, we are still left with an $8$ dimensional integral.  To facilitate numerical computation, we must analytically reduce this expression down to fewer dimensions.  Fortunately, the systems we are interested in have a large amount of symmetry that we can utilize for this purpose.  

The distribution functions we are concerned with are isotropic in some frame defined  by a unit timelike vector $U$, i.e. they depend only on the four-momentum only through $p_i\cdot U$.  The same is true of the basis functions $\hat\psi_k$, and hence the  $g_i$ depend only on $p_i\cdot U$ as well.  In \cite{Madsen,Dolgov_Hansen} approaches are outlined that reduce integrals of this type down to $3$ dimensions.  We outline the method from \cite{Dolgov_Hansen}, as applied to our spectral method solver, in appendix \ref{app:dogov_method}.  However, the integrand one obtains from these methods is only piecewise smooth or has an integration domain with a complicated geometry.  This presents difficulties for the integration routine we employ, which utilizes adaptive mesh refinement to ensure the desired error tolerance.  We take an alternative approach that, for the scattering kernels found in $e^\pm$, neutrino interactions, reduces the problem to three iterated integrals (but not quite to a three dimensional integral) and results in an integrand with better smoothness properties.  In our comparison with the method in \cite{Dolgov_Hansen}, the resulting formula evaluates significantly faster under the numerical integration scheme we used.   The derivation presented expands on what is found in \cite{letessier2002hadrons}.

\section{Simplifying the Collision Integral}\label{coll_simp_sec}
Our strategy for simplifying the collision integrals is as follows.  We first make a change of variables designed to put the 4-momentum conserving delta function in a particularly simple form and allowing us to analytically use that delta function to reduce the integral from $16$ to $12$ dimensions.  The remaining four delta functions, which impose the mass shell constraints, are then seen to reduce to integration over a product of spheres.  The simple form of the submanifold that these delta function restict us to allows us to use the method in chapter \ref{ch:vol_forms} to analytically evaluate all four of the remaining delta functions simultaneously.  During this process, the isotropy of the system in the frame given by the 4-vector $U$ allows us to reduce the dimensionality further, by analytically evaluating several of the angular integrals. 

The change of variables that simplifies the 4-momentum conserving delta function is given by
\begin{equation}
p=p_1+p_2,\hspace{2mm} q=p_1-p_2, \hspace{2mm} p^{'}=p_3+p_4, \hspace{2mm} q^{'}=p_3-p_4.
\end{equation}
The Jacobian of this transformation is $1/2^{8}$.  Therefore using lemma \ref{diffeo_property} we find
{\small
\begin{align}\label{M_eq1}
M=\frac{1}{256(2\pi)^8 }\int & 1_{p^0>|q^0|} 1_{(p^{'})^0>|(q^{'})^0|} G((p+q)\cdot U/2,(p-q)\cdot  U/2,(p^{'}+q^{'})\cdot U/2, (p^{'}-q^{'})\cdot U/2)\notag\\ 
&\times S |\mathcal{M}|^2  \delta(p-p^{'})\delta((p+q)^2/4-m_1^2)\delta((p-q)^2/4-m_2^2)\delta((p^{'}+q^{'})^2/4-m_3^2)\notag\\
&\times\delta((p^{'}-q^{'})^2/4-m_4^2)d^4pd^4qd^4p^{'}d^4q^{'}.
\end{align}
}
First eliminate the integration over $p^{'}$ using $\delta(p-p^{'})$
 and then use Fubini's theorem to write
\begin{align}\label{use_fubini}
M=\frac{1}{256(2\pi)^8 }\int &\bigg[ \int G((p+q)\cdot U/2,(p-q)\cdot  U/2,(p^{'}+q^{'})\cdot U/2, (p^{'}-q^{'})\cdot U/2)\notag\\ 
&\times  1_{p^0>|q^0|} 1_{p^0>|(q^{'})^0|}S |\mathcal{M}|^2 \delta((p+q)^2/4-m_1^2)\delta((p-q)^2/4-m_2^2)\notag\\
&\times \delta((p+q^{'})^2/4-m_3^2)\delta((p-q^{'})^2/4-m_4^2)d^4qd^4q^{'}\bigg]d^4p.
\end{align}
Subsequent computations will justify this use of Fubini's theorem.

Since $p^0>0$ we have $dp\neq 0$ and so we can use the corollary of the coarea formula, \ref{dummy_int}, to decompose this into an integral over the center of mass energy $s=p^2$
{\small
\begin{align}\label{M_eq2}
M=\frac{1}{256(2\pi)^8 }\int_{s_0}^\infty&\int \delta(p^2-s) \left[\int 1_{p^0>|q^0|} 1_{p^0>|(q^{'})^0|}  S |\mathcal{M}|^2  F(p,q,q^{'}) \delta((p+q)^2/4-m_1^2)\right.\notag\\
&\times\delta((p-q)^2/4-m_2^2)\delta((p+q^{'})^2/4-m_3^2)\delta((p-q^{'})^2/4-m_4^2)d^4qd^4q^{'}\bigg]d^4pds,\notag\\
&F(p,q,q^{'})=G((p+q)\cdot U/2,(p-q)\cdot  U/2,(p+q^{'})\cdot U/2, (p-q^{'})\cdot U/2),\notag\\
&s_0=\max\{(m_1+m_2)^2,(m_3+m_4)^2\}.
\end{align}
}
The lower bound on $s$ comes from the fact that both $p_1$ and $p_2$ are future timelike and hence 
\begin{equation}
p^2=m_1^2+m_2^2+2p_1\cdot p_2\geq m_1^2+m_2^2+2m_1m_2=(m_1+m_2)^2.
\end{equation}
The other inequality is obtained by using $p=p^{'}$. 

Note that the integral in brackets in \req{M_eq2} is invariant under $SO(3)$ rotations of $p$ in the frame defined by $U$.  Therefore we obtain
\begin{align}\label{K_def}
M=&\frac{1}{256(2\pi)^8 }\int_{s_0}^\infty\int_0^\infty K(s,p)\frac{4\pi |\vec p|^2}{2p^0}d|\vec p|ds,\hspace{2mm} p^0=p\cdot U=\sqrt{|\vec p|^2+s},\\
K(s,p)=&\int 1_{p^0>|q^0|} 1_{p^0>|(q^{'})^0|}  S |\mathcal{M}|^2  F(p,q,q^{'}) \delta((p+q)^2/4-m_1^2)\delta((p-q)^2/4-m_2^2)\notag\\
&\times\delta((p+q^{'})^2/4-m_3^2)\delta((p-q^{'})^2/4-m_4^2)d^4qd^4q^{'}
\end{align}
where $|\vec p|$ denotes the norm of the spacial component of $p$ and in the formula for $K(s,p)$, $p$ is any four vector whose spacial component has norm $|\vec p|$ and timelike component $\sqrt{|\vec p|^2+s}$. Note that in integrating over $\delta(p^2-s)dp^0$, only the positive root was taken, due to the indicator functions in the $K(s,p)$.

We now simplify $K(s,p)$ for fixed but arbitrary $p$ and $s$ that satisfy $p^0=\sqrt{|\vec p|^2+s}$ and $s>s_0$.  These conditions imply $p$ is future timelike, hence we can we can change variables in  $q,q^{'}$ by an element of $SO(1,3)$ so that 
\begin{equation}
p=(\sqrt{s},0,0,0), \hspace{2mm} U=(\alpha,0,0,\delta)
\end{equation}
where
\begin{equation}
\alpha=\frac{p\cdot U}{\sqrt{s}}, \hspace{2mm} \delta=\frac{1}{\sqrt{s}}\left((p\cdot U)^2-s \right)^{1/2}.
\end{equation}
Note that the delta functions in the integrand imply $p\pm q$ is  timelike (or null if the corresponding mass is zero).  Therefore $p^0>\pm q^0$ iff $p\mp q$ is future timelike (or null).  This condition is preserved by $SO(1,3)$ hence $p^0>|q^0|$ in one frame iff it holds in every frame.  Similar comments apply to $p^0>|(q^{'})^0|$ and so $K(s,p)$ has the same formula in the transformed frame as well.

We now evaluate the measure that is induced by the delta functions, using the method given in chapter \ref{ch:vol_forms}.  We have the constraint function
\begin{equation}
\Phi(q,q^{'})=((p+q)^2/4-m_1^2,(p-q)^2/4-m_2^2,(p+q^{'})^2/4-m_3^2,(p-q^{'})^2/4-m_4^2)
\end{equation}
and must compute the solution set $\Phi(q,q^{'})=0$. Adding and subtracting the first two components and the last two respectively, we have the equivalent conditions
\begin{align}
\frac{s+q^2}{2}=m_1^2+m_2^2,\hspace{2mm} p\cdot q=m_1^2-m_2^2, \hspace{2mm}\frac{s+(q^{'})^2}{2}=m_3^2+m_4^2,\hspace{2mm} p\cdot q^{'}=m_3^2-m_4^2.
\end{align}
If we let $(q^0,\vec{q})$, $((q^{'})^0,\vec{q}^{'})$ denote the spacial components in the frame defined by $p=(\sqrt{s},0,0,0)$ we have another set of equivalent conditions
\begin{align}\label{coord_conditions}
&q^{0}=\frac{m_1^2-m_2^2}{\sqrt{s}},\hspace{2mm} |\vec{q}|^2=\frac{(m_1^2-m_2^2)^2}{s}+s-2(m_1^2+m_2^2),\\
&(q^{'})^{0}=\frac{m_3^2-m_4^2}{\sqrt{s}}, \hspace{2mm} |\vec{q}^{'}|^2=\frac{(m_3^2-m_4^2)^2}{s}+s-2(m_3^2+m_4^2).
\end{align}
Note that if these hold then using $s\geq s_0$ we obtain
\begin{equation}
\frac{|q^0|}{p^0}\leq \frac{|m_1^2-m_2^2|}{(m_1+m_2)^2}<1
\end{equation}
and similarly for $q^{'}$.  Hence the conditions in the indicator functions are satisfied and we can drop them from the formula for $K(s,p)$.

The conditions \req{coord_conditions} imply that our solution set is a product of spheres in $\vec{q}$ and $\vec{q}^{'}$, as long as the conditions are consistent i.e. so long as $|\vec{q}|,|\vec{q}^{'}|>0$. To see that this holds for almost every $s$, first note
\begin{equation}
\frac{d}{ds}|\vec{q}|^2=1-\frac{(m_1^2-m_2^2)^2}{s^2}>0
\end{equation}
since $s\geq (m_1+m_2)^2$.  At $s=(m_1+m_2)^2$, $|\vec{q}|^2=0$.  Therefore, for $s>s_0$ we have $|\vec{q}|>0$ and similarly for $q^{'}$.  Hence we have the result
\begin{equation}
\Phi^{-1}(0)=\{q^{0}\}\times B_{|\vec{q}|}\times \{(q^{'})^{0}\}\times B_{|\vec{q}^{'}|}.
\end{equation}
where $B_r$ denotes the radius $r$ ball centered at $0$.  We will parametrize this by spherical angular coordinates in $q$ and $q^{'}$. 

 We now compute the induced volume form.  First consider the differential 
\begin{equation}
 D\Phi=\left( \begin{array}{c}
\frac{1}{2}(q+p)^\alpha\eta_{\alpha\beta}dq^\beta \\
\frac{1}{2}(q-p)^\alpha\eta_{\alpha\beta}dq^\beta\\
\frac{1}{2}(q^{'}+p)^\alpha\eta_{\alpha\beta}dq^{'^\beta}  \\
\frac{1}{2}(q^{'}-p)^\alpha\eta_{\alpha\beta}dq^{'^\beta}  \end{array} \right).
\end{equation}
Evaluating this on the coordinate vector fields $\partial_{q^0}$, $\partial_r$ we obtain
\begin{equation}
 D\Phi(\partial_{q^0})=\left( \begin{array}{c}
\frac{1}{2}(q^0+\sqrt{s}) \\
\frac{1}{2}(q^0-\sqrt{s}) \\
0\\
0 \end{array} \right), \hspace{2mm}  D\Phi(\partial_{r})=\left( \begin{array}{c}
-\frac{1}{2}|\vec{q}| \\
-\frac{1}{2}|\vec{q}| \\
0\\
0 \end{array} \right)=\left( \begin{array}{c}
-\frac{1}{2}r \\
-\frac{1}{2}r \\
0\\
0 \end{array} \right).
\end{equation}
Similar results hold for $q^{'}$.  Therefore we have the determinant
\begin{equation}
\det\left( \begin{array}{cccc}
D\Phi(\partial_{q^0}) & D\Phi(\partial_{r}) & D\Phi(\partial_{(q^{'})^0}) & D\Phi(\partial_{r^{'}}) \end{array} \right)=\frac{s}{4}rr^{'}.
\end{equation}
Note that this determinant being nonzero implies that our use of Fubini's theorem in \req{use_fubini} was justified.

By \req{vol_form_coords} and \req{delta_def}, the above computations imply that the induced volume measure is
\begin{align}
&\delta((p+q)^2/4-m_1^2)\delta((p-q)^2/4-m_2^2)\delta((p+q^{'})^2/4-m_3^2)\delta((p-q^{'})^2/4-m_4^2)d^4qd^4q^{'}\\
=&\frac{4}{srr^{'}}i_{(\partial_{q^0},\partial_{r},\partial _{(q^{'})^0},\partial_{r^{'}})}\left(r^2\sin(\phi)dq^0drd\theta d\phi\right)\wedge\left((r^{'})^2\sin(\phi^{'})d(q^{'})^0dr^{'}d\theta^{'}d\phi^{'}\right)\\
=&\frac{4rr^{'}}{s}\sin(\phi)\sin(\phi^{'})d\theta d\phi d\theta^{'}d\phi^{'}
\end{align}
where
\begin{equation}
r=\frac{1}{\sqrt{s}}\sqrt{(s-(m_1+m_2)^2)(s-(m_1-m_2)^2)},\hspace{2mm} r^{'}=\frac{1}{\sqrt{s}}\sqrt{(s-(m_3+m_4)^2)(s-(m_3-m_4)^2)}.
\end{equation}

Consistent with our interest in the Boltzmann equation, we assume $F$ factors as
\begin{align}
 F(p,q,q^{'})=&F_{12}((p+q)\cdot U/2,(p-q)\cdot U/2)F_{34}((p+q^{'})\cdot U/2,(p-q^{'})\cdot U/2)\\
\equiv &G_{12}(p\cdot U,q\cdot U)G_{34}(p\cdot U,q^{'}\cdot U).
\end{align}
For now, we suppress the dependence on $p$, as it is not of immediate concern. In our chosen coordinates where $U=(\alpha,0,0,\delta)$ we have
\begin{equation}
q\cdot U=q^0\alpha-r\delta\cos(\phi)
\end{equation}
and similarly for $q^{'}$.
To compute
\begin{align}\label{K_angular1}
K(s,p)=\frac{4rr^{'}}{s}\int \left[\int S |\mathcal{M}|^2 (s,t) G_{34}\sin(\phi^{'})d\theta^{'} d\phi^{'}\right] G_{12}\sin(\phi)d\theta d\phi
\end{align}
first recall 
\begin{align}
t=&(p_1-p_3)^2=\frac{1}{4}(q- q^{'})^2=\frac{1}{4}(q^2+(q^{'})^2-2(q^0(q^{'})^0-\vec{q}\cdot \vec{q}^{'})),\\
\vec{q}\cdot\vec{q}^{'}&=rr^{'}(\cos(\theta-\theta^{'})\sin(\phi)\sin(\phi^{'})+\cos(\phi)\cos(\phi^{'})).
\end{align}
Together, these imply that the integral in brackets in  \req{K_angular1} equals
{\small
\begin{align}
&\int_0^\pi\int_0^{2\pi} S |\mathcal{M}|^2 (s,t(\cos(\theta-\theta^{'})\sin(\phi)\sin(\phi^{'})+\cos(\phi)\cos(\phi^{'})))\\
&\hspace{15mm}\times G_{34}((q^{'})^0\alpha-r^{'}\delta\cos(\phi^{'}))\sin(\phi^{'})d\theta^{'} d\phi^{'}\notag\\
=&\int_{-1}^1\int_0^{2\pi} S |\mathcal{M}|^2 (s,t(\cos(\psi)\sin(\phi)\sqrt{1-y^2}+\cos(\phi)y)) G_{34}((q^{'})^0\alpha-r^{'}\delta y)d\psi dy.
\end{align}
}

Therefore 
\begin{align}
K(s,p)=&\frac{8\pi rr^{'}}{s}\int_{-1}^1 \left[\int_{-1}^1\left(\int_0^{2\pi} S |\mathcal{M}|^2 (s,t(\cos(\psi)\sqrt{1-y^2}\sqrt{1-z^2}+yz))d\psi\right)\right.\\
&\hspace{26mm}\times G_{34}((q^{'})^0\alpha-r^{'}\delta y) dy\bigg] G_{12}(q^0\alpha-r\delta z)dz
\end{align}
where
\begin{align}\label{t_def}
t(x)=&\frac{1}{4}((q^0)^2-r^2+((q^{'})^0)^2-(r^{'})^2-2q^0(q^{'})^0+2rr^{'}x),\\
=&\frac{1}{4}((q^0-(q^{'})^0)^2-r^2-(r^{'})^2+2rr^{'}x).
\end{align}

\section{Electron and Neutrino Collision Integrals}\label{nu_matrix_elements}
In this section, we compute various integrals of the scattering matrix element that appear in the scattering kernels for processes involving $e^\pm$ and neutrinos.  For reference, we collect the important results from section \ref{coll_simp_sec} on evaluation of the scattering kernel integrals \req{collision_integrals}.

\begin{align}\label{M_simp}
M=\frac{1}{256(2\pi)^7 }\int_{s_0}^\infty&\int_0^\infty K(s,p)\frac{ p^2}{p^0}dpds,
\end{align}
\begin{align}\label{matrix_elem_int}
K(s,p)=&\frac{8\pi rr^{'}}{s}\int_{-1}^1 \left[\int_{-1}^1\left(\int_0^{2\pi} S |\mathcal{M}|^2 (s,t(\cos(\psi)\sqrt{1-y^2}\sqrt{1-z^2}+yz))d\psi\right)\right.\\
&\hspace{26mm}\times G_{34}((q^{'})^0\alpha-r^{'}\delta y) dy\bigg] G_{12}(q^0\alpha-r\delta z)dz.
\end{align}
where
\begin{align}\label{t_def}
p^0=&\sqrt{p^2+s},\hspace{2mm} \alpha=\frac{p^0}{\sqrt{s}}, \hspace{2mm} \delta=\frac{p}{\sqrt{s}},\hspace{2mm}q^{0}=\frac{m_1^2-m_2^2}{\sqrt{s}},\hspace{2mm}(q^{'})^{0}=\frac{m_3^2-m_4^2}{\sqrt{s}},\\
r=&\frac{1}{\sqrt{s}}\sqrt{(s-(m_1+m_2)^2)(s-(m_1-m_2)^2)},\\
 r^{'}=&\frac{1}{\sqrt{s}}\sqrt{(s-(m_3+m_4)^2)(s-(m_3-m_4)^2)},\\
t(x)=&\frac{1}{4}((q^0-(q^{'})^0)^2-r^2-(r^{'})^2+2rr^{'}x),\\
s_0=&\max\{(m_1+m_2)^2,(m_3+m_4)^2\}.
\end{align}
and
\begin{align}
 F(p,q,q^{'})=&F_{12}((p+q)\cdot U/2,(p-q)\cdot U/2)F_{34}((p+q^{'})\cdot U/2,(p-q^{'})\cdot U/2)\\
\equiv &G_{12}(p\cdot U,q\cdot U)G_{34}(p\cdot U,q^{'}\cdot U)\notag.
\end{align}

This is as far as we can simplify things without more information about the form of the matrix elements.  The  matrix elements for weak force scattering processes involving neutrinos and $e^\pm$ in the limit $|p|\ll M_W,M_Z$, taken from \cite{Dolgov_Hansen}, are as follows
\begin{table}[H]
\centering 
\begin{tabular}{|c|c|}
\hline
Process &$S|\mathcal{M}|^2$  \\
\hline
$\nu_e+\bar\nu_e\rightarrow\nu_e+\bar\nu_e$ & $128G_F^2(p_1\cdot p_4)(p_2\cdot p_3)$\\
\hline
$\nu_e+\nu_e\rightarrow\nu_e+\nu_e$ & $64G_F^2(p_1\cdot p_2)(p_3\cdot p_4)$\\
\hline
$\nu_e+\bar\nu_e\rightarrow\nu_j+\bar\nu_j$&$32G_F^2(p_1\cdot p_4)(p_2\cdot p_3)$\\
\hline
$\nu_e+\bar\nu_j\rightarrow\nu_e+\bar\nu_j$ & $32G_F^2(p_1\cdot p_4)(p_2\cdot p_3)$\\
\hline
$\nu_e+\nu_j\rightarrow\nu_e+\nu_j$&$32G_F^2(p_1\cdot p_2)(p_3\cdot p_4)$\\
\hline
$\nu_e+\bar\nu_e\rightarrow e^++e^-$ & $128G_F^2[g_L^2(p_1\cdot p_4)(p_2\cdot p_3)+g_R^2(p_1\cdot p_3)(p_2\cdot p_4)+g_Lg_Rm_e^2(p_1\cdot p_2)]$\\
\hline
$\nu_e+e^-\rightarrow\nu_e+e^-$ & $128G_F^2[g_L^2(p_1\cdot p_2)(p_3\cdot p_4)+g_R^2(p_1\cdot p_4)(p_2\cdot p_3)-g_Lg_Rm_e^2(p_1\cdot p_3)]$\\
\hline
$\nu_e+e^+\rightarrow\nu_e+e^+$ & $128G_F^2[g_R^2(p_1\cdot p_2)(p_3\cdot p_4)+g_L^2(p_1\cdot p_4)(p_2\cdot p_3)-g_Lg_Rm_e^2(p_1\cdot p_3)]$\\
\hline
\end{tabular}
\caption{Matrix elements for electron neutrino processes where $j=\mu,\tau$,  $g_L=\frac{1}{2}+\sin^2\theta_W$, $g_R=\sin^2\theta_W$, $\sin^2(\theta_W)\approx 0.23$ is the Weinberg angle, and $G_F=1.16637\times 10^{-5}\text{GeV}^{-2}$ is Fermi's constant.}
\label{table:nu_e_reac}
\end{table}

\begin{table}[H]
\centering 
\begin{tabular}{|c|c|}
\hline
Process &$S|\mathcal{M}|^2$  \\
\hline
$\nu_i+\bar\nu_i\rightarrow\nu_i+\bar\nu_i$ & $128G_F^2(p_1\cdot p_4)(p_2\cdot p_3)$\\
\hline
$\nu_i+\nu_i\rightarrow\nu_i+\nu_i$ & $64G_F^2(p_1\cdot p_2)(p_3\cdot p_4)$\\
\hline
$\nu_i+\bar\nu_i\rightarrow\nu_j+\bar\nu_j$&$32G_F^2(p_1\cdot p_4)(p_2\cdot p_3)$\\
\hline
$\nu_i+\bar\nu_j\rightarrow\nu_i+\bar\nu_j$ & $32G_F^2(p_1\cdot p_4)(p_2\cdot p_3)$\\
\hline
$\nu_i+\nu_j\rightarrow\nu_i+\nu_j$&$32G_F^2(p_1\cdot p_2)(p_3\cdot p_4)$\\
\hline
$\nu_i+\bar\nu_i\rightarrow e^++e^-$ & $128G_F^2[\tilde{g}_L^2(p_1\cdot p_4)(p_2\cdot p_3)+g_R^2(p_1\cdot p_3)(p_2\cdot p_4)+\tilde{g}_Lg_Rm_e^2(p_1\cdot p_2)]$\\
\hline
$\nu_i+e^-\rightarrow\nu_i+e^-$ & $128G_F^2[\tilde{g}_L^2(p_1\cdot p_2)(p_3\cdot p_4)+g_R^2(p_1\cdot p_4)(p_2\cdot p_3)-\tilde{g}_Lg_Rm_e^2(p_1\cdot p_3)]$\\
\hline
$\nu_i+e^+\rightarrow\nu_i+e^+$ & $128G_F^2[g_R^2(p_1\cdot p_2)(p_3\cdot p_4)+\tilde{g}_L^2(p_1\cdot p_4)(p_2\cdot p_3)-\tilde{g}_Lg_Rm_e^2(p_1\cdot p_3)]$\\
\hline
\end{tabular}
\caption{Matrix elements for $\mu$ and $\tau$ neutrino processes where $i=\mu,\tau$, $j=e,\mu,\tau$, $j\neq i$,  $\tilde{g}_L=g_L-1=-\frac{1}{2}+\sin^2\theta_W$, $g_R=\sin^2\theta_W$, $\sin^2(\theta_W)\approx 0.23$ is the Weinberg angle, and $G_F=1.16637\times 10^{-5}\text{GeV}^{-2}$ is Fermi's constant.}
\label{table:nu_mu_reac}
\end{table}
In the following subsections, we will analytically simplify \req{M_simp} for each of these processes as much as possible.

\subsection{$\nu\nu\rightarrow\nu\nu$ }
Using \req{Mandelstam}, the matrix elements for neutrino neutrino scattering can be simplified to
\begin{align}
\label{TA002}
S|\mathcal{M}|^2=C(p_1\cdot p_2)(p_3\cdot p_4)=C\frac{s^2}{4}
\end{align}
 where the coefficient $C$ is given in table \ref{table:nu_nu_coeff}.

\begin{table}[H]
\centering 
\begin{tabular}{|c|c|}
\hline
Process &$C$ \\
\hline
$\nu_i+\nu_i\rightarrow\nu_i+\nu_i,\hspace{2mm} i\in\{e,\mu,\tau\}$& $64 G_F^2$\\
\hline
$\nu_i+\nu_j\rightarrow\nu_i+\nu_j,\hspace{2mm} i\neq j, \hspace{1mm} i,j\in\{e,\mu,\tau\}$& $32 G_F^2$\\
\hline
\end{tabular}
\caption{Matrix element coefficients for neutrino neutrino scattering processes.}
\label{table:nu_nu_coeff}
\end{table}
From here we obtain

\begin{align}
K(s,p)=&\frac{8\pi rr^{'}}{s}\int_{-1}^1 \left[\int_{-1}^1\left(\int_0^{2\pi}S|\mathcal{M}|^2 (s,t(\cos(\psi)\sqrt{1-y^2}\sqrt{1-z^2}+yz))d\psi\right)\right.\notag\\
&\hspace{26mm}\times G_{34}((q^{'})^0\alpha-r^{'}\delta y) dy\bigg] G_{12}(q^0\alpha-r\delta z)dz\\
=& 4\pi^2 Crr^{'}s \int_{-1}^1 G_{12}(q^0\alpha-r\delta z)dz \int_{-1}^1G_{34}((q^{'})^0\alpha-r^{'}\delta y) dy.
\end{align}
{\small
\begin{align}
M_{\nu\nu\rightarrow\nu\nu}=&\frac{C}{256(2\pi)^5 }\int_{s_0}^\infty s^2\int_0^\infty \int_{-1}^1 G_{12}(q^0\alpha-r\delta z)dz \int_{-1}^1G_{34}((q^{'})^0\alpha-r^{'}\delta y) dy\frac{ p^2}{p^0}dpds,\\
=&\frac{C}{256(2\pi)^5 } T^8\!\!\!\int_{0}^\infty\tilde{s}^2\int_0^\infty  \left[\int_{-1}^1 \tilde{G}_{12}(-\tilde{p} z)dz \int_{-1}^1\tilde{G}_{34}(-\tilde{p} y) dy\right]\frac{\tilde{p}^2}{\tilde{p}^0}d\tilde{p}d\tilde{s}
\end{align}
}
where the tilde quantities are obtained by non-dimensionalizing via scaling by $T$. If we want to emphasize the role of $C$ then we write $M_{\nu\nu\rightarrow\nu\nu}(C)$.

\subsection{$\nu\bar\nu\rightarrow \nu\bar\nu$ }
Using \req{Mandelstam}, the matrix elements for neutrino anti-neutrino scattering can be simplified to
\begin{align}
S|\mathcal{M}|^2=C\left(\frac{s+t}{2}\right)^2
\end{align}
where the coefficient $C$ is given in table \ref{table:nu_nubar_coeff}.

\begin{table}[H]
\centering 
\begin{tabular}{|c|c|}
\hline
Process &$C$  \\
\hline
$\nu_i+\bar\nu_i\rightarrow\nu_i+\bar\nu_i,\hspace{2mm} i\in\{e,\mu,\tau\}$& $128 G_F^2$\\
\hline
$\nu_i+\bar\nu_i\rightarrow\nu_j+\bar\nu_j,\hspace{2mm} i\neq j, \hspace{1mm} i,j\in\{e,\mu,\tau\}$& $32 G_F^2$\\
\hline
$\nu_i+\bar\nu_j\rightarrow\nu_i+\bar\nu_j,\hspace{2mm} i\neq j, \hspace{1mm} i,j\in\{e,\mu,\tau\}$& $32 G_F^2$\\
\hline
\end{tabular}
\caption{Matrix element coefficients for neutrino neutrino scattering processes.}
\label{table:nu_nubar_coeff}
\end{table}
Using this we find
 \begin{align}
\int_0^{2\pi} S |\mathcal{M}|^2 (s,t(\cos(\psi)\sqrt{1-y^2}\sqrt{1-z^2}+yz))d\psi=&\frac{\pi C}{16} s^2(3+4 yz-y^2-z^2+3y^2z^2)\notag\\
\equiv&\frac{\pi C}{16} s^2q(y,z),
\end{align}

\begin{align}
K(s,p)=&\frac{\pi^2C}{2}s^2\int_{-1}^1 \left[\int_{-1}^1q(y,z)G_{34}(-p y) dy\right] G_{12}(-p z)dz,
\end{align}
\begin{align}
M_{\nu\bar\nu\rightarrow\nu\bar\nu}=&\frac{C}{2048(2\pi)^5 }T^8\!\!\!\int_0^\infty\!\!\!\int_0^\infty\!\!\! \tilde{s}^2\left[\int_{-1}^1\int_{-1}^1q(y,z)\tilde{G}_{34}(-\tilde{p} y) \tilde{G}_{12}(-\tilde{p} z)dydz\right]\frac{\tilde{p}^2}{\tilde{p}^0}d\tilde{p}d\tilde{s}.
\end{align}
 If we want to emphasize the role of $C$ then we write $M_{\nu\bar\nu\rightarrow\nu\bar\nu}(C)$. Note that due to the polynomial form of the matrix element integral, the double integral in brackets breaks into a linear combination of products of one dimensional integrals, meaning that the nesting of integrals is only three deep.

\subsection{$\nu\bar{\nu}\rightarrow e^+e^-$}\label{nu_nubar_int}
Using \req{Mandelstam}, the matrix elements for neutrino anti-neutrino annihilation into $e^\pm$ can be simplified to
\begin{align}
S|\mathcal{M}|^2=A\left(\frac{s+t-m_e^2}{2}\right)^2+B\left(\frac{m_e^2-t}{2}\right)^2+Cm_e^2\frac{s}{2}
\end{align}
where the coefficients $A,B,C$ are given in table \ref{table:nu_nubar_ee_coeff}.

\begin{table}[H]
\centering 
\begin{tabular}{|c|c|c|c|}
\hline
Process &$A$&$B$&$C$  \\
\hline
$\nu_e+\bar\nu_e\rightarrow e^++e^-$&$128G_F^2g_L^2$&$128G_F^2g_R^2$&$128G_F^2g_Lg_R$\\
\hline
$\nu_i+\bar\nu_i\rightarrow e^++e^-,\hspace{2mm} i\in\{\mu,\tau\}$&$128G_F^2\tilde g_L^2$&$128G_F^2g_R^2$&$128G_F^2\tilde g_Lg_R$\\
\hline
\end{tabular}
\caption{Matrix element coefficients for neutrino neutrino annihilation into $e^\pm$.}
\label{table:nu_nubar_ee_coeff}
\end{table}

  The integral of each of these terms is
 \begin{align}
&\int_0^{2\pi}\frac{(s+t(\psi)-m_e^2)^2}{4}d\psi=\frac{\pi}{16}s(3s-4m_e^2)+\frac{\pi}{4}s^{3/2}\sqrt{s-4m_e^2}yz\\
&-\frac{\pi}{16}s(s-4m_e^2)(y^2+z^2)+\frac{3\pi}{16}s(s-4m_e^2)y^2z^2,\notag\\
&\int_0^{2\pi} \frac{(m_e^2-t(\psi))^2}{4}d\psi=\frac{\pi}{16}s(3s-4m_e^2)-\frac{\pi}{4}s^{3/2}\sqrt{s-4m_e^2}yz\\
&-\frac{\pi}{16}s(s-4m_e^2)(y^2+z^2)+\frac{3\pi}{16}s(s-4m_e^2)y^2z^2,\\
&\int_0^{2\pi} m_e^2\frac{s}{2} d\psi=\pi m_e^2s.
\end{align}
Therefore 
{\small
\begin{align}
\int_0^{2\pi} S |\mathcal{M}|^2 (s,t(\psi))d\psi=&\frac{\pi}{16}s[3s(A+B)+4m_e^2(4C-A-B)]+\frac{\pi}{4}s^{3/2}\sqrt{s-4m_e^2}(A-B)yz\notag\\
&-\frac{\pi}{16}s(s-4m_e^2)(A+B)(y^2+z^2)+\frac{3\pi}{16}s(s-4m_e^2)(A+B)y^2z^2\notag\\
\equiv& \pi q(m_e,s,y,z).
\end{align}
\begin{align}
M_{\nu\bar\nu\rightarrow e^+e^-}=&\frac{1}{128(2\pi)^5 }\int_{4m_e^2}^\infty\int_0^\infty\!\!\!\sqrt{1-4m_e^2/s}\left[\int_{-1}^1\int_{-1}^1q(s,y,z,m_e)G_{34}(-(\sqrt{1-4m_e^2/s})p y)\right.\notag\\
&\hspace{68mm}\times G_{12}(-p z)dydz\bigg]\frac{ p^2}{p^0}dpds,\notag\\
=&\frac{T^8}{128(2\pi)^5 }\int_{4\tilde m_e^2}^\infty\int_0^\infty\!\!\!\sqrt{1-4\tilde m_e^2/\tilde s}\left[\int_{-1}^1 \int_{-1}^1q(\tilde s,y,z,\tilde m_e)\tilde G_{34}(-(\sqrt{1-4\tilde{m}_e^2/\tilde s})\tilde p y)\right.\notag\\
&\hspace{68mm}\tilde G_{12}(-\tilde p z)dydz\bigg] \frac{ \tilde p^2}{\tilde p^0}d\tilde pd\tilde s,
\end{align}
}
where $\tilde{m_e}=m_e/T$.  If we want to emphasize the role of $A,B,C$ then we write $M_{\nu\bar\nu\rightarrow e^+e^-}(A,B,C)$.  Note that this expression is linear in $(A,B,C)\in\mathbb{R}^3$. Also note that, under our assumptions that the distributions of $e^+$ and $e^-$ are the same,  the $G_{ij}$ terms that contain the product of $e^\pm$ distributions are even functions. Hence the term involving the integral of $yz$ vanishes by antisymmetry.

\subsection{ $\nu e^\pm\rightarrow \nu e^\pm$}
 Using \req{Mandelstam}, the matrix elements for neutrino $e^\pm$ scattering can be simplified to
\begin{align}
\label{TA002}
S|\mathcal{M}|^2=A\left(\frac{s-m_e^2}{2}\right)^2+B\left(\frac{s+t-m_e^2}{2}\right)^2+Cm_e^2\frac{t}{2}
\end{align}
 where the coefficients $A,B,C$ are given in table \ref{table:nu_e_coeff}.

\begin{table}[H]
\centering 
\begin{tabular}{|c|c|c|c|}
\hline
Process &$A$&$B$&$C$  \\
\hline
$\nu_e+e^-\rightarrow \nu_e+e^-$&$128G_F^2g_L^2$&$128G_F^2g_R^2$&$128G_F^2g_Lg_R$\\
\hline
$\nu_i+e^-\rightarrow \nu_i+e^-,\hspace{2mm} i\in\{\mu,\tau\}$&$128G_F^2\tilde g_L^2$&$128G_F^2g_R^2$&$128G_F^2\tilde g_Lg_R$\\
\hline
$\nu_e+e^+\rightarrow \nu_e+e^+$&$128G_F^2g_R^2$&$128G_F^2g_L^2$&$128G_F^2g_Lg_R$\\
\hline
$\nu_i+e^+\rightarrow \nu_i+e^+,\hspace{2mm} i\in\{\mu,\tau\}$&$128G_F^2 g_R^2$&$128G_F^2\tilde g_L^2$&$128G_F^2\tilde g_Lg_R$\\
\hline
\end{tabular}
\caption{Matrix element coefficients for neutrino $e^\pm$ scattering.}
\label{table:nu_e_coeff}
\end{table}

  The integral of each of these terms is
 \begin{align}
&\int_0^{2\pi}\frac{(s-m_e^2)^2}{4} d\psi=\pi\frac{(s-m_e^2)^2}{2},\\
&\int_0^{2\pi}\frac{(s+t(\psi)-m_e^2)^2}{4}d\psi=\frac{\pi}{16s^2}(s-m_e^2)^2(3m_e^4+2m_e^2s+3s^2)+\frac{\pi}{4s^2}(s-m_e^2)^3(s+m_e^2)yz,\notag\\
&-\frac{\pi}{16s^2}(s-m_e^2)^4(y^2+z^2)+\frac{3\pi}{16s^2}(s-m_e^2)^4y^2z^2,\\
&\int_0^{2\pi} m_e^2\frac{t(\psi)}{2}d\psi=-\frac{\pi}{2s}m_e^2(s-m_e^2)^2(1-yz).
\end{align}
Therefore we have
\begin{align}
\int_0^{2\pi} S |\mathcal{M}|^2 (s,t(\psi))d\psi=&\pi\left[\frac{A}{2}+\frac{B}{16s^2}(3m_e^4+2m_e^2s+3s^2)-\frac{C}{2s}m_e^2\right](s-m_e^2)^2\notag\\
&+\pi\left[\frac{B}{4s^2}(s-m_e^2)(s+m_e^2)+\frac{C}{2s}m_e^2\right](s-m_e^2)^2yz\notag\\
&-B\frac{\pi}{16s^2}(s-m_e^2)^4(y^2+z^2)+B\frac{3\pi}{16s^2}(s-m_e^2)^4y^2z^2\notag\\
\equiv& \pi q(m_e,s,y,z)
\end{align}
and
\begin{align}\label{matrix_elem_int}
K(s,p)=&\frac{8\pi^2 rr^{'}}{s}\int_{-1}^1 \left[\int_{-1}^1 q(m_e,s,y,z) G_{34}((q^{'})^0\alpha-r^{'}\delta y) dy\right] G_{12}(q^0\alpha-r\delta z)dz,\\
r=r^{'}=&\frac{s-m_e^2}{\sqrt{s}},\hspace{2mm} q^0=(q^{'})^0=-\frac{m_e^2}{\sqrt{s}},\hspace{2mm} \delta=\frac{p}{\sqrt{s}},\hspace{2mm}  \alpha=\frac{p^0}{\sqrt{s}}.
\end{align}

\begin{align}\label{M_simp}
M_{\nu e\rightarrow\nu e}=&\frac{1}{128(2\pi)^5 }\int_{m_e^2}^\infty\!\int_0^\infty (1-m_e^2/s)^2\left(\int_{-1}^1 \int_{-1}^1 q(m_e,s,y,z) G_{34}((q^{'})^0\alpha-r^{'}\delta y)\right.\\
&\hspace{68mm}\times  G_{12}(q^0\alpha-r\delta z)dydz\bigg)\frac{ p^2}{p^0}dpds.\notag
\end{align}
As above, after scaling all masses by $T$, we obtain a prefactor of $T^8$. If we want to emphasize the role of $A,B,C$ then we write $M_{\nu e\rightarrow\nu e}(A,B,C)$.  Note that this expression is also linear in $(A,B,C)\in\mathbb{R}^3$.

\subsection{Total Collision Integral}
We now give the total collision integrals for neutrinos.    In the following, we indicate which distributions are used in each of the four types of scattering integrals discussed above by using the appropriate subscripts. For example, to compute $M_{\nu_e\bar\nu_\mu\rightarrow\nu_e\bar\nu_\mu}$  we set $G_{1,2}=\hat\psi_jf^1f^2$, $G_{3,4}=f_3f_4$, $f_1= f_{\nu_e}$, $f_3=f_{\nu_e}$, and $f_2=f_4=f_{\bar\nu_\mu}$ in the expression for $M_{\nu\bar\nu\rightarrow\nu\bar\nu}$ from section \ref{nu_nubar_int} and then, to include the reverse direction of the process, we must {\emph subtract}  the analogous expression whose only difference is $G_{1,2}=\hat\psi_jf_1f_2$, $G_{3,4}=f^3f^4$.
With this notation the collision integral for $\nu_e$ is
\begin{align}\label{M_tot}
M_{\nu_e}=&[M_{\nu_e\nu_e\rightarrow\nu_e\nu_e}+M_{\nu_e\nu_\mu\rightarrow\nu_e\nu_\mu}+M_{\nu_e\nu_\tau\rightarrow\nu_e\nu_\tau}]\\
&+[M_{\nu_e\bar\nu_e\rightarrow\nu_e\bar\nu_e}+M_{\nu_e\bar\nu_e\rightarrow\nu_\mu\bar\nu_\mu}+M_{\nu_e\bar\nu_e\rightarrow\nu_\tau\bar\nu_\tau}+M_{\nu_e\bar\nu_\mu\rightarrow\nu_e\bar\nu_\mu}+M_{\nu_e\bar\nu_\tau\rightarrow\nu_e\bar\nu_\tau}]\notag\\
&+M_{\nu_e\bar\nu_e\rightarrow e^+e^-}+[M_{\nu_e e^-\rightarrow\nu_e e^-}+M_{\nu_e e^+\rightarrow\nu_e e^+}].
\end{align}

Symmetry among the interactions implies that the distributions of $\nu_\mu$ and $\nu_\tau$ are equal.  We also neglect the small matter anti-matter asymmetry and so we take the distribution of each particle to be equal to that of the corresponding antiparticle.  Therefore there are only three independent distributions, $f_{\nu_e}$, $f_{\nu_\mu}$, and $f_e$ and so we can combine some of the terms in \req{M_tot} to obtain
\begin{align}
M_{\nu_e}=&M_{\nu_e\nu_e\rightarrow\nu_e\nu_e}(64G_F^2)+M_{\nu_e\nu_\mu\rightarrow\nu_e\nu_\mu}(2\times 32 G_F^2)+M_{\nu_e\bar\nu_e\rightarrow\nu_e\bar\nu_e}(128G_F^2)  \\
&+M_{\nu_e\bar\nu_e\rightarrow\nu_\mu\bar\nu_\mu}(2\times 32G_F^2)+M_{\nu_e\bar\nu_\mu\rightarrow\nu_e\bar\nu_\mu}(2\times 32 G_F^2)\notag\\
&+M_{\nu_e\bar\nu_e\rightarrow e^+e^-}(128G_F^2g_L^2,128G_F^2g_R^2,128G_F^2g_Lg_R)\notag\\
&+M_{\nu_e e\rightarrow\nu_e e}(128 G_F^2( g_L^2+g_R^2),128 G_F^2 (g_L^2+ g_R^2),256G_F^2g_Lg_R)\notag.
\end{align}
Introducing one more piece of notation, we use a subscript $k$ to denote the orthogonal polynomial basis element that multiplies $f_1$ or $f^1$ in the inner product.  The inner product of the $k$th basis element with the total scattering operator for electron neutrinos is therefore 
\begin{align}
R_k=&2\pi^2T^{-3} M_{k,\nu_e}.
\end{align}
Under these same assumptions and conventions, the total collision integral for the combined $\nu_\mu$, $\nu_\tau$ distribution (which we label $\nu_\mu$) is
\begin{align}
M_{\nu_\mu}=&M_{\nu_\mu\nu_\mu\rightarrow\nu_\mu\nu_\mu}(64G_F^2+32G_F^2)+M_{\nu_\mu\nu_e\rightarrow\nu_\mu\nu_e}(32 G_F^2)+M_{\nu_\mu\bar\nu_\mu\rightarrow\nu_\mu\bar\nu_\mu}(128G_F^2+32G_F^2+32G_F^2) \notag \\
&+M_{\nu_\mu\bar\nu_\mu\rightarrow\nu_e\bar\nu_e}(32G_F^2)+M_{\nu_\mu\bar\nu_e\rightarrow\nu_\mu\bar\nu_e}( 32 G_F^2)\notag\\
&+M_{\nu_\mu\bar\nu_\mu\rightarrow e^+e^-}(128G_F^2\tilde g_L^2,128G_F^2g_R^2,128G_F^2\tilde g_Lg_R)\notag\\
&+M_{\nu_\mu e\rightarrow\nu_\mu e}(128 G_F^2(\tilde  g_L^2+g_R^2),128 G_F^2 (\tilde g_L^2+ g_R^2),256G_F^2\tilde g_Lg_R),\\
R_k=&2\pi^2T^{-3} M_{k,\nu_\mu}.
\end{align}


\subsection{ Neutrino Freeze-out Test}
Now that we have the above expressions for the neutrino scattering integrals, we can abandon the model problem used to test our method in chapter \ref{ch:boltz_orthopoly} and compare the chemical equilibrium and non-equilibrium methods on the problem of neutrino freeze-out using the full $2$-$2$ scattering kernels for neutrino processes.  We solve the Boltzmann equation, \req{boltzmann}, for both the electron neutrino distribution and the combined $\mu$, $\tau$ neutrino distribution, including all of the  processes outlined above in the scattering operator, together with the Hubble equation for $a(t)$, \req{Hubble_eq}.  The total energy density  appearing in the Hubble equation consists of the contributions from both independent neutrino distributions as well as chemical equilibrium $e^\pm$ and photon distributions at some common temperature $T_\gamma$, all computed using \req{moments}.  The dynamics of $T_\gamma$ are fixed by the divergence freedom condition of the total stress energy tensor, \req{divTmn}, implied by Einstein's equations.  In addition, we include the QED corrections to the $e^\pm$ and photon equations of state from appendix \ref{app:QED_corr}.

To compare our results with Ref.~\cite{Mangano2005}, where neutrino freeze-out was simulated using $\sin^2(\theta_W)=0.23$ and $\eta=\eta_0$, in table \ref{table:method_comp} we present $N_\nu$ together with the following quantities
\begin{align}
 z_{fin}=T_\gamma a,\hspace{2mm}  \rho_{\nu 0}=\frac{7}{120}\pi^2a^{-4}, \hspace{2mm}  \delta\bar\rho_{\nu}= \frac{\rho_\nu}{\rho_{\nu 0}}-1.
\end{align}
This quantities were introduced in Ref.~\cite{Mangano2005}, but some additional discussion of their significance is in order.  The normalization of the scale factor $a$ is chosen so that at the start of the computation $T_\gamma=1/a$.  This means that $1/a$ is the temperature of a (hypothetical) particle species that is completely decoupled throughout the computation.  Here we will call it the free-streaming temeprature.  

 $z_{fin}$ is the ratio of photon temperature to the free-streaming temperature.  It is a measure of the amount of reheating that photons underwent due to the annihilation of $e^\pm$.  For completely decoupled neutrinos, whose temperature is the free-streaming temperature, the well known value can be computed from conservation of entropy
\begin{equation}
z_{fin}=(11/4)^{1/3}\approx 1.401.
\end{equation}
For coupled neutrinos, one expects this value to be slightly reduced, due to the  transfer of some entropy from annihilating $e^\pm$ into neutrinos. This is reflected in table \ref{table:method_comp}.

 $\rho_{\nu0}$ is the energy density of a massless fermion with two degrees of freedom and temperature equal to the free-streaming temperature.  In other words, it is the energy density of a single neutrino species, assuming it decoupled before reheating. Consequently, $\delta\bar\rho_\nu$ is the fractional increase in the energy density of a coupled neutrino species, due to its participation in reheating.

We compute the above using both the chemical equilibrium and non-equilibrium methods. For the following results, we used $\sin^2(\theta_W)=0.23$ and $\eta=\eta_0$. 
\begin{table}[h]\label{table:method_comp}
\centering 
\begin{tabular}{|c|c|c|c|c|c|}
\hline
Method &Modes&$z_{fin}$ & $\delta\bar\rho_{\nu_e}$&   $\delta\bar\rho_{\nu_{\mu,\tau}}$ & $N_{\nu}$  \\
\hline
Chemical Eq& 4 &1.39785 &0.009230 &0.003792 &3.044269\\
\hline
Chemical Non-Eq& 2&1.39784 &0.009269 & 0.003799&3.044383 \\
\hline
Chemical Non-Eq& 3&1.39785&0.009230 & 0.003791&3.044264 \\
\hline
\end{tabular}
\end{table}
We see that $\Delta N_\nu\equiv N_\nu-3$ agrees to $2$ digits and $4$ digits when using $2$ and $3$ modes respectively for the chemical non-equilibrium method, and similar behavior holds for the other quantities. Due to the reduction in the required number of modes, the chemical non-equilibrium method with the minimum number of required modes ($2$ modes) is more than $20\times$ faster than the chemical equilibrium method with its minimum number of required modes ($4$ modes), a very significant speed-up when the minimum number of modes meets the required precision.  The value of $N_\nu$ we obtain agrees with that found by \cite{Mangano2005}, up to their cited error tolerance of $\pm 0.002$.

\section{Conservation Laws and Scattering Integrals}
For some processes, some of the $R_k$'s vanish exactly.  As we now show, this is an expression of various conservation laws. First consider processes in which $f_1=f_3$ and $f_2=f_4$, such as kinetic scattering processes. Since $m_1=m_3$ and $m_2=m_4$ we have $r=r^{'}$, $q^0=(q^{'})^0$.  The scattering terms are all two dimensional integrals of some function of $s$ and $p$ multiplied by 

{\small
\begin{align}
I_k\equiv&\int_{-1}^1 \left[\int_{-1}^1\left(\int_0^{2\pi}S|\mathcal{M}|^2 (s,t(\cos(\psi)\sqrt{1-y^2}\sqrt{1-z^2}+yz))d\psi\right) f_1(h_1(y))f_2(h_2(y)) dy\right] \notag\\
&\hspace{26mm}\times f_k^1(h_1(z))f^2(h_2(z))dz\\
&-\int_{-1}^1 \left[\int_{-1}^1\left(\int_0^{2\pi}S|\mathcal{M}|^2 (s,t(\cos(\psi)\sqrt{1-y^2}\sqrt{1-z^2}+yz))d\psi\right) f^1(h_1(y))f^2(h_2(y)) dy\right] \notag\\
&\hspace{26mm}\times f_{1,k}(h_1(z))f_2(h_2(z))dz\\
h_1(y)=&(p^0+(q^{'})^0\alpha-r^{'}\delta y)/2,\hspace{2mm} h_2(y)=(p^0-q^0\alpha+r\delta y)/2, \hspace{2mm} f_{1,k}=\hat\psi_k f_1,\hspace{2mm} f^1_k=\hat\psi_k f^1.
\end{align}
}
For $k=0$, $\hat\psi_0$ is constant.  After factoring it out of $I_k$, the result is obviously zero and so $R_0=0$.  

We further specialize to a distribution scattering from itself i.e. $f_1=f_2=f_3=f_4$.  Since $m_1=m_2$ and $m_3=m_4$ we have $q^0=(q^{'})^0=0$ and
\begin{equation}
h_1(y)=(p^0-r^{'}\delta y)/2,\hspace{2mm} h_2(y)=(p^0+r\delta y)/2.
\end{equation}
 By the above, we know that $R_0=0$.  $\hat\psi_1$ appears in $I_1$ in the form $\hat\psi_1(h_1(z))$, a degree one polynomial in $z$.  Therefore $R_1$ is a sum of two terms, one which comes from the degree zero part and one from the degree one part.  The former is zero, again by the above reasoning.  Therefore, to show that $R_1=0$ we need only show $I_1=0$, except with $\hat\psi_1(h_1(z))$ replaced by  $z$.  Since $h_1(-y)=h_2(y)$, changing variables  $y\rightarrow -y$ and $z\rightarrow -z$ in the following shows that this term is equal to its own negative, and hence is zero
\begin{align}
&\int_{-1}^1 \left[\int_{-1}^1\left(\int_0^{2\pi}S|\mathcal{M}|^2 (s,t(\cos(\psi)\sqrt{1-y^2}\sqrt{1-z^2}+yz))d\psi\right) f_1(h_1(y))f_1(h_2(y)) dy\right] \notag\\
&\hspace{26mm}\times  zf^1(h_1(z))f^1(h_2(z))dz\\
&-\int_{-1}^1 \left[\int_{-1}^1\left(\int_0^{2\pi}S|\mathcal{M}|^2 (s,t(\cos(\psi)\sqrt{1-y^2}\sqrt{1-z^2}+yz))d\psi\right) f^1(h_1(y))f^1(h_2(y)) dy\right] \notag\\
&\hspace{26mm}\times zf_{1}(h_1(z))f_1(h_2(z))dz.
\end{align}
We note that the corresponding scattering integrals do not vanish for the chemical equilibrium spectral method.  This is another advantage of the method developed in chapter \ref{ch:boltz_orthopoly} and leads to a further reduction in cost of the method, beyond just the reduction in minimum number of modes.

Finally, we point out how the vanishing of these inner products is a reflection of certain conservation laws. From \req{n_div}, \req{collision_integrals}, and the fact that $\hat\psi_0,\hat\psi_1$ span the space of polynomials of degree $\leq 1$, we have the following expressions for the change in number density and energy density of a massless particle
\begin{align}
\frac{1}{a^3} \frac{d}{dt}(a^3n)=&\frac{g_p}{2\pi^2}\int \frac{1}{E}C[f]p^2dp=c_0 R_0,\\
\frac{1}{a^4}\frac{d}{dt}(a^4\rho)=&\frac{g_p}{2\pi^2}\int C[f] p^2dp=d_0R_0+d_1R_1\notag
\end{align}
for some $c_0,d_0,d_1$. Therefore, the vanishing of $R_0$ is equivalent to conservation of comoving particle number.  The vanishing of $R_0$ and $R_1$ implies $\rho\propto 1/a^4$ i.e. that the reduction in energy density is due entirely to redshift; energy is not lost from the distribution due to scattering.  These findings match the situations above where we found one or both of $R_0=0$, $R_1=0$.  $R_0$ vanished for all kinetic scattering processes and we know that all such processes conserve comoving particle number.  Both $R_0$ and $R_1$ vanished for a distribution scattering from itself and in such a process one expects that no energy is lost from the distribution by scattering, it is only redistributed among the particles corresponding to that distribution.

\section{Freeze-Out Temperature and Relaxation Time}
To make a connection with our development from chapter \ref{ch:model_ind}, we now give a definition of the kinetic freeze-out temperature that is applicable to the Boltzmann equation model. Any such definition will be only approximate, as the freeze-out process is not a sharp transition.  Our definition is motivated in part the treatment in \cite{kolb}. 

We first define a characteristic length between scatterings. Recalling the formula \req{n_div}, we obtain the fractional rate of change of comoving particle number
\begin{align}
\frac{\frac{d}{dt}(a^3 n)}{a^3n}=\frac{g_\nu}{2\pi^2n}\int C[f]p^2/Edp.
\end{align}
Here we don't want the net change, but rather to count the number of interactions.  For that reason, we imagine that only one direction of the process is operational and define the relaxation rate
\begin{align}
\Gamma\equiv\frac{g_\nu}{2\pi^2n}T^2\int \tilde C[f]zdz
\end{align}
where the one way collision is $\tilde C[f]$ is computed as in \req{coll} except with $F$ replaced by 
\begin{equation}
\tilde F=f_1(p_1)f_2(p_2)f^3(p_3)f^4(p_4).
\end{equation}
If particle type $1$ also participates in the reverse of the reaction $1+2\rightarrow 3+4$ then a corresponding term for the reverse reaction must also be added.  The key difference is there is no minus sign-we are counting reactions, not net particle number change.

Using the average velocity, which for neutrinos is $\bar v=c=1$, we obtain what we call the scattering length
\begin{align}
L&\equiv\frac{\bar v}{\Gamma}=\frac{\int_0^\infty\frac{1}{\Upsilon^{-1}e^z+1}z^2dz}{\int_0^\infty \tilde C[f] z^2/E dz}.
\end{align}
This can be compared to the Hubble length $L_H=c/H$ and the temperature at which $L=L_H$ we call the freeze-out temperature for that reaction.  Figure \ref{fig:scatt_length} shows the scattering length and $L_H$ for various types of neutrino reactions.  The solid line corresponds to the annihilation process $e^+e^-\rightarrow \nu\bar\nu$, the dashed line corresponds to the scattering $\nu e^\pm\rightarrow \nu e^\pm$, and the dot-dashed line corresponds to the combination of all processes involving only neutrinos.  The freeze-out temperatures in MeV are given in table \ref{table:freezeout_temp}.

\begin{table}[H]
\centering 
\begin{tabular}{|c|c|c|c|}
\hline
              & $e^+e^-\rightarrow \nu\bar\nu$ & $\nu e^\pm\rightarrow \nu e^\pm$ & $\nu$-only processes\\
\hline
$\nu_e$ &2.29 & 1.15&0.910\\
\hline
$\nu_{\mu,\tau}$ &3.83 & 1.78& 0.903\\
\hline
\end{tabular}
\caption{Freeze-out temperatures in MeV for electron neutrinos and for $\mu$,$\tau$ neutrinos.}
\label{table:freezeout_temp}
\end{table}

\begin{figure}[h]
\centerline{\includegraphics[height=6cm]{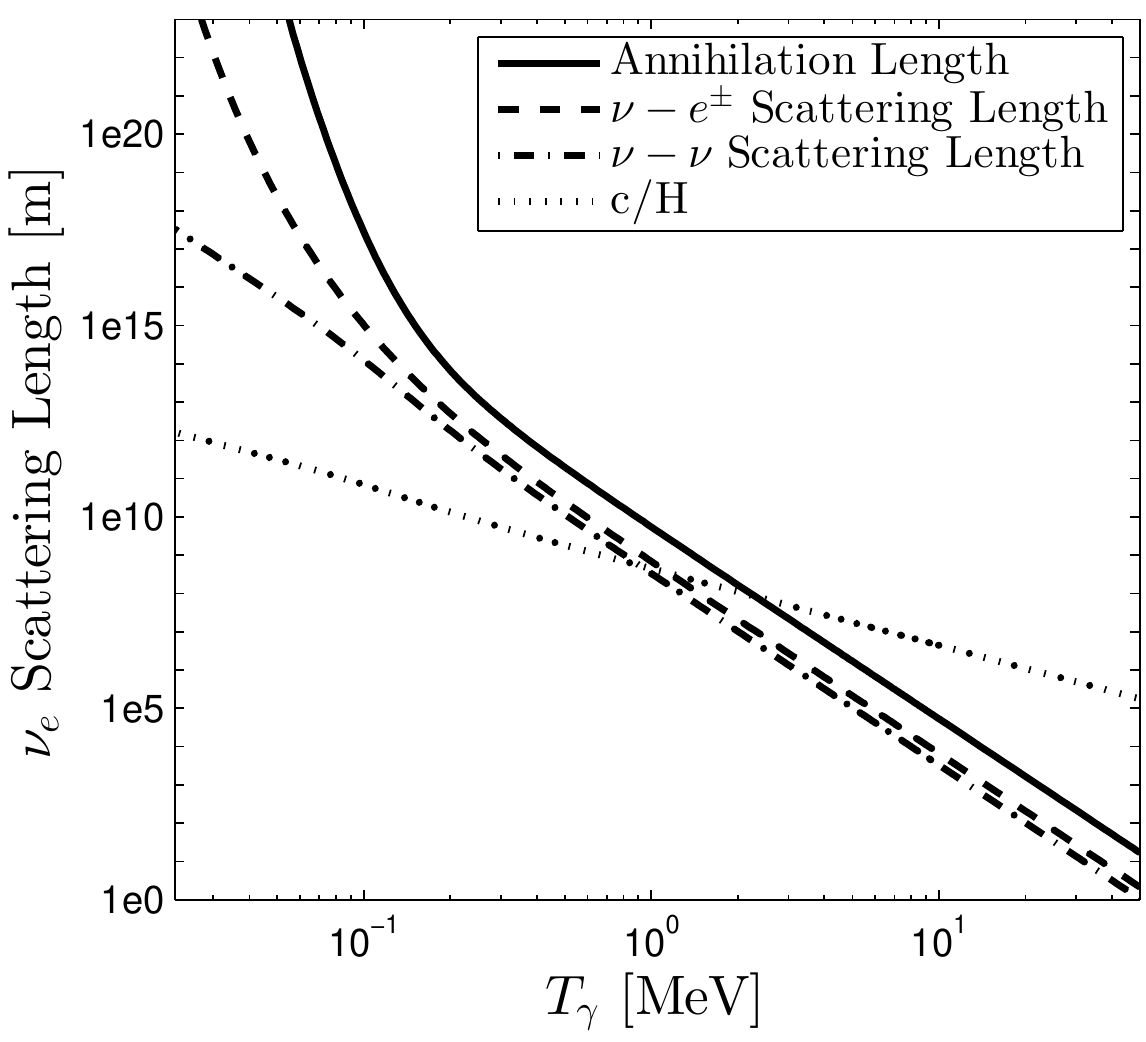}\includegraphics[height=6cm]{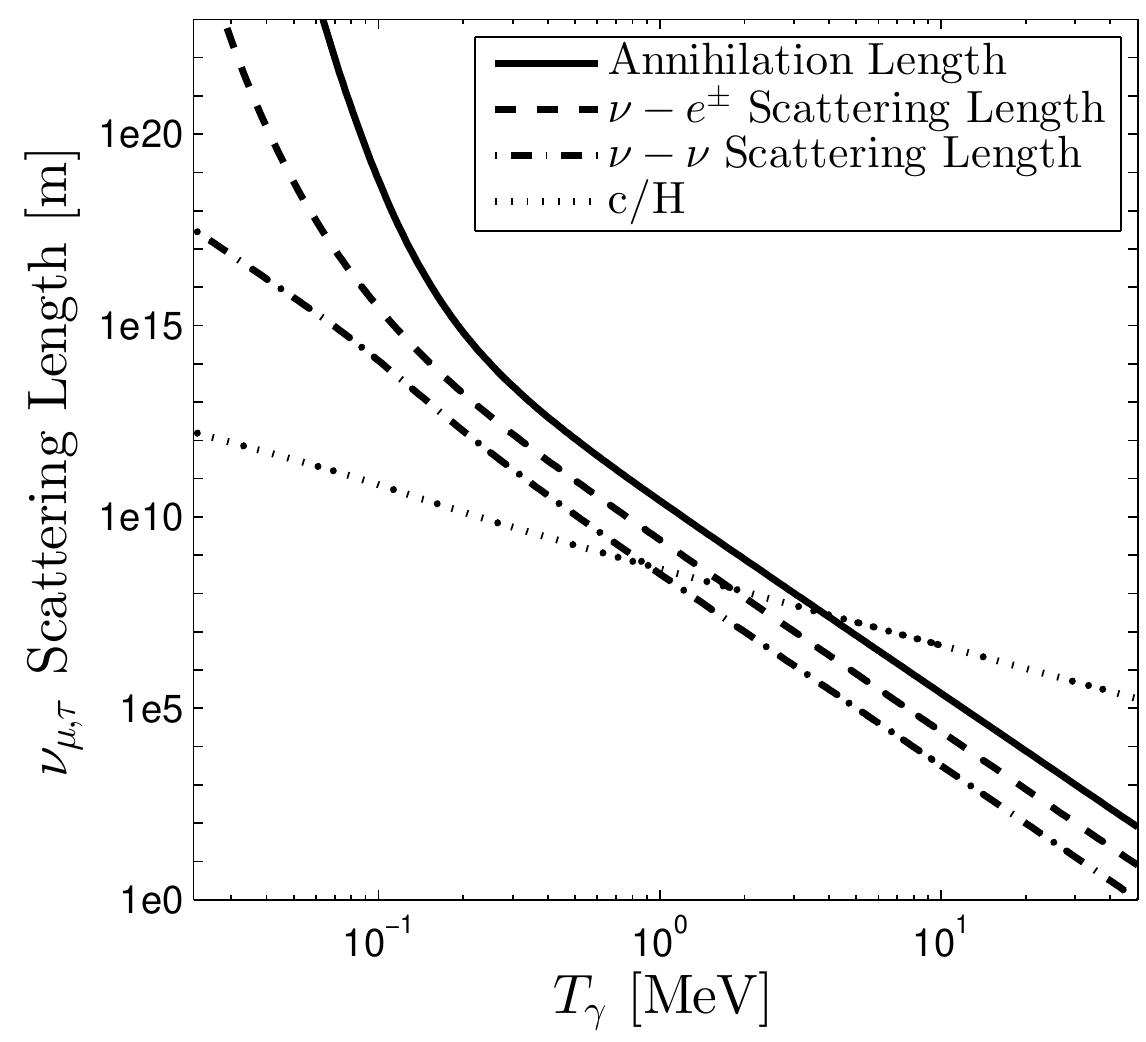}}
\caption{Comparison of Hubble parameter to neutrino scattering length for various types of processes for $\sin^2(\theta_W)=.23$. }\label{fig:scatt_length}
 \end{figure}

We now consider the the relaxation time for a given reaction, defined by $\tau=1/\Gamma$.  Suppose we have a time interval $t_f>t_i$  and corresponding temperature interval $T_f<T_i$ during which there is no reheating and the Universe is radiation dominated.  Normalizing time so $t=0$ corresponds to the temperature $T_i$ we have
\begin{equation}\label{ch6:H_eq}
\dot a/a=-\dot T/T,\hspace{2mm} H=\frac{C}{2Ct+T_i^2}\propto T^2
\end{equation}
where $C$ is a constant that depends on the energy density and the Planck mass.  Its precise form will not be significant for us.  Note that \req{ch6:H_eq} implies
\begin{equation}
1/H(t)-1/H(0)=2t.
\end{equation}

At $T\gg m_e$, the rates for reactions under consideration from tables \ref{table:nu_e_reac} and \ref{table:nu_mu_reac} scale as $\Gamma\propto T^5$.  Therefore, supposing $H(T_f)/\Gamma(T_f)=1$ (which occurs at $T_f=O(1\MeV)$ as seen in the above figures), at any time $t_f>t>t_i$ we find 
\begin{align}\label{relax_time}
\tau(t)/t=&\frac{2}{\Gamma(t)}\left(\frac{1}{H(t)}-\frac{1}{H(0)}\right)^{-1}=\frac{2T_f^5}{\Gamma(T_f)T^5}\left(\frac{T_f^2}{H(T_f)T^2}-\frac{T_f^2}{H(T_f)T_i^2}\right)^{-1}\\
=&\frac{2T_f^3}{T^3}\left(1-\frac{T^2}{T_i^2}\right)^{-1}.
\end{align}
Therefore, given any time $t_i<t_0<t_f$ we have
\begin{equation}\label{ch6:tau_eq}
\tau(t)<\tau(t_0)=\frac{2T_f^3}{T_0^3}\left(1-\frac{T_0^2}{T_i^2}\right)^{-1}\Delta t \text{ for all } t<t_0
\end{equation}
where $\Delta t=t_0-t_i=t_0$.

 The first reheating period that precedes neutrino freeze-out is the disappearance of muons and pions around $O(100\MeV)$, as seen in figure \ref{fig:energy_frac}, and so we let $T_i=100\MeV$. \req{ch6:tau_eq} is minimized at $T_0\approx 77.5\MeV$ at which point we have 
\begin{equation}
\tau(t)<10^{-5} \Delta t_0 \text{ for } t<t_0.
\end{equation}
This shows that the relaxation time during the period between $100\MeV$ and $77.5\MeV$ is at least five orders of magnitude smaller than the corresponding time interval.  Therefore the system has sufficient time to relax back to equilibrium after any potential non-equilibrium aspects developed during the reheating period.  Thus justifies our assumption that the neutrino distribution has the equilibrium Fermi Dirac form at $T=O(10 \MeV)$ when we begin our numerical simulation.

We demonstrate this numerically in figure \ref{fig:relax} where we have initialized the system at $T_\gamma=12\MeV$ with a non-equilibrium distriubtion of $\mu$ and $\tau$ neutrinos, giving them $\Upsilon=0.9$, and let them evolve.  We see that after approximately $10^{-3}$ seconds the system relaxes back to equilibrium, well before neutrino freeze-out near $t=1$s.

\begin{figure} 

\begin{minipage}{\linewidth}
\makebox[0.5\linewidth]%
{\includegraphics[height=6.2cm]{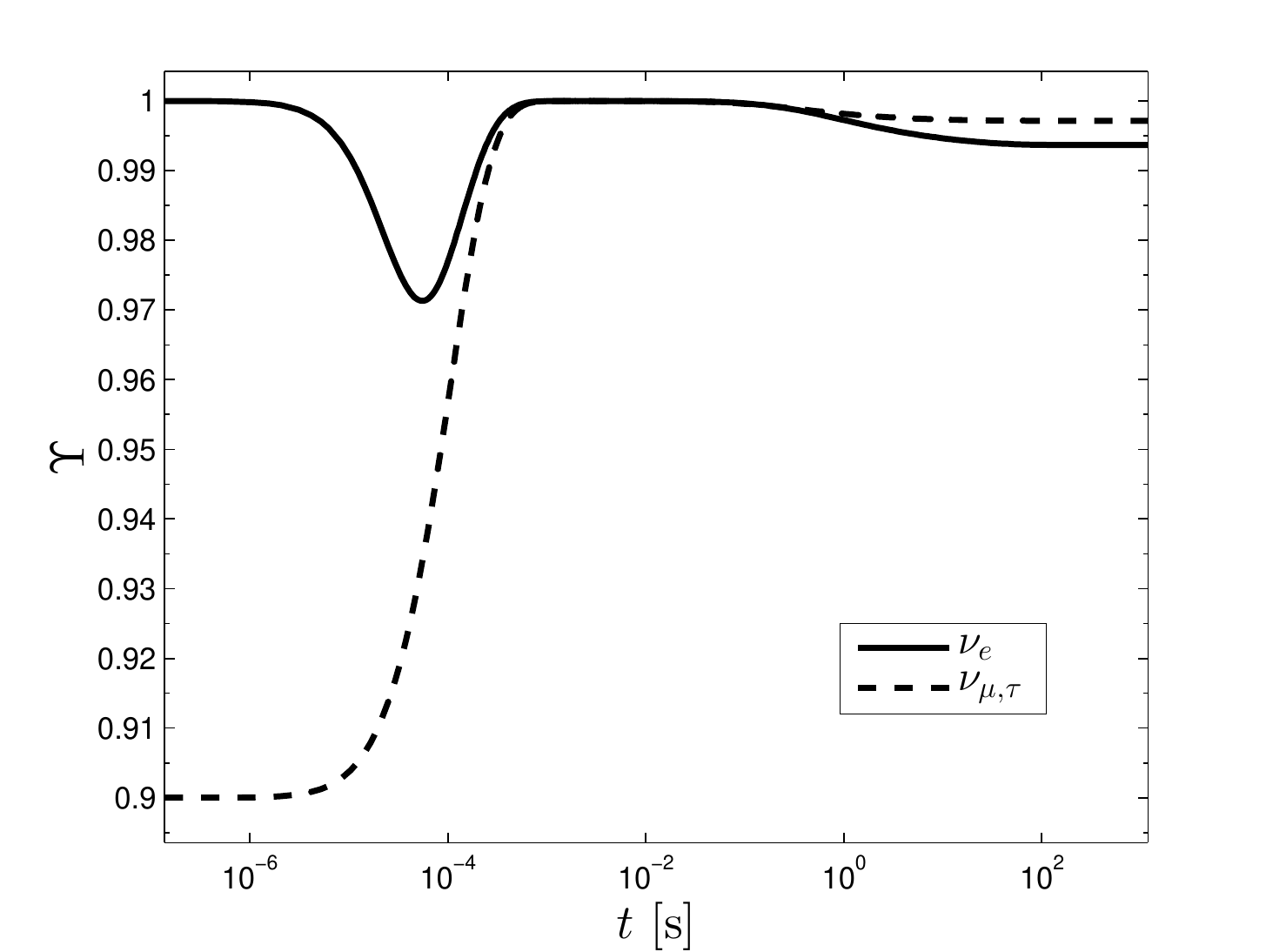}}
\makebox[0.5\linewidth]%
{\includegraphics[height=6.2cm]{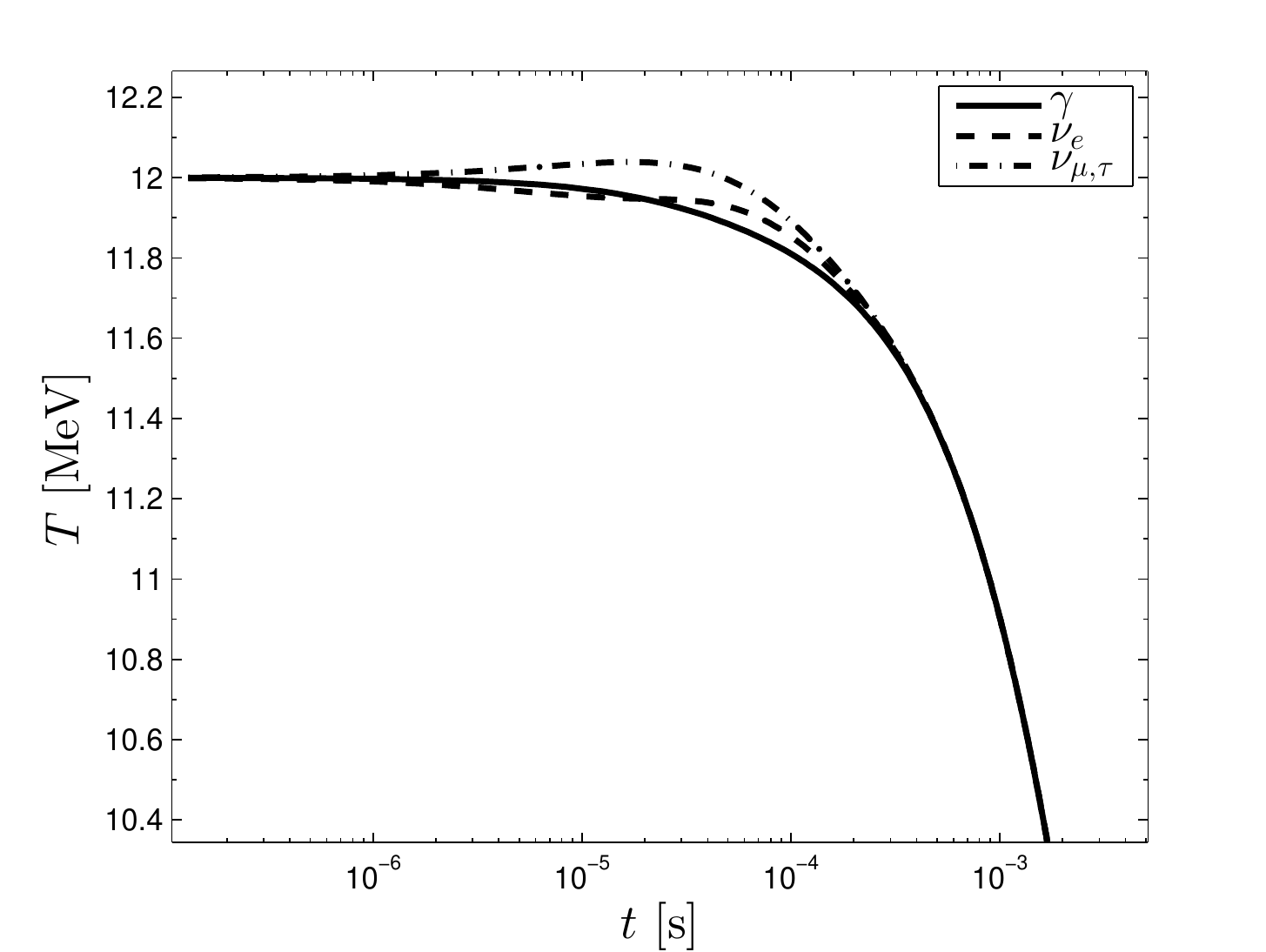}}
\end{minipage}
\caption{Starting at $12\MeV$, this figure shows the relaxation of a non-equilibrium $\mu,\tau$-neutrino distribution towards equilibrium. The fugacities are shown in the left frame while the temperatures are shown in the right frame.\label{fig:relax}}
 \end{figure}

The attentive reader will notice that we have omitted here a discussion of flavor neutrino oscillations. If it weren't for the differences between the matrix elements for the interactions between $e^\pm$ and $\nu_e$ on one hand and $e^\pm$ and $\nu_\mu,\nu_\tau$ on the other, oscillations would have no effect on the flow of entropy into neutrinos and hence no effect on $N_\nu$, but these differences do lead to a modification of $N_\nu$.  In \cite{Mangano2005} the impact of oscillations on neutrino freeze-out for the present day measured values of $\theta_W$ and $\eta$ was investigated.  It was found  that while oscillations redistributed energy amongst the neutrino flavors, the impact on $N_\nu$ was negligible. We have neglected oscillations in our study and so, once the relevant neutrino properties are fully understood, the precision of the results could be improved by incorporating the effect of oscillations.

\begin{subappendices}

\section{Another Method for Computing Scattering Integrals}\label{app:dogov_method}
As a comparison and consistency check for our method of computing the scattering integrals, in this appendix we analytically reduce the collision integral down to $3$ dimensions by a method adapted from \cite{Dolgov_Hansen}.  The only difference between our treatment in this section and theirs being that they solved the Boltzmann equation numerically on a grid in momentum space and not via a spectral method.  Therefore we must take an inner product of the collision operator with a basis function and hence we are integrating over all particle momenta, whereas they integrate over all momenta except that of particle one.  For completeness we give a detailed discussion of their method.

 Writing the conservation of four-momentum enforcing delta function
\begin{equation}
\delta(\Delta p)=\frac{1}{(2\pi)^3}\delta(\Delta E)e^{i\vec z\cdot \Delta \vec p}d^3z,
\end{equation}
where the arrow denoted the spatial component, we can simplify the collision integral as follows
\begin{align}
R\equiv&\int G(E_1,E_2,E_3,E_4) S|\mathcal{M}|^2(s,t)(2\pi)^4\delta(\Delta p)\prod_{i=1}^4\frac{d^3p_i}{2(2\pi)^3 E_i}\\
=&\frac{1}{16(2\pi)^{11}}\int G(E_i) S|\mathcal{M}|^2(s,t)\delta(\Delta E)e^{i\vec z\cdot\Delta p}\prod_{i=1}^4\frac{d^3p_i}{ E_i}d^3z\\
=&\frac{2}{(2\pi)^{6}}\int G(E_i)K(E_i) \delta(\Delta E)\prod_{i=1}^4\frac{p_i}{E_i}dp_i z^2dz,\\
K=&\frac{p_1p_2p_3p_4}{(4\pi)^5}\int S|\mathcal{M}|^2(s,t)e^{i\vec z\cdot\Delta\vec p}\prod_{i=1}^4d\Omega_id\Omega_z.
\end{align}
We can change variables from $p_i$ to $E_i$ in the outer integrals and use the delta function to eliminate the integration over $E_4$ to obtain
\begin{align}
R=&\frac{2}{(2\pi)^{6}}\int1_{E_1+E_2-E_3>m_4}G(E_i)\left[\int_0^\infty K(z,E_i)z^2dz\right]dE_1dE_2dE_3,\\
p_i=&\sqrt{E_i^2-m_i^2},\hspace{2mm} E_4=E_1+E_2-E_3.
\end{align}
From tables \ref{table:nu_e_reac} and \ref{table:nu_mu_reac} we see that the matrix elements for weak scattering involving neutrinos are linear combinations of the terms
\begin{equation}
p_1\cdot p_2,\hspace{2mm} p_1\cdot p_3,\hspace{2mm}(p_1\cdot p_4)(p_2\cdot p_3), \hspace{2mm} (p_1\cdot p_2)(p_3\cdot p_4),\hspace{2mm} (p_1\cdot p_3)(p_2\cdot p_4).
\end{equation}
Therefore we must compute the angular integral term $K$ with $S|\mathcal{M}|^2$ replaced by elements from the following list
\begin{align}\label{matrix_element_pieces}
&1,\hspace{2mm}\vec p_1 \cdot\vec p_2,\hspace{2mm}\vec p_1 \cdot\vec p_3,\hspace{2mm}\vec p_1 \cdot\vec p_4 ,\hspace{2mm}\vec p_2\cdot\vec p_3,\hspace{2mm}\vec p_2\cdot\vec p_4 ,\hspace{2mm}\vec p_3\cdot\vec p_4 ,\\
& (\vec p_1 \cdot\vec p_2)(\vec p_3\cdot\vec p_4 ),\hspace{2mm}(\vec p_1 \cdot\vec p_4 )(\vec p_2\cdot\vec p_3),\hspace{2mm} (\vec p_1 \cdot\vec p_3)(\vec p_2\cdot\vec p_4 ),
\end{align}
producing $K_0$, $K_{12}$, $K_{13}$,...,$K_{1324}$.  All of these are rotationally invariant, and so we can always rotate coordinates so that $\vec z=z\hat z$.  This allows us to evaluate the $z$ angular integral
\begin{equation}
K=\frac{p_1p_2p_3p_4}{(4\pi)^4}\int S|\mathcal{M}|^2(s,t)e^{iz \hat z\cdot\Delta\vec p}\prod_{i=1}^4d\Omega_i.
\end{equation}

The remaining angular integrals are straightforward to evaluate analytically for each expression in \req{matrix_element_pieces}
\begin{align}
K_0&=\prod_{i=1}^4\frac{\sin(p_iz)}{z},\\
K_{12}&=-\frac{(\sin(p_1z)-p_1z\cos(p_1z))(\sin(p_2z)-p_2z\cos(p_2z))\sin(p_3z)\sin(p_4z)}{z^6},\\
K_{13}&=\frac{(\sin(p_1z)-p_1z\cos(p_1z))\sin(p_2z)(\sin(p_3z)-p_3z\cos(p_3z))\sin(p_4z)}{z^6},\\
K_{14}&=\frac{(\sin(p_1z)-p_1z\cos(p_1z))\sin(p_2z)\sin(p_3z)(\sin(p_4z)-p_4z\cos(p_4z))}{z^6},\\
K_{23}&=\frac{\sin(p_1z)(\sin(p_2z)-p_2z\cos(p_2z))(\sin(p_3z)-p_3z\cos(p_3z))\sin(p_4z)}{z^6},\\
K_{24}&=\frac{\sin(p_1z)(\sin(p_2z)-p_2z\cos(p_2z))\sin(p_3z)(\sin(p_4z)-p_4z\cos(p_4z))}{z^6},\\
K_{34}&=-\frac{\sin(p_1z)\sin(p_2z)(\sin(p_3z)-p_3z\cos(p_3z))(\sin(p_4z)-p_4z\cos(p_4z))}{z^6},\\
K_{1234}&=K_{1423}=K_{1324}=\prod_{i=1}^4\frac{(\sin(p_iz)-p_iz\cos(p_iz))}{z^2}.
\end{align}

To compute $\int_0^\infty K(z) z^2 dz$ we need to evaluate the following three integrals
\begin{align}
D_1=&\int_0^\infty \frac{\sin(p_1z)\sin(p_2z)\sin(p_3z)\sin(p_4z)}{z^2}dz,\\
D_2=&\int_0^\infty\frac{\sin(p_1z)\sin(p_2z)(\sin(p_3z)-p_3z\cos(p_3z))(\sin(p_4z)-p_4z\cos(p_4z))}{z^4}dz,\\
D_3=&\int_0^\infty\frac{\prod_{i=1}^4(\sin(p_iz)-p_iz\cos(p_iz))}{z^6}dz.
\end{align}
These expressions are symmetric under $1\leftrightarrow 2$ and $3\leftrightarrow 4$ and so without loss of generality we can assume $p_1\geq p_2$, $p_3\geq p_4$. We require $p_1\leq p_2+p_3+p_4$ (and cyclic permutations) by conservation of energy.  In the case where the above conditions all hold, we separate things into four additional cases in which the integrals can be evaluated analytically, as given in \cite{Dolgov_Hansen},\\
${\bf p_1+p_2>p_3+p_4\text{, \hspace{1mm} }p_1+p_4>p_2+p_3}${\bf :}
\begin{align}
D_1=&\frac{\pi}{8}(p_2+p_3+p_4-p_1),\\
D_2=&\frac{\pi}{48}((p_1-p_2)^3+2(p_3^3+p_4^3)-3(p_1-p_2)(p_3^2+p_4^2),\\
D_3=&\frac{\pi}{240}(p_1^5-p_2^5+5p_2^3(p_3^2+p_4^2)-5p_1^3(p_2^2+p_3^2+p_4^2)-(p_3+p_4)^3(p_3^2-3p_3p_4+p_4^2)\notag\\
&\hspace{7mm}+5p_2^2(p_3^3+p_4^3)+5p_1^2(p_2^3+p_3^3+p_4^3)).
\end{align}
${\bf p_1+p_2<p_3+p_4\text{, \hspace{1mm} }p_1+p_4>p_2+p_3}${\bf :}
\begin{align}
D_1=&\frac{\pi }{4}p_2,\\
D_2=&\frac{\pi }{24}p_2(3(p_3^2+p_4^2-p_1^2)-p_2^2),\\
D_3=&\frac{\pi}{120}p_2^3(5(p_1^2+p_3^2+p_4^2)-p_2^2).
\end{align}
${\bf p_1+p_2>p_3+p_4\text{, \hspace{1mm} }p_1+p_4<p_2+p_3}${\bf :}
\begin{align}
D_1=&\frac{\pi }{4}p_4,\\
D_2=&\frac{\pi}{12} p_4^3,\\
D_3=&\frac{\pi }{120}p_4^3(5(p_1^2+p_2^2+p_3^2)-p_4^2).
\end{align}
${\bf p_1+p_2<p_3+p_4\text{, \hspace{1mm} }p_1+p_4<p_2+p_3}${\bf :}
\begin{align}
D_1=&\frac{\pi}{8}(p_1+p_2+p_4-p_3),\\
D_2=&\frac{\pi}{48}(-(p_1+p_2)^3-2p_3^3+2p_4^3+3(p_1+p_2)(p_3^2+p_4^2)),\\
D_3=&\frac{\pi}{240}(p_3^5-p_4^5-(p_1+p_2)^3(p_1^2-3p_1p_2+p_2^2)+5(p_1^3+p_2^3)p_3^2-5(p_1^2+p_2^2)p_3^3\\
&\hspace{7mm}+5(p_1^3+p_2^3-p_3^3)p_4^2+5(p_1^2+p_2^2+p_3^2)p_4^3).\notag
\end{align}
We computed the remaining integrals numerically in several test cases for each of the reaction types in section \ref{nu_matrix_elements} and obtained agreement between this method and ours, up to the integration tolerance used.  However, the method we have developed in this chapter has the distinct advantage of resulting in a smooth integrand which then must be evaluated numerically.  The expressions obtained here are only piecewise smooth and therefore much costlier to integrate numerically.  In tests, the difference in integration time was found to be $1000$ times longer in some instances for the non-smooth integrand using an adaptive mesh integration method.  Since the cost of numerically solving the Boltzmann equation is dominated by the cost of computing the collision integrals, this is a very significant optimization.

\section{QED Corrections to Equation of State}\label{app:QED_corr}
At the time of neutrino freeze-out, the universe is at sufficiently high temperature for photons and $e^\pm$ to be in chemical and kinetic equilibrium.  The temperature is also sufficiently high for QED corrections to the photon and $e^\pm$ equation of state to be non-negligible.  We use the results given in \cite{Heckler:1994tv,Mangano2002} to include these in our computation by modifying the combined photon, $e^\pm$ equation of state
\begin{align}
P=P^0+P^{int},\hspace{2mm} \rho=-P+T\frac{dP}{dT}
\end{align}
where
\begin{align}
P^{int}=&-\frac{1}{2\pi^2}\int_0^\infty\left[\frac{k^2}{E_k}\frac{\delta m_e^2}{e^{E_k/T}+1}+\frac{k}{2}\frac{\delta m_\gamma^2}{e^{k/T}-1}\right]dk,\hspace{2mm} E_k=\sqrt{k^2+m_e^2}\\
\delta m_e^2=&\frac{2\pi\alpha^2}{3}+\frac{4\alpha}{\pi}\int_0^\infty \frac{k^2}{E_k}\frac{1}{e^{E_k/T}+1}dk,\hspace{2mm} \delta m_\gamma^2=\frac{8\alpha}{\pi}\int_0^\infty \frac{k^2}{E_k}\frac{1}{e^{E_k/T}+1}dk.
\end{align}
and $P^0$ is the pressure of a noninteracting gas of photons and $e^\pm$ in chemical equilibrium.

\end{subappendices}

\chapter{Dependence of Neutrino Freeze-out on Parameters}\label{ch:param_studies}
Having developed an improved method for solving the Boltzmann equation and computing scattering integrals that greatly reduces the computational cost, we are now able to characterize the dependence of neutrino freeze-out on parameters.  This will allow us to identify potential avenues by which the tension between observed and theoretical values of $N_\nu$ may be alleviated.  See also our paper \cite{Birrell:2014uka}.

Our study will also us to constrain the time and/or temperature variation of certain natural constants by comparing the results with measurements of $N_\nu$.  The topic of time variation of natural constants is a very active field with a long history. For a comprehensive review of this area, with which we make only slight contact,  see \cite{Uzan:2010pm}.

\section{Weinberg Angle}

As mentioned above, the Weinberg angle is one of the standard model parameters that impacts the neutrino freeze-out process.  More specifically, it is found in the matrix elements of weak force processes, including the reactions $e^+e^-\rightarrow \nu\bar\nu$ and $\nu e^\pm\rightarrow \nu e^\pm$ found in tables \ref{table:nu_e_reac} and \ref{table:nu_mu_reac}.  It is determined by the $SU(2)\times U(1)$ coupling constants $g$, $g^{'}$  by
\begin{equation}
\sin(\theta_W)=\frac{g^{'}}{\sqrt{g^2+(g^{'})^2}}.
\end{equation}
It is also related to the mass of the $W$ and $Z$ bosons and the Higgs vacuum expectation value $v$ by
\begin{equation}
M_Z=\frac{1}{2}\sqrt{g^2+(g^{'})^2}v,\hspace{2mm}  M_W=\frac{1}{2}gv,\hspace{2mm} \cos(\theta_W)=\frac{M_W}{M_Z}
\end{equation}
as well as the electromagnetic coupling strength
\begin{equation}
e=2M_W\sin(\theta_W)/v=\frac{gg^{'}}{\sqrt{g^2+(g^{'})^2}}.
\end{equation}
It has a measured value in vacuum $\theta_W\approx 30^\circ$, giving $\sin(\theta_W)\approx 1/2$, but its value is not fixed within the Standard Model. For this reason, a time or temperature variation can be envisioned and this would have an observable impact on the neutrino freeze-out process, as measured by $N_\nu$.

In letting $\sin(\theta_W)$, and hence $g$ and $g^{'}$, vary we must fix the electromagnetic coupling $e$ so as not to impact sensitive cosmological observables such as Big Bang Nucleosynthesis.  Fixing $v$, the smallest $M_W$ can become is when $\sin(\theta_W)=1$, yielding a reduction in $M_W$ by a factor of $2$.  This implies that $M_Z>M_W\gg |p|$ for neutrino momentum $p$ in the energy range of neutrino freeze-out, around $1\MeV$, even as we vary $\sin(\theta_W)$.  This approximation is inherent in the formulas for the matrix elements  in tables  \ref{table:nu_e_reac} and \ref{table:nu_mu_reac} and continues to be valid here. We will characterize the dependence of $N_\nu$ on $\sin(\theta_W)$ in section \ref{sec:param_char} below, but first we identify the remaining parameter dependence in the Einstein Boltzmann system

\section{Interaction Strength}
 In order to isolate the dependence of the Einstein Boltzmann system for neutrino freeze-out on dimensioned quantities, we now convert it to dimensionless form. Letting $m_e$ be the mass scale and $M_p/m_e^2$ be the time scale the Einstein equations take the form
\begin{equation}
H^2=\frac{\rho}{3},\hspace{2mm}\dot\rho=-3H(\rho+P).
\end{equation}
 Since $e^\pm$ are the only (effectively) massive particles in the system, by scaling all energies, momenta, energy densities, pressures, and temperatures by $m_e$ we have removed all scale dependent parameters from the Einstein equations.  The Boltzmann equation becomes
\begin{equation}\label{eta_def}
\partial_tf-pH\partial_pf=\eta\frac{C[f]}{E},\hspace{2mm}\eta\equiv M_p m_e^3G_F^2
\end{equation}
where we have also factored out of $C[f]$ the $G_F^2$ term that is common to all of the neutrino interaction matrix elements. 

Aside from the $\theta_W$ dependence of the matrix elements seen in tables \ref{table:nu_e_reac} and \ref{table:nu_mu_reac}, the complete dependence on natural constants  is now contained in a single dimensionless interaction strength parameter $\eta$ with the vacuum present day value,
\begin{equation}\label{eta0_def}
\eta_0\equiv \left.M_p m_e^3 G_F^2\right|_0  = 0.04421 .
\end{equation}
In the following section we characterize the dependence of $N_\nu$ on the interaction strength.

\section{Dependence of $N_\nu$ on Parameters}\label{sec:param_char}

The main result  of this chapter is the  dependence of $N_\nu$ on  the SM parameters   $\sin^2\theta_W$ and $\eta$. These results are shown in  figure \ref{N_nu_params}, presented as a function of  Weinberg angle $\sin^2 \theta_W $ for $\eta/\eta_0=1,2,5,10$. The effects of an increase in both parameters above the vacuum values superpose  in the parameter range  considered, amplifying the effect and generating a significant increase in  $N_\nu\to 3.5$. The present day vacuum value of Weinberg angle puts the $\nu_\mu,\nu_\tau$ freeze-out temperature, seen in the right pane of figure \ref{fig:freezeoutT},  near its maximum value.  This is why a comparatively large change in $\sin^2(\theta_W)$ is needed to produce a change in $N_\nu$ for $\sin^2(\theta_W)\approx0.23$.
 
\begin{figure}
\centerline{\includegraphics[width=0.70\columnwidth]{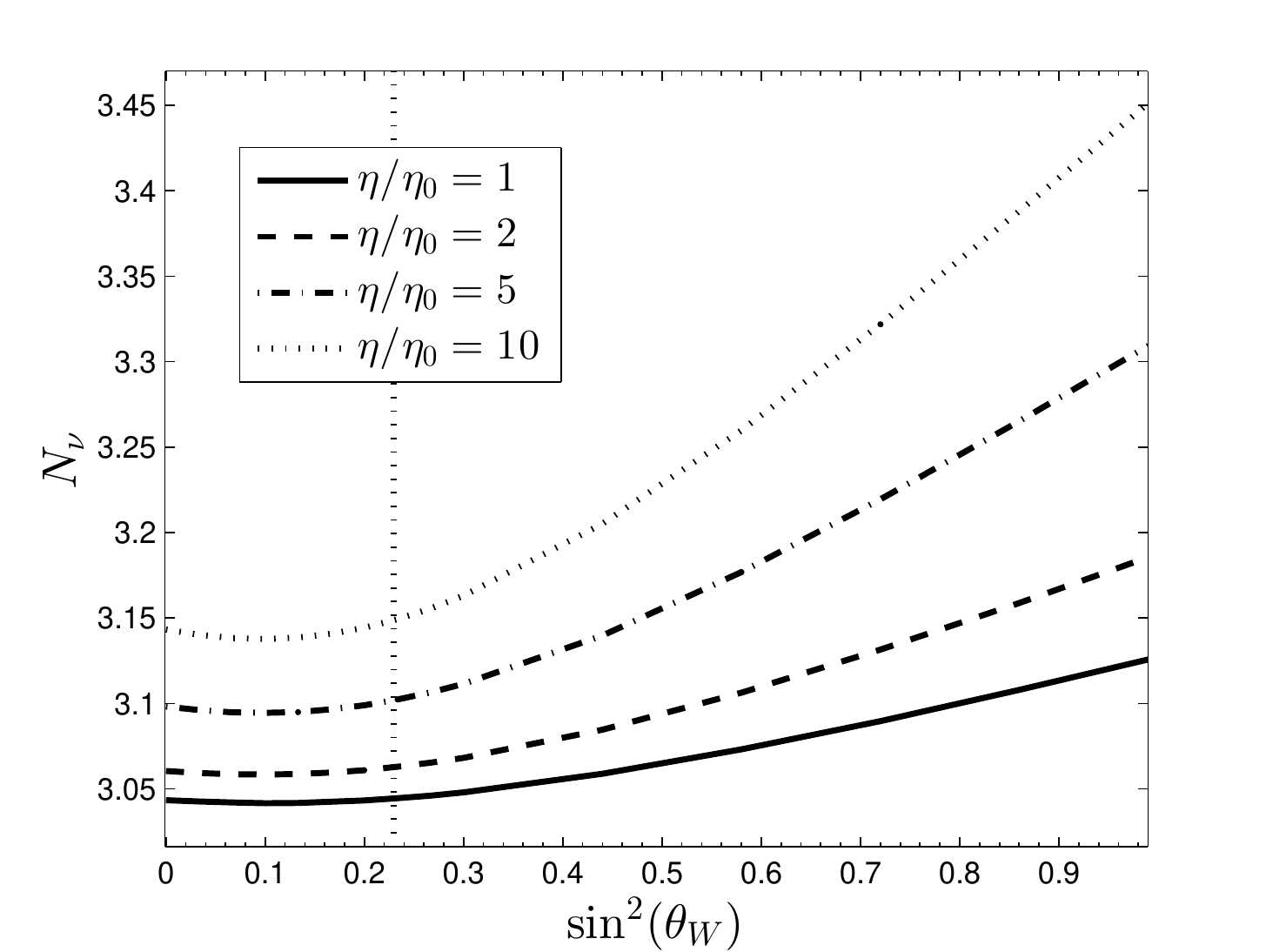}
}
\caption{Change in effective number of neutrinos  $N_\nu$ as a function of Weinberg angle for  several values of $\eta/\eta_0=1,2,5,10$. Vertical line is $\sin^2(\theta_W)=0.23$.}
\label{N_nu_params}  
 \end{figure}
We performed a least squares fit of $N_\nu$ over the range $0\leq \sin^2(\theta_W)\leq 1$, $1\leq \eta/\eta_0\leq 10$ shown in figure \ref{N_nu_params}, obtaining a result with relative error less than $0.2\%$,
\begin{align}
N_\nu=&3.003-0.095\sin^2\theta_W +0.222\sin^4\theta_W  -0.164\sin^6\theta_W \notag\\
+&\sqrt{\frac{\eta}{\eta_0}}\left(0.043+0.011\sin^2\theta_W +0.103\sin^4\theta_W\right).
\end{align}
$N_\nu$ is monotonically increasing in $\eta/\eta_0$ with dominant behavior  scaling as $\sqrt{ \eta/\eta_0}$. Monotonicity is to be expected, as increasing $\eta$ decreases the freeze-out temperature and the longer neutrinos are able to remain coupled to $e^\pm$, the more energy and entropy from annihilation is transfered to neutrinos.

We complement this with fits to the photon to neutrino temperature ratios $ T_\gamma / T_{\nu_e}, T_\gamma / T_{\nu_\mu}= T_\gamma / T_{\nu_\tau} $, and the neutrino fugacities, $\Upsilon_{\nu_e}, \Upsilon_{\nu_\mu}=\Upsilon_{\nu_\tau}$, again with relative error less than $0.2\%$  
\begin{align}
\frac{T_\gamma}{T_{\nu_\mu}}=&1.401+0.015x-0.040x^2+0.029x^3-0.0065y+0.0040xy-0.017x^2y, \label{fit1}\\
\Upsilon_{\nu_e}=&1.001+0.011x-0.024x^2+0.013x^3-0.005y-0.016xy+0.0006x^2y,\label{fit2}\\ 
\frac{T_\gamma}{T_{\nu_e}}=&1.401+0.015x-0.034x^2+0.021x^3-0.0066y-0.015xy-0.0045x^2y,\label{fit3}\\
\Upsilon_{\nu_\mu}=&1.001+0.011x-0.032x^2+0.023x^3-0.0052y+0.0057xy-0.014x^2y.\label{fit4}
\end{align}
where
\begin{equation}
x\equiv \sin^2 \theta_W ,\qquad
y\equiv  \sqrt{\frac{\eta}{\eta_0}}.
\end{equation}

\begin{figure}
\centerline{\includegraphics[width=0.75\columnwidth]{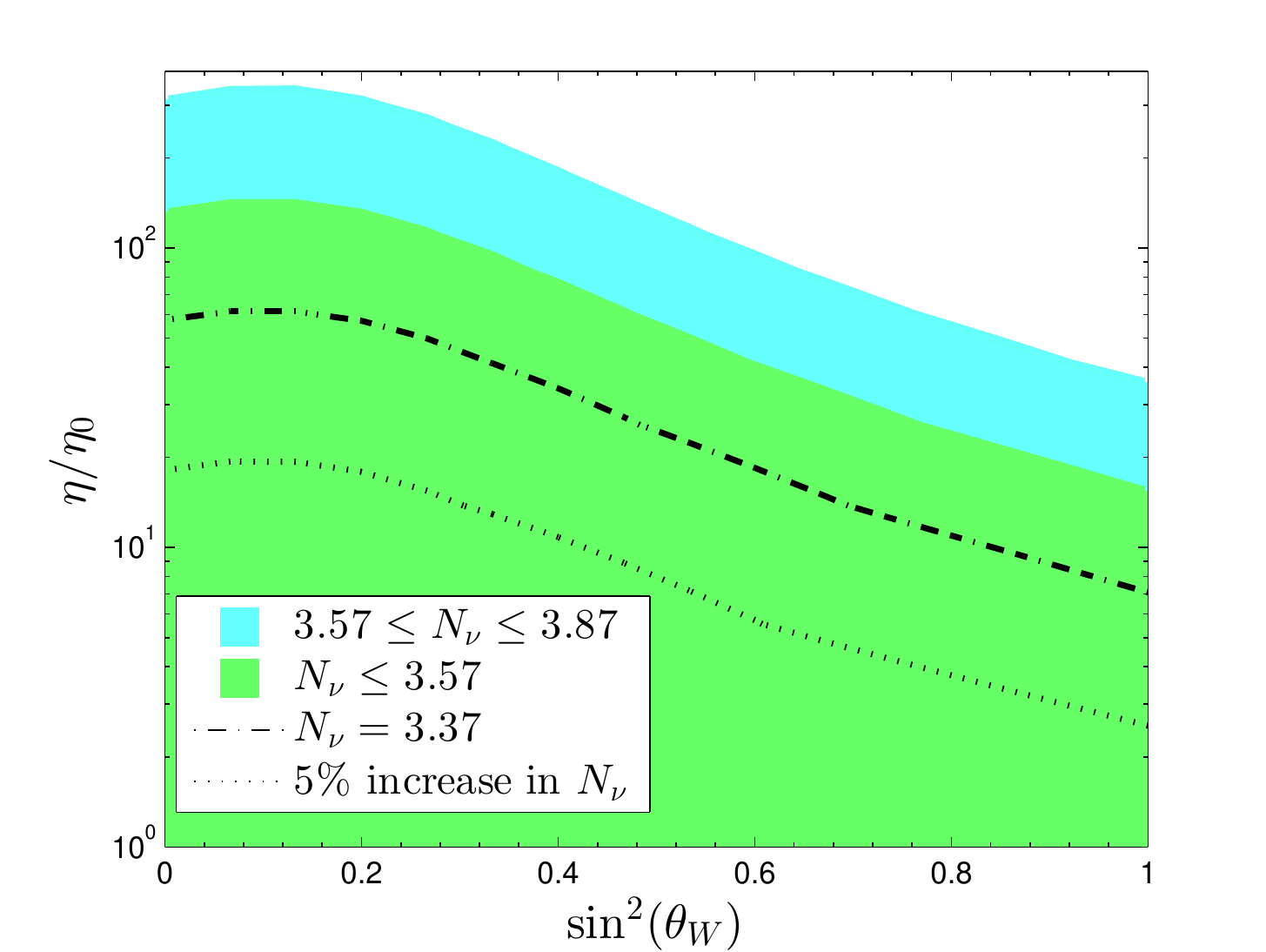}
}
\caption{$N_\nu$ bounds in the $\eta/\eta_0, \sin^2(\theta_W)$ plane. Dark (green) for $N_\nu\in (3.03,3.57)$ corresponding to Ref.\cite{Planck} CMB+BAO analysis and light(teal) extends the region to $N_\nu<3.87$ i.e. to CMB+$H_0$. Dot-dashed line delimits the 1s.d. lower boundary of the second analysis.}
\label{N_nu_domain}
 \end{figure}
The bounds on $N_\nu$ from the Planck analysis \cite{Planck} can be  used to constrain time or temperature variation of $\sin^2\theta_W$ and $\eta$.  
In  Figure \ref{N_nu_domain} the dark (green) color shows the combined range of  variation of natural constants  compatible with CMB+BAO and the light (teal) color shows  the extension in the range of  variation of  natural constants for CMB+$H_0$, both at a $68\%$ confidence level. The dot-dashed line within the dark (green) color  delimits   this latter domain. The dotted line shows the limit of a 5\% change in $N_\nu$.    Any increase in  $\eta/\eta_0$ and/or $\sin^2(\theta_W)$ moves the value of $N_\nu$ into the domain favored by current experimental results. Further parameter study is found in \ref{app:weinberg} and \ref{app:int_strength}.


\section{Summary, Discussion, and Conclusions}\label{sec:concl}
We have employed a novel spectral method Boltzmann solver and a new procedure for evaluating the Boltzmann scattering integrals in order to characterize the impact of a potential time and/or temperature variation of SM parameters on the effective number of neutrinos. Specifically, we identified a dimensionless combination of $m_e$, $M_p$, and $G_F$, called the interaction strength $\eta$, that, along with the Weinberg angle $\sin^2 \theta_W$, control neutrino freeze-out and the resulting value of the effective number of neutrinos, $N_\nu$.  

\subsection{Primordial Variation of Natural Constants}
The question which we addressed in this section is: What neutrino decoupling in the early Universe can tell us about the values of natural constants when the Universe was about one second old and at an ambient temperature near to 1 MeV (11.6 billion degrees K). Our results were presented assuming that the Universe contains no other effectively massless particles but the three left handed neutrinos and corresponding, three right handed anti-neutrinos. 

We found that near to the physical value of the Weinberg angle  $\sin^2 \theta_W\simeq 0.23$ the effect of changing $\sin^2\theta_W$ on the decoupling of neutrinos is small. Thus as seen in Figure \ref{N_nu_params}  the dominant variance is due to the change  in the coupling strength $\eta/\eta_0$, \req{eta_def}  and \req{eta0_def}. The dotted line in  Figure \ref{N_nu_domain} shows that in order to achieve a change in $N_\nu$ at the level of up to 5\% that is  $N_\nu\lesssim 3.2 $  both $\sin^2 \theta_W$ and $\eta/\eta_0$ must change significantly, with e.g. $\eta$ increasing by an order of magnitude.

Let us look closer at what an increase in the strength parameter $\eta$ by factor 10 means, looking case by case on all the natural constant contributions as if each were responsible for the entire change:
\begin{itemize}
\item 
Considering that  $\eta\propto M_p\propto G_N^{-1/2}$ this translates into a decrease  in the strength of Gravity at neutrino freeze-out by a factor 100.  This effect would need to become much smaller by the time the age of the Universe is 1000 times longer (1s compared to 10 min) for Big Bang nucleosynthesis to be unaffected. This presumably means that, conversely, as we go further back in time we would need the gravity to continue to rapidly become very much weaker yet. In models of emergent gravity we can  imagine a  `melting' of gravity in the hot primordial Universe. Whether such a model can be realized will be a topic for future consideration. The attractive aspect of Gravity weakening rapidly with increasing temperature is that for  exponentially disappearing $G_N\to 0$ as $t\to 0$ and/or $T\to \infty$ the dynamics can be arranged to be similar to an inflationary  Universe.
\item 
Since $\eta\propto m_e^3$ electron mass would need to go up `only' by factor 2.15 . Compared to all other particles the electron mass has an anomalously  low value. Appearance of a mechanism just when $T\simeq m_e$ that `restores' the electron mass to where intuition would like it to be, a few MeV, arising from  the systematics of other Yukawa Higgs coupling $g_{Ye}$ compared to the Yukawa coupling of other charged light particles, where $m_e= g_{Ye} v $ seems to us also  a possible scenario. Interestingly,   laboratory limits for these conditions could be attainable in the foreseeable future.
\item
Since $\eta\propto G_F^2\propto 1/v^4$  we would need to find a mechanism that would decrease the vacuum value $v_0\simeq 246$ GeV by factor 1.8 already at temperature $T\simeq m_e$.  Allowing three powers of $v$ to cancel by using the Higgs minimal coupling formula for electron mass  we need to change $v$ by an order of magnitude near to $T\simeq m_e$. This appears impossible.
\end{itemize}
While ideas justifying strong variation of $\eta$ can be developed as two of the above three cases argue, a model for temperature or time dependence of  $\sin^2 \theta_W$ seems at this time without a theoretical anchor point, mainly so since we do not have a valid grand unified theoretical framework in which the electro-weak mixing or equivalently the masses $M_W, M_Z$ would be anchored.

\subsection{Two Different Ways to Change $N_\nu$}

However, there are additional challenges we have not at this time addressed. This is so since the immediate observable is the energy content of the invisible Universe as defined by the effective number of neutrinos $N_\nu$. Considering a value of   $N_\nu>3$,  there could  be a contribution from presently not discovered, more weakly interacting massless particles that decoupled even before neutrinos, and which therefore could contribute fractionally to $N_\nu$, see our discussion in Ref.\cite{Birrell:2014connect}. 

Of particular relevance could be a so called light sterile neutrino~\cite{Abazajian:2012ys}, possibly the right handed complement to the left handed neutrinos. If such particles exist and freeze-out well before regular neutrinos, their contribution to $N_\nu$ would be subject to dilution by reheating~\cite{Birrell:2014connect} and thus their contribution to $N_\nu$ would depend on when precisely they begin free-streaming.

These unknown dark `radiation' particles as well as neutrinos could have a mass that is at the scale of the temperature of photon decoupling $T_{\gamma 0}=0.25$ eV, for which an analysis of the Universe density fluctuations akin to Planck~\cite{Planck} would need to be adapted. We have  discussed in Ref.\cite{Birrell:2013_2} a consistent treatment of neutrino mass and $N_\nu$,  in the case of a particular type of delayed massive neutrino  freeze-out. This approach is exactly the same as would be the case for dark radiation:  Near to $T_{\gamma 0}=0.25$ eV massive neutrinos are indistinguishable from massive dark radiation, which contributes as an additional particle with reduced contribution to $N_\nu$~\cite{Birrell:2014connect}.

The alternative explanation of $N_\nu>3$ in terms of variation of  of natural constants that we have presented comprises  speculative beyond the standard model ideas akin in this aspect to new dark `radiation' particles. We believe that our present contribution provides a viable alternative  mechanism  capable of influencing $N_\nu$. In order to achieve an increase in $N_\nu$ the change in natural constants must cause a delay in neutrino freeze-out and thus  a greater participation of neutrinos in reheating during $e^\pm$ annihilation. The changes in the natural constants  which are required to make a large and visible contribution in $N_\nu$ appear at first sight to reach beyond a variation that one could tolerate at the time of big bang nucleosynthesis only a factor 1000 in time later. We have argued that a change in the electron mass $m_e$ by factor larger than two, and/or Newtons constant $G_N$ even by several orders of magnitude  could be present.

\begin{subappendices}
\section{Weinberg Angle Plots}\label{app:weinberg}

\begin{figure}[H]
\centerline{\includegraphics[height=5.8cm]{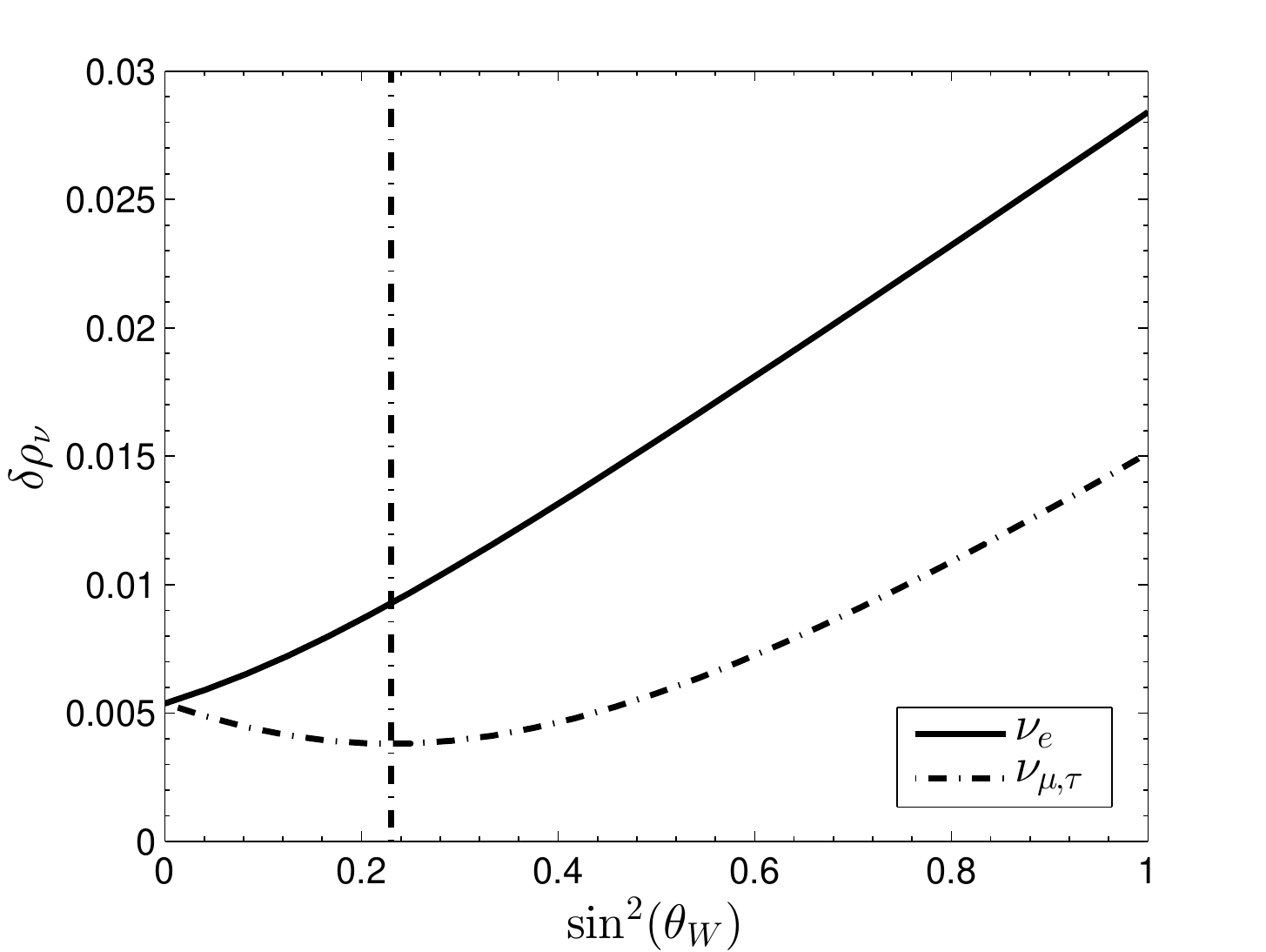}\hspace{-5mm}\includegraphics[height=5.8cm]{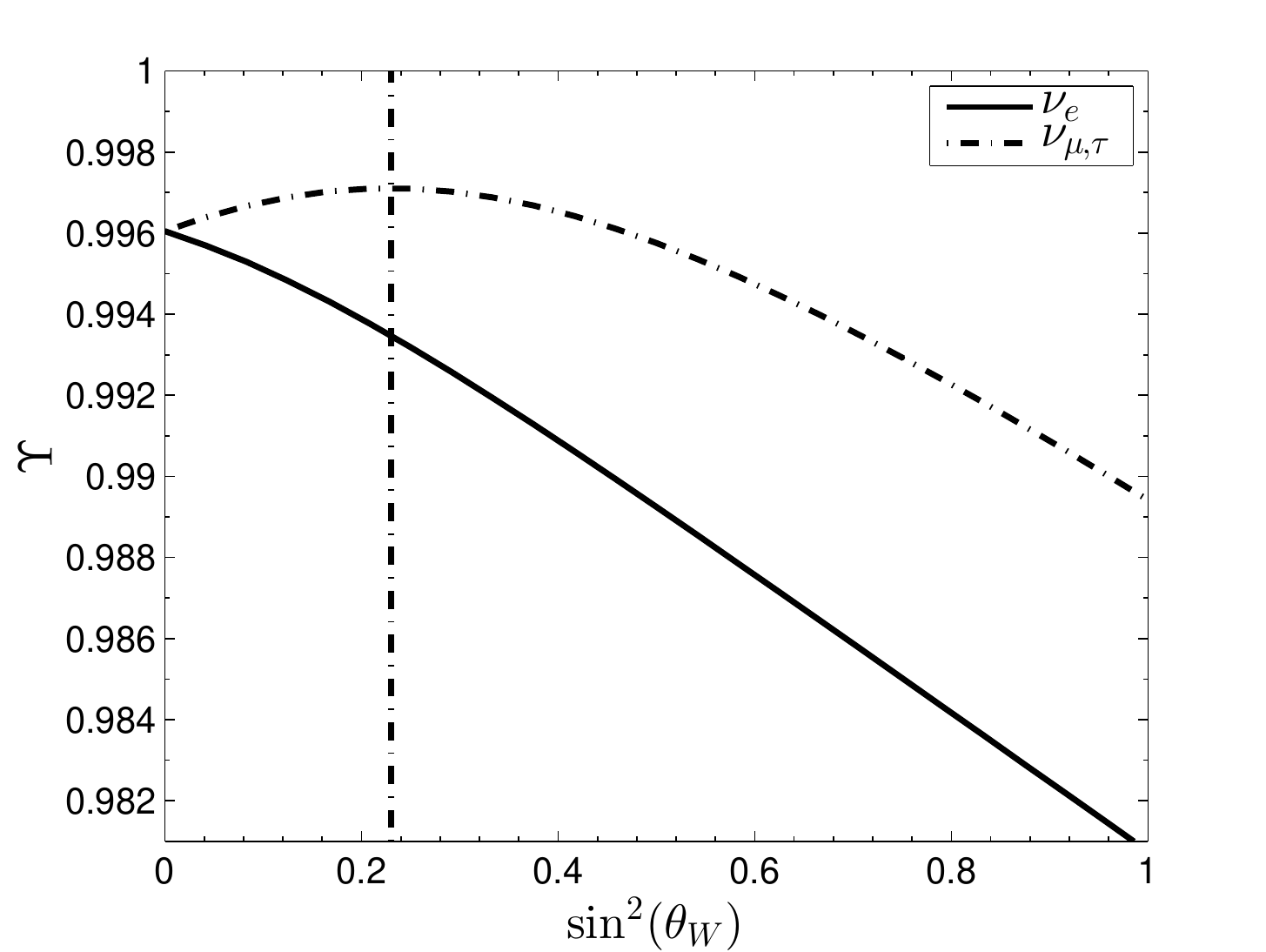}}
\caption{Fractional increase in neutrino energy (left) and neutrino fugacities (right), as functions of Weinberg angle. Vertical line is $\sin^2(\theta_W)=0.23$.}
 \end{figure}

\begin{figure}[H]
\centerline{\includegraphics[height=5.8cm]{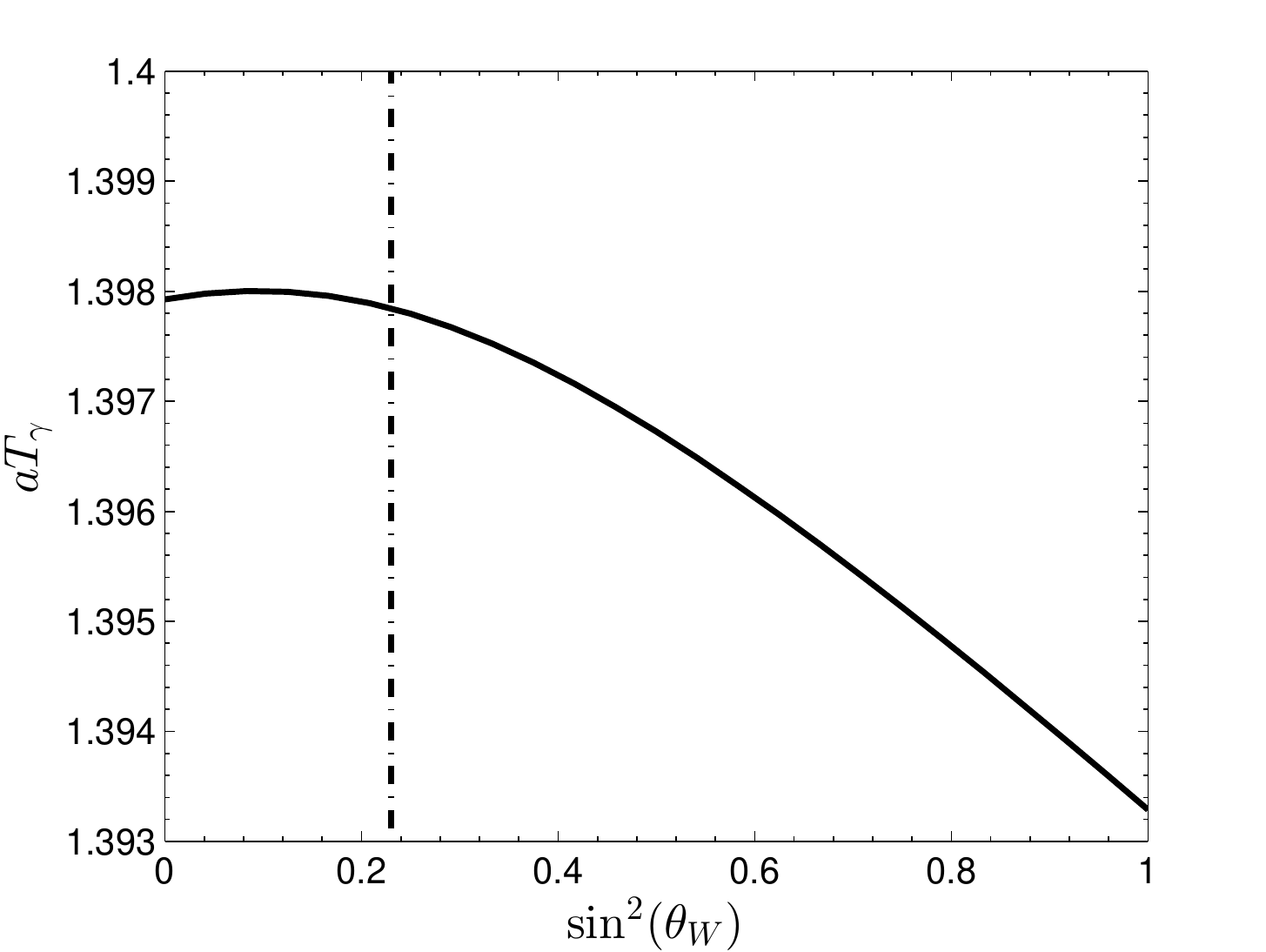}\hspace{-5mm}\includegraphics[height=5.8cm]{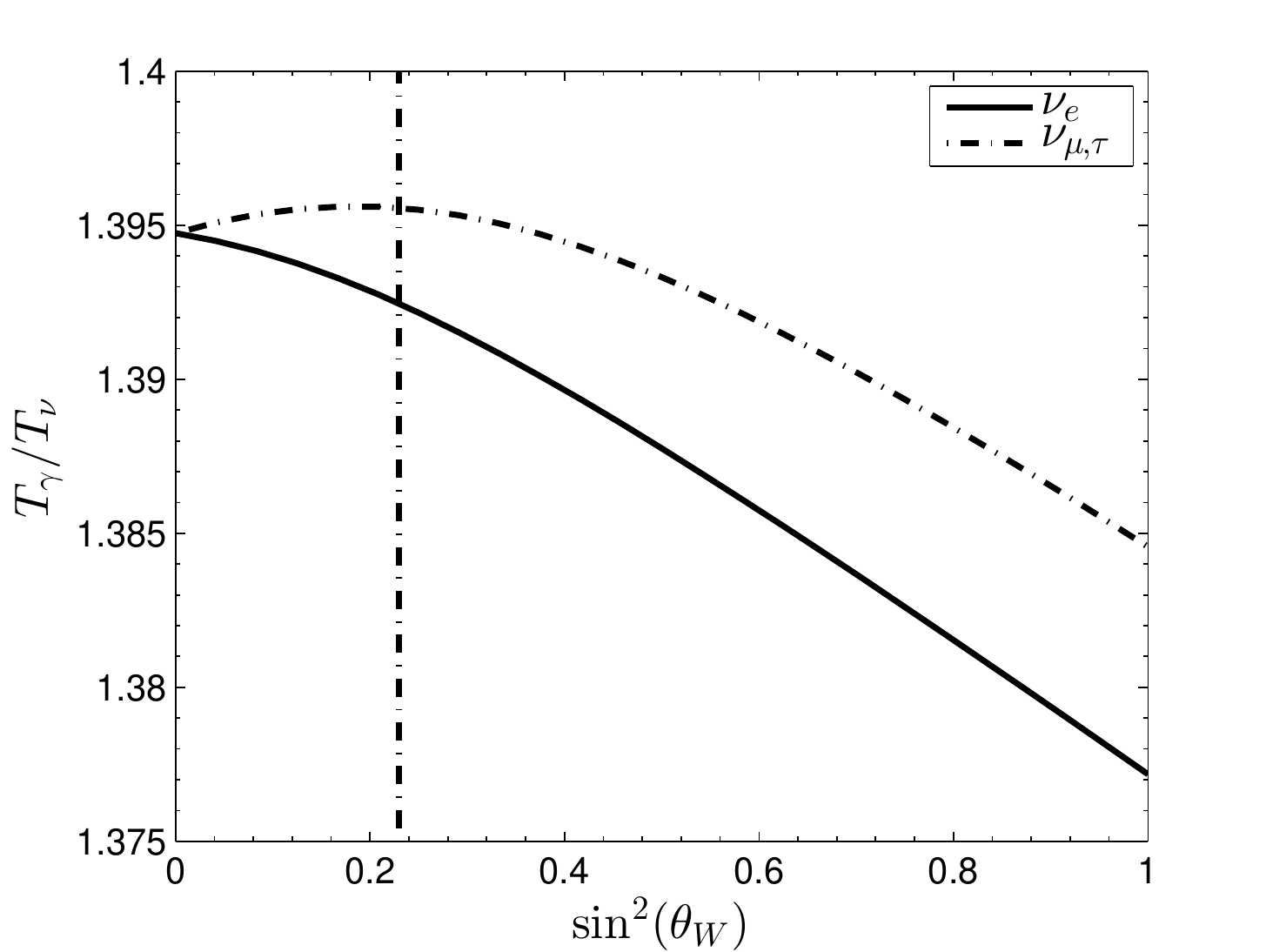}}
\caption{Dependence on Weinberg angle of photon reheating  (left) and photon-neutrino temperature ratios (right) after freeze-out. Vertical line is $\sin^2(\theta_W)=0.23$.}
 \end{figure}

\begin{figure}[H]
\centerline{\includegraphics[height=6.3cm]{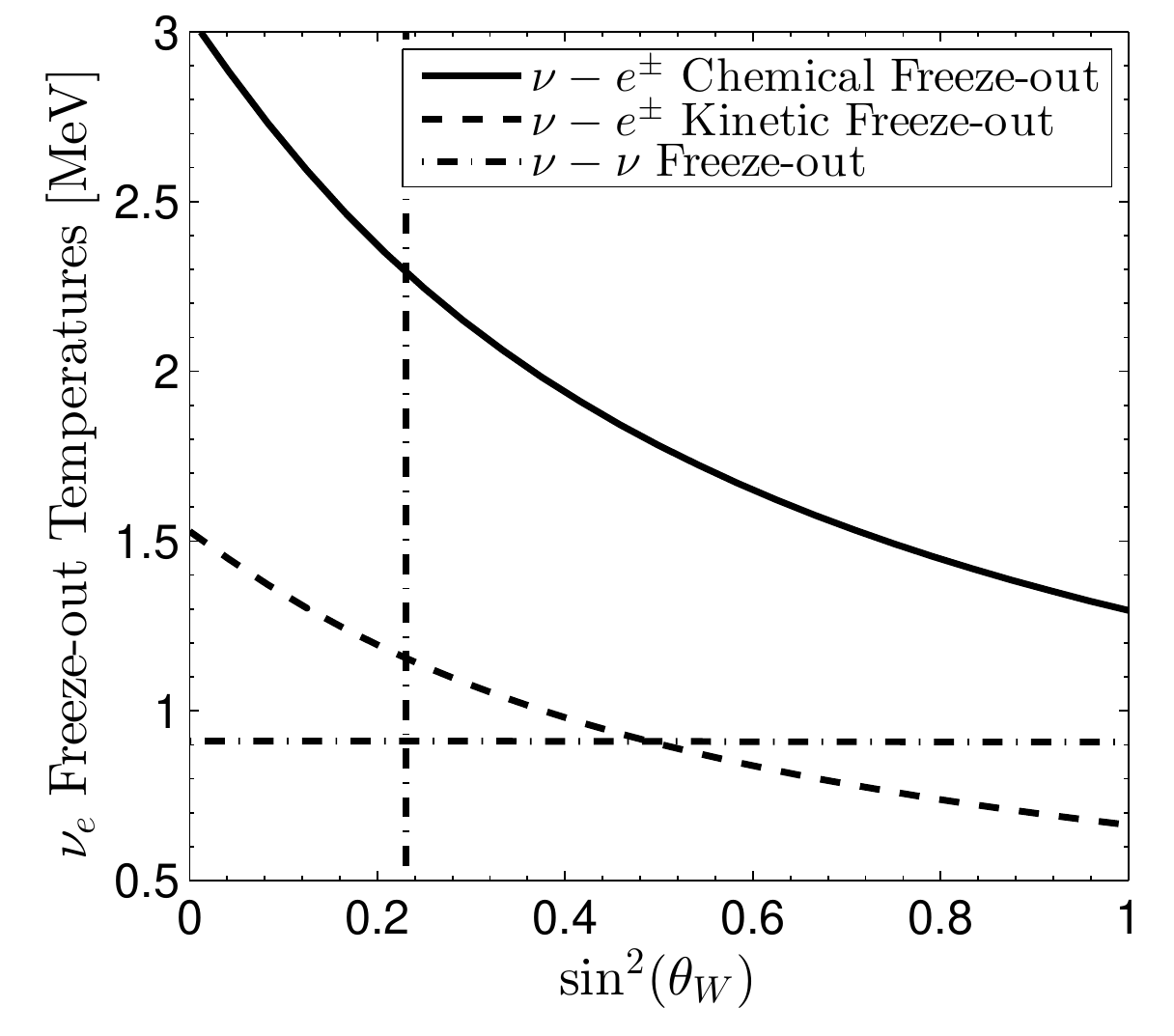}\hspace{5mm}\includegraphics[height=6.3cm]{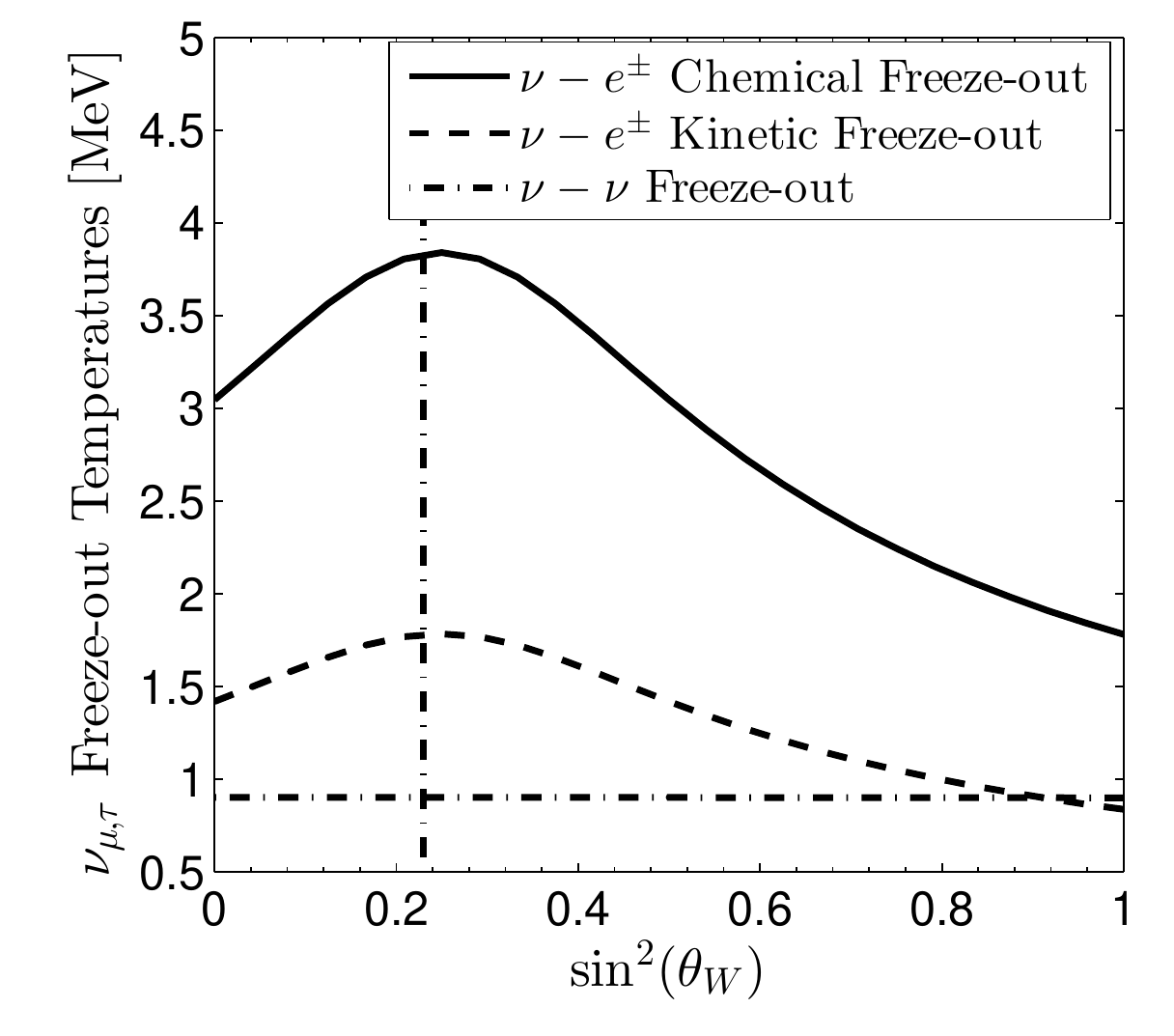}}
\caption{Freeze-out temperatures for electron neutrinos (left) and $\mu$, $\tau$ neutrinos (right) for various types of processes, as functions of Weinberg angle. Vertical line is $\sin^2(\theta_W)=0.23$.}\label{fig:freezeoutT}
 \end{figure}

\section{Interaction Strength Plots}\label{app:int_strength}

\begin{figure}[H]
\centerline{\includegraphics[height=5.8cm]{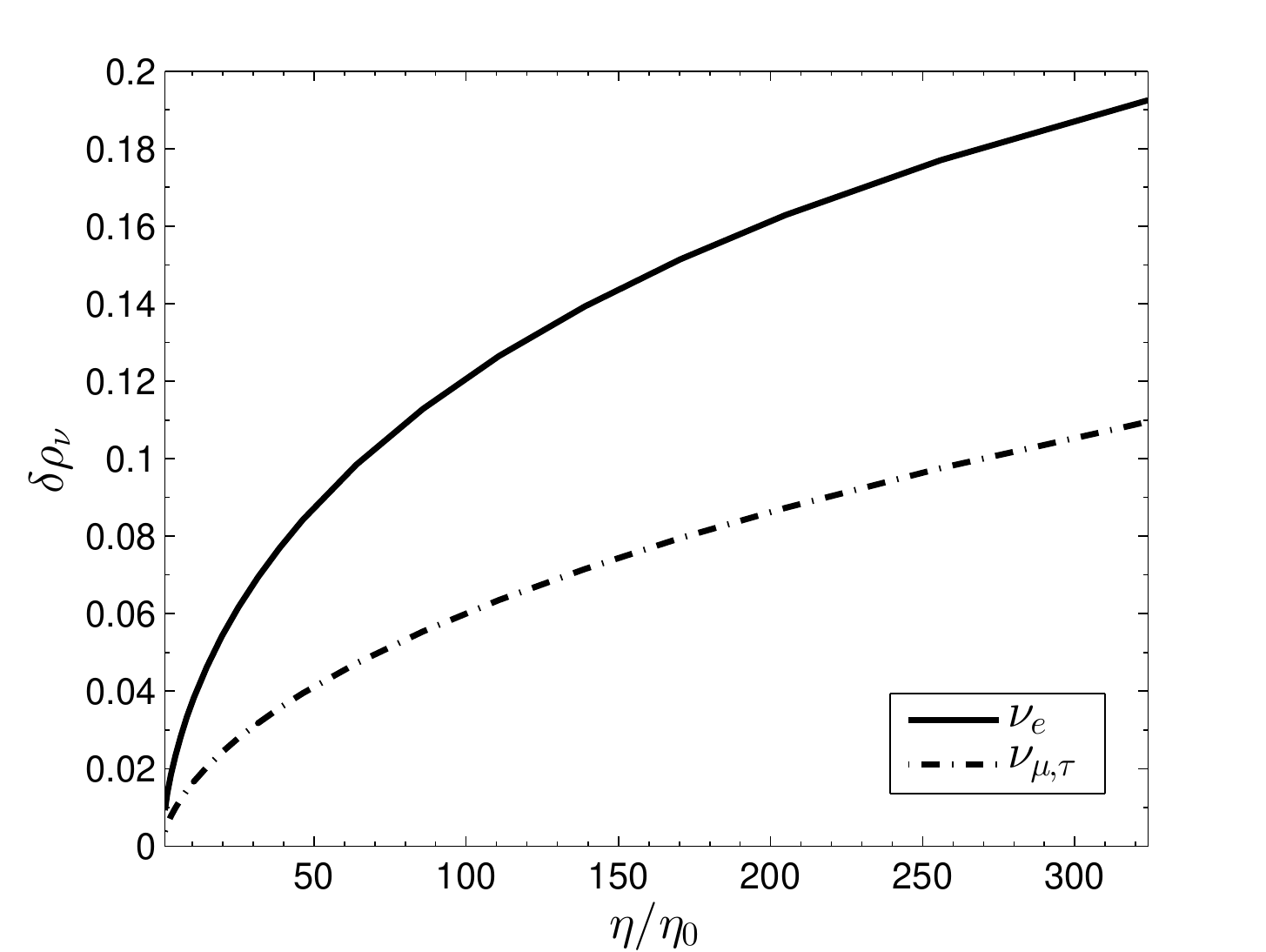}\hspace{-5mm}\includegraphics[height=5.8cm]{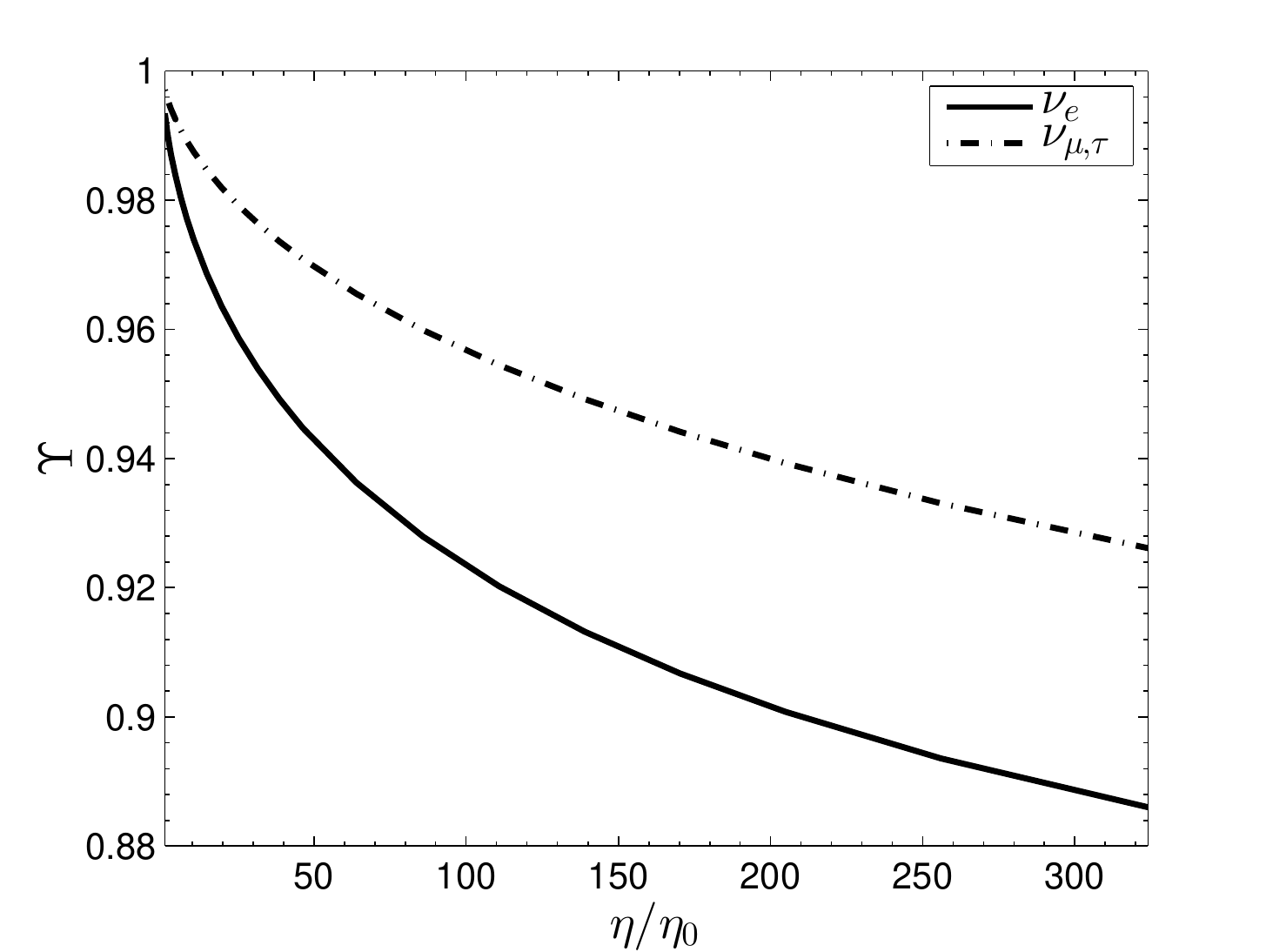}}
\caption{Fractional increase in neutrino energy (left) and neutrino fugacities (right), as functions of interaction strength.}
 \end{figure}

\begin{figure}[H]
\centerline{\includegraphics[height=5.8cm]{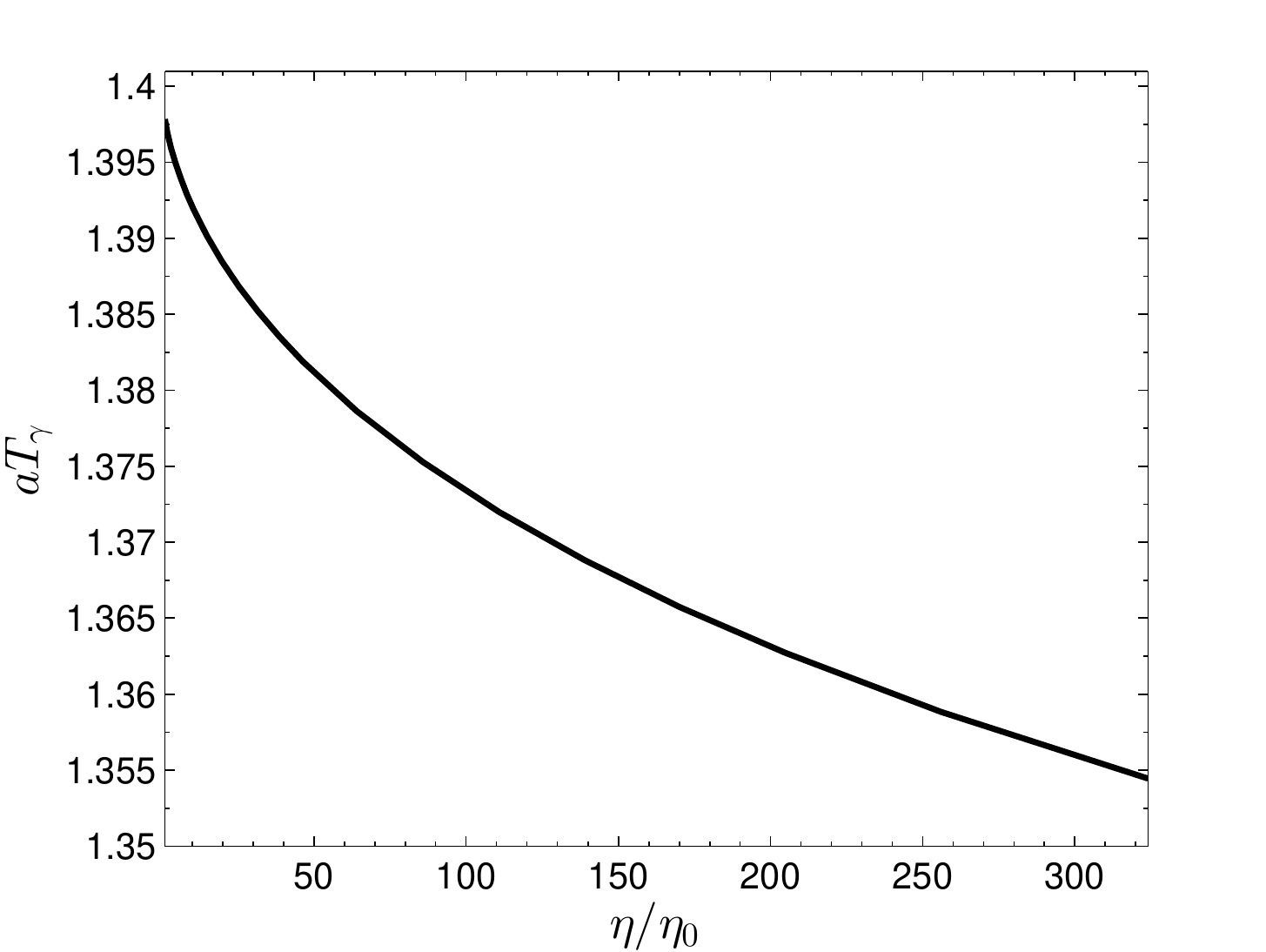}\hspace{-5mm}\includegraphics[height=5.8cm]{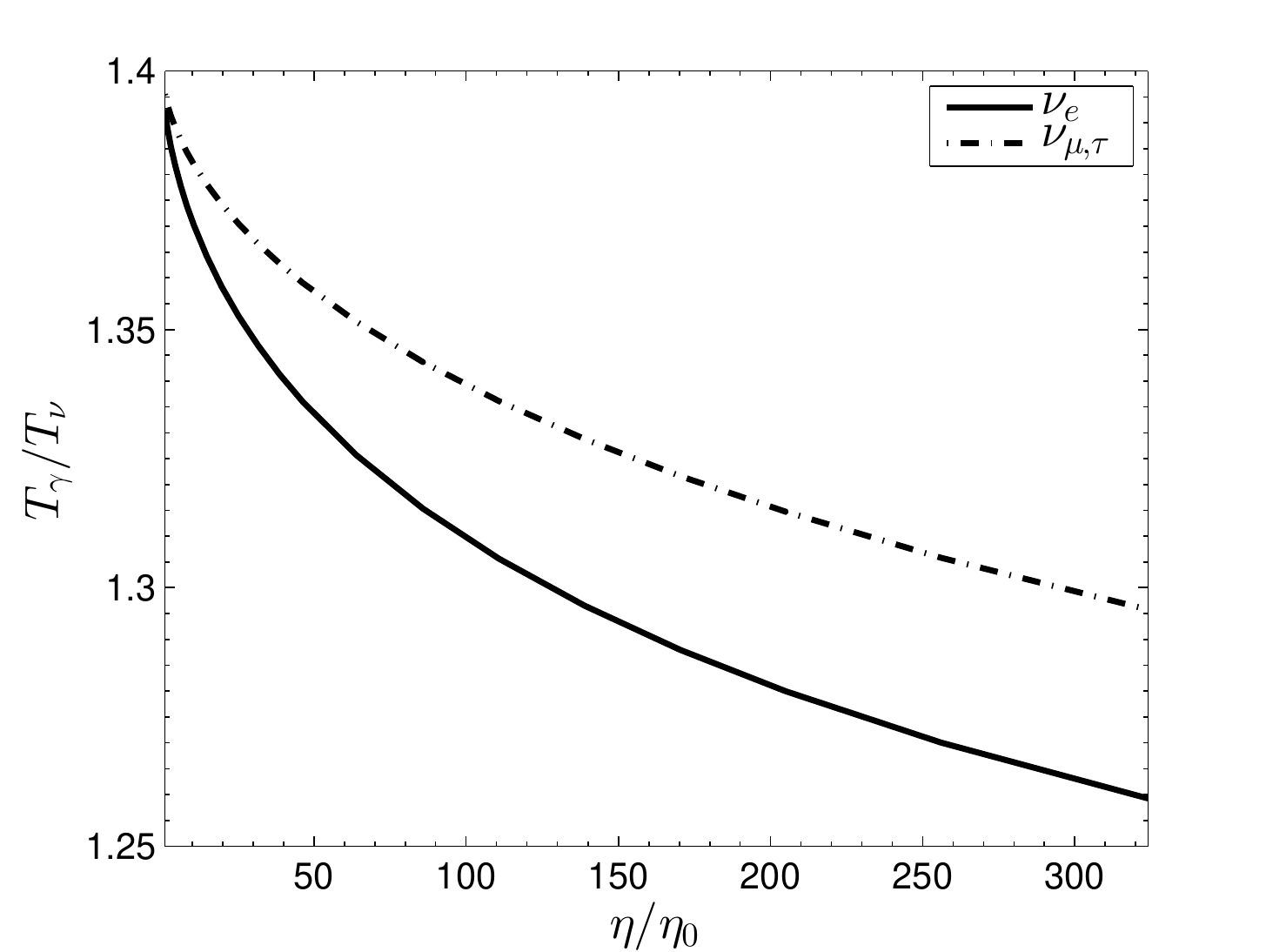}}
\caption{Dependence on  interaction strength of photon reheating (left) and photon-neutrino temperature ratios (right) after freeze-out.}
 \end{figure}

\begin{figure}[H]
\centerline{\includegraphics[height=6.3cm]{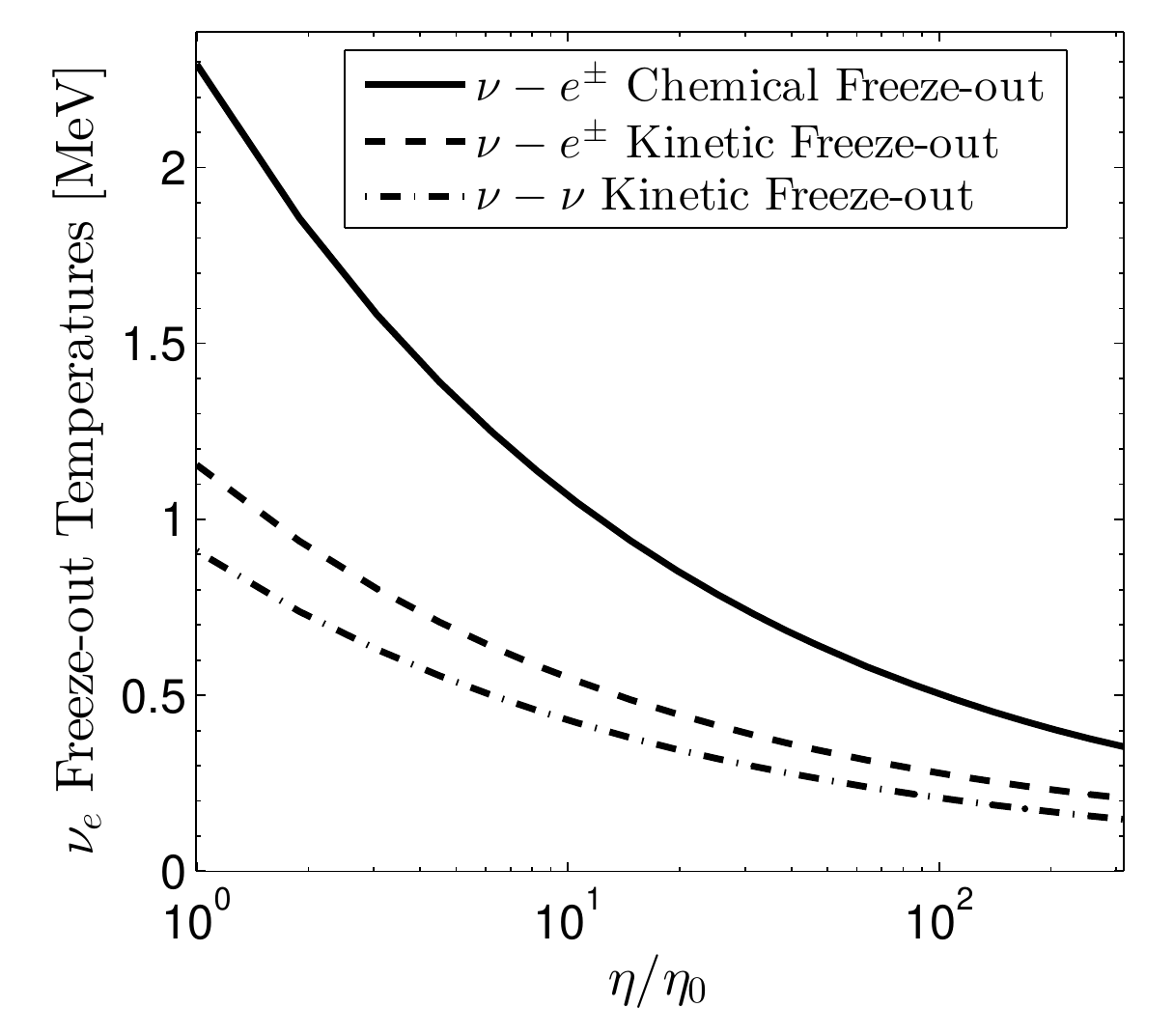}\hspace{5mm}\includegraphics[height=6.3cm]{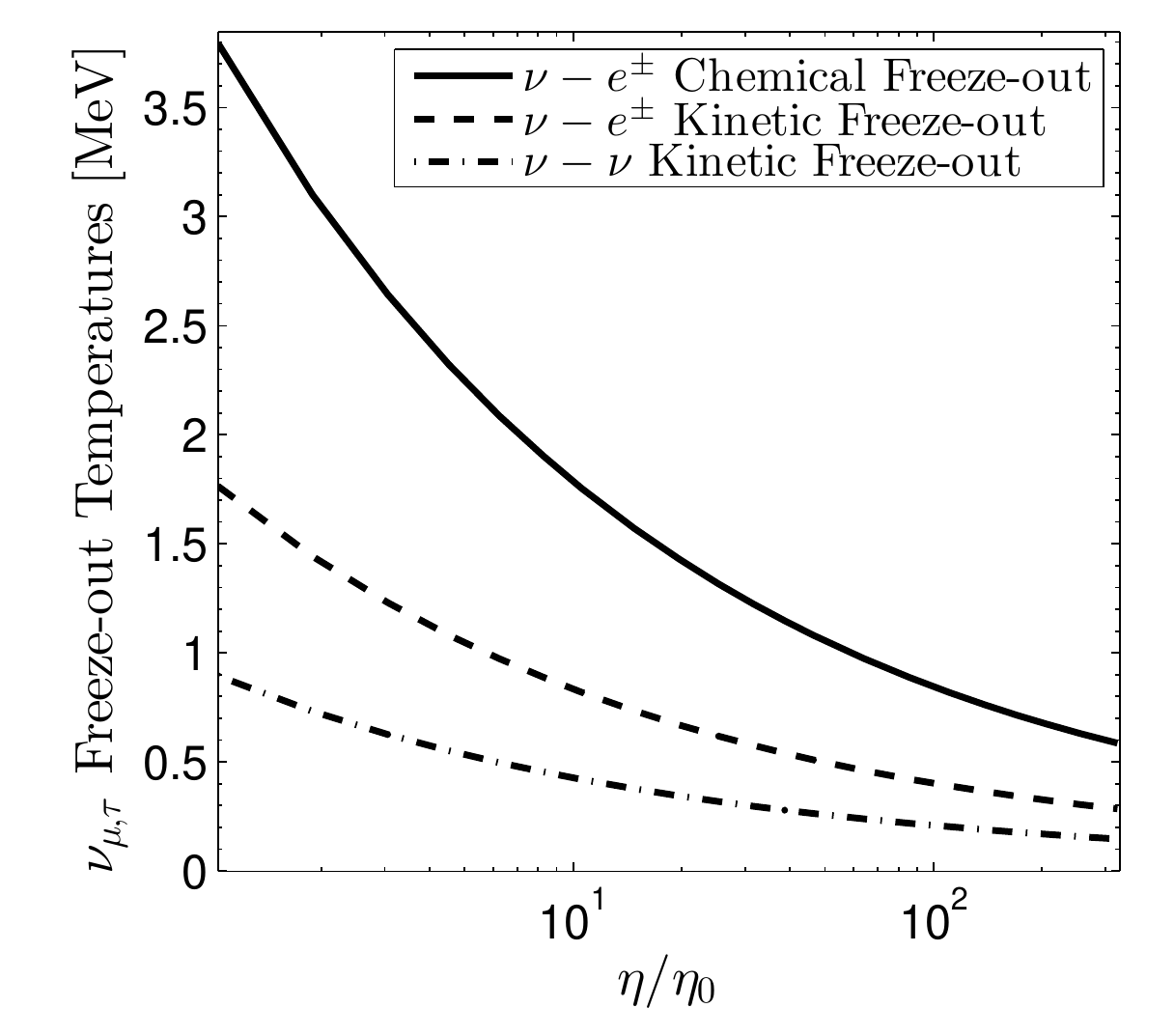}}
\caption{Freeze-out temperatures for electron neutrinos (left) and $\mu$, $\tau$ neutrinos (right) for various types of processes, as functions of interaction strength.}
 \end{figure}

\end{subappendices}

\chapter{Summary and Future Work}
We have studied the evolution of the relic neutrino background, focusing on the deviation from equilibrium, starting at freeze-out and moving to the present day. 
In the first part of this dissertation we focused on chemical non-equilibrium and characterized the neutrino distribution after freeze-out by a model independent approach that  used only conservation laws and was independent of the details of the scattering processes.  In particular, we computed the dependence on the kinetic freeze-out temperature $T_k$ of the reheating ratio $T_\gamma/T_\nu$, the effective number of neutrinos $N_\nu$, and the deviation from chemical equilibrium captured by the fugacity $\Upsilon$.  In this model a measurement of $N_\nu$ constitutes a measurement of $T_k$ and we found that a modified freeze-out temperature is capable of matching a value of $N_\nu>3.046$ as seen in the Planck CMB data.   We also characterized the free-streaming neutrino distribution today, as seen from Earth, a small step towards detector design.  In particular, we computed the drag force on a millimeter sized coherent detector due to the anisotropy of the neutrino distribution induced by the relative Earth-CMB motion.

In the second part of this dissertation we focused on the aspects of neutrino freeze-out that depend on the details of the scattering processes and hence require solution of the full general relativistic Boltzmann equation. Such a description is needed in order to model the non-thermal distortions of the neutrino distribution and eliminates the need to estimate $T_k$ or treat it as a free parameter, as we did in part one. We developed a novel spectral method for solving the Boltzmann equation that is adapted to emergent chemical non-equilibrium. The method incorporates three key improvements over the spectral method used by previous authors: a modified weight function, which reduces the minimum required number of modes from $4$ to $2$, together with a dynamical effective temperature and fugacity. We also developed an improved method for analytically simplifying the scattering integrals that substantially reduces the numerical integration cost.

Finally, we used our improved method to perform parametric studies of the dependence of neutrino freeze-out on the Weinberg angle, weak force interaction strength, the strength of gravity, and the electron mass in order to identify which mechanisms in the neutrino freeze-out process are capable of leading to the measured value of  $N_\nu$ in the environment of a hot Universe in which freeze-out  occurs.  Furthermore, our study allows us  to constrain time and/or temperature variation of these parameters using measurements of $N_\nu$.  

In the future, we will refine our study of neutrino freeze-out by investigating the impact of collective plasma effects on the system. For example, we are in the process of including photon decay due to the nonzero photon plasma mass, a process which has been overlooked in prior work. We also hope to better understand the limits on the variation of natural constants set by neutrino freeze-out, especially as the precision of measurements of $N_\nu$ improves.

Extending beyond the realm of neutrino freeze-out, we plan to adapt our spectral method for solving the Boltzmann equation to systems of massive bosons/fermions as well as to spatially varying systems with the hope of generalizing the method to form a broadly applicable method for studying the emergence of chemical non-equilibrium in systems of bosons and fermions that is also capable of capturing non-thermal distortions. In particular, we will use this to study the dynamics of massive quarks in Quark Gluon Plasma, both as found in the laboratory and in the early Universe.

\bibliographystyle{h-physrev}
\bibliography{refs}

\begin{thebibliography}{10}

\bibitem{Planck}
Planck Collaboration, P.~Ade {\em et~al.},
\newblock (2013), 1303.5076.

\bibitem{ErasOfUniverse}
J.~Rafelski and J.~Birrell,
\newblock (2013), 1311.0075.

\bibitem{hartle2003gravity}
J.~Hartle,
\newblock {\em Gravity: an introduction to Einstein's general relativity}
  (Addison-Wesley, 2003).

\bibitem{hobson}
M.~Hobson, G.~Efstathiou, and A.~Lasenby,
\newblock {\em General Relativity: An Introduction for Physicists} (Cambridge
  University Press, Cambridge, 2006).

\bibitem{misner1973gravitation}
C.~Misner, K.~Thorne, and J.~Wheeler,
\newblock {\em Gravitation} (W. H. Freeman, 1973).

\bibitem{Borsanyi:2013bia}
S.~Borsanyi {\em et~al.},
\newblock Phys.Lett. {\bf B730}, 99 (2014), 1309.5258.

\bibitem{Letessier:2002gp}
J.~Letessier and J.~Rafelski,
\newblock {\em Hadrons and Quark-Gluon Plasma}Cambridge monographs on particle
  physics, nuclear physics, and cosmology (Cambridge University Press, 2002).

\bibitem{Zaroubi:2012in}
S.~Zaroubi,
\newblock (2012), arXiv:1206.0267.

\bibitem{Iocco:2008va}
F.~Iocco, G.~Mangano, G.~Miele, O.~Pisanti, and P.~D. Serpico,
\newblock Phys.Rept. {\bf 472}, 1 (2009), 0809.0631.

\bibitem{Madsen}
S.~Hannestad and J.~Madsen,
\newblock Phys. Rev. D {\bf 52}, 1764 (1995).

\bibitem{Dolgov_Hansen}
A.~Dolgov, S.~Hansen, and D.~Semikoz,
\newblock Nucl.Phys. {\bf B543}, 269 (1999), hep-ph/9805467.

\bibitem{Gnedin}
N.~Y. Gnedin and O.~Y. Gnedin,
\newblock Astrophys.J. {\bf 509}, 11 (1998), astro-ph/9712199.

\bibitem{Esposito2000}
S.~Esposito, G.~Miele, S.~Pastor, M.~Peloso, and O.~Pisanti,
\newblock Nucl.Phys. {\bf B590}, 539 (2000), astro-ph/0005573.

\bibitem{Mangano2002}
G.~Mangano, G.~Miele, S.~Pastor, and M.~Peloso,
\newblock Phys.Lett. {\bf B534}, 8 (2002), astro-ph/0111408.

\bibitem{Mangano2005}
G.~Mangano {\em et~al.},
\newblock Nucl.Phys. {\bf B729}, 221 (2005), hep-ph/0506164.

\bibitem{Weinberg:2013kea}
S.~Weinberg,
\newblock Phys.Rev.Lett. {\bf 110}, 241301 (2013), arXiv:1305.1971.

\bibitem{andre}
H.~Andréasson,
\newblock Living Reviews in Relativity {\bf 14} (2011).

\bibitem{cercignani}
C.~Cercignani and G.~Kremer,
\newblock {\em The Relativistic Boltzmann Equation: Theory and Applications}
  (Birkhäuser Verlag, Basel, 2000).

\bibitem{bruhat}
Y.~Choquet-Bruhat,
\newblock {\em General Relativity and the Einstein Equations} (Oxford
  University Press, Oxford, 2009).

\bibitem{ehlers}
J.~Ehlers,
\newblock Survey of general relativity theory,
\newblock in {\em Relativity, Astrophysics and Cosmology}, pp. 1--125, D.
  Reidel Publishing Company, Dordrecht-Holland, 1973.

\bibitem{kolb}
E.~Kolb and M.~Turner,
\newblock {\em The Early Universe}Frontiers in physics (Westview Press, 1994).

\bibitem{bernstein2004kinetic}
J.~Bernstein,
\newblock {\em Kinetic Theory in the Expanding Universe}Cambridge Monographs on
  Mathematical Physics (Cambridge University Press, 2004).

\bibitem{PhysRevLett.48.1066}
J.~Rafelski and B.~M\"uller,
\newblock Phys. Rev. Lett. {\bf 48}, 1066 (1982).

\bibitem{Bernstein:1985}
J.~Bernstein, L.~S. Brown, and G.~Feinberg,
\newblock Phys. Rev. D {\bf 32}, 3261 (1985).

\bibitem{Dolgov:1993}
A.~Dolgov and K.~Kainulainen,
\newblock Nucl.Phys. {\bf B402}, 349 (1993), hep-ph/9211231.

\bibitem{Birrell2013}
J.~Birrell, C.-T. Yang, P.~Chen, and J.~Rafelski,
\newblock Mod.Phys.Lett. {\bf A28}, 1350188 (2013), arXiv:1303.2583.

\bibitem{Birrell:2013_2}
J.~Birrell, C.-T. Yang, P.~Chen, and J.~Rafelski,
\newblock Phys.Rev. {\bf D89}, 023008 (January 2014), 1212.6943.

\bibitem{Stodolsky:1975}
L.~Stodolsky,
\newblock Phys.Rev.Lett. {\bf 34}, 110 (1975).

\bibitem{Cabibbo:1982}
N.~Cabibbo and L.~Maiani,
\newblock Phys.Lett. {\bf B114}, 115 (1982).

\bibitem{Shvartsman}
B.~Shvartsman {\em et~al.},
\newblock JETP Lett. {\bf 36}, 277 (1982).

\bibitem{Langacker:1982}
P.~Langacker, J.~P. Leveille, and J.~Sheiman,
\newblock Phys.Rev. {\bf D27}, 1228 (1983).

\bibitem{Smith}
P.~Smith and J.~Lewin,
\newblock Physics Letters B {\bf 127}, 185  (1983).

\bibitem{Ferreras:1995wf}
I.~Ferreras and I.~Wasserman,
\newblock Phys.Rev. {\bf D52}, 5459 (1995).

\bibitem{Hagmann:1999kf}
C.~Hagmann,
\newblock (arXiv:astro-ph/9905258, 1999).

\bibitem{Duda:2001hd}
G.~Duda, G.~Gelmini, and S.~Nussinov,
\newblock Phys.Rev. {\bf D64}, 122001 (2001), hep-ph/0107027.

\bibitem{Gelmini}
G.~B. Gelmini,
\newblock Physica Scripta {\bf 2005}, 131 (2005).

\bibitem{Ringwald:2009}
A.~Ringwald,
\newblock Nuclear Physics A {\bf 827}, 501c  (2009),
\newblock \{PANIC08\} Proceedings of the 18th Particles and Nuclei
  International Conference.

\bibitem{Liao:2012}
W.~Liao,
\newblock Phys.Rev. {\bf D86}, 073011 (2012), arXiv:1207.6847.

\bibitem{Hedman}
M.~Hedman,
\newblock Journal of Cosmology and Astroparticle Physics {\bf 2013}, 029
  (2013).

\bibitem{PTOLEMY}
S.~{Betts} {\em et~al.},
\newblock ArXiv e-prints  (2013), 1307.4738.

\bibitem{nu_today}
J.~Birrell and J.~Rafelski,
\newblock (2014), arXiv:1402.3409.

\bibitem{Wong}
Y.~Y. Wong,
\newblock Ann.Rev.Nucl.Part.Sci. {\bf 61}, 69 (2011), arXiv:1111.1436.

\bibitem{Battye:2013xqa}
R.~A. Battye and A.~Moss,
\newblock Phys.Rev.Lett. {\bf 112}, 051303 (2014), arXiv:1308.5870.

\bibitem{Beringer:1900zz}
Particle Data Group, J.~Beringer {\em et~al.},
\newblock Phys.Rev. {\bf D86}, 010001 (2012).

\bibitem{Aseev}
Troitsk Collaboration, V.~Aseev {\em et~al.},
\newblock Phys.Rev. {\bf D84}, 112003 (2011), arXiv:1108.5034.

\bibitem{CUDO}
J.~Rafelski, L.~Labun, and J.~Birrell,
\newblock Phys. Rev. Lett. {\bf 110}, 111102 (2013).

\bibitem{Opher}
R.~{Opher},
\newblock Astron. Astrophys. {\bf 37}, 135 (1974).

\bibitem{Lewis}
R.~R. Lewis,
\newblock Phys. Rev. D {\bf 21}, 663 (1980).

\bibitem{Opher2}
R.~{Opher},
\newblock Astron. Astrophys. {\bf 108}, 1 (1982).

\bibitem{Divari:2012zz}
P.~Divari,
\newblock Adv.High Energy Phys. {\bf 2012}, 379460 (2012).

\bibitem{Biercuk}
M.~J. Biercuk, H.~Uys, J.~W. Britton, A.~P. Vandevender, and J.~J. Bollinger,
\newblock Nature Nanotechnology {\bf 5}, 646 (2010).

\bibitem{Ringwald:2004np}
A.~Ringwald and Y.~Y. Wong,
\newblock JCAP {\bf 0412}, 005 (2004), hep-ph/0408241.

\bibitem{Safdi}
B.~R. Safdi, M.~Lisanti, J.~Spitz, and J.~A. Formaggio,
\newblock (2014), arXiv:1404.0680.

\bibitem{lee2003introduction}
J.~Lee,
\newblock {\em Introduction to Smooth Manifolds}Graduate Texts in Mathematics
  (Springer, 2003).

\bibitem{lee1997riemannian}
J.~Lee,
\newblock {\em Riemannian Manifolds: An Introduction to Curvature}Graduate
  Texts in Mathematics (Springer, 1997).

\bibitem{Birrell:2014uka}
J.~Birrell, C.-T. Yang, and J.~Rafelski,
\newblock (2014), arXiv:1406.1759.

\bibitem{chavel1995riemannian}
I.~Chavel,
\newblock {\em Riemannian Geometry: A Modern Introduction}Cambridge tracts in
  mathematics (Cambridge University Press, 1995).

\bibitem{tsamparlis}
M.~Tsamparlis,
\newblock General Relativity and Gravitation {\bf 17}, 831 (1985),
\newblock 10.1007/BF00773681.

\bibitem{pettini}
M.~Pettini,
\newblock {\em Geometry and Topology in Hamiltonian Dynamics and Statistical
  Mechanics}Interdisciplinary Applied Mathematics (Springer, 2007).

\bibitem{Wilkening}
J.~{Wilkening}, A.~{Cerfon}, and M.~{Landreman},
\newblock ArXiv e-prints  (2014), 1402.2971.

\bibitem{Wilkening2}
J.~{Wilkening} and A.~{Cerfon},
\newblock ArXiv e-prints  (2013), 1310.5074.

\bibitem{Birrell_orthopoly}
J.~{Birrell} {\em et~al.},
\newblock ArXiv e-prints  (2014), arXiv:1403.2019.

\bibitem{Anderson_Witting}
J.~Anderson and H.~Witting,
\newblock Physica {\bf 74}, 466  (1974).

\bibitem{Olver}
F.~Olver,
\newblock {\em Asymptotics and Special Functions}AKP classics (A.K. Peters,
  1997).

\bibitem{letessier2002hadrons}
J.~Letessier and J.~Rafelski,
\newblock {\em Hadrons and Quark-Gluon Plasma}Cambridge Monographs on Particle
  Physics, Nuclear Physics and Cosmology (Cambridge University Press, 2002).

\bibitem{Heckler:1994tv}
A.~Heckler,
\newblock Phys.Rev. {\bf D49}, 611 (1994).

\bibitem{Uzan:2010pm}
J.-P. Uzan,
\newblock Living Rev.Rel. {\bf 14}, 2 (2011), arXiv:1009.5514.

\bibitem{Birrell:2014connect}
J.~Birrell and J.~Rafelski,
\newblock (2014), arXiv:1404.6005.

\bibitem{Abazajian:2012ys}
K.~Abazajian {\em et~al.},
\newblock (2012), arXiv:1204.5379.

\end{thebibliography}

\appendix
\chapter{Fugacity and Reheating of Primordial Neutrinos.}\label{app:chem_freezeout}
J. Birrell, C.T. Yang, P. Chen, J. Rafelski, Mod. Phys. Lett. A 28, 1350188 (2013) 
DOI: 10.1142/S0217732313501885 

\section*{ Summary}

In this paper we studied numerically the dependence of the neutrino distriubtion, including fugacity $\Upsilon$, effective number of neutrinos $N_\nu$, and reheating ratio $T_\gamma/T_\nu$ on the kinetic freezeout temperature $T_k$ under the assumption of kinetic equilibrium. In particular, we showed that  a measurement of $N_\nu$ constitues a measurement of $T_k$ in this model. We demonstrated that the instantaneous chemical freeze-out approximation, and hence entropy conservation,  holds to a good approximation for a subset of the neutrino reactions.  We also numerically identified an approximate power law relation between $\Upsilon$ and $T_\gamma/T_\nu$.

Others worked on deriving the neutrino reaction rates, but I was responsible for deriving all other equations, numerically simulating the system, and producing both figures. I was also responsible for the creation of the initial draft of the manuscript.

\includepdf[pages={-}]{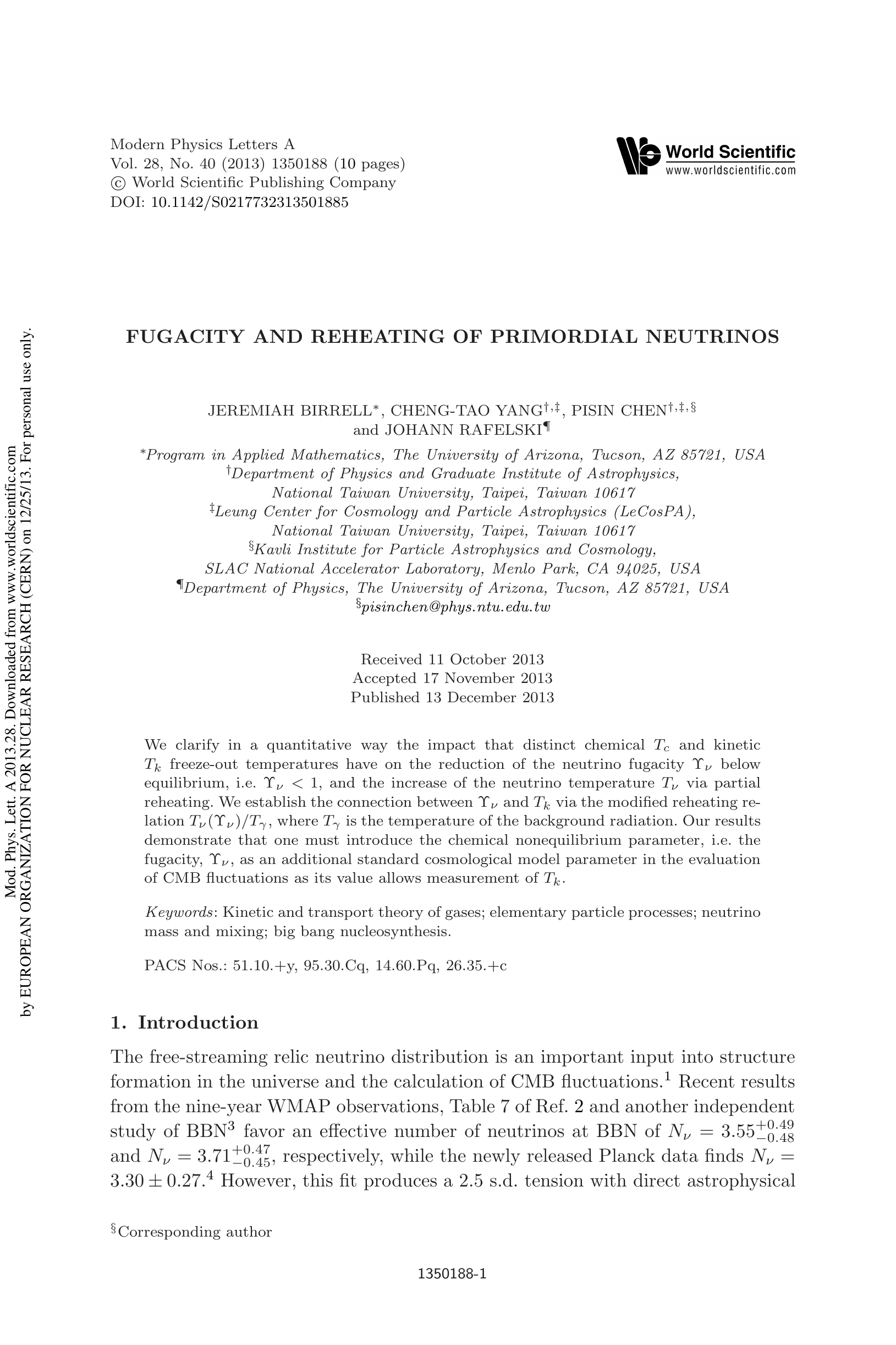}
\chapter{Relic neutrinos: Physically consistent treatment of effective number of neutrinos and neutrino mass.}\label{app:model_ind}

J. Birrell, C.T. Yang, P. Chen, J. Rafelski, Phys. Rev. D 89, 023008 (2014)
DOI: 10.1103/PhysRevD.89.023008

\section*{ Summary}

In this paper, we performed a model independent characterization of the neutrino distribution after freeze-out as a function of the kinetic freeze-out temperature $T_k$.  We showed how chemical non-equilibrium, in the form of a fugacity $\Upsilon<1$, emerges during freeze-out as a result of $T_k$ being on the order of the electron mass, the period when $e^\pm$ annihilation begins in earnest.

Using conservation of energy and entropy we were able to compute the neutrino fugacity, the photon to neutrino reheating temperature ratio $T_\gamma/T_\nu$, and the effective number of neutrinos $N_\nu$, all as functions of $T_k$.  In particular, we presented an analytic derivation of an approximate power law relation between the  reheating ratio and the fugacity.  We also showed that a delayed neutrino freeze-out is capable of matching the value of $N_\nu>3$, as seen in the recent Planck CMB results.

We also derived an analytic expression for the free-streaming neutrino distribution after freeze-out and used that to find fits of the neutrino energy density and pressure as functions of both the observed value of $N_\nu$ and the neutrino mass.  Such a parameterization is required in order to include the effects of both neutrino mass and neutrino reheating, $N_\nu>3,$ into CMB studies in a physically consistent way.  These constitute a new insight that was made possible by the model independent approach.  Prior studies had difficulty when attempting to include both effects simultaneously.

I was responsible for all mathematical derivations, numerical computations, and creation of figures.  I was also responsible for the creation of the initial draft of the manuscript.

\includepdf[pages={-}]{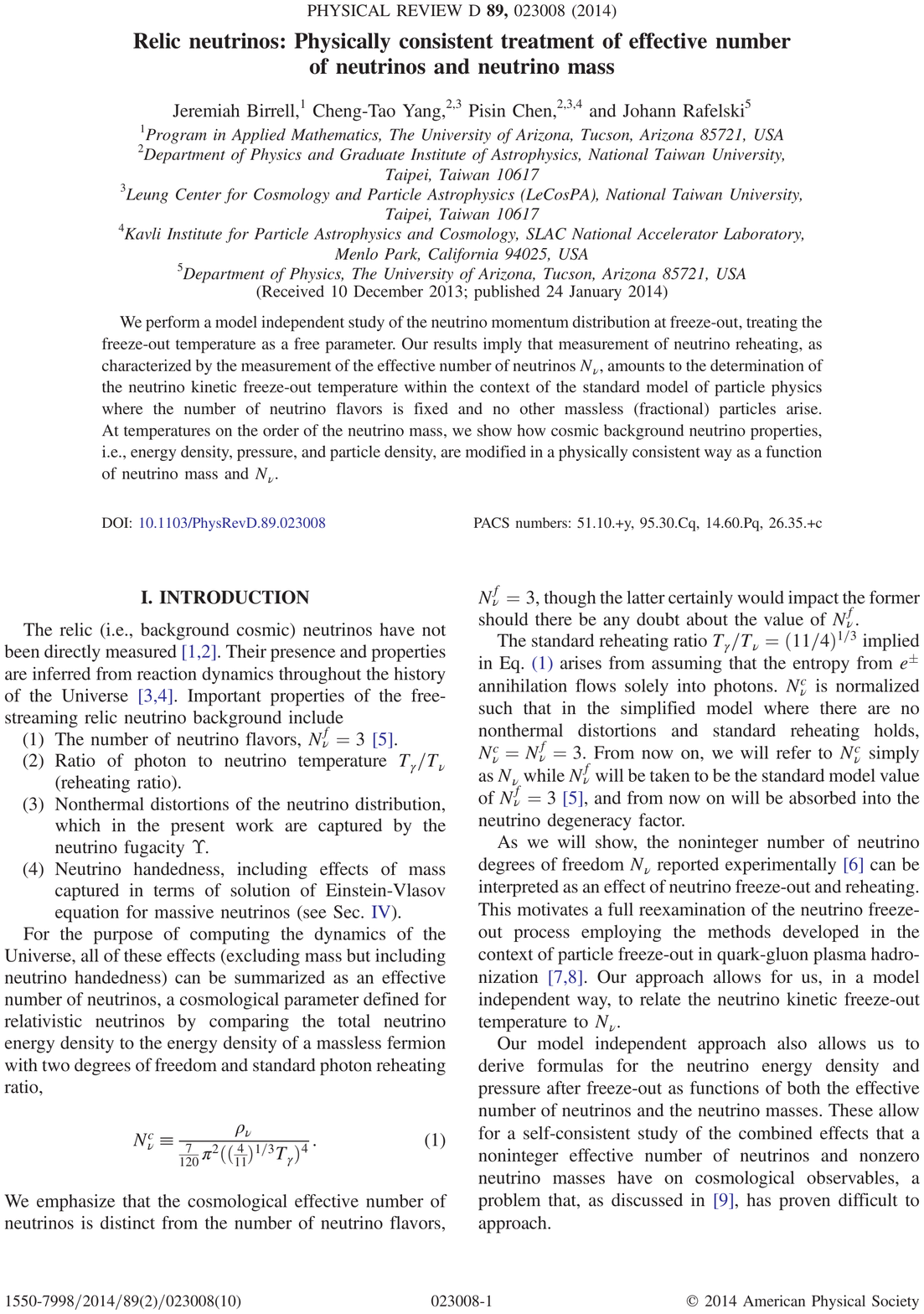}

\end{document}